\newif\ifabstract
\abstracttrue
\newif\iffull
\ifabstract \fullfalse \else \fulltrue \fi

\documentclass[11pt,letterpaper]{article}
\usepackage{amsfonts}
\usepackage{amssymb}
\usepackage{amstext}
\usepackage{amsmath}
\usepackage{xspace}
\usepackage{ntheorem}
\usepackage{graphicx}

\usepackage{url}
\usepackage{graphics}
\usepackage{colordvi}
\usepackage[colorlinks,linkcolor=black,citecolor=black,filecolor=blue,urlcolor=blue]{hyperref}
\usepackage{hyperref}
\usepackage{subfigure}
\usepackage{xcolor}
\usepackage{cleveref}

\usepackage{thmtools}
\usepackage{thm-restate}
\usepackage{enumitem}

\advance \oddsidemargin by -0.8in
\newcommand{\myparskip}{3pt}
\parskip \myparskip

\newcommand{\val}{\operatorname{val}}

\newcommand\tO{\ensuremath{\tilde O}}
\newcommand{\DS}{\operatorname{DS}}

\newcommand{\algterm}{\ensuremath{\operatorname{AlgSelectTerminals}}\xspace}

\newcommand{\alg}{\ensuremath{\operatorname{Alg}}\xspace}

\newcommand{\ceil}[1]{\ensuremath{\left\lceil#1\right\rceil}}
\newcommand{\floor}[1]{\ensuremath{\left\lfloor#1\right\rfloor}}

\newcommand{\band}{\wedge}
\newcommand{\bor}{\vee}

\newcommand{\event}{{\cal{E}}}

\newcommand{\opt}{\mathsf{OPT}}

\newcommand{\set}[1]{\left\{ #1 \right\}}

\newcommand{\tset}{{\mathcal T}}

\newcommand{\pset}{{\mathcal{P}}}
\newcommand{\qset}{{\mathcal{Q}}}

\newcommand{\bset}{{\mathcal{B}}}
\newcommand{\aset}{{\mathcal{A}}}

\newcommand{\be}{\begin{enumerate}}
\newcommand{\ee}{\end{enumerate}}
\newcommand{\bd}{\begin{description}}
\newcommand{\ed}{\end{description}}
\newcommand{\bi}{\begin{itemize}}
\newcommand{\ei}{\end{itemize}}

\newtheorem{theorem}{Theorem}[section]
\newtheorem{fact}[theorem]{Fact}
\newtheorem{lemma}[theorem]{Lemma}
\newtheorem{observation}[theorem]{Observation}
\newtheorem{corollary}[theorem]{Corollary}
\newtheorem{claim}[theorem]{Claim}

\newtheorem{definition}[theorem]{Definition}
\newenvironment{proof}{\par \smallskip{\bf Proof:}}{\hfill\stopproof}
\def\stopproof{\square}
\def\square{\vbox{\hrule height.2pt\hbox{\vrule width.2pt height5pt \kern5pt
\vrule width.2pt} \hrule height.2pt}}

\newenvironment{proofof}[1]{\noindent{\bf Proof of #1.}}%
        {\hfill\stopproof}

\newenvironment{prog}[1]{
\begin{minipage}{5.8 in}
\begin{center}
{\sc #1}
\end{center}
}
{
\end{minipage}
}

\newcommand{\program}[3]{\begin{figure} \fbox{\vspace{2mm}\begin{prog}{#1} #3 \end{prog}\vspace{2mm}} 
			\caption{#1 \label{#2}} \end{figure}}

\renewcommand{\phi}{\varphi}
\newcommand{\eps}{\epsilon}

\newcommand{\half}{\ensuremath{\frac{1}{2}}}

\newcommand{\poly}{\operatorname{poly}}

\newcommand{\expect}[2][]{\text{\bf E}_{#1}\left [#2\right]}
\newcommand{\prob}[2][]{\text{\bf Pr}_{#1}\left [#2\right]}

\newenvironment{properties}[2][0]
{
\begin{enumerate} \setcounter{enumi}{#1}}{\end{enumerate}}

\newcommand{\edel}{E^{\operatorname{del}}}

\newcommand{\mynote}[2][red]{\textcolor{red}{\sc\bf{[#2]}}}

\newcommand{\cmax}{C_{\mbox{\textup{\footnotesize{max}}}}}
\newcommand{\wmax}{W_{\mbox{\textup{\footnotesize{max}}}}}

\newcommand{\OUT}{\mathsf{OUT}}
\newcommand{\IN}{\mathsf{IN}}
\newcommand{\reg}{\operatorname{reg}}
\newcommand{\spec}{\operatorname{spec}}
\newcommand{\inn}{\operatorname{in}}
\newcommand{\out}{\operatorname{out}}

\newcommand{\explore}{\mbox{\sf{Explore} $H'$}\xspace}

\makeatletter
\newcommand{\customlabel}[2]{%
   \protected@write \@auxout {}{\string \newlabel {#1}{{#2}{\thepage}{#2}{#1}{}} }%
   \hypertarget{#1}{}
}
\makeatother

\makeatletter
\def\ifempty#1{%
 \def\@tmp@a{#1}%
 \ifx\@tmp@a\@empty%
}

\newtheoremstyle{ams-theorem}%
  {\item[\hskip\labelsep \theorem@headerfont ##1\ ##2\theorem@separator]}%
  {\item[\hskip\labelsep {\theorem@headerfont ##1\ ##2}{\normalfont\ (##3)}{\theorem@headerfont
  \theorem@separator}]}
\newtheoremstyle{ams-restatedtheorem}
  {\item[\hskip\labelsep \theorem@headerfont ##1\ ##2\theorem@separator]}%
  {\item[\hskip\labelsep {\theorem@headerfont ##1\ ##2}{\normalfont\ (##3)}{\theorem@headerfont
  \theorem@separator}]}
\newtheoremstyle{nonumberams-restatedtheorem}%
  {\item[\theorem@headerfont \hskip\labelsep ##1\theorem@separator]}%
  {\item[\hskip\labelsep \theorem@headerfont ##3\theorem@separator]}%

\begin{document}

\begin{titlepage}
	
	\title{Breaking the $O(mn)$-Time Barrier for Vertex-Weighted Global Minimum Cut\footnote{STOC 2025, to appear}}
	\author{Julia Chuzhoy\thanks{Toyota Technological Institute at Chicago. Email: {\tt cjulia@ttic.edu}. Supported in part by NSF grant CCF-2402283 and NSF HDR TRIPODS award 2216899.}\and Ohad Trabelsi\thanks{Toyota Technological Institute at Chicago. Email: {\tt ohadt@ttic.edu}.}}
	\date{}
	\maketitle
\pagenumbering{gobble}
	
\thispagestyle{empty}
	
	\begin{abstract}
We consider the Global Minimum Vertex-Cut problem: given an undirected vertex-weighted graph $G$, compute a minimum-weight subset of its vertices whose removal disconnects $G$. The problem is closely related to Global Minimum Edge-Cut, where the weights are on the graph edges instead of vertices, and the goal is to compute a minimum-weight subset of edges whose removal disconnects the graph. Global Minimum Cut is one of the most basic and extensively studied problems in combinatorial optimization and graph theory. While an almost-linear time algorithm was known for the edge version of the problem for awhile (Karger, STOC 1996 and J. ACM 2000), the fastest previous algorithm for the vertex version (Henzinger, Rao and Gabow, FOCS 1996 and J. Algorithms 2000) achieves a running time of $\tO(mn)$, where $m$ and $n$ denote the number of edges and vertices in the input graph, respectively. For the special case of unit vertex weights, this bound was broken only recently  (Li {et al.}, STOC 2021);
their result, combined with the recent breakthrough almost-linear time algorithm for Maximum $s$-$t$ Flow (Chen {et al.}, FOCS 2022, van den Brand {et al.}, FOCS 2023), yields an almost-linear time algorithm for Global Minimum Vertex-Cut with unit vertex weights.

In this paper we break the $28$ years old bound of Henzinger {et al.}  for the general weighted Global Minimum Vertex-Cut, by providing a randomized algorithm for the problem with running time $O(\min\{mn^{0.99+o(1)},m^{1.5+o(1)}\})$.
	\end{abstract}

\newpage
\end{titlepage}
\tableofcontents{}
\newpage

\pagenumbering{arabic}

\section{Introduction}

We consider the Global Minimum Vertex-Cut problem: given an undirected $n$-vertex and $m$-edge graph $G=(V,E)$ with integral weights $0\leq w(v)\leq W$ on its vertices $v\in V$, compute a subset $S\subseteq V$ of vertices, such that the graph $G\setminus S$ is not connected. A closely related and extensively studied problem is Global Minimum Edge-Cut: here, the input graph has weights on its edges instead of vertices, and the goal is to compute a minimum-weight subset of edges, whose removal disconnects the graph. The Global Minimum Cut problem is among the most fundamental and extensively studied in computer science, with algorithms dating back to the 1960's ~\cite{Kle69,Pod73,ET75,LLW88}. 
The famous edge-contraction based algorithm  of Karger and Stein~\cite{KS93} for the edge version of the problem is one of the gems of Theoretical Computer Science and is routinely taught in basic Algorithms classes. 
The notion of a Global Minimum Cut is closely related to the central Graph-Theoretic notions of edge- and vertex-connectivity. The edge-connectivity of a graph $G$ is the smallest value $\lambda$, such that, for every pair $u,v$ of vertices of $G$, there is a collection $\lambda$ of edge-disjoint paths connecting $u$ to $v$. Vertex-connectivity is defined similarly, except that the paths are required to be internally vertex-disjoint. By Menger's Theorem~\cite{Men1927}, if $G$ is a graph with unit vertex weights, then the value of the Global Minimum Vertex-Cut  in $G$ is equal to its vertex-connectivity. One can naturally extend the notion of vertex-connectivity to the vertex-weighted setting, where a collection of $\lambda$ internally disjoint $u$-$v$ paths is replaced with  a $u$-$v$ flow of value $\lambda$, that obeys the vertex capacities defined by their weights. Similarly, the value of the Global Minimum Edge-Cut in a graph with unit edge weights is equal to its edge-connectivity, and the value of the Global Minimum Edge-Cut in a graph with arbitrary edge weights can be thought as corresponding to the ``weighted'' notion of edge-connectivity, where edge-disjoint paths are replaced by flows.

The edge version of the Global Minimum Cut problem has been studied extensively, and is now reasonably well understood in undirected graphs. The classical algorithm of Karger and Stein~\cite{KS93} mentioned above achieves running time $\tilde{O}(n^2)$, improving upon the previous fastest $\tO(mn)$-time algorithm by Hao and Orlin~\cite{HO94}. Later, Karger~\cite{Kar00} obtained a randomized algorithm with near optimal running time $\tO(m)$, using tree-packing and dynamic programming. This running time was recently matched by a deterministic algorithm~\cite{HLRW24}.
The same problem was also studied in the directed setting: Cen et al.~\cite{CLN21}, building on some of the ideas from~\cite{Kar00}, recently presented a randomized algorithm for the directed edge version with running time $O(\min\{n/m^{1/3}, \sqrt{n}\}\cdot m^{1+o(1)})$.

Despite this, progress on the vertex-cut version of the problem has been much slower. In 1996, Henzinger et al.~\cite{HRG00} presented an algorithm for vertex-weighted Global Minimum Cut with running time $\tO(mn)$.
In the 28 years since then, no significant progress has been made on the general vertex-weighted version.
This lack of progress is likely due to several existing technical barriers to adapting the approaches used for the edge-cut version of the problem.  For example, while there are at most $\binom{n}{2}$ minimum cuts in the undirected edge-cut setting, the number of minimum vertex-cuts can be as large as $\Theta(2^k\cdot (n/k)^2)$ even in the unweighted setting, where $k$ is the value of the optimal solution~\cite{Kan90}. Furthermore, the tree-packing method, such as the one used in~\cite{Kar00,CLN21}, is not known to provide similar guarantees in the vertex capacitated setting.
Only recently, the $O(mn)$ barrier was broken for the \emph{unweighted} version of Global Minimum Vertex-Cut by Li et al.~\cite{LNP21}. Their algorithm, combined with the recent breakthrough almost-linear time algorithm for Maximum $s$-$t$ Flow \cite{CKLP22,BCP23}, achieves a running time of $O(m^{1+o(1)})$. They also obtain an $O(\min\{\tO(n^2),mn^{1-1/12+o(1)}\})$-time algorithm for the directed version of the problem with unit vertex weights.

We note that the recent progress on fast algorithms for Maximum $s$-$t$ Flow mentioned above does not appear to directly imply algorithms for Global Minimum Cut with general vertex weights whose running time is below $O(mn)$.
Moreover, while, in the context of Maximum $s$-$t$ Flow, it is well known that the vertex-capacitated version of the problem in undirected graphs can be cast as a special case of the directed edge-capacitated version, this connection is not known in the context of Global Minimum Vertex-Cut. Unfortunately, the standard reduction, where an undirected vertex-weighted graph is replaced by the corresponding edge-weighted directed \emph{split graph}\footnote{A split graph of a vertex-weighted undirected graph $G$ is obtained by replacing every vertex $v\in V(G)$ by a directed edge $(v^{\inn},v^{\out})$ of capacity $w(v)$, and every edge $(x,y)$ of $G$ by a pair of edges  $(x^{\out},y^{\inn})$ and $(y^{\out},x^{\inn})$ of infinite capacity.}, does not work for the Global Minimum Vertex-Cut problem (see \Cref{fig: reduction_fails}), and so the recent algorithms for the directed Global Minimum Edge-Cut problem of~\cite{CLN21} cannot be directly leveraged to obtain a faster algorithm for undirected Global Minimum Vertex-Cut.
This raises the following fundamental question:

\begin{center}\customlabel{Q1}{Q1}
\textit{\textbf{Q1:} Can the Global Minimum Vertex-Cut problem be solved in time faster than $\Theta(mn)$?}
\par\end{center}

 
In this paper, we answer this question affirmatively, by presenting a randomized algorithm for weighted Global Minimum Vertex-Cut with running time $O(\min\{mn^{0.99+o(1)}),m^{3/2+o(1)}\})$. Our result also resolves an open question 
posed by \cite{NSY19} regarding the existence of an $o(n^2)$-time algorithm for the problem, for the setting where $m=O(n)$.
We introduce several new technical tools and expand the scope of some existing techniques; we hope that these tools and techniques will lead to even faster algorithms for the problem, and will be useful in other graph-based problems. We now provide a more detailed overview of previous work, followed by the formal statement of our results and a high-level overview of our techniques.

\begin{figure}[h]
	\centering
	\subfigure[A vertex-weighted graph $G$. The unique global minimum vertex-cut in $G$ is $(\{x\},\{y\},\{z\})$, of value $\opt=10$.]{\scalebox{0.6}{\includegraphics{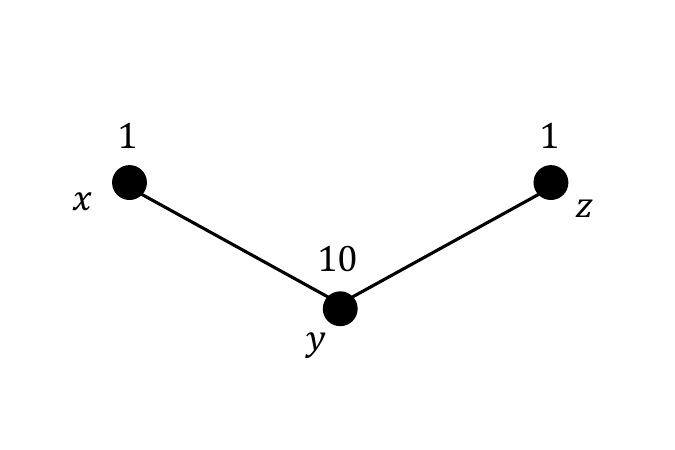}}}
	\hspace{1cm}
	\subfigure[The corresponding split graph $G'$; the blue edges have infinite capacity. Consider the cut $(X,Y)$, where $Y=\{z^{\out}\}$ and $X=V(G')\setminus Y$. The cut value is $1$, but it does not define a vertex cut in $G$ of the same value.]{
		\scalebox{0.4}{\includegraphics{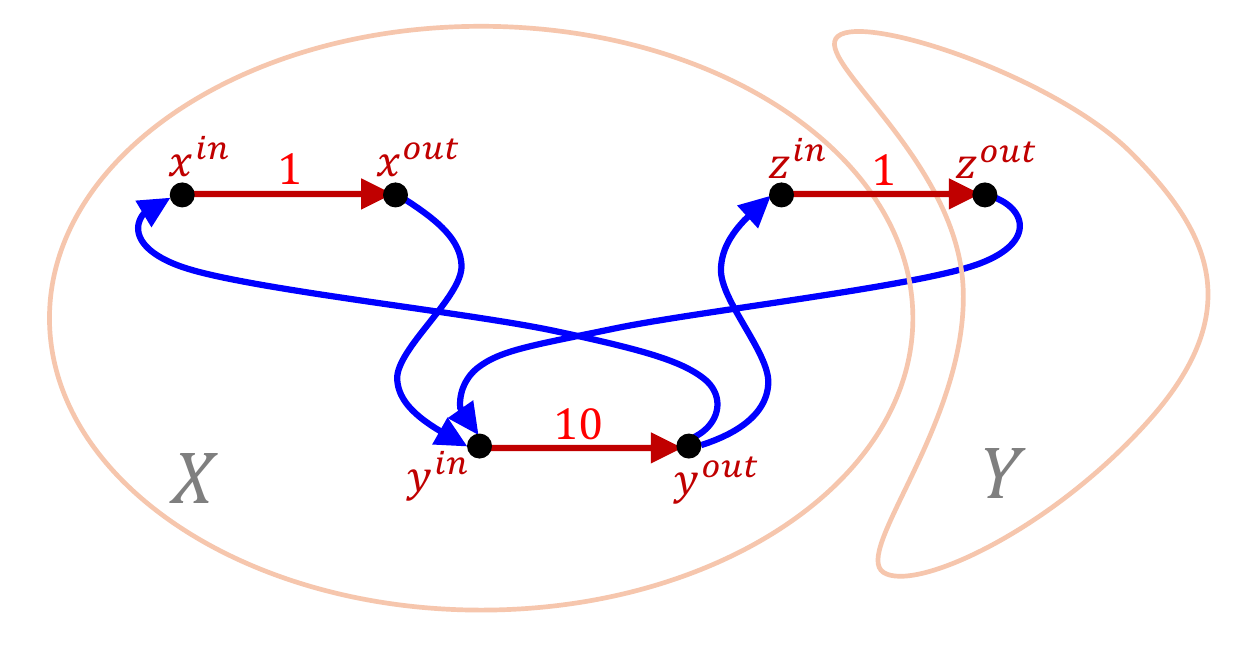}}}
	\caption{An illustration of the classic split-graph reduction from  undirected vertex-weighted graphs to directed edge-weighted graphs, in the context of Global Minimum Vertex-Cut.\label{fig: reduction_fails}}
\end{figure}


\subsection{Overview of Previous Related Work}

In this subsection we provide a high-level overview of previous results for central variants of the Global Minimum Cut problem, and highlight the main challenges in adapting the techniques used in these results to the weighted Vertex-Cut setting.

\paragraph*{Undirected edge-cut version.}
The famous algorithm of Karger and Stein~\cite{KS93} mentioned above achieves an $\tilde O(n^2)$ running time for the problem. The algorithm is iterative. In every iteration, an edge $(u,v)$ is sampled with a probability proportional to its weight, and then $u$ and $v$ are contracted into a single vertex. This process continues until only two vertices remain in the graph, that naturally define an edge-cut, which is then returned by the algorithm. Intuitively, the algorithm succeeds in computing a Global Minimum Edge-Cut with a reasonably high probability, because, if we fix some optimal cut $(X,Y)$, then the edges crossing the cut have a relatively small weight, making it not very likely that an edge crossing the cut is chosen for contraction.  This method is clearly tailored to the (undirected) edge-cut version of the problem.

In a later paper, Karger~\cite{Kar00} presented an $\tO(m)$-time algorithm, that achives a near-optimal running time for the same setting. This algorithm is based on the seminal result by Tutte and Nash-Williams~\cite{Tut61, Nash61}, which showed that any unweighted $c$-edge-connected graph contains a collection of $c/2$ edge-disjoint spanning trees.
Karger~\cite{Kar00} made a crucial observation that the edges of a minimum cut must be partitioned across these 
$c/2$ spanning trees, and so there must be some tree $T$ in the packing, such that the number of edges of $T$ crossing the cut is at most $2$; if the number of edges of a tree $T$ crossing a cut $(X,Y)$ is at most $i$, then we say that the cut \emph{$i$-respects} the tree. 
%
To find the global minimum cut, the algorithm identifies, for each tree in the packing, the minimum cut among those that $2$-respect the tree, and returns the smallest resulting cut. To make this approach efficient, the following strategy was used.
First, in a preprocessing step, the input graph is sparsified by randomly subsampling edges in proportion to their capacities, and treating the resulting ``skeleton" graph $H$ as unweighted. This reduces the number of graph edges to $\tO(n)$ and the minimum cut value to 
$O(\log n)$, 
while ensuring that, if $(X,Y)$ is a minimum cut in $G$, then its value in $H$ is close to the value of the optimal cut in $H$. Next, $O(\log n)$ trees are packed into $H$ using the algorithm of Gabow~\cite{Gabow95}. Finally, dynamic programming is used to examine all trees in $\tO(m)$ time, in order to compute a minimum cut in the original graph.
A sequence of works aimed at achieving fast deterministic algorithms for undirected Global Minimum Edge-Cut \cite{HRW17,KawT19,Li21} recently culminated with an $\tilde O(m)$-time algorithm \cite{HLRW24}.

\paragraph*{Directed edge-cut version.}
%
The fastest current algorithm for this version, due to~\cite{CLN21}, solves the problem in the time of $O(\min\{n/m^{1/3},\sqrt{n}\})$ applications of Maximum $s$-$t$ Flow, that, combined with the recent almost-linear time algorithm for the latter problem by \cite{CKLP22,BCP23}, yields an $O(\min\{n/m^{1/3},\sqrt{n}\}\cdot m^{1+o(1)})$-time algorithm for directed edge-weighted Global Minimum Cut. At a high level, the algorithm uses the ideas from \cite{Kar00} described above.
However, one significant barrier in this setting is that, in the context of directed graphs, a sparsifier that (approximately) preserves all cuts may not exist (see, e.g.,~\cite{CCPS21}). 
In order to overcome this difficulty, the authors only rely on this method if the global minimum cut is unbalanced; since the number of unbalanced cuts can be suitably bounded even in directed graphs, arguments similar to those from~\cite{Kar00} can be used. 
The case of a balanced global minimum cut is handled by sampling vertices and computing Maximum $s$-$t$ Flow between a subset of pairs of vertices from that sampled set; it is then shown that, with high probability, this subset contains at least one pair of vertices from opposite sides of the cut.

We note that the tree-packing based techniques used in the above algorithms appear to be unhelpful in the vertex-cut setting, as tree-packings have no known analogues in the vertex-connectivity regime that achieve similar guarantees (see~\cite{CGK14}).

\paragraph*{Unweighted vertex-cut version.}
This version of the problem received a considerable amount of attention over the years~\cite{Kle69,Pod73,ET75,LLW88}. Recently, a series of results~\cite{NSY19,FNY20} culminated in an almost linear $O(m^{1+o(1)})$-time algorithm for this setting in undirected graphs, along with a somewhat similar $O(\min\{\tO(n^2),mn^{1-1/12+o(1)}\})$-time algorithm for directed graphs~\cite{LNP21}; both of the above bounds follow from the recent breakthrough almost-linear time algorithm for directed Maximum $s$-$t$ Flow of~\cite{CKLP22,BCP23}.
Our algorithm for weighted Global Minimum Vertex-Cut builds on several high-level ideas of~\cite{LNP21}, but requires a significant additional amount of work  and new ideas in order to handle the weighted version of the problem. We provide a more detailed description of the algorithm of~\cite{LNP21}, the barriers to adapting it to the weighted version of the problem, and describe how our algorithm overcomes them in \Cref{sec: results}.

\paragraph*{Weighted vertex-cut version.}
Henzinger et al.~\cite{HRG00} presented an algorithm with running time $\tO(mn)$ for this version of the problem, by generalizing the algorithm of Hao and Orlin~\cite{HO94} for the edge-cut version. 
At a high level, the algorithm of~\cite{HO94} identifies $n-1$ instances of the Maximum $s$-$t$ Flow problem, such that one of the instances is guaranteed to correspond to the global minimum cut. However, instead of performing $n-1$ applications of the algorithm for Maximum $s$-$t$ Flow, which would be too time consuming, Hao and Orlin showed that one can solve all these $n-1$ instances in time that is amortized to match the time of a single Maximum $s$-$t$ Flow computation. However, this approach only works provided that the Maximum $s$-$t$ Flow algorithm follows the Push-Relabel scheme~\cite{GT88}. Unfortunately, the best current such algorithm has running time $O(mn)$, and the recent almost linear-time algorithms for maximum $s$-$t$ flow~\cite{CKLP22,BCP23} do not follow the Push-Relabel paradigm and hence cannot be leveraged.
Henzinger et al.~\cite{HRG00} refined this method of~\cite{HO94} and adapted it to the vertex-weighted version of the problem. Their algorithm uses the ideas from \cite{BDD82} 
to identify a collection of $\tO(n)$ instances of Maximum $s$-$t$ Flow, whose solution is sufficient in order to compute global minimum vertex-cut. A combination of these two approaches then yields their $\tilde O(mn)$-time algorithm for Global Minimum Vertex-Cut.
Another relevant work in this context is the recent combinatorial $O\left(n^{2+o(1)}\right)$-time algorithm for Maximum s-t Flow of~\cite{BBS24}. The algorithm proceeds by setting the initial $s$-$t$ flow $f$ to $0$, and then iteratively augmenting it. In every iteration, an approximate Maximum $s$-$t$ Flow $f'$ in the residual flow network $G_f$ is computed, via an algorithm that can be viewed as following the Push-Relabel scheme. The flow $f$ is then augmented via this flow $f'$, and the algorithm continues until an optimal flow $f$ is obtained. Since this adaptation of the Push-Relabel scheme only computes an approximate flow in the residual flow network, to the best of our understanding, it cannot be directly used to obtain faster algorithms for Global Minimum Vertex-Cut using the approach of~\cite{HO94,HRG00} described above.
Other recent developments on vertex-weighted Global Minimum Cut include a $(1+\varepsilon)$-approximation algorithm with running time $\tO(n^2/\varepsilon^2)$ by~\cite{CLN21}, and an exact algorithm  with pseudo-polynomial running time $\tO(\opt \cdot n \cdot \sum_{u\in V} w(u))$ by ~\cite{CQ21}.

\paragraph{Other Related Work.}
%
%
%
As mentioned already, there is a long and extensive line of work on Global Minimum Cut and its many variations.
We summarize here some additional previous results that are most relevant to us, namely algorithms for the unweighted version of Global Minimum Vertex-Cut. We denote by $k$ the value of the optimal solution (that may be as high as $\Omega(n)$), and by $T(m)$ the running time of the Maximum $s$-$t$ Flow algorithm on an $m$-edge graph.
One of the earliest algorithms for this problem by~\cite{Kle69} achieved a running time of $O(nk \cdot T(m))$ for undirected graphs. Subsequently,~\cite{Pod73, ET75} developed algorithms for directed graphs, with a similar running time $O(nk \cdot T(m))$.
 Even~\cite{Even75} provided an algorithm with a running time of $(n + k^2) \cdot T(m)$, and later~\cite{BDD82} designed an algorithm with a running time of $\tO(n) \cdot T(m)$ for undirected graphs.
Linial, Lov\'asz, and Wigderson~\cite{LLW88} discovered a surprising connection between the problem and convex hulls, that led to an algorithm with running time $\tO(n^{\omega} + nk^{\omega})$ for undirected graphs, where $\omega$ is the matrix multiplication exponent, that simplifies to $\tO(n^{1+\omega})$ in the worst case when $k=\Theta(n)$.
Their proof mimics a physical process, where an embedding of a graph is computed by viewing its edges as ideal springs, and  then letting its vertices settle into an equilibrium. In a follow-up work,~\cite{CR94} extended their result to directed graphs.
An algorithm with running time $O(k^3n^{1.5} + k^2n^2)$ was designed for the undirected version of the problem by~\cite{NI92, CKT93}.
Finally,~\cite{Gab06} provided an algorithm with running time $\tO(n + k\sqrt{n}) \cdot T(m) = \tO(n^{1.5}) \cdot T(m)$ by using expanders in order to identify a small number of vertex pairs for which an algorithm for the Maximum $s$-$t$ flow is then executed.
More recently,~\cite{FNY20} gave an algorithm with $\tO(m+nk^3)$ running time for undirected graphs and $\tO(mk^2)$ time for directed graphs, and~\cite{SY22} provided an algorithm with  running time $\tO\left(m^{1+o(1)} \cdot 2^{O(k^2)}\right)$ for undirected graphs.
We note that in the worst-case setting where $k=\Omega(n)$, the running times of all these algorithms are at least as high as $\Theta(mn)$.

\subsection{Our Results and Techniques}\label{sec: results}
We note that, if a graph $G$ is a complete graph, then it does not have any vertex-cuts. Given a graph $G$, we can check whether this is the case in time $O(|E(G)|)$. Therefore, in the statements of our results we assume that this is not the case.
Our main result is summarized in the following theorem.

\begin{theorem}\label{thm: main}
	There is a randomized algorithm, that, given a simple undirected $n$-vertex and $m$-edge graph $G$ with integral weights $0\leq w(v)\leq W$ on vertices $v\in V(G)$, such that $G$ is not a complete graph, returns a vertex-cut $(A,B,C)$ in $G$, such that, with probability at least $1-1/n$, cut $(A,B,C)$ is a global minimum vertex-cut in $G$. The running time of the algorithm is: $$O\left(\min\{m^{3/2+o(1)},mn^{0.99+o(1)}\}\cdot (\log W)^{O(1)}\right ).$$
\end{theorem}

We now provide a high-level overview of the techniques used in the proof of \Cref{thm: main}, where we mostly focus on the algorithm with running time $O\left(mn^{0.99+o(1)}\cdot (\log W)^{O(1)}\right )$, the most technically challenging part of our result.
The high-level intuitive description provided here is somewhat oversimplified and omits some technical details, since we prioritize clarity of exposition over precision in this high-level overview.

 Let $G$ be the input simple undirected $n$-vertex and $m$-edge graph with integral weights  $0\leq w(v)\leq W$ on its vertices. In order to simplify the discussion in this section, we assume that $W\leq \poly(n)$. For simplicity, we assume that all vertex weights are non-zero; we later show a simple transformation that allows us to achieve this property without changing the problem. Throughout, we use a parameter $0<\eps<1$. We fix a global minimum vertex-cut $(L,S,R)$, and we assume w.l.o.g. that $w(L)\leq w(R)$. We denote by $\opt=w(S)$ the value of the optimal solution.
 It appears that the most difficult special case of the problem is where $|L|\leq n^{\eps}$ and the graph $G$ is dense (e.g., $m\geq n^{1.98}$). We start by providing simple algorithms for other special cases, that essentially reduce the problem to this hard special case.
 
 \paragraph{Algorithm $\alg_1$: when $|S|$ is small.} Our first algorithm, $\alg_1$, is designed to work in the regime where $|S|\leq |L|\cdot n^{1-\eps}$. The algorithm for this special case follows  the paradigm that  was introduced by~\cite{LNP21} in the context of computing global minimum vertex-cut in unweighted graphs. We start by carefully subsampling a subset $T$ of vertices of $G$ in a way that ensures that, with probability at least $\Omega(1/n^{1-\eps})$, $|T\cap L|=1$ and $T\cap S=\emptyset$ holds. If the sampled set $T$ of vertices has these properties, then we say that the sampling algorithm is \emph{successful}. We then use the Isolating Cuts Lemma of~\cite{LP20,AKT21_stoc}, combined with the recent almost linear-time algorithm for directed Maximum $s$-$t$ Flow of \cite{CKLP22,BCP23} (see also \Cref{thm: min cuts via isolating}) to compute, in time $O(m^{1+o(1)})$, for every vertex $v\in T$, a minimum vertex-cut separating $v$ from $T\setminus\set{v}$ in $G$. Lastly, we return the smallest of the resulting cuts. It is immediate to verify that, if the sampling algorithm is successful, then the cut returned by the algorithm is indeed a global minimum vertex-cut. By repeating the algorithm $\tilde O(n^{1-\eps})$ times, we ensure that with high probability, in at least one of its executions, the sampling algorithm is successful. The total running time of the algorithm for this special case is $O(mn^{1-\eps+o(1)})$. 

\paragraph{Algorithm $\alg_2$: when $|S|$ is large and $G$ is not dense.}
Our second algorithm, $\alg_2$, is designed for the regime where $|S|\geq n^{1-\eps}\cdot |L|$ and $m\leq n^{2-2\eps}$. In this case, we let $T$ be the set of all vertices of $G$ whose degree is at least $n^{1-\eps}$. Since $m\leq n^{2-2\eps}$, it is easy to verify that $|T|\leq 2n^{1-\eps}$. Our main observation is that at least one vertex of $T$ must lie in $L$. Indeed, it is easy to verify that every vertex of $S$ must have an edge connecting it to some vertex of $L$, and so the total number of edges incident to the vertices of $L$ is at least $|S|$. Since $|S|\geq n^{1-\eps}\cdot |L|$, at least one vertex in $L$ must have degree at least $n^{1-\eps}$, so it must lie in $T$. Next, we provide a random procedure that receives as input a vertex $x\in V(G)$, and returns another vertex $y_x\in V(G)$, such that, if $x\in L$, then with a sufficiently high probability $y_x\in R$ holds. The procedure uses the fact that, if $x\in L$, then the total weight of all vertices of $S$ that are not neighbors of $x$ is bounded by $w(L)\leq w(R)$. Vertex $y_x$ is then selected uniformly at random from $V(G)\setminus\left(\set{x}\cup N_G(x)\right )$, and, if $x\in L$, then it is guaranteed to lie in $R$ with a constant probability. Lastly, for every vertex $x\in T$, we compute a minimum vertex-cut separating $x$ from $y_x$ in $G$ using the almost-linear algorithm of \cite{CKLP22,BCP23} for directed Maximum $s$-$t$ Flow (see also \Cref{cor: min_vertex_cut}), and return the cut of the smallest weight from among all the resulting cuts. Since at least one vertex $x\in T$ must lie in $L$, with a sufficiently high probability $y_x\in R$ must hold, and in this case, the cut that the algorithm returns is indeed a minimum vertex-cut. 

By combining Algorithms $\alg_1$ and $\alg_2$, we obtain an algorithm for global Minimum Vertex-Cut in non-dense graphs: specifically, if $m\leq n^{2-2\eps}$, then the running time of the resulting algorithm is $O\left (mn^{1-\eps+o(1)}\right )$. We note that this algorithm can be used directly to obtain a randomized algorithm for global minimum vertex-cut with running time $O\left(m^{3/2+o(1)}\right )$; see \Cref{subsec: finishing the alg} for details. We use this algorithm with $\eps=1/100$ to obtain an algorithm with running time $O\left (mn^{0.99+o(1)}\right )$ for graphs in which $m\leq n^{1.98}$ holds. From now on it is sufficient to design an algorithm for dense graphs, where $m\geq n^{1.98}$. In the remainder of this exposition, we focus on such graphs, and we set $\eps=1/45$ for this part of the algorithm. Note that Algorithm $\alg_1$ works well even in dense graphs, provided that $|S|\leq |L|\cdot n^{1-\eps}$ holds. From now on we assume that this is not the case, so in particular, $|L|\leq \frac{|S|}{n^{1-\eps}}\leq n^{\eps}$. 
We denote our algorithm for this remaining special case, that we describe below, by $\alg_3$.

\paragraph{A good set of terminals.}

At a very high level, we employ the notion of \emph{a good set of terminals}, that was (implicitly) introduced by~\cite{LNP21}, in their algorithm for  global Minimum Vertex-Cut in unweighted graphs. 
Recall that we have fixed a global minimum vertex-cut $(L,S,R)$, where $|L|\leq n^{\eps}$. For a vertex set $T\subseteq V(G)$, we say that it is a \emph{good set of terminals} if $|T\cap L|=1$ and $T\cap R\neq \emptyset$ hold, and, additionally, if we denote by $x$ the unique vertex in $T\cap L$, then $T\cap S\subseteq N_G(x)$. For every vertex $x\in T$, we denote by $T_x=T\setminus\left (\set{x}\cup N_G(x)\right )$. We say that $(x,T)$ is a \emph{good pair}, if $T$ is a good set of terminals, and $x$ is the unique vertex in $L\cap T$. Notice that, if $(x,T)$ is a good pair, then $T_x\subseteq R$ and $T_x\neq\emptyset$, and so any minimum vertex-cut separating $x$ from $T_x$ in $G$ is a global minimum cut.

The algorithm of \cite{LNP21} starts by selecting a set $T$ of vertices at random, where every vertex is added to $T$ with probability roughly $1/|L|$. It is not hard to show that, if the graph $G$ is unweighted, then with a reasonably high probability, the resulting set $T$ is a good set of terminals, and moreover, $|T|\leq \tilde O(n/|L|)$. For every vertex $x\in T$, let $G_x$ be the graph obtained from $G$ by unifying the vertices of $T_x$ into a destination vertex $t_x$. For every vertex $x\in T$, the algorithm of~\cite{LNP21} computes the value of the minimum $x$--$t_x$ vertex cut in $G_x$, and then returns the smallest of the resulting cuts. In order to do so efficiently, they first \emph{sparsify} each such graph $G_x$, by cleverly employing  sparse recovery sketches. With a sufficiently high probability\footnote{In fact they provide a different simple algorithm for the special case where $S$ contains few low-degree vertices, and the above bound only holds for the remaining case; we ignore these technical details in this informal exposition.}, each resulting sparsified graph $\tilde G_x$ contains at most $\tilde O(\opt\cdot |L|)$ edges. 
Applying the almost-linear time algorithm for maximum $s$-$t$ flow of \cite{CKLP22,BCP23} to each such graph $\tilde G_x$ for $x\in T$ then yields an algorithm for computing global minimum cut in $G$, whose running time, with high probability, is bounded by $O\left((\opt\cdot |L|)^{1+o(1)}\cdot |T|\right )\leq O\left(\opt\cdot n^{1+o(1)}\right )$. Finally, it is easy to see that, in the unweighted setting, the degree of every vertex in $G$ must be at least $\opt$, so $m\geq \Omega(\opt\cdot n)$ must hold, leading to the bound of $O\left(m^{1+o(1)}\right )$ on the running time of their algorithm.

Our algorithm also follows this high-level scheme: we use a randomized procedure, that is adapted to vertex-weighted graphs, in order to select a subset $T\subseteq V(G)$ of vertices, so that, with a reasonably high probability, $T$ is a good set of terminals. Then  for every terminal $x\in T$, we compute the value of the minimum $x$-$t_x$ cut in the corresponding graph $G_x$, and output the smallest such value. But unfortunately, the algorithm of~\cite{LNP21} can no longer be used in order to sparsify each such graph $G_x$ in the weighted setting, and it is not clear how to extend it to this setting. Among other problems, in the weighted setting, the value $\opt$ of the optimal solution may be much higher than $n$, and we can no longer claim that the degree of every vertex in $G$ is at least $\opt$. Additionally, in order to ensure that the set $T$ of vertices is a good set of terminals with a sufficiently high probability, we need to sample vertices into $T$ with probability that is proportional to their weight, and other parts of the analysis of the algorithm of~\cite{LNP21}  do not work in this setting. We now proceed to describe our algorithm in more detail.

We start by computing the \emph{split graph} $G'$ corresponding to graph $G$, that is a standard step in transforming cut and flow problems in undirected vertex-capacitated graphs to directed edge-capacitated graphs. The set of vertices of $G'$ is $V(G')=\set{v^{\inn},v^{\out}\mid v\in V(G)}$; we refer to $v^{\inn}$ and $v^{\out}$ as the \emph{in-copy} and the \emph{out-copy} of $v$. Given a vertex set $X\subseteq V(G)$, we will also denote by $X^{\inn}=\set{x^{\inn}\mid x\in X}$, $X^{\out}=\set{x^{\out}\mid x\in X}$, and $X^*=X^{\inn}\cup X^{\out}$. We denote by $V^{\inn}=\set{v^{\inn}\mid v\in V(G)}$ and by $V^{\out}=\set{v^{\out}\mid v\in V(G)}$. 
The set $E(G')$ of edges of $G'$ is partitioned into two subsets: the set of \emph{special edges}, that contains, for every vertex $v\in V(G)$, the special edge $e_v=(v^{\inn},v^{\out})$ of capacity $c(e_v)=w(v)$, and the set of \emph{regular edges}, that contains, for every edge $(u,v)\in E(G)$, the regular edges $(u^{\out},v^{\inn})$ and $(v^{\out},u^{\inn})$, whose capacities are set to a value $\wmax>4n\cdot W$, that is an integral power of $2$. 

For every vertex $x\in T$, let $G'_x$ be the graph that is obtained from $G'$ by adding a destination vertex $t$, and, for every vertex $y\in T_x$, adding a regular edge $(y^{\inn},t)$ of capacity $\wmax$ to the graph. It is immediate to verify that the value of the minimum $x$-$t_x$ vertex-cut in $G_x$ is equal to the value of the minimum $x^{\out}$-$t$ edge-cut in $G'_x$, which, in turn, from the Max-Flow / Min-Cut theorem, is equal to the value of the maximum $x^{\out}$-$t$ flow in $G'_x$.

Our algorithm computes, for every vertex $x\in T$, a value $c_x$, that is at least as high as the value of the maximum $x^{\out}$-$t$ flow in $G'_x$. Additionally, if $(x,T)$ is a good pair, then our algorithm guarantees that, with a sufficiently high probability, $c_x$ is equal to the value of the maximum $x^{\out}$-$t$ flow in $G'_x$ (which, in turn, must be equal to $\opt$). Let $x\in T$ be the vertex for which $c_x$ is the smallest, and let $y$ be any vertex in $T_x$. If $T$ is a good set of terminals, then the value of the global minimum cut in $G$ is then guaranteed, with a sufficiently high probability, to be equal to the value of the minimum vertex-cut separating $x$ from $y$ in $G$, which, in turn, is equal to $c_x$. Computing the minimum $x$-$y$ vertex-cut in $G$ then yields the  global minimum vertex-cut with a sufficiently high probability.

At the heart of the algorithm is our main technical subroutine, that receives as input a set $T\subseteq V(G)$ of vertices and a vertex $x\in T$, together with two copies of the graph $G'_x$ in the adjacency-list representation. The algorithm outputs a value $c_x$, that is at least as high as the value of the maximum $x^{\out}$-$t$ flow in $G'_x$. Moreover, if $(x,T)$ is a good pair, then with probability at least $1/2$, $c_x$ is guaranteed to be equal to the value of the maximum $x^{\out}$-$t$ flow in $G'_x$.
In the remainder of this exposition, we focus on the algorithm for this subroutine. For convenience, we denote $G'_x$ by $G'$ and $x^{\out}\in V(G')$ by $s$. We also denote by $\opt_x$ the value of the maximum $s$-$t$ flow in $G'$.
We note that, since $|T|$ may be as large as $\Omega(n)$, we need to ensure that the running time of the subroutine is significantly lower than $n^2$. In particular, if $G$ is sufficiently dense, this running time may need to be lower than $|E(G')|$.

Recall that we have set the parameter $\eps=1/45$. We also use another parameter $n^{7/9}\leq \gamma\leq 2n^{7/9}$, that is an integral power of $2$. This setting ensures that $n^{14\eps}\leq \gamma\leq 2n^{1-10\eps}$.
The algorithm for the subroutine consists of $z=O(\poly\log n)$ \emph{phases}. For all $0\leq i\leq z$, we also use the parameter $M_i=\frac{\wmax'}{2^i\cdot \gamma}$, where $\wmax'$ is an integral power of $2$ that is greater than the largest weight of any vertex in $G$.

Intuitively, our algorithm maintains an $s$-$t$ flow $f$ in $G'$, starting with $f=0$, and then gradually augments it. The \emph{deficit} of the flow $f$ is $\Delta(f)=\opt_x-\val(f)$.
Specifically, for all $1\leq i\leq z$, 
we denote by $f_i$ the flow $f$ at the beginning of the $i$th phase. We ensure that, for all $1\leq i\leq z$, flow $f_i$ is $M_i$-integral, and its deficit is bounded by $4nM_i$. Additionally, for technical reasons, we require that, for every vertex $v\in V(G)$ with $w(v)\geq M_i$, if $e=(v^{\inn},v^{\out})$ is the corresponding special edge, then $f_i(e)\leq w(v)-M_i$; in other words, the residual capacity of the edge $e$ remains at least $M_i$. Assume first that we can indeed ensure these properties. If $z$ is sufficiently large, then the flow $f_z$ has deficit at most $4nM_z<1$. We can then construct a graph $\hat G\subseteq G'$, that only contains edges $e\in E(G')$ with $f(e)>0$; the number of such edges is bounded by the total running time of our subroutine so far, which, in turn, is significantly less than $\Theta(n^2)$. Since all edge capacities in $\hat G$ are integral, the value of the maximum $s$-$t$ flow in $\hat G$ must be $\ceil{\val(f_z)}=\opt_x$. We then compute the value of the maximum $s$-$t$ flow in $\hat G$ and output it as $c_x$. Before we describe our algorithm for the preprocessing step and for each phase, we need to introduce the notion of \emph{shortcut operations}.

\paragraph{Shortcut Operations.}
Over the course of the algorithm, we may slightly modify the graph $G'$ by performing \emph{shortcut operations}. In a shortcut operation, we insert an edge $(v,t)$ into $G'$, for some vertex $v\in V(G')$. If the pair $(x,T)$ is not good, then we say that such an operation is always \emph{valid}. Notice that the insertion of the edge may only increase the value of the maximum $s$-$t$ flow in $G'$. Otherwise, if the pair $(x,T)$ is good, then we say that the operation is valid only if $v\in S^{\out}\cup R^*$, that is, $v$ is either an out-copy of a vertex of $S$, or it is an in-copy or an out-copy of a vertex of $R$. By inspecting the minimum $s$-$t$ edge-cut in $G'$ (whose corresponding edge set is $\set{(u^{\inn},u^{\out})\mid u\in S}$), it is easy to verify that a valid shortcut operation may not change the value of the maximum $s$-$t$ flow in $G'$ in this case. We ensure that, with a sufficiently high probability, all shortcut operations that the algorithm performs are valid. For clarity, we continue to denote the original graph $G'$ by $G'$, and we use $G''$ to denote the graph obtained from $G'$ after our algorithm possibly performed a number of shortcut operations. The flow $f$ that our algorithm maintains is in graph $G''$. Recall that our algorithm is given as input 2 copies of the graph $G'$ in the adjacency-list representation. We use one of these copies to maintain the graph $G''$, and we use the second copy in order to maintain the residual graph of $G''$ with respect to the current flow $f$, that we denote by $H$.

\paragraph{The Preprocessing Step.} The purpose of the preprocessing step is to compute a flow $f_1$ that can serve as the input to the first phase. Recall that we require that $f_1$ is $M_1$-integral, where $M_1=\frac{\wmax'}{2\gamma}$, and that its deficit is at most $4nM_1$. We also require that, for every vertex $v\in V(G)$ with $w(v)\geq M_1$, if $e=(v^{\inn},v^{\out})$ denotes the corresponding edge of $G''$, then $f_1(e)\leq w(v)-M_1$.
Consider the graph $\hat G\subseteq G''$, that is obtained from $G''$ as follows. First, we delete, for every vertex $v\in V(G)$ with $w(v)<M_1$, both copies of $v$ from the graph. Next, for every remaining special edge $e$, we decrease the capacity of $e$ by at least $M_1$ and at most $2M_1$ units, so that the new capacity is an integral multiple of $M_1$. It is easy to verify that, by computing a maximum $s$-$t$ flow in the resulting flow network $\hat G$, we can obtain the flow $f_1$ with all desired properties. Unfortunately, it is possible that the graph $\hat G$ is very dense, so we cannot compute the flow directly. Instead, we \emph{sparsify} $\hat G$ first, using the following two observations, that essentially generalize Equation (2) and Observation $2$ of~\cite{LNP21} to arbitrary vertex capacities. Let $N=\set{v\in V(\hat G)\mid (s,v)\in E(\hat G)}$, and let $N'\subseteq V(G)$ contain all vertices $u$ such that a copy of $u$ lies in $N$. The first observation is that there is an optimal flow in $\hat G$ that does not use any edges that enter the vertices of $N$, except for the edges that start at $s$, so all such edges can be deleted. Second, it is not hard to prove that the total number of vertices $v\in S\setminus N'$ of weight at least $M_1$ is bounded by $2|L|\gamma$. It then follows that, if $u\in V(G)$ has more than $3|L|\gamma$ neighbors in $V(G)\setminus N'$ that have weight at least $M_1$, then at least one such neighbor must lie in $R$, and so $u$ may not belong to $L$. We can then safely add a shortcut edge $(v^{\out},t)$ to both $G''$ and $\hat G$, and we can ignore all other edges leaving $v^{\out}$. These two observations allow us to sparsify the graph $\hat G$, obtaining a smaller graph $\hat G'\subseteq \hat G$, such that, on the one hand, $|E(\hat G')|\leq \tilde O(n^{2-4\eps})$, while on the other hand, all shortcut operations performed so far are valid with a high enough probability, so the value of the maximum $s$-$t$ flow in $\hat G'$ is at least as high as that in $\hat G$, and the two flow values are equal if $(x,T)$ is a good pair.

\paragraph{Executing a Single Phase.}
We now provide the description of the $i$th phase, for $1\leq i\leq z$. Recall that we receive as input a flow $f=f_i$ in the current graph $G''$, that is $M_i$-integral, and whose deficit is bounded by $4nM_i$. Additionally, for every vertex $v\in V(G)$ with $w(v)\geq M_i$, if $e=(v^{\inn},v^{\out})$ is its corresponding edge in $G''$, then $f(e)\leq w(v)-M_i$. Therefore, if we denote by $U_i=\set{v\in V(G)\mid w(v)\geq M_i}$, then, for every vertex $v\in U_i$, the residual capacity of the edge $e=(v^{\inn},v^{\out})$ is at least $M_i$.
Throughout, we denote by $H$ the residual flow network of $G''$ with respect to the current flow $f$. By the properties of the residual flow network, the value of the maximum $s$-$t$ flow in $H$ is at least $\opt_x-\val(f)$ (and it is equal to $\opt_x-\val(f)$ if $(x,T)$ is a good pair and all shortcut operations so far have been valid).
 As before, we can consider the graph $\hat H$, that is obtained from $H$ by first deleting, for every vertex $u\in V(G)\setminus U_i$, both copies of $u$, and then reducing, for each remaining forward special edge $(v^{\inn},v^{\out})$, its residual capacity by at least $M_{i+1}=M_i/2$ and at most $2M_{i+1}$, so that it becomes an integral multiple of $M_{i+1}$. As before, computing a maximum $s$-$t$ flow in the resulting flow network $\hat H$ would yield the flow $f_{i+1}$ with all desired properties. But unfortunately graph $\hat H$ may be very dense, and we may not be able to afford to compute the flow directly. As before, we will sparsify this graph first, though the sparsification procedure is much more challenging now. In fact our algorithm may first gradually augment the flow $f$ and add some new shortcut edges, before sparsifying the resulting new graph $\hat H\subseteq H$.
 
 The key step in our sparsification procedure is to compute a partition $(Q,Q')$ of the set $U_i$ of vertices of $G$, so that $|Q'\cap S|<n^{1-4\eps}$. Additionally, if $|Q|>n^{1-4\eps}$, then we also compute an $s$-$t$ edge-cut $(X^*,Y^*)$ in the current residual flow network $H$, that has the following property: the total residual capacity of the edges $E_H(X^*,Y^*)$ is close to the value of the maximum $s$-$t$ flow in $H$. In other words, in any maximum $s$-$t$ flow in $H$, the edges of $E_H(X^*,Y^*)$ are close to being saturated. Lastly, we ensure that $|E_H(X^*,Y^*)|$ is relatively small, and, for every vertex $v\in Q$, $v^{\inn}\in X^*$ and $v^{\out}\in Y^*$ holds. Assume for now that we have computed  the partition $(Q,Q')$ of $U_i$ and the cut $(X^*,Y^*)$ as above. Observe that, if some vertex $v\in V(G)$ has more than $2n^{1-4\eps}$  neighbors in $Q'$, then, since $|L|\leq n^{\eps}$ and $|Q'\cap S|\leq n^{1-4\eps}$, at least one of these neighbors must lie in $R$, and so $v$ may not lie in $L$. We can then safely add a shortcut edge $(v^{\out},t)$, and ignore all other edges leaving $v^{\out}$ in $\hat H$. Assume now that $|Q|>n^{1-4\eps}$, and let $E'=\set{e=(y,x)\in E(H)\mid y\in Y^*,x\in X^*}$. Consider any maximum $s$-$t$ flow $f^*$ in $H$, and denote its value by $F$. Since the total capacity of all edges in $E_H(X^*,Y^*)$ is close to $F$, it is easy to see that the total flow that $f^*$ sends on the edges of $E'$ is relatively low. Indeed, if $\pset$ is a flow-path decomposition of $f^*$, and $\pset'\subseteq \pset$ is the subset of paths containing the edges of $E'$, then every path in $\pset'$ must use at least two edges of $E_H(X^*,Y^*)$, so these paths may only carry a relatively small amount of flow. We can then delete the edges of $E'$ from the sparsified graph without decreasing the value of the maximum $s$-$t$ flow by too much. We use these and other observation in order to sparsify the graph $\hat H$, to obtain a graph $\hat H'$ with $|E(\hat H')|\leq O(n^{2-4\eps+o(1)})$, such that the value of the maximum $s$-$t$ flow in $\hat H'$ is close to that in $\hat H$, and all edge capacities in $\hat H'$ are integral multiples of $M_{i+1}$. Computing the maximum $s$-$t$ flow in $\hat H'$ then yields the desired flow $f_{i+1}$, that serves as the output of the current phase.

\paragraph{Computing the Partition $(Q,Q')$.}
The most technically challenging part of our algorithm is computing the partition $(Q,Q')$ of $U_i$, and, if needed, the cut $(X^*,Y^*)$ in $H$ with the properties described above. Our first step towards this goal is to compute a subgraph $J\subseteq H$ with $s\in J$ and $t\not \in J$, such that $|V^{\out}\cap V(J)|$ is small, and additionally, for every vertex $u\not \in J$, the total number of edges in set $\set{(x,u)\mid x\in V(J)}$ is relatively small, as is their total capacity in $H$. Assume first that we have computed such a graph $J$. We start with an initial partition $(Q,Q')$ of $U_i$, where $Q$ contains all vertices $v\in U_i$ for which $v^{\inn}\in V(J)$ and $v^{\out}\not\in V(J)$, and $Q'=U_i\setminus Q$. We show that $|Q'\cap S|$ must be relatively small. Consider now a graph $\tilde H$, that is obtained from $H$, by unifying all vertices of $V(H)\setminus V(J)$ into a new destination vertex $t'$; we keep parallel edges and discard self-loops. Note that, for every edge $e=(x,y)\in E(H)$ with $x\in V(J)$ and $y\not\in V(J)$, there is an edge $(x,t')$ corresponding to $e$ in $\tilde H$; we do not distinguish between these two edges.
Let $F$ denote the value of the maximum $s$-$t$ flow in $H$, and let $F'$ denote the value of the maximum $s$-$t'$ flow in $\tilde H$.
 Our key observation is that, on the one hand, $F'$ is quite close to $F$, while, on the other hand, if $f'$ is any maximum $s$-$t'$ flow in $\tilde H$, then the total flow on the edges $\set{(v^{\inn},v^{\out})\mid v\in Q\cap S}$ must be close to their capacity. In other words, if we compute any maximum $s$-$t'$ flow $f'$ in $\tilde H$, and a subset $B\subseteq Q$ of vertices, such that the total amount of flow on the edges of  $\set{(v^{\inn},v^{\out})\mid v\in B}$ is significantly lower than their capacity, then only a small fraction of vertices of $B$ may lie in $S$; we call such a vertex set $B\subseteq Q$ a \emph{bad batch}. Our algorithm iteratively computes bad batches $B\subseteq Q$ of a sufficiently large cardinality, and then moves all vertices of $B$ from $Q$ to $Q'$. In order to compute a bad batch, we set up an instance of Min-Cost Maximum $s$-$t'$ Flow in $\tilde H$, by setting the \emph{cost} of every edge $(v^{\inn},t')$ corresponding to vertices $v\in Q$ to $1$, and setting the costs
of all other edges to $0$. Once the algorithm terminates, we are guaranteed that $Q$ may no longer contain large bad batches of vertices, or, equivalently, in any maximum $s$-$t'$ flow in $\tilde H$, the vast majority of the edges in $\set{(v^{\inn},v^{\out})\mid v\in Q}$ are close to being saturated. If $|Q|\leq n^{1-4\eps}$, then we terminate the algorithm with the resulting partition $(Q,Q')$ of $U_i$. Otherwise, we perform one additional step that computes a cut $(X,Y)$ in the graph $\tilde H$, that in turn defines an $s$-$t$ cut $(X^*,Y^*)$ in $H$ with the desired properties. Note that this step crucially relies on the fact that the graph $\tilde H$ is relatively small, which, in turn, follows from the fact that $|V(J)\cap V^{\out}|$ is relatively small, and so is the number of edges $(x,y)\in E(H)$ with $x\in V(J)$ and $y\not\in V(J)$.

\paragraph{Computing the Subgraph $J$.} Recall that, in order to compute the partition $(Q,Q')$ as described above, we need first to  compute a subgraph $J\subseteq H$ with $s\in J$ and $t\not \in J$, such that $|V^{\out}\cap V(J)|$ is small, and additionally, for every vertex $u\not \in J$, the total number of edges in set $\set{(x,u)\mid x\in V(J)}$ is relatively small, and so is their total capacity in $H$. We do not compute this subgraph directly. Instead, we perform a number of iterations. In every iteration, we either add a shortcut edge and augment the current flow $f$ by a significant amount; or we compute the subgraph $J$ of the current residual network $H$ with the desired properties and terminate the algorithm. We note that, as we add shortcut edges, and as the flow $f$ is augmented, the residual flow network $H$ is updated accordingly, and the graph $J$ that our algorithm eventually computes is a subgraph of the current residual flow network $H$.
Our algorithm for this step relies on the technique of \emph{local flow augmentations}, that was first introduced by~\cite{CHILP17} in the context of the Maximal $2$-Connected Subgraph problem. 
The technique was later refined and applied to \emph{unweighted} and \emph{approximation} versions of the global minimum vertex-cut problem~\cite{NSY19,FNY20,CQ21}.

We first provide a high-level intuitive description of this technique in the context of past work on unweighted global minimum vertex-cut, and then describe our approach that extends and generalizes this technique. 
Our description here is somewhat different from that provided in previous work, as we rephrase it to match the terminology used in this overview.
Consider an instance $G$ of the unweighted Global Minimum Vertex-Cut problem, and assume that there is an optimal solution $(L,S,R)$ where $|L|$ is small. Suppose we are given a vertex\footnote{In fact a vertex $v$ is selected from $G$ via some random process, and we are only guaranteed that with some reasonably high probability, $v\in L$; this is similar to our setting where we first select a random set $T$ of vertices that is a good set of terminals with a sufficiently high probability, and then select a random vertex $v\in T$.} $v\in L$. We construct a split graph $G'$ as before, and add a destination vertex $t$ to it. Intuitively, our goal is to compute a flow of value $|S|$ from $s=v^{\out}$ to $t$ in $G'$, while we are allowed to add valid shortcut edges to $G'$. We start with a flow $f=0$, and then iterate. In every iteration, we perform a DFS search in the residual flow network $H$ of $G'$ corresponding to the current flow $f$, and, once the search discovers roughly $|L|\cdot |S|$ vertices, we terminate it. Let $X$ be the set of vertices that the search has discovered. We note that, if $G$ is a graph with unit vertex weights, and $f$ is an integral flow, then every vertex of $V^{\inn}$ has exactly one edge leaving it in $H$, and the other endpoint of the edge lies in $V^{\out}$. Additionally, every vertex of $V^{\out}$ has exactly one incoming edge in $H$. Therefore, at least half the vertices of $X$ must lie in $V^{\out}$. We select a vertex $u\in X\cap V^{\out}$ uniformly at random, and add a shortcut edge $(u,t)$ to $G'$. Since $|X|\geq 2|L|\cdot |S|$, with probability at least $\left (1-\frac 1{2|S|}\right )$, $u\in S^{\out}\cup R^{\out}$, and so the shortcut operation is valid. Let $P$ be the path connecting $s$ to $u$ that the search computed. We augment $f$ by sending $1$ flow unit via the path $P\circ (u,t)$ and terminate the iteration. Lastly, if the DFS search only discovers fewer than $2|L|\cdot |S|$ vertices, then the total number of vertices reachable from $s$ in $H$ is small, and we can compute a maximum $s$-$t$ flow in $H$ directly. Since the number of iterations is bounded by $\opt=|S|$, with a reasonably high probability, all shortcut operation are valid, and so we compute the value of the Global Minimum Vertex-Cut correctly.

There are several issues with adapting this algorithm to the vertex-weighted setting. The first issue is that, if $G$ is a graph with arbitrary vertex weights, then the residual flow network $H$ may no longer have the property that every vertex of $V^{\out}$ has exactly one incoming edge; the number of such edges may be large. As the result, it is possible that the vast majority of the vertices in the set $X$ discovered by the DFS lie in $V^{\inn}$. But we cannot select vertices of $V^{\inn}$ for shortcut operations, since we are not allowed to add edges connecting vertices of $S^{\inn}$ to $t$, and it is possible that $|S|=\Omega(n)$. Therefore, our DFS search needs to ensure that the number of vertices of $V^{\out}$ that it discovers is sufficiently large. To achieve this high number of vertices of $V^{\out}$, it may be forced to explore a significantly larger number of vertices of $V^{\inn}$, leading to a much higher running time of each iteration. In order to obtain a fast overall running time, we then need to ensure that the number of iterations is sufficiently low, and this, in turn, can be ensured by sending a large amount of flow in each iteration. Therefore, when performing a DFS search in $H$, we only consider edges whose capacities are quite high, at least $\gamma\cdot M_i$. Once the DFS search terminates with a set $X$ containing a relatively low number of vertices of $V^{\out}$, we could let $J$ be the subgraph of $H$ induced by the vertices of $X$. This approach could indeed be used to guarantee that $|E(J)|$ is low, but unfortunately it does not ensure the other crucial property that we need: that every vertex $u\in V(H)\setminus V(J)$ has relatively few edges $(u',u)$ with $u'\in V(J)$, and that all such edges have a relatively low capacity. 

In order to overcome these difficulties, we expand the scope of the local-search technique. In every iteration of the local-search algorithm, we explore the graph $H$, starting from the vertex $s$, as follows. Initially, we let $X=\set{s}$. Whenever we encounter an edge $(u,v)$ with $u\in X$ and $v\not \in X$, whose capacity in $H$ is at least $\gamma\cdot M_i$, we add $v$ to $X$. Additionally, whenever we discover a vertex $a\not \in X$, such that, for some integer $i$, there are at least $2^i$ edges in set $\set{(b,a)\in E(H)\mid b\in X}$, whose capacity is at least $M_i\cdot \gamma/2^i$, we add $a$ into $X$. We show that, for every vertex $v$ that is ever added to $X$ via one of the above rules, there is an $s$-$v$ flow in $H[X]$ of value $\gamma\cdot M_i$. An iteration terminates once $|X\cap V^{\out}|$ is sufficiently large, or once neither of the above rules can be applied to expand $X$. In the latter case, we let $J=H[X]$, and we show that $J$ has all required properties. In the former case, we select a vertex $v\in X\cap V^{\out}$ uniformly at random, and connect it to $t$ via a shortcut operation. We then compute an $s$-$v$ flow $f'$ of value $\gamma\cdot M_i$ in graph $H[X]$, which is guaranteed to contain relatively few edges, and augment the current flow $f$ by sending $\gamma\cdot M_i$ flow units from $s$ to $t$ via the flow $f'$, that is extended via the edge $(v,t)$ to reach $t$. We note that, unlike previous work that used the local-flow technique, where the flow augmentation in every iteration was performed via a single path, we augment the flow $f$ via the flow $f'$ that may be more general.


To summarize, the algorithm for a single phase consists of three steps. In the first step, we compute the subgraph $J\subseteq H$, after possibly augmenting the initial flow and adding new shortcut edges. This step exploits the expanded local-search technique. In the second step we compute the partition $(Q,Q')$ of $U_i$, and, if required, an $s$-$t$ cut $(X^*,Y^*)$ in $H$. The former is done by repeated applications of the algorithm for min-cost maximum flow, and the latter is achieved by computing a minimum $s$-$t$ cut in an appropriately constructed graph. Lastly, in the third step, we sparsify the residual flow network $H$, and compute the desired flow $f_{i+1}$. This step relies on the properties of the partition  $(Q,Q')$ and of the cut $(X^*,Y^*)$.

\paragraph{A Subgraph Oracle.}
The algorithm described above needs an access to a subroutine, that we refer to as the \emph{subgraph oracle}. The oracle is given access to an undirected graph $G$ in the adjacency-list representation, a subset $Z$ of its vertices, and a parameter $\delta$. The oracle must return a partition of $V(G)$ into two subsets $Y^{\ell}$ and $Y^h$ such that, for every vertex $v\in Y^h$, the number of neighbors of $v$ that lie in $Z$ is at least roughly $n^{1-\delta}$, while for every vertex $u\in Y^{\ell}$, the number of such neighbors is at most roughly $n^{1-\delta}$. Additionally, the algorithm must return the set $E'=E_G(Y^{\ell},Z)$ of edges, and its running time should be  at most $n^{2-\Theta(\delta)}$. 
Our main technical subroutine, when processing a vertex $x\in T$, performs at most $z\leq \poly\log n$ calls to the oracle, with at most one such call performed per phase.
While it is not hard to design an algorithm that computes the partition $(Y^{\ell},Y^h)$ of $V(G)$ with the desired properties by using the standard sampling technique, computing the set $E'$ of edges within the desired running time appears much more challenging. Instead, we design an algorithm that can simultaneously process up to $n$ such subgraph queries in bulk, in time $O(n^{3-\Theta(\delta)})$. This algorithm casts the problem of computing the sets $E'$ of edges for all these queries simultaneously as an instance of the Triangle Listing problem in an appropriately constructed graph, and then uses the algorithm of \cite{bjorklund2014listing} for this problem.
Our final algorithm for the Global Minimum Vertex-Cut problem proceeds in $z$ iterations. In each iteration $i$, for every vertex $x\in T$, we run the algorithm for processing $x$ described above, while recording all random choices that the algorithm makes, until it performs the $i$th call to the subgraph oracle, and then we terminate it. Once we obtain the $i$th query to the subgraph oracle from each of the algorithms that process vertices of $T$, we ask the subgraph oracle all these queries in bulk. Then iteration $(i+1)$ is performed similarly, except that we follow all random choices made in iteration $i$. Since we now have the response to the $i$th query to the subgraph oracle, we can continue executing the algorithm for processing each vertex $x\in T$ until it performs the $(i+1)$th such query.

\paragraph{Organization.} We start with preliminaries in Section~\ref{sec: prelims}. We provide a high-level outline of our algorithm in Section~\ref{sec: the_algorithm}, including the complete descriptions of Algorithms $\alg_1$ and $\alg_2$, and a description of Algorithm $\alg_3$, with the proof of our main technical result that implements the main subroutine deferred to subsequent sections. In Section~\ref{sec: main subroutine} we provide the algorithm for the main subroutine, with the algorithms describing Steps $1,2$ and $3$ of a single phase provided in sections~\ref{sec: step 1},~\ref{sec: step 2},~\ref{sec: step 3}, respectively.

\section{Preliminaries}
\label{sec: prelims}

All logarithms in this paper are to the base of $2$, unless stated otherwise.
We assume that all graphs are given in the adjacency list representation, unless stated otherwise.

\subsection{Graph-Theoretic Notation, Cuts and Flows}

Given an undirected graph $G$ and a vertex $v\in V(G)$, we denote by $\delta_G(v)$ the set of all edges incident to $v$ in $G$, by $\deg_G(v)=|\delta_G(v)|$ the degree of $v$ in $G$, and by $N_G(v)$ the set of all neighbors of $v$ in $G$: namely, all vertices $u\in V(G)$ with $(u,v)\in E(G)$. If the graph $G$ is directed, then we denote by 
$N^-_G(v)=\set{u\in V(G)\mid (u,v)\in E(G)}$, and we refer to $N^-_G(v)$ as the \emph{set of all in-neighbors of $v$}. Similarly, we denote by $N^+_G(v)=\set{u\in V(G)\mid  (v,u)\in E(G)}$ the set of all \emph{out-neighbors} of $v$. We also denote by $\delta^-_G(v)$ and by $\delta^+_G(v)$ the sets of all incoming and all outgoing edges of $v$ in $G$, respectively. Finally, we denote the degree of $v$ in $G$ by $\deg_G(v)=|\delta^-_G(v)|+|\delta^+_G(v)|$, and we denote by  $N_G(v)=N^+_G(v)\cup N^-_G(v)$; we refer to $N_G(v)$ as the set of all neighbors of $v$ in $G$.
Given a pair $Z,Z'$ of subsets of vertices of $G$ (that are not necessarily disjoint), we denote by $E_G(Z,Z')$ the set of all edges $(x,y)$ with $x\in Z$ and $y\in Z'$; this definition applies to both directed and undirected graphs.
 We may omit the subscript $G$ when the graph is clear from context.

\paragraph{Vertex-Cuts in Undirected Graphs.}
Let $G=(V,E)$ be an undirected graph with weights $w(v)\geq 0$ on its vertices $v\in V$. For any subset $X\subseteq V$ of its vertices, we denote the \emph{weight of $X$} by $w(X)=\sum_{v\in X}w(v)$. A \emph{global vertex-cut}, or just a \emph{vertex-cut}, in $G$ is a partition $(L,S,R)$ of its vertices, such that (i) $L,R\neq \emptyset$; and (ii) no edge of $G$ connects a vertex of $L$ to a vertex of $R$. The \emph{value}, or the \emph{weight} of the cut is $w(S)$. Whenever we consider global vertex-cuts $(L,S,R)$, we will always assume w.l.o.g. that $w(L)\leq w(R)$. In the Global Minimum Vertex-Cut problem, given an undirected graph $G$ with vertex weights as above, the goal is to compute a global vertex-cut of minimum value; we call such a cut a \emph{global minimum vertex-cut}, and we denote its value by $\opt$. We note that some graphs may not have any vertex cuts (for example a clique graph). In such a case, we may say that the value of the global minimum vertex-cut is infinite.

Consider now an undirected vertex-weighted graph $G$ as above and a pair $s,t$ of its vertices. We say that a vertex-cut $(L,S,R)$ \emph{separates} $s$ from $t$, or that $(L,S,R)$ is an \emph{$s$-$t$ vertex-cut}, if $s\in L$ and $t\in R$ holds. We say that it is a \emph{minimum vertex-cut separating $s$ from $t$}, if the value of the cut is minimized among all such cuts. Similarly, given a vertex $s\in V(G)$ and a subset $T\subseteq V(G)\setminus\set{s}$ of vertices, we say that a vertex-cut $(L,S,R)$ separates $s$ from $T$ if $s\in L$ and $T\subseteq R$ holds. As before, we say that $(L,S,R)$ is the minimum $s$-$T$ vertex-cut, or a minimum cut separating $s$ from $T$ if it is a smallest-value cut separating $s$ from $T$. We define cuts separating two subsets $X,Y$ of vertices similarly. Notice that, if there is an edge connecting a vertex of $X$ to a vertex of $Y$ in $G$, then $G$ does not contain a vertex-cut separating $X$ from $Y$; in such a case, we say that the value of the minimum vertex-cut separating $X$ from $Y$ is infinite, and that the cut is undefined.

\paragraph{Edge-cuts in directed graphs.} 
Let $G=(V,E)$ be a directed graph with capacities $c(e)\geq 0$ on its edges $e\in E$. Given a subset $E'\subseteq E$ of edges, we denote by $c(E')=\sum_{e\in E'}c(e)$ the \emph{total capacity of the edges in $E'$}. An \emph{edge-cut} in $G$ is a partition $(X,Y)$ of the vertices of $G$ into two disjoint non-empty subsets, and the \emph{value} of the cut is $c(E_G(X,Y))$. We say that a cut $(X,Y)$ is a \emph{global minimum edge-cut for $G$}, if it is an edge-cut of $G$ of minimum value. Given a pair $s,t$ of vertices of $G$, we say that a cut $(X,Y)$ is an \emph{$s$-$t$ edge-cut}, or that it is an \emph{edge-cut that separates $s$ from $t$}, if $s\in X$ and $t\in Y$ holds. We define edge-cuts separating a vertex $s$ from a subset $T\subseteq V\setminus \set{s}$ of vertices, and edge-cuts separating disjoint subsets $S,T\subseteq V(G)$ of vertices similarly.

\paragraph{Flows in edge-capacitated directed graphs.}
Let $G=(V,E)$ be a directed graph with capacities $c(e)\geq 0$ on its edges $e\in E$, and let $s,t\in V$ be a pair of vertices of $G$. An \emph{$s$-$t$ flow} in $G$ is an assignment of flow values $0\leq f(e)\leq c(e)$ to every edge $e\in E$, such that, for every vertex $v\in V\setminus\set{s,t}$, the following \emph{flow conservation constraint} holds: $\sum_{e\in \delta^+(v)}f(e)=\sum_{e\in \delta^-(v)}f(e)$. We also require that, if $e\in \delta^-(s)\cup \delta^+(t)$, then $f(e)=0$. 
We say that the flow $f$ is \emph{acyclic}, if there is no cycle $C$ in the graph, such that every edge $e\in E(C)$ has $f(e)>0$. For a value $M>0$, we say that the flow $f$ is \emph{$M$-integral}, if, for every edge $e\in E$, $f(e)$ is an integral multiple of $M$ (note that it is possible that $M<1$ under this definition).
The \emph{value} of the flow is $\val(f)=\sum_{e\in \delta^+(s)}f(e)$.
A \emph{flow-path decomposition} of an $s$-$t$ flow $f$ is a collection $\pset$ of $s$-$t$ paths in $G$, together with a flow value $f(P)>0$ for every path $P\in \pset$, such that $\sum_{P\in \pset}f(P)=\val(f)$, and, for every edge $e\in E(G)$, $\sum_{\stackrel{P\in \pset:}{e\in E(P)}}f(P)\leq f(e)$. 

A for a vertex $s\in V$ and a subset $T\subseteq V\setminus\set{s}$ of vertices, we can define the $s$-$T$ flow $f$ similarly: the only difference is that now the flow conservation constraints must hold for all vertices $v\in V\setminus (T\cup \set{s})$, and the flow-paths decomosition of the flow only contains paths connecting $s$ to vertices of $T$.

\subsection{A Modified Adjacency-List Representation of Graphs}
\label{subsec: modified adj list}

Suppose we are given a directed $n$-vertex and $m$-edge graph $G=(V,E)$ with integral capacities $0<c(e)\leq \cmax$ on its edges $e\in E$. In order to make our algorithms efficient, we will sometimes use a slight modification of the standard adjacency-list representation of such graphs, that we refer to as the \emph{modified adjacency-list representation} of $G$.

In the modified adjacencly-list representation of $G$, for every vertex $v\in V(G)$, for all $0\leq i\leq \ceil{\log(\cmax)}+1$, there are two linked lists $\OUT_i(v)$ and $\IN_i(v)$, where $\OUT_i(v)$ contains all edges $e\in \delta^+(v)$ with $2^i\leq c(e)<2^{i+1}$, and $\IN_i(v)$ similarly contains  all edges $e\in \delta^-(v)$ with $2^i\leq c(e)<2^{i+1}$.
 Notice that, given a directed graph $G$ in the adjacency-list representation, we can compute the modified adjacency-list representation of $G$ in time $O(m\log \cmax)$.
Consider now some index $0\leq i\leq \ceil{\log(\cmax)}+1$, and let $G^{\geq i}$ be the graph obtained from $G$ by deleting edges $e$ with $c(e)<2^i$. Notice that, if we are given the modified adjacency-list representation of $G$, then we can use it to simulate the modified adjacency-list representation of $G^{\geq i}$, by simply ignoring the lists $\OUT_{i'}(v)$ and $\IN_{i'}(v)$ for all $v\in V$ and $i'<i$.

\subsection{A Split Graph and Parameters $\wmax'$, $\wmax$}
\label{subsec: split graph}

Given an undirected graph $G=(V,E)$ with integral weights $w(v)\geq 0$ on its vertices $v\in V(G)$, we denote by $W'_{\max}(G)$ the smallest integral power of $2$, such that $\wmax'(G)> \max_{v\in V}\set{w(v)}$ holds. We let $W_{\max}(G)$ be the smallest integral power of $2$ with $\wmax\geq2n\cdot W'_{\max}$. When the graph $G$ is clear from context, 
we denote $\wmax'(G)$ and $\wmax(G)$ by $\wmax'$ and $\wmax$, respectively. Notice that:

\begin{equation}\label{eq: bound on wmax}
\wmax(G)\leq 4n\cdot \wmax'(G)\leq 8n\cdot \max_{v\in V}\set{w(v)}.
\end{equation}

A \emph{split graph} $G'$ corresponding to $G$ is 
is defined as follows. The set of vertices of $G'$ contains two copies of every vertex $v\in V$, an \emph{in-copy} $v^{\inn}$, and an \emph{out-copy} $v^{\out}$, so $V(G')=\set{v^{\inn},v^{\out}\mid v\in V}$. The set of edges of $G'$ is partitioned into two subsets: the set $E^{\spec}$ of \emph{speical} edges, and the set $E^{\reg}$ of \emph{regular edges}. For every vertex $v\in V(G)$, the set $E^{\spec}$ of special edges contains the edge $e_v=(v^{\inn},v^{\out})$, whose capacity $c(e_v)$ is set to be $w(v)$; we refer to $e_v$ as the \emph{special edge representing $v$}. Additionally, for every edge $e=(x,y)\in E(G)$, we add two regular edges representing $e$: the edge $(x^{\out},y^{\inn})$, and the edge $(y^{\out},x^{\inn})$. The capacities of both these edges are set to  $\wmax(G)$.
See Figure \ref{fig: transformation} for the transformation of the neighborhood of a vertex $v\in V(G)$. Given any subset $Z\subseteq V(G)$ of vertices of $G$, we denote by $Z^{\inn}=\set{v^{\inn}\mid v\in Z}$, $Z^{\out}=\set{v^{\out}\mid v\in Z}$, and $Z^*=Z^{\inn}\cup Z^{\out}$.
We also denote by  $V^{\inn}=\set{v^{\inn}\mid v\in V(G)}$ and $V^{\out}=\set{v^{\out}\mid v\in V(G)}$.

\begin{figure}[h]
	\centering
	\subfigure[Before]{\scalebox{0.6}{\includegraphics{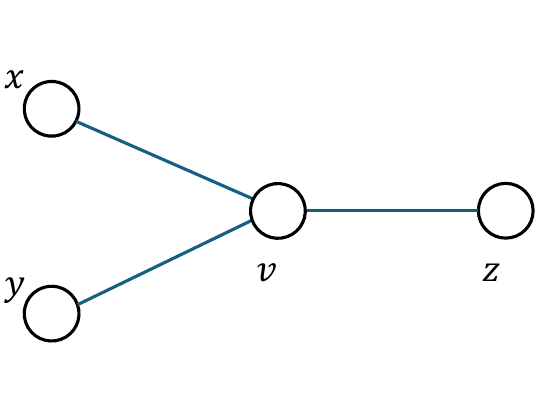}}\label{fig: before}}
	\hspace{1cm}
	\subfigure[After]{
		\scalebox{0.6}{\includegraphics{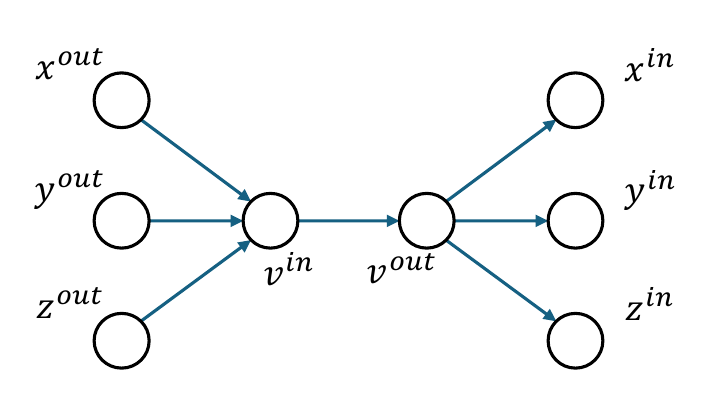}}\label{fig: after}}
	\caption{The local transformation for vertex $v$ for computing the split graph. \label{fig: transformation}}
\end{figure}

Observe that, for every pair $s,t\in V(G)$ of vertices, the value of minimum $s$-$t$ vertex-cut in $G$ is equal to the value of the minimum $s^{\out}$--$t^{\inn}$ edge-cut in $G'$, which, in turn, is equal to the value of the maximum  $s^{\out}$--$t^{\inn}$ flow in $G'$ from the Max-Flow/Min-Cut theorem. In particular, the value of the global minimum vertex-cut in $G$ is equal to maximum value of the $s^{\out}$--$t^{\inn}$ flow in $G'$, over all vertex pairs  $s,t\in G$.
We will use the following simple observation.

\begin{observation}\label{obs: from split to regular}
	Let $G$ be an undirected graph with integral weights $0\leq w(v)\leq W$ on its vertices $v\in V(G)$, and let $s,t\in V(G)$ be a pair of distinct vertices that are not connected by an edge in $G$. 
	Consider the corresponding split graph $G'$, and let $(S,T)$ is a minimum $s^{\out}$-$t^{\inn}$ edge-cut in $G'$. Then there is a deterministic algorithm that, given a cut $(S,T)$ in $G'$ with the above properties, computes, in time $O(|E(G)|)$ a minimum $s$-$t$ vertex-cut $(X,Y,Z)$ in $G$.
\end{observation}
\begin{proof}
	 From the definition of the parameter $\wmax$, we can assume that every edge in $E_G(S,T)$ is a special edge.
	 If $s^{\inn}\in T$, then we move $s^{\inn}$ from $T$ to $S$. This move does not increase the capacity of the cut, since the only edge leaving $s^{\inn}$ in $G'$ is the special edge $(s^{\inn},s^{\out})$. Next, 
	  for every vertex $v\in V(G)$ with $v^{\inn}\in T$ and $v^{\out}\in S$,  we move $v^{\out}$ from $S$ to $T$. Note that this move does not increase the capacity of the cut, since the only edge entering $v^{\out}$ in $G'$ is the special edge $(v^{\inn},v^{\out})$. Once all such vertices are processed, we are guaranteed that, for every vertex $v\in V(G)$, if $v^{\inn}\in T$, then $v^{\out}\in T$ holds. We let $X$ contain all vertices $v\in V(G)$ with $v^{\inn},v^{\out}\in S$ and $Z$ contain all vertices $v\in V(G)$ with $v^{\inn},v^{\out}\in T$. All remaining vertices are added to $Y$. Notice that, for every vertex $v\in Y$, $v^{\inn}\in S$ and $v^{\out}\in T$ holds. Therefore, $w(Y)=\sum_{e\in E_{G'}(S,T)}c(e)$. As observed already, the value of the minimum $s^{\out}$-$t^{\inn}$ edge-cut in $G'$ is equal to the value of the minimum $s$-$t$ vertex-cut in $G$, so the observation follows.
\end{proof}

Lastly, observe that $|V(G')|=2|V(G)|$ and $|E(G')|\leq O(|E(G)|)$; moreover, there is an algorithm that, given $G$, computes $G'$ in time $O(|E(G')|)$.

\subsection{Properties of the Global Minimum Vertex-Cut}

In this subsection we establish some useful properties of a global minimum vertex-cut in undirected graphs.
Consider a vertex-weighted graph $G$, let $(L,S,R)$ be a global minimum vertex-cut in $G$, and let $x\in L$ be any vertex. The following simple claim shows that, in a sense, almost all vertices of $S$ must be neighbors of $x$ in $G$. In other words, the total weight of all vertices of $S$ that are not neighbors of $x$ is bounded by $w(L)$. 

The proof of the claim needs the definition of flows in undirected vertex-capacitated graphs. For convenience, we define such flows via their flow-path decompositions. Specifically, suppose we are given an undirected graph $G$ with weights $w(v)\geq 0$ on its vertices $v\in V$, and two special vertices $s$ and $t$. An \emph{$s$-$t$ flow} in $G$ consists of a collection $\pset$ of paths, each of which connects $s$ to $t$, and, for every path $P\in \pset$, a flow value $f(P)$. We require that, for every vertex $v\in V(G)\setminus\set{s,t}$, $\sum_{\stackrel{P\in \pset:}{v\in P}}f(P)\leq w(P)$ holds. We refer to the paths in $\pset$ as the \emph{flow-paths} of $f$. From the Max-Flow/Min-cut Theorem, the value of the maximum $s$-$t$ flow in a vertex-weighted graph $G$ is equal to the value of the minimum $s$-$t$ vertex-cut.

The following claim can be thought of as a generalization of Proposition $3.2$ in~\cite{LNP21} from unit vertex weights to general weights.

\begin{claim}\label{claim: neighbors in S}
	Let $G$ be a simple undirected graph with integral weights $w(v)> 0$ on its vertices $v\in V(G)$, and let  $(L,S,R)$ be a global minimum vertex-cut in $G$. Let $x\in L$ be any vertex, and denote by $S'=S\setminus N_G(x)$. Then $w(S')\leq w(L)$ must hold. 
\end{claim}
\begin{proof}	
	Let $t\in R$ be any vertex. From the Max-Flow/Min-Cut theorem, there is a flow $f$ in $G$ from $x$ to $t$, that respects all vertex capacities (that are equal to their weights), whose value is $\opt=w(S)$. Let $\pset$ be the corresponding collection of flow-paths. We partition $\pset$ into two sets: set $\pset'\subseteq \pset$ contains all paths $P$, whose second vertex lies in $S$, and $\pset'' = \pset\setminus\pset'$ contains all remaining flow-paths. Notice that, for every path $P\in \pset''$, the second vertex on $P$ must lie in $L$, and so $\sum_{P\in \pset''}f(P)\leq w(L)$ must hold. Therefore, $\sum_{P\in \pset'}f(P)\geq \opt-w(L)=w(S)-w(L)$. For each flow-path $P\in \pset'$, the second vertex on $P$ lies in $S\cap N_G(x)$, so $w(S\cap N_G(x))\geq w(S)-w(L)$, and $w(S')=w(S)-w(S\cap N_G(x))\leq w(L)$.
\end{proof}

We obtain the following immediate corollary of the claim.

\begin{corollary}\label{cor:  heavy neighbors in S}
	Let $G$ be a simple undirected graph with integral weights $0<w(v)\leq W$ on its vertices $v\in V(G)$, and let  $(L,S,R)$ be a global minimum vertex-cut in $G$. Let $x\in L$ be any vertex, let $\gamma>0$ be a parameter, and denote by $S'_{\gamma}$ the set of all vertices $v\in S\setminus N_G(x)$ with $w(v)\geq\frac{W}{\gamma}$. Then $|S'_{\gamma}|\leq |L|\cdot \gamma$ must hold. 
\end{corollary}
\begin{proof}	
	Let $S'=S\setminus N_G(x)$.
	Observe first that $w(L)\leq |L|\cdot W$ must hold, and, from \Cref{claim: neighbors in S}, $w(S')\leq w(L)\leq |L|\cdot W$ holds. Since $S'_{\gamma}\subseteq S'$, and since, for every vertex $v\in S'_{\gamma}$, $w(v)\geq \frac{W}{\gamma}$ holds, we get that:
	
	\[|S'_{\gamma}|\leq \frac{w(S')}{W/\gamma}\leq |L|\cdot \gamma.\]
\end{proof}

\subsection{Fast Algorithms for Flows and Cuts}

Our algorithm relies on the recent breakthrough almost linear time algorithm for Minimum Cost Maximum $s$-$t$ Flow \cite{CKLP22,BCP23}. The algorithm of \cite{BCP23} first computes a fractional solution to the problem, whose cost is close to the optimal one, and then rounds the resulting flow via the Link-Cut Tree data structure~\cite{ST83,KP15} in order to compute the optimal flow. In particular, the flow that they obtain is integral, provided that all edge capacities are integral as well (see Lemma $4.1$ in~\cite{BCP23}). This result is summarized in the following theorem (see Theorem $1.1$ and Lemma 4.1 in \cite{BCP23}).

\begin{theorem}[\cite{BCP23} (see also~\cite{CKLP22})]
	\label{thm: maxflow}
	There is a deterministic algorithm, that receives as input a directed $m$-edge graph $G$ with integral capacities $1\leq u(e)\leq U$ and  integral costs $1\leq c(e)\leq C$ on edges $e\in E(G)$, together with two vertices $s,t\in V(G)$. The algorithm computes a minimum-cost maximum $s$-$t$ flow $f$ in $G$, such that $f$ is integral. The running time of the algorithm is $O\left (m^{1+o(1)}\cdot \log U\cdot  \log C\right )$.
\end{theorem}

It is well known that,  given a maximum $s$-$t$ flow in a flow network $G$, one can compute a minimum $s$-$t$ cut in $G$ in time $O(m)$. In order to do so, we compute a residual flow network $H$ of $G$ with respect to $f$, and then compute the set $S$ of all vertices of $H$ that are reachable from $s$ in $H$. Let $T=V(G)\setminus S$.
Then, by the Max-Flow / Min-Cut theorem~\cite{FF56}, $(S,T)$ is a minimum $s$-$t$ cut in $G$. By combining this standard algorithm with the algorithm from \Cref{thm: maxflow}, we obtain the following immediate corollary.

\begin{corollary}\label{cor: mincut}
	There is a deterministic algorithm, that receives as input a directed $m$-edge graph $G$ with integral capacities $1\leq u(e)\leq U$ on edges $e\in E(G)$, together with two distinct vertices $s,t\in V(G)$. The algorithm computes a minimum $s$-$t$ edge-cut $(S,T)$ in $G$. The running time of the algorithm is $O\left (m^{1+o(1)}\cdot \log U\right )$.
\end{corollary}

Assume now that we are given an undirected graph $G=(V,E)$ with integral weights $w(v)$ on its vertices $v\in V$, and two disjoint vertex subsets $S,T\subseteq V$. We can use the algorithm from
\Cref{cor: mincut} in order to compute a minimum $S$-$T$ vertex-cut in $G$, as follows. 
Observe first that, if there is an edge connecting a vertex of $S$ to a vertex of $T$ in $G$, or if $S\cap T\neq \emptyset$, then the minimum $S$-$T$ vertex-cut is undefined; we assume that this is not the case.
Let $\hat G$ be the graph obtained from $G$ by setting the weights of the vertices in $S\cup T$ to $\infty$, and then adding vertices $s$ and $t$ of unit weight, where $s$ connects with an edge to every vertex in $S$, and $t$ connects with an edge to every vertex in $T$. It is immediate to verify that, if $(X,Y,Z)$ is a minimum $s$-$t$ vertex-cut in $\hat G$, then $(X\setminus\set{s},Y,Z\setminus\set{t})$ is a minimum $S$-$T$ vertex-cut in $G$ and vice versa.
We then compute a split graph $\hat G'$ corresponding to $\hat G$, and
use the algorithm from \Cref{cor: mincut} to compute the minimum $s^{\out}$-$t^{\inn}$ edge-cut $(S,T)$ in $\hat G'$. Lastly, we use the algorithm from \Cref{obs: from split to regular} to compute a minimum $s$-$t$ vertex-cut $(X,Y,Z)$ in $\hat G$, and then output 
$(X\setminus\set{s},Y,Z\setminus\set{t})$ as a minimum $S$-$T$ vertex-cut in $G$. We summarize this algorithm in the following corollary.

\begin{corollary}\label{cor: min_vertex_cut}
	There is a deterministic algorithm, that receives as input an undirected $m$-edge graph $G$ with integral weights $1\leq w(v)\leq W$ on vertices $v\in V(G)$, together with two disjoint subsets $S,T\subseteq V(G)$ of vertices. The algorithm computes a minimum $S$-$T$ vertex-cut $(X,Y,Z)$ in $G$. The running time of the algorithm is $O\left (m^{1+o(1)}\cdot \log W\right )$.
\end{corollary}

\subsection{Computing Minimum Cuts via Isolating Cuts}

The Isolating Cuts lemma was introduced independently in~\cite{LP20,AKT21_stoc}, and found many applications since (see, e.g.~\cite{LP21,CQ21,MN21,LNP21,AKT21_focs,LPS21, Zha22,AKT22_soda,CLP22,LNPS23,FPZ23,HHS24,HLRW24,ACD24}).
We use the following theorem, that generalizes the version of the Isolating Cuts lemma from~\cite{LNP21} for unweighted graphs to graphs with arbitrary vertex weights. The proof closely follows the arguments from \cite{LNP21}, and is provided in Section~\ref{sec: appendix_isolating} of Appendix for completeness.

%

\begin{restatable}[Isolating Cuts]{theorem}{isolating}
\label{thm: min cuts via isolating}
There is a deterministic algorithm, that, given as input an undirected $m$-edge graph $G$ with integral weights $0\leq w(v)\leq W$ on its vertices $v\in V(G)$, together with  a subset $T\subseteq V(G)$ of vertices of $G$, computes, for every vertex $v\in T$, the value of the minimum vertex-cut separating $v$ from $T\setminus\set{v}$. The running time of the algorithm is $O\left(m^{1+o(1)}\cdot\log W\right )$.
\end{restatable}



\subsection{The Chernoff Bound and a Useful Inequality}

We use the following standard version of the Chernoff Bound (see. e.g., \cite{dubhashi2009concentration}).

\begin{lemma}[Chernoff Bound]
	\label{lem: Chernoff}
	Let $X_1,\ldots,X_n$ be independent randon variables taking values in $\set{0,1}$. Let $X=\sum_{1\le i\le n}X_i$, and let $\mu=\expect{X}$. Then for any $t>2e\mu$,
	\[\Pr\Big[X>t\Big]\le 2^{-t}.\]
	Additionally, for any $0\le \delta \le 1$,
	\[\Pr\Big[X<(1-\delta)\cdot\mu\Big]\le e^{-\frac{\delta^2\cdot\mu}{2}}.\]
\end{lemma}

We also use an inequality summarized in the following simple claim.

\begin{claim}\label{claim: simple math}
	For all $r\geq 2$, $\left(1-\frac 1 r \right )^{r-1}\geq \frac{1}{e^2}$ holds. 
\end{claim}
\begin{proof}
	From Taylor's expansion of function $\ln(1+x)$, it is easy to see that, for all $0\leq \eps\leq 1/2$, $\ln(1-\eps)>-\eps-\eps^2$ holds. Therefore: 
	
	\[\ln\left(1-\frac 1 r\right )\geq -\frac{1}{r}-\frac{1}{r^2}\geq -\frac{2}{r}.\] 
	
	We then get that:
	
	\[\ln\left (\left(1-\frac 1 r \right )^{r-1}\right )\geq -\frac{2(r-1)}{r}\geq -2,\]

	and:
	
	\[\left(1-\frac 1 r \right )^{r-1}\ge \frac{1}{e^2}.\]
\end{proof}

\subsection{The Degree Estimation Problem}

We will sometimes use a natural randomized algorithm in order to efficiently estimate the degrees of some vertices in a subgraph of a given graph $G$. Specifically, suppose we are given a graph $G$, two subsets $Z,Z'$ of its vertices, and a threshold $\tau$. Our goal is to compute a subset $A\subseteq Z'$ of vertices, such that, on the one hand, every vertex $v\in A$ has at least $\tau$ neighbors that lie in $Z$, while, on the other hand, every vertex $v\in Z'$ with $|N_G(v)\cap Z|\gg \tau$ is in $A$. We now define the Undirected Degree Estimation problem more formally and provide a simple and natural randomized algorithm for it.

\begin{definition}[Undirected Degree Estimation problem.]\label{def: deg est}
	In the \emph{Undirected Degree Estimation problem}, the input is a simple undirected $n$-vertex graph $G=(V,E)$ with integral weights $0\leq w(v)\leq W$ on its vertices $v\in V$, given in the  adjacency-list representation, two (not necessarily disjoint) subsets $Z,Z'\subseteq V(G)$ of vertices of $G$, a threshold $1\leq \tau \leq n$. The output to the problem is a subset $A\subseteq Z'$ of vertices.
	We say that an algorithm for the problem \emph{errs}, if either (i) for some vertex $v\in A$, $|N_G(v)\cap Z|< \tau$; or (ii) for some vertex $v'\in Z'\setminus A$, $|N_G(v)\cap Z|\geq 1000\tau\log n\cdot\log(\wmax(G))$.
\end{definition}

The following claim provides a simple algorithm for the Degree Estimation problem, using standard techniques. The proof is provided in Appendix~\ref{sec: appendix_degree}.
%

\begin{restatable}[Algorithm for the Undirected Degree Estimation Problem]{claim}{degree}
\label{claim: degree est}
There is a randomized algorithm for the Undirected Degree Estimation problem, whose running time is $O\left(\frac{n\cdot |Z|\cdot \log n\cdot\log(\wmax)}{\tau}\right )$. The probability that the algorithm errs is at most $\frac{1}{n^{10}\cdot \log^6(\wmax)}$, where $\wmax=\wmax(G)$.
\end{restatable}

We will also use the following directed version of the degree estimation problem.

\begin{definition}[Directed Heavy Degree Estimation Problem]\label{def: heavy weight est}	
	In the \emph{Directed Heavy Degree Estimation problem}, the input is a simple directed $n$-vertex graph $G=(V,E)$ with capacities $c(e)>0$ on edges $e\in E$, given in the adjacency-list representation, two (not necessarily disjoint)  subsets $Z,Z'\subseteq V(G)$ of vertices of $G$, a threshold $\tau \geq 1$, and two other parameters $c^*>0$ and $W$ that is greater than a large enough constant. For a vertex $v\in Z$, denote by $E_v$ the set of all edges $(v,u)\in E$ whose capacity is at least $c^*$, such that $u\in Z'$. The output of the problem is a subset $A\subseteq Z$ of vertices.
	We say that an algorithm for the problem \emph{errs}, if one of the following hold:
	\begin{itemize}
		\item either there is some vertex $u\in A$ with $|E_u|<\tau$; or
		\item there is a vertex $u'\in Z\setminus A$ with $|E_u|\geq 1000\tau\log n\cdot\log W$.
	\end{itemize} 
\end{definition}

The following claim provides an algorithm for the Directed Heavy Degree Estimation Problem, that is a simple adaptation of the algorithm from Claim~\ref{claim: degree est}. We provide a proof sketch of this claim in Appendix~\ref{sec: appendix_degree}.

\begin{restatable}[Algorithm for the Directed Heavy Degree Estimation Problem]{claim}{heavyweight}
\label{claim: heavy weight est}
There is a randomized algorithm for the Directed Heavy Degree Estimation problem, whose running time is $O\left(\frac{n\cdot |Z'|\cdot \log n\cdot\log W}{\tau}\right )$. The probability that the algorithm errs is at most $\frac{1}{n^{10}\cdot \log^6 W }$.
\end{restatable}

\subsection{The Heavy Vertex Problem}

We will also sometimes use an algorithm for the Heavy Vertex problem, that is very similar to the Directed Degree Estimation problem. The main difference is that  goal of the algorithm is to either (i) compute a vertex $v\in Z'$, such that there are at least $\tau$ vertices $u_1,\ldots,u_{\tau}\in Z$, where for all $1\leq j\le \tau$, an edge $(u_j,v)$ of capacity at least $c^*$ lies in the graph; or (ii) correctly establish that no such vertex exists. Like in the Directed Degree Estimation problem, in order to ensure that the algorithm is efficient, we will solve the problem approximately. We now define the Heavy Vertex problem formally.

\begin{definition}[Heavy Vertex problem.]\label{def: heavy vertex problem}
	In the \emph{Heavy Vertex problem}, the input is a simple directed $n$-vertex graph $G=(V,E)$  with capacities $c(e)>0$ on edges $e\in E$, given in the adjacency-list representation, two disjoint subsets $Z,Z'\subseteq V(G)$ of vertices of $G$, an integer $\tau\geq 1$, and two other parameters $c^*>0$, and $W>0$ that is greater than a sufficiently large constant.
	The output of the problem is a bit $b\in\set{0,1}$, and, if $b=1$, a vertex $v\in Z'$, together with $\tau$ distinct vertices $u_1,\ldots,u_{\tau}\in Z$, such that, for all $1\leq j\leq \tau$, there is an edge $(u_j,v)$ of capacity at least $c^*$ in $G$. We say that an algorithm for the problem \emph{errs}, if it returns $b=0$, and yet 
	there exists a vertex $v\in Z'$, and a collection $U\subseteq Z$ of vertices with $|U|\geq 1000\tau\log n\log W$, such that, for every vertex $u\in U$, there is an edge $(u,v)$ of capacity at least $c^*$ in $G$.
\end{definition}

For convenience, we denote the input to the Heavy Vertex problem by $(G,Z,Z',\tau,c^*,W)$.
The following simple claim provides an algorithm for the Heavy Vertex problem. The proof of the claim is deferred to Section~\ref{sec: appendix_heavy} of Appendix.

\begin{restatable}[Algorithm for the Heavy Vertex problem]{claim}{heavy}
	\label{claim: heavy}
	There is a randomized algorithm for the Heavy Vertex problem, whose running time, on input $(G,Z,Z',\tau,c^*,W)$, is $O\left(\frac{n\cdot |Z|\cdot \log n\cdot\log W}{\tau}\right )$. The probability that the algorithm errs is at most $1/\left(n^{10}\cdot \log^6 W \right)$.
\end{restatable}

\section{The Algorithm}\label{sec: the_algorithm}

In this section we provide our main result, namely an algorithm for the vertex-weighted global minimum vertex cut, proving \Cref{thm: main}. We defer some of the technical details to subsequent sections.

We assume that we are given a simple undirected $n$-vertex and $m$-edge graph $G$ with integral weights $0\leq w(v)\leq W$ on its vertices.
It will be convenient for us to assume that all vertex weights are strictly positive. In order to achieve this, we can modify the vertex weights as follows: for every vertex $v\in V(G)$, we let  $w'(v)=n^2\cdot w(v)+1$. Let $G'$ be the graph that is identical to $G$, except that the weight of every vertex $v$ is set to $w'(v)$, and let $W'=n^2\cdot W+1$, so for every vertex $v\in V(G)$, $1\leq w'(v)\leq W'$ holds. Note that, if $(L,S,R)$ is a global minimum vertex cut in $G'$, then it must also be a global minimum vertex-cut in $G$. Therefore, it is now enough to compute a global minimum vertex-cut in $G'$.
In order to avoid clutter, we will denote $G'$ by $G$, and the vertex weights $w'(v)$ by $w(v)$ for $v\in V(G)$. From now on we assume that, for every vertex $v\in V(G)$, $1\leq w(v)\leq W'$, where $W'=n^2\cdot W+1$. We define the values $\wmax'(G)$ and $\wmax(G)$ as in \Cref{subsec: split graph}, and denote them by $\wmax'$ and $\wmax$, respectively, throughout this section. Note that:

\begin{equation} \label{eq: wmax}
\wmax\leq O(nW')\leq O(n^3\cdot W).
\end{equation}

Let $\rho=\ceil{\log(W')}$. For all  $0\leq i\leq \rho$, we let $U_i\subseteq V$ denote the set of all vertices 
$v$ with $2^i\leq w(v)< 2^{i+1}$, so $(U_0,\ldots,U_{\rho})$ is a partition of $V(G)$.

Whenever we are given any vertex cut $(L,S,R)$ of $G$, we assume without loss of generality that $w(L)\leq w(R)$ holds. For all $0\leq i\leq \rho$, we then denote $L_i=L\cap U_i$, $S_i=S\cap U_i$, and $R_i=R\cap U_i$, so that $(L_i,S_i,R_i)$ is a partition of $U_i$.

We use a combination of several algorithms, that are designed to handle different edge densities and different tradeoffs between the values $|L_i|$ and $|S_i|$ for the optimal cut $(L,S,R)$, for $0\leq i\leq \rho$. Each such algorithm must return a pair $s$, $t$ of vertices of $G$, together with a value $c$ that is at least as large as the value of the minimu $s$-$t$ vertex-cut in $G$. We will ensure that, with high probability, at least one of the algorithms returns a value $c=\opt$. 
Taking the smallest value $c$ returned by any of these algorithms then yields the optimal solution value with high probability, and computing the minimum $s$-$t$ vertex cut for the corresponding pair $s,t$ of vertices provides an optimal global minimum cut in $G$.
We now describe each of the algorithms in turn.

\subsection{Algorithm $\alg_1$: when $S_i$ is Small for Some $i$}

Our first algorithm, $\alg_1$, is summarized in the following theorem.

\begin{theorem}\label{thm: alg 1}
	There is a randomized algorithm, that we denote by $\alg_1$, whose input is a simple undirected $n$-vertex and $m$-edge graph $G$ with integral weights $1\leq w(v)\leq W'$ on vertices $v\in V(G)$, and a parameter $0<\eps<1$. The algorithm returns two distinct vertices $x,y\in V(G)$ with $(x,y)\not\in E(G)$, and a value $c$ that is at least as large as the value of the minimum $x$-$y$ vertex-cut in $G$. Moreover, if there is a global minimum vertex-cut $(L,S,R)$ in $G$, and an integer $0\leq i\leq \ceil{\log W'}$, such that $L_i\neq\emptyset$ and $|S_i|\leq |L_i|\cdot n^{1-\eps}$, then, with probability at least $1-1/n^4$, $c=\opt$. The running time of the algorithm is $ O\left (mn^{1-\eps+o(1)}\cdot \log^2 W'\right )$.
\end{theorem}

In the remainder of this subsection, we prove 
the above theorem. At a high level, Algorithm $\alg_1$ uses the paradigm of subsampling a set $T$ of vertices of $G$ and then applying the Isolating Cuts Lemma (\Cref{thm: min cuts via isolating}) to $T$; this paradigm was introduced by~\cite{LNP21} in the context of computing global minimum vertex-cut in unweighted graphs.
The following lemma is central to the proof of the theorem.

\begin{lemma}\label{lem: alg 1 i}
	There is a randomized algorithm, whose input is a simple undirected $n$-vertex and $m$-edge graph $G$ with integral weights $1\leq w(v)\leq W'$ on vertices $v\in V(G)$, two integral parameters $0\leq i\leq \ceil{\log W'}$ and  $0<\lambda\leq n$, and another parameter $0<\eps<1$. The algorithm returns two distinct vertices $x,y\in V(G)$ with $(x,y)\not\in E(G)$, and a value $c$ that is at least as large as the value of the minimum $x$-$y$ vertex-cut in $G$.
	 Moreover, if there is a global minimum vertex-cut $(L,S,R)$ in $G$ with $\lambda\leq |L_i|+|S_i|\leq 2\lambda$, such that $L_i\neq\emptyset$ and $|S_i|\leq |L_i|\cdot n^{1-\eps}$, then, with probability at least $1-1/n^4$, $c=\opt$ holds. The running time of the algorithm is $ O\left (mn^{1-\eps+o(1)}\cdot \log W'\right )$.
\end{lemma}

We prove \Cref{lem: alg 1 i} below, after we complete the proof of \Cref{thm: alg 1} using it. Let $z=\ceil{\log W'}$ and let $z'=\ceil{\log n}$. Our algorithm performs $O(z\cdot z')$ iterations, as follows. For all $0\leq i\leq z$ and for all $0\leq j\leq z'$, we apply the algorithm from \Cref{lem: alg 1 i} to graph $G$, with parameter $i$ and $\lambda=2^j$, and parameter $\eps$ remaining unchanged; we denote this application of the algorithm from \Cref{lem: alg 1 i}  by $\aset_{i,j}$, and we denote by $(c_{i,j},x_{i,j},y_{i,j})$ its output. We let $i^*,j^*$ be the indices for which the value $c_{i^*,j^*}$ that Algorithm $\aset_{i,j}$ returned is the smallest, breaking ties arbitrarily. We then output $(c_{i^*,j^*},x_{i^*,j^*},y_{i^*,j^*})$. This completes the description of Algorithm $\alg_1$.

Since the running time of the algorithm from \Cref{lem: alg 1 i} is 
$ O(mn^{1-\eps+o(1)}\cdot \log W')$, the running time of $\alg_1$
is $ O(mn^{1-\eps+o(1)})\cdot O(z\cdot z')\leq O(mn^{1-\eps+o(1)}\cdot \log^2 W')$.
Moreover,  \Cref{lem: alg 1 i} ensures that the value of the minimum $x_{i^*,j^*}$--$y_{i^*,j^*}$ vertex-cut in $G$ is at most $c_{i^*,j^*}$.

Assume now that there is a global minimum vertex cut $(L,S,R)$ in $G$, and an integer $0\leq i\leq \ceil{\log W'}$, such that $L_i\neq\emptyset$ and $|S_i|\leq |L_i|\cdot n^{1-\eps}$. Let $0\leq j\leq z'$ be the unique integer with $2^j\leq |L_i|+|S_i|<2^{j+1}$.
We say that Algorithm $\aset_{i,j}$ is \emph{successful}, if its output 
$(c_{i,j},x_{i,j},y_{i,j})$ has the property that $c_{i,j}=\opt$. From \Cref{lem: alg 1 i}, with probability at least $1-1/n^4$, Algorithm $\aset_{i,j}$ is successful. In this case we are guaranteed that $c_{i^*,j^*}=\opt$ as well. 
In order to complete the proof of \Cref{thm: alg 1}, it is now enough to prove \Cref{lem: alg 1 i}.

\begin{proofof}{\Cref{lem: alg 1 i}}
The proof follows the high-level idea of~\cite{LNP21}, that was introduced in the context of computing global minimum vertex-cuts in unweighted graphs, of subsampling a subset $T$ of vertices of $G$, and then applying the Isolating Cuts Lemma (\Cref{thm: min cuts via isolating}) to it.

We start by describing an algorithm, that we denote by $\aset'$. Our final algorithm will execute $\aset'$ several times.
	Algorithm $\aset'$ starts by constructing a random set $T$ of vertices as follows: every vertex $v\in U_i$ is added to $T$ independently with probability $\frac{1}{2\lambda}$. We then use the algorithm from  \Cref{thm: min cuts via isolating} to compute, for every vertex $v\in T$, the value $c_v$ of the minimum vertex-cut separating $v$ from $T\setminus\set{v}$ in $G$. 
	We let $x\in T$ be the vertex for which the value $c_x$ is the smallest, breaking ties arbitrarily, and we let $y$ be an arbitrary vertex in $T\setminus \set{x}$. We then return $c_x,x$ and $y$ as the outcome of the algorithm $\aset'$. Note that the running time of Algorithm  $\aset'$ is $O(m^{1+o(1)})$.  
	It is easy to verify that 
	 there is an $x$-$y$ vertex-cut of value $c_x$ in $G$ (the cut separating $x$ from $T\setminus\set{x}$), and the value of the minimum $x$-$y$ vertex-cut in $G$ is at most $c_x$.
	We say that the algorithm $\aset$ is \emph{successful} if $c_x=\opt$; otherwise we say that it is \emph{unsuccessful}.
	We will use the following claim to analyze algorithm $\aset'$.
	
	\begin{claim}\label{claim: alg A'}
		Suppose that there is a global minimum vertex-cut $(L,S,R)$ in $G$ with $\lambda\leq |L_i|+|S_i|\leq 2\lambda$, such that $L_i\neq \emptyset$ and $|S_i|\leq |L_i|\cdot n^{1-\eps}$ holds. Then the probability that Algorithm $\aset'$ is successful is at least $\Omega\left(\frac{1}{n^{1-\eps}}\right )$.
	\end{claim}
\begin{proof}
We say that the good Event $\event$ happens if $|T\cap L_i|=1$ and $|T\cap S_i|=\emptyset$. Assume first that Event $\event$ happened, and let $v$ be the unique vertex in $T\cap L_i$, so $T\setminus\set{v}\subseteq R$. Since $(L,S,R)$ is a minimum vertex-cut separating $v$ from $T\setminus\set{v}$, we get that $c_v=w(S)=\opt$ must hold. Moreover, for every vertex $u\in T\setminus\set{v}$, $c_u$ is at least as large as the value of some vertex-cut in $G$, so $c_u\geq \opt$ holds as well. Therefore, if we denote the outcome of the algorithm by $(c_x,x,y)$, then $c_x=\opt$ must hold. Since $c_x$ is the value of the minimum vertex-cut separating $x$ from $T\setminus\set{x}$, we get that there is an $x$-$y$ cut of value $\opt$. 
 We conclude that, if Event $\event$ happened, then the algorithm is successful. The following observation will then complete the proof of \Cref{claim: alg A'}.


\begin{observation}\label{obs exactly one good}
	$\prob{\event}\geq \Omega\left(\frac{1}{n^{1-\eps}}\right )$.
\end{observation}
\begin{proof}
	For every vertex $v\in L_i$, we define the event $\event(v)$ that $v\in T$, and $T$ contains no other vertices of $S_i\cup L_i$. Denote $|S_i\cup L_i|$ by $z$, and recall that $\lambda\leq z\leq 2\lambda$.  Then for all $v\in L_i$:
	
	\[
	\begin{split}
	\prob{\event(v)}&=\frac{1}{2\lambda}\cdot \left(1-\frac 1 {2\lambda} \right )^{z-1}\\
	&\geq \frac{1}{2\lambda}\cdot \left(1-\frac 1 {2\lambda} \right )^{2\lambda-1}\\
	&\geq \frac{1}{a\lambda},	\end{split}
	\]
	for some constant $a\geq 1$.
	For the first inequality, we have used the fact that
$\lambda\leq z\leq 2\lambda$ holds, and the second inequality follows from \Cref{claim: simple math}. 

	Recall that $|L_i|\geq \frac{|S_i|}{n^{1-\eps}}$, and so $|L_i|\geq \frac{|S_i\cup L_i|}{2n^{1-\eps}}\geq \frac{\lambda}{2n^{1-\eps}}$ must hold.
	Since Events $\event(v)$ for $v\in L_i$ are all disjoint, and since event $\event$ only happens if one of the events $\event_v$ for $v\in L_i$ happens, we get that: 
	
\[\prob{\event}=\sum_{v\in L_i}\prob{\event_i}\geq \frac{|L_i|}{a\lambda}\geq \Omega\left(\frac{1}{n^{1-\eps}}\right ). \]
\end{proof}
\end{proof}

We perform $\Theta(n^{1-\eps}\cdot \log n)$ executions of Algorithm $\aset'$, and return the triple $(c_x,x,y)$ with smallest value $c_x$ that any execution of the algorithm produced (if $x=y$ or $(x,y)\in E(G)$ holds, then we instead return an arbitrary pair of distinct vertices of $G$ that are not connected by an edge, and the value $c=\wmax(G)$). It is easy to verify that, if there is a global minimum vertex-cut $(L,S,R)$ in $G$ with $\lambda\leq |L_i|+|S_i|\leq 2\lambda$, such that $L_i\neq\emptyset$ and $|S_i|\leq |L_i|\cdot n^{1-\eps}$, then the probability that at least one of the executions of $\aset'$ is successful is at least $1-1/n^4$, and in this case, if $(c,x,y)$ is the output of the algorithm, then we are guaranteed that $c=\opt$ holds. The running time of the entire algorithm is $O\left(m\cdot n^{1-\eps+o(1)}\cdot \log W'\right )$.
\end{proofof}

\subsection{Algorithm $\alg_2$: when $G$ is not Dense and $S$ is Large}

Our second algorithm, $\alg_2$, will be used in the case where the input graph $G$ is not very dense; we summarize it in the following theorem.

\begin{theorem}\label{thm: alg 2}
	There is a randomized algorithm, that we denote by $\alg_2$, whose input is a simple undirected $n$-vertex and $m$-edge graph $G$ with integral weights $1\leq w(v)\leq W'$ on vertices $v\in V(G)$, and a parameter $0<\eps<1$, such that $m\leq n^{2-2\eps}$ holds. 
	The algorithm returns two distinct vertices $x,y\in V(G)$ with $(x,y)\not\in E(G)$, and a value $c$ that is at least as large as the value of the minimum $x$-$y$ vertex-cut in $G$.
	Moreover, if there is a global minimum vertex-cut $(L,S,R)$ in $G$ with $|S|\geq n^{1-\eps}\cdot |L|$, then, with probability at least $1-1/n^4$, $c=\opt$ holds. The running time of the algorithm is $O\left (mn^{1-\eps+o(1)}\cdot \log W'\right )$.
\end{theorem}

The remainder of this subsection is dedicated to the proof of \Cref{thm: alg 2}. 

Let $T$ be the set of all vertices $v\in V(G)$ with $\deg_G(v)\geq n^{1-\eps}$. Since $\sum_{v\in V(G)}\deg_G(v)=2m\leq 2n^{2-2\eps}$, we get that $|T|\leq \frac{2n^{2-2\eps}}{n^{1-\eps}}\leq 2n^{1-\eps}$. We use the following simple observation.

\begin{observation}\label{obs: high degree in L}
	Suppose there is a global minimum vertex-cut $(L,S,R)$ in $G$ with $|S|\geq n^{1-\eps}\cdot |L|$. Then $L\cap T\neq \emptyset$ must hold.
\end{observation}
\begin{proof}
	Let  $(L,S,R)$ be a global minimum vertex-cut in $G$ with $|S|\geq n^{1-\eps}\cdot |L|$. We claim that, for every vertex $v\in S$, there must be an edge connecting $v$ to some vertex of $L$. Indeed, assume otherwise. Let $v\in S$ be any vertex that has no neighbors in $L$. Consider the partition $(L,S',R')$ of vertices of $G$ with $S'=S\setminus \set{v}$ and $R'=R\cup \set{v}$. It is easy to verify that $(L,S',R')$ is a valid vertex-cut in $G$, and moreover that $w(S')<w(S)$,
contradicting the assumption that $(L,S,R)$ is a global minimum vertex-cut in $G$.
	
We conclude that the total number of edges in $G$ connecting the vertices of $L$ to the vertices of $S$ is at least $|S|$, and so there must be some vertex $u\in L$ with $\deg_G(u)\geq\frac{|S|}{|L|}\geq n^{1-\eps}$, so $u\in T$ must hold.
\end{proof}

The following lemma provides a central tool for the proof of \Cref{thm: alg 2}.

\begin{lemma}\label{lem: cut if vertex of L}
	There is a randomized algorithm, that, given a simple $n$-vertex and $m$-edge graph $G$ with integral weights $1\leq w(v)\leq W'$ on vertices $v\in V(G)$, and a vertex $x\in V(G)$, returns a vertex $y\in V(G)$, and the value $c$ of the minimum $x$-$y$ vertex-cut in $G$. Moreover, if there is a global minimum vertex-cut $(L,S,R)$ in $G$ with $x\in L$, then with probability at least $1-1/n^4$, $c=\opt$ holds. The running time of the algorithm is $O\left (m^{1+o(1)}\cdot \log W'\right )$.
\end{lemma}

We prove \Cref{lem: cut if vertex of L} below, after we complete the proof of \Cref{thm: alg 2} using it. Our algorithm starts by computing the set $T$ of vertices of $G$, whose degrees in $G$ are at least $n^{1-\eps}$. Clearly, $T$ can be computed in time $O(m)$, and, as observed already, $|T|\leq 2n^{1-\eps}$ must hold. Next, we apply the algorithm from \Cref{lem: cut if vertex of L} to every vertex $x\in T$ in turn, and we let $(c_x,y_x)$ be the outcome of this application of the algorithm from the lemma. We then let $x'\in T$ be the vertex for which $c_{x'}$ is minimized, and return $(c_{x'},x',y_{x'})$ as the outcome of the algorithm. \Cref{lem: cut if vertex of L} guarantees that $c_{x'}$ is the value of the minimum $x'$-$y_{x'}$ vertex cut in $G$. Assume now that 
there is a global minimum vertex cut $(L,S,R)$ in $G$ with $|S|\geq n^{1-\eps}\cdot |L|$.
Then, from \Cref{obs: high degree in L}, $L\cap T\neq \emptyset$ must hold. Let $v$ be any vertex in $L\cap T$. Then, when the algorithm from \Cref{lem: cut if vertex of L} was applied to vertex $v$, with probability at least $1-1/n^4$, the value $c_{v}$ of the minimum $v$-$y_v$ vertex cut that it returned is equal to $\opt$. In this case, we are guaranteed that the value $c$ that our algorithm returns is equal to $\opt$ as well.
Since $|T|\leq n^{1-\eps}$, and since the running time of the algorithm from \Cref{lem: cut if vertex of L} is $O\left (m^{1+o(1)}\cdot \log W' \right )$, we get that the running time of our algorithm is $O\left (mn^{1-\eps+o(1)}\cdot \log W'\right )$.

In order to complete the proof of \Cref{thm: alg 2}, it now remains to prove \Cref{lem: cut if vertex of L}.

\begin{proofof}{\Cref{lem: cut if vertex of L}}
	Let $Z=V(G)\setminus\left(\set{x}\cup N_G(x)\right )$.  Let $\aset''$ be an algorithm that selects a vertex $y\in Z$ at random, where the probability to select a vertex $v\in V(G)$ is $\frac{w(v)}{w(Z)}$. It then uses the algorithm from \Cref{cor: min_vertex_cut} to compute the value $c$ of minimum $x$-$y$ vertex cut in $G$, in time $O\left (m^{1+o(1)}\cdot \log W'\right )$, and then returns vertex $y$ and the value $c$. 
	We use the following observation.

	\begin{observation}\label{obs: alg2 success}
		Suppose there is a global minimum vertex cut $(L,S,R)$ in $G$ with $x\in L$. Let $(c,y)$ be the outcome of the algorithm $\aset''$. Then with probability at least $\frac{1}{3}$, $y\in R$ and $c=\opt$ holds.
	\end{observation}
\begin{proof}
	Consider the global minimum cut $(L,S,R)$, and assume that $x\in L$. Let $S'=S\setminus N_G(x)$. Then from \Cref{claim: neighbors in S},  $w(S')\leq w(L)$ must hold. 
	
	We conclude that $w(Z)\leq w(L)+w(R)+w(S')\leq 2w(L)+w(R)\leq 3w(R)$, since we always assume by default that, in any vertex cut $(L,S,R)$, $w(L)\leq w(R)$ holds. Therefore,  the probability that Algorithm $\aset''$ chooses a vertex $y\in R$ is $\frac{w(R)}{w(Z)}\geq \frac{1}{3}$. Clearly, if $y\in R$ holds, then the value $c$ of the minimum $x$-$y$ vertex cut in $G$ is equal to $\opt$.	
\end{proof}

We are now ready to complete the proof of \Cref{lem: cut if vertex of L}. We execute Algorithm $\aset''$ $\Theta(\log n)$ times, and output the smallest value $c$ that the algorithm returned, together with the corresponding vertex $y$. It is immediate to verify that $c$ is indeed the value of the minimum $x$-$y$ vertex cut in $G$. Assume now that there is a global minimum vertex cut $(L,S,R)$ in $G$ with $x\in L$. We say that an application of Algorithm $\aset''$ is \emph{successful}, if it returns a value $c'=\opt$. From \Cref{obs: alg2 success}, the probability that a single application of the algorithm is successful is at least $1/3$. It is then easy to verify that the probability that every application of Algorithm $\aset''$ was unsuccessful is at most $1/n^4$. If at least one application of the algorithm was successful, then we are guaranteed that the value $c$ returned by our algorithm is equal to $\opt$. Since the running time of Algorithm $\aset''$ is $O(m^{1+o(1)}\cdot \log W')$, the total running time of the entire algorithm is $O(m^{1+o(1)}\cdot \log W')$.
\end{proofof}

\subsection{Summary: an Algorithm for Non-Dense Graphs}

By combining Algorithms $\alg_1$ and $\alg_2$, we obtain an algorithm for global minimum cut on graphs that are not very dense, that is summarized in the following corollary.

\begin{corollary}\label{cor: alg non dense}
	There is a randomized algorithm, whose input is a simple undirected $n$-vertex and $m$-edge graph $G$ with integral weights $1\leq w(v)\leq W'$ on vertices $v\in V(G)$, and a parameter $0<\eps\leq 1/2$, such that $m\leq n^{2-2\eps}$ holds. 
	The algorithm returns two distinct vertices $x,y\in V(G)$ with $(x,y)\not\in E(G)$, and a value $c$ that is at least as large as the value of the minimum $x$-$y$ vertex-cut in $G$.
	Moreover, with probability at least $1-1/n^4$, $c=\opt$ holds. The running time of the algorithm is $O\left (mn^{1-\eps+o(1)}\cdot \log W'\right )$.
\end{corollary}
\begin{proof}
Assume that we are given as input  an undirected $n$-vertex and $m$-edge graph $G$ with integral weights $1\leq w(v)\leq W'$ on vertices $v\in V(G)$, and a parameter $0<\eps\leq \half$, such that $m\leq n^{2-2\eps}$ holds. Consider the following algorithm, that we denote by $\alg^*$: the algorithm starts by applying Algorithm $\alg_1$ from \Cref{thm: alg 1} to $G$, and denotes by $(c,x,y)$ its outcome. It then applies Algorithm $\alg_2$ from \Cref{thm: alg 2} to $G$, and denotes by $(c',x',y')$ its outcome. If $c'\leq c$, it returns  value $c^*=c'$, and vertices $x^*=x'$ and $y^*=y'$; otherwise, it returns value $c^*=c$ and vertices $x^*=x$ and $y^*=y$. Since the running times of both Algorithms $\alg_1$ and $\alg_2$ are bounded by $O(mn^{1-\eps+o(1)}\cdot \log W')$, we get that the running time of Algorithm $\aset^*$ is $O(mn^{1-\eps+o(1)}\cdot \log W')$. It is easy to verify that the value of the minimum $x^*$-$y^*$ cut in $G$ is indeed at most $c^*$. 

Consider now any global minimum cut $(L,S,R)$ in $G$. Assume first that $|S|\geq n^{1-\eps}\cdot |L|$ holds. Then, from \Cref{thm: alg 2}, with probability at least $1-1/n^4$, $c'=\opt$ holds, and in this case, we are guaranteed that $c^*=\opt$.

Assume now that $|S|<n^{1-\eps}\cdot |L|$. We claim that, in this case, there must be an integer $0\leq i\leq \ceil{\log W'}$, such that $L_i\neq\emptyset$ and $|S_i|\leq |L_i|\cdot n^{1-\eps}$. Indeed, assume otherwise. Then for all $0\leq i\leq \ceil{\log W'}$, $|S_i|\geq|L_i|\cdot n^{1-\eps}$ must hold, and so:

\[|S|=\sum_{i=0}^{\ceil{\log W'}}|S_i| \geq  n^{1-\eps}\cdot \sum_{i=0}^{\ceil{\log W'}}|L_i|= |L|\cdot n^{1-\eps}, \]
a contradiction.
\Cref{thm: alg 1} then guarantees that, in this case,  with probability at least $1-1/n^4$, $c=\opt$ holds, and so we are guaranteed that $c^*=\opt$.

In either case, we get that, with probability at least $1-1/n^4$, $c^*=\opt$ holds.
\end{proof}

\subsection{Algorithm $\alg_3$: when $G$ is Dense and $L$ is Small}

Our final and the most involved algorithm, $\alg_3$, will only be used in the case where the input graph $G$ is dense; we summarize it in the following theorem. The theorem statement uses the parameter $\wmax=\wmax(G)$, defined in \Cref{subsec: split graph}.

\begin{theorem}\label{thm: alg 3}
	There is a randomized algorithm, whose input is a simple undirected $n$-vertex graph $G$ with integral weights $1\leq w(v)\leq W'$ on vertices $v\in V(G)$. 
	The algorithm returns two distinct vertices $x,y\in V(G)$ with $(x,y)\not\in E(G)$, and a value $c$ that is at least as large as the value of the minimum $x$-$y$ vertex-cut in $G$.
Moreover, if there is a global minimum vertex-cut $(L,S,R)$ in $G$, such that, for all $0\leq i\leq \ceil{\log W'}$, $|S_i|\geq n^{44/45}\cdot |L_i|$ holds, then, with probability at least $\frac{1}{2^{11}\cdot \log^3n\cdot \log W'}$, $c=\opt$ holds. The running time of the algorithm is 
$O\left (n^{2.96}\cdot (\log(\wmax))^{O(1)}\right )$. 
\end{theorem}

The majority of the remainder of the paper is dedicated to the proof of \Cref{thm: alg 3}.
We now complete \Cref{thm: main} using it.

\subsection{Completing the Proof of~\Cref{thm: main}}
\label{subsec: finishing the alg}

Recall that we are given as input simple undirected graph $G$ with integral  weights $1\leq w(v) \leq W'$ on its vertices, where $W'\leq O(n^2\cdot W)$. 
Our goal is to compute a global minimum vertex-cut in $G$, whose value is denoted by $\opt$. We start by presenting an algorithm with running time $O\left (mn^{0.99+o(1)}\cdot (\log W)^{O(1)}\right )$, followed by an algorithm with running time $O\left(m^{3/2+o(1)}\cdot (\log W)^{O(1)}\right )$.

\subsubsection*{Algorithm with Running Time $O\left (mn^{0.99+o(1)}\cdot (\log W)^{O(1)}\right )$}

Let $\eps'=\frac{1}{100}$. We consider two cases.
The first case happens if $m\leq n^{2-2\eps'}=n^{1.98}$ holds. In this case, we apply the algorithm from \Cref{cor: alg non dense}, and we denote by $(c,x,y)$ its output. We then compute a minimum $x$-$y$ vertex-cut $(X,C,Y)$ in $G$ using the algorithm from \Cref{cor: min_vertex_cut} in time $O(m^{1+o(1)}\cdot \log W')$, and return the cut $(X,C,Y)$ as the algorithm's output. Notice that, from \Cref{cor: alg non dense}, with probability at least $1-1/n^4$, $(X,C,Y)$ is indeed a global minimum vertex-cut in $G$. The running time of the algorithm from \Cref{cor: alg non dense} is $O(mn^{1-\eps'+o(1)}\cdot \log W')=O(mn^{0.99+o(1)}\cdot \log W')$, and so the total running time of the algorithm in this case is $O(mn^{0.99+o(1)}\cdot \log W')$.

It now remains to consider the second case, where $m> n^{1.98}$ holds. We let $\eps=1/45$.
Consider the following algorithm, that we denote by $\alg^{**}$: the algorithm starts by applying Algorithm $\alg_1$ from \Cref{thm: alg 1} to $G$, and denotes by $(c_1,x_1,y_1)$ its outcome. It then
performs $\hat z=\ceil{2^{18}\cdot \log W'\cdot \log^4n}$ iterations, and in every iteration it
 applies Algorithm $\alg_3$ from \Cref{thm: alg 3} to $G$. We denote by $(c_3,x_3,y_3)$ the outcome of the iteration in which the value $c$ returned by the algorithm was the smallest. If $c_3\leq c_1$, then we use the algorithm from \Cref{cor: min_vertex_cut} 
 to compute a minimum $x_3$-$y_3$ vertex-cut $(X_3,C_3,Y_3)$ in $G$, and return this cut as the algorithm's outcome. 
 Otherwise, we use the algorithm from \Cref{cor: min_vertex_cut}  to  compute a minimum $x_1$-$y_1$ vertex-cut $(X_1,C_1,Y_1)$, and return this cut.

Recall that the running time of Algorithms $\alg_1$ from \Cref{thm: alg 1} is $O\left (mn^{1-\eps+o(1)}\cdot \log W'\right )\leq O(mn^{44/45+o(1)}\cdot \log W')$, and the running time of Algorithm $\alg_3$ from \Cref{thm: alg 3} is:

$$O\left (n^{2.96}\cdot (\log(\wmax))^{O(1)}\right )\leq O\left(n^{2.96+o(1)}\cdot (\log W')^{O(1)}\right ),$$

since $\wmax\leq O(n\cdot W')$ from Inequality \ref{eq: bound on wmax}.

Since the running time of the algorithm from \Cref{cor: min_vertex_cut} is $O(m^{1+o(1)}\cdot \log W')$, we get that the total running time of algorithm $\aset^{**}$ is bounded by:

\[
\begin{split}
\hat z\cdot O\left(n^{2.96+o(1)}\cdot (\log W')^{O(1)}\right )&+O(mn^{44/45+o(1)}\cdot \log W')\\
&\leq O\left(n^{2.96+o(1)}\cdot (\log W')^{O(1)}\right )+O(mn^{44/45+o(1)}\cdot \log W')\\
&\le O\left(mn^{0.99+o(1)}\cdot (\log W')^{O(1)}\right )+O(mn^{44/45+o(1)}\cdot \log W')\\
&\leq O\left (mn^{0.99+o(1)}\cdot (\log W')^{O(1)}\right ).
\end{split}
\]

since $\hat z=O\left(\log W'\cdot\log^4 n\right )$ and $m>n^{1.98}$.

Consider now any global minimum vertex-cut $(L,S,R)$ in $G$.
Assume first that, for some integer $0\leq i\leq \ceil{\log W'}$,  $L_i\neq\emptyset$ and $|S_i|\leq |L_i|\cdot n^{44/45}= |L_i|\cdot n^{1-\eps}$ hold. From
\Cref{thm: alg 1}, in this case, with probability at least $1-1/n^4$, $c_1=\opt$ holds, and so our algorithm returns a vertex cut of value $\opt$.

Otherwise, we are guaranteed that for all $0\leq i\leq \ceil{\log W'}$, $|S_i|\geq n^{44/45}\cdot |L_i|$ holds. 
We say that the application of the algorithm from \Cref{thm: alg 3} is successful, if the estimate $c$ that it returns is equal to $\opt$. From  \Cref{thm: alg 3}, the probability that a single application of the algorithm is successful is at least $\frac{1}{2^{11}\cdot \log^3 n\cdot \log W'}$. Since we execute the algorithm from \Cref{thm: alg 3} $\hat z=\ceil{2^{18}\log W'\cdot \log^4n}$ times, the probability that every application of the theorem is unsuccessful is at most $\left(1-\frac{1}{2^{11}\cdot \log^3 n\cdot \log W'}\right )^{\hat z}\leq \frac{1}{n}$. If at least one application of the algorithm is successful, then our algorithm is guaranteed to return a vertex cut of value $\opt$. Note that in either case, our algorithm returns global a minimum vertex-cut with probability at least $1-1/n$.

The total running time of the algorithm is: 

\[ O\left (mn^{0.99+o(1)}\cdot \log W'\right ) \leq O\left (mn^{0.99+o(1)}\cdot \log W \right ),\]

since $W'=O(n^2\cdot W)$.

\subsubsection*{Algorithm with Running Time $O\left(m^{3/2+o(1)}\cdot (\log W)^{O(1)}\right )$}

Note that we can assume that the input graph $G$ is connected, and that it is not a tree, since otherwise the problem is trivial. 
We can also assume that $m<n^2$, since otherwise, the $\tO(mn)$-time algorithm of~\cite{HRG00} achieves the desired bound of $O\left(m^{3/2+o(1)}\cdot (\log W)^{O(1)}\right )$ on the  running time. 
Therefore, we can denote $m=n^x$, where $1\leq x< 2$. Let $\eps=1-x/2$, so that $m= n^{2-2\eps}$ and $\eps>0$. We apply the algorithm from Corollary~\ref{cor: alg non dense} to graph $G$, which must return a global minimum vertex cut with probability at least $1-1/n^4$. The running time of the algorithm is $O\left (mn^{1-\varepsilon+o(1)}\cdot(\log W')^{O(1)}\right ) \leq O\left(m^{3/2+o(1)}\cdot (\log W)^{O(1)}\right )$.

In order to complete the proof of \Cref{thm: main}, it is now enough to prove \Cref{thm: alg 3}, which we do in the remainder of the paper. We start with a high-level overview of the proof.
\subsection{Proof of \Cref{thm: alg 3} -- High Level Overview}
\label{subsec: alg 3 setup}

Throughout the proof, we will use the following simple claim.

\begin{claim}\label{claim: bounds on S and L}
	Let $G$ be a simple undirected $n$-vertex graph with integral weights $1\leq w(v)\leq W'$ on vertices $v\in V(G)$, and let $0<\eps<1$ be a parameter. 
	Let $(L,S,R)$ be any vertex cut  in $G$, and assume that, for all $0\leq i\leq \ceil{\log W'}$, $|S_i|\geq n^{1-\eps}\cdot |L_i|$ holds. Then $|L|\leq \frac{|S|}{n^{1-\eps}}\leq n^{\eps}$ holds. 
\end{claim}
\begin{proof}
Since, for all $0\leq i\leq \ceil{\log W'}$, $|S_i|\geq n^{1-\eps}\cdot |L_i|$ holds, we get that:
	
	\[|S|=\sum_{i=0}^{\ceil{\log W'}}|S_i|\geq \sum_{i=0}^{\ceil{\log W'}}|L_i|\cdot n^{1-\eps} \geq |L|\cdot n^{1-\eps}.\] 
	
	Therefore, $|L|\leq \frac{|S|}{n^{1-\eps}}\leq n^{\eps}$.
%
%
%
\end{proof}

We assume that we are given a simple undirected $n$-vertex graph $G$ with integral weights $1\leq w(v)\leq W'$ on vertices $v\in V(G)$. Throughout, we set $\eps=1/45$.
We say that a global minimum vertex-cut $(L,S,R)$ is \emph{acceptable}, if, for all $0\leq i\leq \ceil{\log W'}$, $|S_i|\geq n^{1-\eps}\cdot |L_i|$ holds. 
We define a specific global minimum vertex-cut $(L,S,R)$ as follows: if an acceptable global minimum vertex-cut exists in $G$, then we let $(L,S,R)$ be any such cut; otherwise, we let $(L,S,R)$ be an arbitrary vertex cut. The cut $(L,S,R)$ remains fixed throughout this section. We denote by $\opt=w(S)$ the value of the global minimum vertex-cut.

\subsubsection{A Good Set of Terminals and a Good Pair}

 We use the notion of a good set of terminals, that is defined next.

\begin{definition}[A good set of terminals]\label{def: good set of terminals}
Let $T\subseteq V(G)$ be any subset of vertices. We say that $T$ is a \emph{good set of terminals} with respect to a global minimum cut $(L,S,R)$ if $|T\cap L|=1$ and $T\cap R\neq \emptyset$ hold, and, additionally, if we denote by $x$ the unique vertex in $T\cap L$, then $T\cap S\subseteq N_G(x)$.
\end{definition}

Since we have fixed the global minimum cut $(L,S,R)$, if vertex set $T$ is a good set of terminals with respect to this cut, then we simply say that $T$ is a good set of terminals.

Next, we provide a randomized algorithm, that we refer to as $\algterm$, that selects a subset $T$ of vertices of $G$ so that, if the cut $(L,S,R)$ is acceptable, then, with a high enough probability, set $T$ of terminals is good. 

The notion of a good set of terminals was implicitly introduced by~\cite{LNP21} in their algorithm for computing a global minimum vertex-cut in unweighted graphs. They also provide a simple randomized algorithm that, with a high enough probability, computes a good set of terminals in this setting. Our algorithm $\algterm$ is a natural extension of their algorithm to vertex-weighted graphs.

We now describe Algorithm $\algterm$ formally. The algorithm starts by selecting an integer $i\in[1,\ceil{\log (n\cdot W')}]$ uniformly at random. Next, we let $T$ be a random subset of vertices of $G$, obtained by adding every vertex $v\in V(G)$ to $T$ independently with probability $\frac{w(v)}{2^{i+1}}$. 
We say that a good event $\event$ happens if the resulting set $T$ of vertices is a good set of terminals with respect to $(L,S,R)$. If the event $\event$ happens, then we say that Algorithm \algterm is \emph{successful}. In the following claim, we show that, if the cut $(L,S,R)$ is acceptable, then algorithm \algterm is successful with a reasonably high probability.

\begin{claim}\label{claim: good set of terminals}
If the cut $(L,S,R)$ is acceptable, then	$\prob{\event}\geq \frac{1}{2^{11}\log n \cdot \log W'}$.
\end{claim}
\begin{proof}
	Assume that the cut $(L,S,R)$ is acceptable.
	Let $1\leq i^*\leq \ceil{\log (n\cdot W')}$ be the unique integer for which $2^{i^*-1}\leq w(L)<2^{i^*}$ holds. We say that a good event $\event'$ happens if $i=i^*$. Clearly, $\prob{\event'}\geq \frac{1}{\ceil{\log (n\cdot W')}}\geq \frac{1}{2\log n\cdot \log W'}$.
	
	For every vertex $v\in L$, we say that a good event $\event(v)$ happens, if all of the following hold: (i) $v\in T$; (ii) $v$ is the unique vertex of $L$ that lies in $T$; and (iii) $T\cap S\subseteq N_G(v)$. 
	Let $S'_v=S\setminus N_G(v)$, and observe that from \Cref{claim: neighbors in S}, $w(S'_v)\leq w(L)$ must hold. Let $X=(L\cup S'_v)\setminus \set{v}$. Then Event $\event(v)$ happens if and only if $v$ is added to $T$, but none of the vertices of $X$ are added to $T$. From our discussion so far, $w(X)\leq 2w(L)$, and for every vertex $x\in X$, $w(x)\leq w(L)\leq 2^{i^*}$. Altogether, we get that:
	
	\[\begin{split}
	\prob{\event(v)\mid \event'}&=\frac{w(v)}{2^{i+1}}\cdot \prod_{x\in X}\left(1-\frac {w(x)} {2^{i+1}} \right )\\
	&= \frac{w(v)}{2^{i+1}}\cdot e^{\sum_{x\in X}\ln\left(1-\frac {w(x)} {2^{i+1}} \right )}\\
	&\geq\frac{w(v)}{16\cdot 2^{i}}.	\end{split}
	\]

For the last inequlatiy, we used the fact that, from Taylor's expansion of function $\ln(1-\eps)$,  for all $0\leq \eps\leq 1/2$, $\ln(1-\eps)>-\eps-\eps^2\geq -2\eps$ holds. Therefore, for all $x\in X$, $\ln\left(1-\frac {w(x)} {2^{i+1}} \right )\geq -\frac{w(x)}{2^{i}}$, and $\sum_{x\in X}\ln\left(1-\frac {w(x)} {2^{i+1}} \right )\geq -\frac{\sum_{x\in X}w(x)}{2^{i}}= -\frac{w(X)}{2^{i}}\geq -2$, since $w(X)\leq 2w(L)\leq 2^{i+1}$.

Let $\event_1$ be the good event that Event $\event(v)$ happened for any vertex $v\in L$.
Notice that the events $\event(v)$ are disjoint for vertices $v\in L$. Therefore:

\[\prob{\event_1\mid \event'}=\sum_{v\in L}\prob{\event(v)\mid\event'}
\geq \frac{w(L)}{16 \cdot 2^i}\geq \frac{1}{32},
\]

since $w(L)\geq 2^{i-1}$.
Next, we let $\event_2$ be the good event that $R\cap T\neq \emptyset$. Since we assumed that $w(R)\geq w(L)$, it is immediate to verify that:

\[\prob{\event_2\mid \event'}\geq \prob{T\cap L\neq \emptyset\mid \event'}\geq \prob{\event_1\mid \event'}\geq \frac{1}{32}.\]

Since Events $\event_1$ and $\event_2$ are independent, we get that:

\[\prob{\event_1\band \event_2\mid \event'}=\prob{\event_1\mid \event'}\cdot \prob{\event_2\mid \event'}\geq \frac{1}{2^{10}}.\]
Finally, we get that:

\[\prob{\event}\geq \prob{\event_1\band\event_2\mid \event'}\cdot \prob{\event'}\geq \frac{1}{2^{11}\cdot \log n\cdot \log W'}.\]

\end{proof}

Given a set $T$ of vertices of $G$, and a vertex $x\in T$, we denote by  $T_x=T\setminus\left(\set{x}\cup N_G(x)\right )$. Next, we define a notion of a good pair.

\begin{definition}[A good pair]
	Let $T\subseteq V(G)$ be a set of vertices, let $x\in T$ be a vertex, and let $(L,S,R)$ be a global minimum vertex-cut. We say that $(x,T)$ is a \emph{good pair} with respect to the cut $(L,S,R)$, if $T$ is a good set of terminals with respect to $(L,S,R)$, and $x$ is the unique vertex of $T\cap L$. 
\end{definition}	
	
Notice that, if $(x,T)$ is a good pair with respect to $(L,S,R)$, then $T\cap S\subseteq N_G(x)$ and $T\cap R\neq \emptyset$ must hold; in particular, $T_x\subseteq R$. Therefore, the value of the minimum $x$-$T_x$ vertex-cut in $G$ is $\opt$ (since the set $S$ of vertices disconnects $x$ from $T_x$ in $G$, and $S$ is disjoint from $x$ and from $T_x$). Clearly, if $T$ is a good set of terminals with respect to $(L,S,R)$, then there is exactly one vertex $x\in T$, such that $(x,T)$ is a good pair with respect to $(L,S,R)$.

The main technical subroutine of our algorithm, that we formally describe below, receives as input a set $T$ of vertices of $G$, and a vertex $x\in T$. The subroutine outputs a value $c_x$ that is at least as high as the value of the minimum $x$-$T_x$ cut in $G$. If, additionally, $(L,S,R)$ is an acceptable cut, and $(x,T)$ is a good pair with respect to $(L,S,R)$, then the subroutine guarantees that, with a reasonably high probability, $c_x$ is equal to the value of the minimum $x$-$T_x$ vertex-cut in $G$, and hence to $\opt$. One subtlety in the description of the subroutine is that it needs an access to a \emph{subgraph oracle}, that we define below. We will ensure that the subroutine only accesses the oracle $O(\log n \log (\wmax))$ times over its execution. We will first provide the statement of the theorem summarizing the subroutine, assuming that it has access to the subgraph oracle, and then show an algorithm that, in a sense, implements this oracle. 
 We now proceed to define a subgraph oracle.

\subsubsection{A Subgraph Oracle}

We now define a subgraph oracle, that receives as input an undirected graph $G$ with integral weights on its vertices. The definition uses the parameter $\wmax=\wmax(G)$ defined in \Cref{subsec: modified adj list}.

\begin{definition}[A Subgraph Oracle]
	A $\delta$-subgraph oracle,  for a parameter $0<\delta\leq 1/2$, receives as input a simple undirected $n$-vertex graph $G$ with integral weights $1\leq w(v)\leq W'$ on its vertices $v\in V(G)$ in the adjacency list representation, where $n$ is sufficiently large, so $\delta\geq\frac{1}{\sqrt{\log n}}$, and a set $Z\subseteq V(G)$ of vertices. The oracle must return a partition $(Y^h,Y^{\ell})$ of $V(G)$, and a collection $E'$ of at most $n^{2-\delta}\cdot \log^2(\wmax)$ edges of $G$, such that 
	$E'=E_G(Y^{\ell},Z)$ holds. We say that the oracle \emph{errs}, if some vertex $v\in Y^{h}$ has fewer than $\frac{n^{1-\delta}}{1000\log n}$ neighbors in $Z$, or some vertex $v'\in Y^{\ell}$ has more than $n^{1-\delta}\cdot \log^2(\wmax)$ such neighbors. We require that the probability that the oracle errs is at most $\frac{1}{n^5\cdot \log^4(\wmax)}$. 
\end{definition}

\subsubsection{The Main Technical Subroutine}
The following theorem summarizes our main technical subroutine. The theorem also uses the notion of the split graph and the parameter 
$\wmax=\wmax(G)$, both of which are defined in \Cref{subsec: split graph}.

\begin{theorem}\label{thm: main subroutine}
	There is a randomized algorithm, whose input consists of a simple undirected $n$-vertex graph $G$ with integral weights $1\leq w(v)\leq W'$ on vertices $v\in V(G)$ in the adjacency-list representation, a subset $T\subseteq V(G)$ of its vertices, and a vertex $s\in T$. Additionally, the algorithm is given two copies of the split graph $G'$ corresponding to $G$ in the modified adjacency-list representation. Lastly, it is given an access to the $\delta$-subgraph oracle for $G$, for $\delta=\frac{4}{45}-\frac{\log(4000\log n)}{\log n}=\frac{4}{45}-o(1)$. The algorithm returns a value $c_s$, that is at least as high as the value of the minimum $s$-$T_s$ vertex-cut in $G$. Moreover, if there is a global minimum vertex-cut $(L,S,R)$ in $G$, such that $|S|\geq n^{44/45}\cdot |L|$ holds, and $(s,T)$ is a good pair with respect to cut $(L,S,R)$, then, with probability at least $\frac{1}{2}$, $c_s$ is equal to the value of the minimum $s$-$T_s$ vertex-cut in $G$. The running time of the algorithm is $O\left(n^{86/45+o(1)}\cdot (\log(\wmax))^{O(1)}\right )$, and it accesses the $\delta$-subgraph oracle at most $8\log n\cdot \log (\wmax)$ times.
\end{theorem}

We note that the running time of the algorithm from \Cref{thm: main subroutine} may be faster than $|E(G)|$, and that it might modify the adjacency-list representations of the graph $G'$ that it is given as input. Specifically, the algorithm gradually modifies the graph $G'$ in order to ``simplify'' it, and also maintains an $s$-$T_s$ flow $f$ in the resulting modified graph, that it gradually augments.
One copy of the graph will be used to maintain the modified graph $G'$,  while the second copy will be used to maintain the residual flow network of $G'$ with respect to the flow that the algorithm maintains.
We prove \Cref{thm: main subroutine} in the remainder of the paper, after we complete the proof of \Cref{thm: alg 3} using it.
We start by providing an algorithm that we will use to implement the oracle.

\subsubsection{Implementing the Oracle}

Recall that our algorithm will apply the subroutine from \Cref{thm: main subroutine} to every vertex $s\in T$. Each such execution of the subroutine will make at most $8\log n\cdot \log(\wmax)$ calls to the oracle. In order to make our algorithm efficient, we will design an oracle that responds to several $\delta$-subgraph queries in bulk. Intuitively, for all $1\leq i\leq \floor{8\log n\cdot \log(\wmax)}$, we will run the subroutine from \Cref{thm: main subroutine} for each vertex $s\in T$, until its $i$th attempt to access the oracle. We will then simultaneously compute responses to all resulting $|T|$ subgraph queries. We summarize our algorithm for implementing the oracle in the following lemma. The lemma statement uses the matrix multiplication exponent $\omega$, whose current bound is $\omega\leq 2.371552$~\cite{WXXZ24}.

\begin{lemma}\label{lem: oracle in bulk}
	There is a deterministic algorithm, that is given as input an undirected $n$-vertex graph $G$ with weights $1\leq w(v)\leq W'$ on the vertices $v\in V(G)$, a parameter $\frac{1}{\sqrt{\log n}}<\delta\leq \half$, and $q\leq n$ queries $(Z_1,\ldots,Z_q)$ to the $\delta$-subgraph oracle. The algorithm computes responses $(Y^h_1,Y^{\ell}_1,E'_1),\ldots,(Y^h_q,Y^{\ell}_q,E'_q)$ to the queries, such that the probability that the algorithm errs in its response to any query is at most $\frac{1}{n^9\cdot\log^4(\wmax)}$. The running time of the algorithm is $O\left (n^{3-\delta\cdot 2(3-\omega)/(5-\omega)+o(1)}\cdot \poly\log(\wmax)\right )\leq O\left (n^{3-0.478\delta+o(1)}\cdot \poly\log(\wmax)\right )$, where $\omega$ is the matrix multiplication exponent.
\end{lemma}

\begin{proof}
	We start by computing, for all $1\leq i\leq q$, the partition $(Y^h_i,Y^{\ell}_i)$ of $V(G)$, as follows. Fix an index $1\leq i\leq q$, and set $\tau=\frac{n^{1-\delta}}{1000\log n}$.
	We apply the algorithm for the Undirected Degree Estimation problem from 
	\Cref{claim: degree est} to graph $G$, the sets $Z=Z_i$ and $Z'=V(G)$ of its vertices, and threhsold $\tau$. Let $A\subseteq V(G)$ be the subset of vertices that the algorithm returns. We then set $Y_i^h=A$ and $Y_i^{\ell}=V(G)\setminus A$. Let $\event_i$ be the bad event that the algorithm from \Cref{claim: degree est} erred, and recall that $\prob{\event_i}\le \frac{1}{n^{10}\log^4(\wmax)}$. Note that, if the event $\event_i$ did not happen, then we are guaranteed that (i) every vertex $v\in Y_i^{h}$ has at least $\tau=\frac{n^{1-\delta}}{1000\log n}$ neigbhors in $Z_i$; and (ii) every vertex $v'\in Y_i^{\ell}$ has at most $1000\tau\log n\log (\wmax)=n^{1-\delta}\cdot \log(\wmax)$ neighbors in $Z$.
Recall that the running time of the algorithm from \Cref{claim: degree est} is:

\[O\left(\frac{n\cdot |Z_i|\cdot \log n\cdot\log(\wmax)}{\tau}\right )\leq O\left(n^{1+\delta}\cdot \log^2n \cdot\log(\wmax) \right )\leq  O\left(n^{2-\delta}\cdot \log^2n \cdot\log(\wmax) \right ),
\]

since $\delta\leq 1/2$,
 and so the time required to execute this step for all $1\leq i\leq q$ is $O\left(n^{3-\delta}\cdot \log^2n\cdot\log(\wmax)\right )$. We let $\event$ be the bad event that $\event_i$ happened for any $1\leq i\leq q$. By using the Union Bound, and since $q\leq n$, we get that $\prob{\event}\leq \frac{1}{n^9\log^4(\wmax)}$.

Next, we construct a tripartite graph $H=(A\cup B\cup C,E')$, as follows. We let $A=\set{a_v\mid v\in V(G)}$, $B=\set{b_v\mid v\in V(G)}$, and $C=\set{c_1,\ldots,c_{q}}$. For every edge $e=(u,v)\in E(G)$, we add the edges $(a_v,b_u)$ and $(a_u,b_v)$ to $H$. Additionally, for every index $1\leq i\leq q$, for every vertex $v\in Z_i$, we add the edge $(a_v,c_i)$, and for every vertex $u\in Y^{\ell}_i$, we add the edge $(c_i,b_u)$ to $H$.

For a triple $\Pi=(a_v,b_u,c_i)$ of vertices of $H$ with $a_v\in A$, $b_u\in B$, and $c_i\in C$, we say that $\Pi$ is a \emph{good triple}, if $(v,u)\in E_G(Z_i,Y^{\ell}_i)$. Notice that, if Event $\event$ did not happen, then $|E_G(Z_i,Y^{\ell}_i)|\leq n^{2-\delta}\cdot\log(\wmax)$ must hold. Therefore, if Event $\event$ did not happen, then the total number of good triples in $H$ is bounded by $n^{2-\delta}\cdot q\cdot\log(\wmax)\leq n^{3-\delta}\cdot\log(\wmax)$. Moreover, every good triple defines a triangle in $H$ and vice versa. Therefore, in order to compute, for every index $1\leq i\leq q$, the set $E_G(Z_i,Y^{\ell}_i)$ of edges, it is sufficient to compute all triangles in $H$. We use the following algorithm of 
\cite{bjorklund2014listing} in order to do so.

\begin{theorem}[Theorem 1 in \cite{bjorklund2014listing}]\label{thm: triangle counting}
There is a deterministic algorithm that lists all triangles in an $n$-vertex graph in time:  

\[O\left(n^{\omega+o(1)}+n^{3(\omega-1)/(5-\omega)+o(1)}\cdot t^{2(3-\omega)/(5-\omega)}\right ),\] 

where $\omega$ is the matrix multiplication exponent, and $t$ is the total number of triangles in $G$.
\end{theorem}

By substituting $t\leq n^{3-\delta}\cdot\log(\wmax)$, we get that the running time of the algorithm from \Cref{thm: triangle counting} on graph $H$ is bounded by:

\[
\begin{split}
&O\left(\left(n^{\omega}+n^{3(\omega-1)/(5-\omega)}\cdot n^{2(3-\omega)\cdot (3-\delta)/(5-\omega)}\right )\cdot n^{o(1)}\cdot \poly\log(\wmax)\right )\\
&\quad\quad\quad\quad\quad\quad\quad\quad\quad\leq 
 O\left (n^{3-\delta\cdot 2(3-\omega)/(5-\omega)}\right )\cdot n^{o(1)}\cdot \poly\log(\wmax).
\end{split}
\]

Finally, recall that the running time of the first step, that performed $q$ applications of the algorithm from \Cref{claim: degree est} for the Degree Estimation problem is $O(n^{3-\delta}\cdot\log^2n\cdot\log(\wmax))$, and that the time required to compute the graph $H$, given the graph $G$ and the sets $\set{Z_i,Y_i^{\ell}}_{1\leq i\leq q}$ of vertices is bounded by $O(n^2)$. Therefore, the total running time of the algorithm is 
$O\left ( n^{3-\delta\cdot 2(3-\omega)/(5-\omega)}\cdot n^{o(1)}\cdot \poly\log(\wmax)\right )$, or $O\left (n^{3-0.478\delta+o(1)}\cdot \poly\log(\wmax)\right )$ using the current bound $\omega\leq 2.371552$~\cite{WXXZ24}.
\end{proof}

\subsubsection{Completing the Proof of \Cref{thm: alg 3}}

We are now ready to complete the proof of \Cref{thm: alg 3}. We assume that we are given as input a simple undirected $n$-vertex and $m$-edge graph $G$ with integral weights $1\leq w(v)\leq W'$ on vertices $v\in V(G)$. We fix a global minimum vertex-cut $(L,S,R)$ in $G$ as follows.  If there exists an acceptable global minimum vertex cut $(L,S,R)$ (that is, for all $0\leq i\leq \ceil{\log W'}$, $|S_i|\geq n^{44/45}\cdot |L_i|$ holds), then we fix this cut. Otherwise, we let $(L,S,R)$ be an arbitrary global minimum vertex-cut. 

We start by computing a split graph $G'$ of $G$, and two copies of $G'$  in  modified the adjacency-list representation. We also compute a matrix representation of $G$. All this can be done in time $O(m\log (\wmax))\leq O(n^2\log(\wmax))$. Next, we apply algorithm \algterm to compute a set $T\subseteq V(G)$ of vertices that we call terminals from now on. It is easy to verify that the running time of Algorithm \algterm is $O(n)$. We let $\event_T$ be the good event that $T$ is a good set of terminals with respect to the cut $(L,S,R)$. From \Cref{claim: good set of terminals}, if $(L,S,R)$ is an acceptable cut, then 
$\prob{\event_T}\geq \frac{1}{2^{11}\cdot \log n \cdot \log W'}$.

Our algorithm consists of at most $z=\floor{8\log n\cdot \log \wmax}$ phases. For all $1\leq i\leq z$, Phase $i$ has $|T|$ iterations, each of which processes a different vertex $s\in T$. When vertex $s\in T$ is processed in the $i$th phase, we apply the algorithm from \Cref{thm: main subroutine} to graph $G$, set $T$ of terminals and vertex $s$. We supply the algorithm with the 2 copies of the adjacency-list representation of $G'$. We execute the algorithm until its $i$th call to the $\delta$-subgraph oracle, and we denote by $Z_{i,s}$ the resulting query to the oracle. We also record all random choices made by the algorithm, and, when the algorithm from  \Cref{thm: main subroutine} is executed with vertex $s$ as an input in Phase $(i+1)$, we repeat all the recorded random choices, so that the first $i$ calls to the $\delta$-subgraph oracle are identical to the execution of the algorithm in Phase $i$. Lastly, once the algorithm from 
\Cref{thm: main subroutine} performs the $i$th call to the $\delta$-subgraph oracle, we terminate it, and we undo all the modifications that it made to the two copies of graph $G'$ in the modified adjacency-list representation. We then continue to process the next vertex of $T$.

Once all vertices of $T$ are processed, we execute the algorithm from \Cref{lem: oracle in bulk} in order to compute responses to all queries $\set{Z_{i,s}\mid s\in T}$. We then continue to Phase $(i+1)$. Note that, at the beginning of Phase $(i+1)$, for all $s\in T$, we have computed the responses to the first $i$ queries to the $\delta$-subgraph oracle that the algorithm from \Cref{thm: main subroutine} asks when applied to $s$.

Let $\event'_T$ be the bad event that the algorithm from \Cref{lem: oracle in bulk} errs in any of the iterations. Then, from \Cref{lem: oracle in bulk} and the Union Bound, $\prob{\event'_T}\leq \frac{8\log n\cdot \log (\wmax)}{n^9\cdot\log^4(\wmax)}\leq \frac{8}{n^9\log^2(\wmax)}$, since $\wmax\geq n$ by definition.

After the last phase terminates, we obtain, for every vertex $s\in T$, the value $c_s$ that the algorithm from \Cref{thm: main subroutine} outputs. Recall that $c_s$ is at least as high as the value of the minimum $s$-$T_s$ vertex cut in $G$. We let $x\in T$ be the vertex minimizing the value $c_x$, and we let $y$ be any vertex of $T_x$. We compute the minimum $x$-$y$ vertex cut in $G$ using the algorithm from \Cref{cor: min_vertex_cut}, whose running time is ${O}\left (m^{1+o(1)}\cdot \log(\wmax)\right )\leq {O}\left (n^{2+o(1)}\cdot \log(\wmax)\right )$, and return the value $c^*_x$ of the cut, and the vertices $x$ and $y$. Clearly, $c^*_x\ge \opt$ must hold. 

Assume now that the cut $(L,S,R)$ is acceptable, so for all $0\leq i\leq \ceil{\log W'}$, $|S_i|\geq n^{44/45}\cdot |L_i|$ holds. Assume further that Event $\event_T$ happened. Let $s$ be the unique vertex in $T\cap L$, and let $\event''_s$ be the good event that the algorithm from \Cref{thm: main subroutine}, when applied to $s$, returned a value $c_s$ that is equal to the value of the minimum $s$-$T_s$ vertex cut in $G$. 
Then, from \Cref{thm: main subroutine}, $\prob{\event''_s\mid \event_T\band \neg\event'_T}\geq \frac{1}{2}$. From our discussion, if Events $\event_T$ and $\event''_s$ happen, and Event $\event'_T$ does not happen, then the algorithm is guaranteed to return an estimate $c^*_x=\opt$. Therefore, the probability that $c^*_x=\opt$ is at least:

\[\begin{split}
\prob{\event_T\band\event''_s\band\neg\event'_T}&=\prob{\event''_s\mid \event_T\band \neg\event'_T}\cdot \prob{\neg \event'_T\mid \event_T}\cdot \prob{\event_T}\\
&=\prob{\event''_s\mid \event_T\band \neg\event'_T}\cdot \prob{\neg \event'_T}\cdot \prob{\event_T}\\
&\geq \frac{1}{2}\cdot \left (1-\frac{8}{n^{9}\cdot\log^2(\wmax)}\right )\cdot \frac{1}{2^{11}\cdot \log n \cdot \log (W')}\\
&\geq \frac{1}{2^{13}\cdot \log n\cdot \log (W')}.
\end{split}\] 


We now bound the running time of the algorithm. Recall that the running time of the algorithm from \Cref{thm: main subroutine} is $O\left (n^{86/45+o(1)}\cdot (\log(\wmax))^{O(1)}\right )$. We execute this algorithm for $|T|\leq n$ terminals, and for each terminal $s\in T$, we excute the algorithm $O(\log n\cdot \log \wmax)$ times (though some of these executions are incomplete). Therefore, the total time that we spend on executing the algorithm from \Cref{thm: main subroutine} is 
$O\left (n^{3-4/45+o(1)}\cdot (\log(\wmax))^{O(1)}\right )$.
Additionally, we perform $O(\log n\cdot \log \wmax)$ calls to the algorithm from \Cref{lem: oracle in bulk} for the $\delta$-subgraph oracle, where $\delta=\frac{4}{45}-o(1)$. From \Cref{lem: oracle in bulk}, the total running time of all such calls is bounded by:

\[O\left (n^{3-0.478\delta+o(1)}\cdot \poly\log(\wmax)\right )\leq O\left ( n^{2.96}\cdot (\log(\wmax))^{O(1)}\right ).\]

The running time of Algorithm $\algterm$ is bounded by $O(n\log(\wmax))$, and so the total running time of the entire algorithm is  bounded by $O\left ( n^{2.96}\cdot (\log(\wmax))^{O(1)}\right )$.

In order to complete the proof of \Cref{thm: alg 3} and of \Cref{thm: main}, it is now enough to prove {\Cref{thm: main subroutine}, which we do in the remainder of the paper.

\section{Main Subroutine: the Proof of \Cref{thm: main subroutine}}
	\label{sec: main subroutine}
	
	
We assume that we are given as input a simple undirected $n$-vertex graph $G$ with integral weights $1\leq w(v)\leq W'$ on vertices $v\in V(G)$, a subset $T\subseteq V(G)$ of vertices that we call terminals, and a vertex $s\in T$. Additionally, the algorithm is given two copies of the split graph $G'$ corresponding to $G$ in the modified adjacency-list representation. 
For every edge $e\in E(G')$, its capacity in $G'$ is denoted by $c(e)$. Recall that we denoted by $\wmax'$ the smallest integral power of $2$ with $\wmax'>\max_{v\in V(G)}\set{w(v)}$, and by $\wmax$ the smallest integral power of $2$ with $\wmax\geq 2n\cdot \wmax'$. For every edge $e\in E(G')$, if $e$ is a special edge then $c(e)< \wmax'$, and if it is a regular edge, then $c(e)=\wmax$ hold.

For the sake of the analysis, we fix a minimum vertex-cut $(L,S,R)$ in $G$, as follows. 
If there exists a minimum vertex-cut $(L',S',R')$ with $|S'|\geq n^{44/45}\cdot |L'|$, such that $(s,T)$ is a good pair with respect to $(L',S',R')$, then we set $(L,S,R)=(L',S',R')$. Otherwise, we let  $(L,S,R)$ be any global minimum vertex-cut in $G$. For brevity, if $(s,T)$ is a good pair with respect to this fixed cut $(L,S,R)$, we will simply say that $(s,T)$ is a \emph{good pair}; in particular, $|S|\ge n^{44/45}\cdot |L|$ must hold in this case.
Recall that we have denoted by $T_s=T\setminus\left(\set{s}\cup N_G(s)\right )$, and, if $(s,T)$ is a good pair, then $T_s\subseteq R$ holds.

We start by modifying the graph $G'$ as follows. We add a new destination vertex $t$ to it, and, for every vertex $v\in T_s$, we add the edge $(v^{\inn},t)$ of capacity $\wmax$ to the graph. These newly added edges are considered regular edges. Notice that this modification can be done in time $O(n)$, and that the value of the 
minimum $s^{\out}$-$t$ edge-cut in $G'$ is equal to the value of the minimum $s$-$T_s$ vertex-cut in $G$. Moreover, if $(s,T)$ is a good pair, then the set $E_S=\set{(v^{\inn},v^{\out})\mid v\in S}$ of edges of $G'$ is the minimum $s^{\out}$-$t$ edge-cut in $G'$, and $c(E_S)=\opt$, the value of the global minimum vertex-cut in $G$.

We denote by $\opt_s$ the value of the minimum $s^{\out}$-$t$ edge-cut in this modified graph $G'$.
Our goal now is to output a value $c_s\geq \opt_s$. Additionally, if $(s,T)$ is a good pair, then we need to ensure that, with probability at least $\half$, $c_s=\opt_s$ holds.

Over the course of the algorithm, we may modify the graph $G'$ by performing \emph{shortcut operations}, that we define next.

\paragraph{Shortcut Operations.}
In a shortcut operation, we add an edge $(v,t)$ of capacity $\wmax$ to graph $G'$, for some vertex $v\in V(G')$. All such newly added edges are considered regular edges. If $(s,T)$ is a good pair, then the shortcut operation in which an edge $(v,t)$  is inserted into $G'$ is a \emph{valid shortcut operation}, and edge $(v,t)$ is a \emph{valid shortcut edge}, if and only if either (i) $v$ is a copy of a vertex $u\in R$; or (ii) $v$ is an out-copy of a vertex $u'\in S$.  If the pair $(s,T)$ is not good, then any such shortcut operation is a valid shortcut operation, and any edge $(v,t)$ is a valid shortcut edge.  We will use the following claim that summarizes the properties of the graph obtained from $G'$ via a sequence of shortcut operations.

\begin{claim}\label{claim: shortcut}
	Let $G''$ be a graph that is obtained from $G'$ via a sequence of shortcut operations. Then the value of the minimum $s^{\out}$-$t$ edge-cut in $G''$ is at least $\opt_s$. Moreover, if $(s,T)$ is a good pair, and all shortcut operations in the sequence are valid, then the value of the minimum $s^{\out}$-$t$ edge-cut in $G''$ is exactly $\opt_s=\opt$.
\end{claim}
\begin{proof}
From the Max-Flow/Min-Cut theorem, the value of the minimum $s^{\out}$-$t$ edge-cut in $G'$ is equal to the value of the maximum $s^{\out}$-$t$ flow in $G'$. The insertion of edges into $G'$ may not decrease the value of the maximum $s^{\out}$-$t$ flow. Therefore, the value of the minimum $s^{\out}$-$t$ edge-cut in $G''$ is at least $\opt_s$.

	Assume now that $(s,T)$ is a good pair, and that all shortcut opreations in the sequence are valid. Let $E_S=\set{(v^{\inn},v^{\out})\mid  v\in S}$ be the set of special edges in graph $G'$ that represent the vertices of $S$.
	Let $X=L^*\cup S^{\inn}$ and let $Y=V(G')\setminus X$. Note that $(X,Y)$ is an $s^{\out}$-$t$ cut in $G'$. Clearly, $E_{G'}(X,Y)=E_S$, and so the edges of $E_S$ separate the vertices of $X$ from the vertices of $Y$ in $G'$. Moreover, since $(L,S,R)$ is a  minimum vertex-cut in $G$,  $\sum_{e\in E_S}c(e)=\sum_{v\in S}w(v)=\opt_s=\opt$. 
	Since $s^{\out}\in X$ and $t\in Y$, and since the valid shortcut edges that we add all connect vertices of $Y$ to $t$, set $E_S$ of edges also separates $X$ from $Y$ in $G''$. Therefore, the value of the minimum $s^{\out}$-$t$ edge-cut in $G''$ is $\sum_{e\in E_S}c(e)=\opt_s=\opt$.
\end{proof}

For convenience, we denote by $G''$ the graph that our algorithm maintains by starting from $G'$, and then gradually adding shortcut edges to it, while graph $G'$ remains fixed throughout the algorithm. We use the first copy of the modified adjacency-list representation of $G'$ that our algorithm receives as input in order to maintain $G''$. 
We mark in the modified adjacency-list representation of $G''$ every vertex $v$ for which a shortcut edge $(v,t)$ was added to $G''$.
Our algorithm also maintains an $s^{\out}$-$t$ flow $f$ in $G''$, by starting with $f=0$, and then gradually augmenting it. 
Specifically, we maintain a set $E_f=\set{e\in E(G'')\mid f(e)>0}$ of edges  in an efficient search structure (such as a Binary Search Tree), and, for every edge $e\in E_f$, we list the current flow value $f(e)$. For convenience, we denote this data structure by $\DS_f$.
Lastly, we maintain a graph $H=G''_f$, the residual flow network of $G''$ with respect to $f$. Initially, when $G''=G'$ and $f=0$, graph $H$ is identical to $G'$. We use the second copy of the modified adjacency-list representation of $G'$ in order to maintain $H$. 
For every edge $e\in E(H)$, its capacity in $H$ is denoted by $\hat c(e)$; we sometimes refer to $\hat c(e)$ as the \emph{residual capacity} of $e$.
Recall that, for an edge $e=(x,y)\in E(G'')$, if $f(e)<c(e)$, then a copy of $e$ (that we refer to as its \emph{forward copy}) is present in $H$, and $\hat c(e)=c(e)-f(e)$. Additionally, if $f(e)>0$, then an anti-parallel edge $(y,x)$ (that we refer to as the \emph{backward copy of $e$}) is present in $H$, with $\hat c(y,x)=f(e)$. If edge $e\in E(G'')$ is special, then all its copies in $H$ are special edges as well; otherwise, the copies of $e$ in $H$ are regular edges.
We emphasize that $\opt_s$ is the value of the maximum $s^{\out}$-$t$ flow in the initial graph $G'$; the value of the $s^{\out}$-$t$ flow in $G''$ may be higher.
Next, we define several parameters that are used throughout the algorithm.

\paragraph{Parameters.}
Throughout the algorithm, we use a parameter $\eps=\frac{1}{45}$.
Recall that, if $(s,T)$ is a good pair, then $|S|\geq n^{44/45}\cdot |L|=n^{1-\eps}\cdot |L|$ must hold, and so:

\begin{equation}\label{eq: small L}
|L|\leq \frac{|S|}{n^{1-\eps}}\leq n^{\eps}
\end{equation}

We assume that $n$ is greater than a sufficiently large constant, so that, for example:

\begin{equation}\label{eq: eps n} 
\eps>\frac{1}{\sqrt{\log n}}
\end{equation}
holds. Additionally, we assume that: 

\begin{equation}\label{eq: neps large}
n^{\eps/10}>2^{20}\log^2 n \cdot \log^2 (\wmax).
\end{equation}

Notice that, if either of these assumptions does not hold, then we can use the algorithm from \Cref{thm: maxflow}, in order to compute the value $\opt_s$ of the maximum $s^{\out}$-$t$ flow in $G'$, which is equal to the value of the minimum $s$-$T_s$ vertex-cut in $G$, as observed above. The running time of the algorithm in this case is $O\left (n^{2+o(1)}\cdot \log(\wmax)\right )\leq O\left(n \cdot (\log(\wmax))^{O(1)}\right )$, as required. Therefore, we assume from now on that $n$ is large enough, and that, in particular, Inequalities \ref{eq: eps n}  and \ref{eq: neps large} hold.
Our second main parameter is $\gamma$, defined to be the smallest integral power of $2$, so that $\gamma\geq n^{7/9}$ holds. It is immediate to verify that:

\begin{equation}\label{eq: gamma}
n^{7/9}\leq \gamma\leq 2n^{7/9}
\end{equation}

and so:

\begin{equation}\label{eq: gamma2}
n^{2/3+5\eps}\leq \gamma\leq 2n^{1-10\eps},
\end{equation}

and additionally:

\begin{equation}\label{eq: gamma3}
\gamma\geq n^{14\eps}.
\end{equation}

Our last main parameter is $\delta=4\eps-\frac{\log(4000\log n)}{\log n}=\frac{4}{45}-\frac{\log(4000\log n)}{\log n}=\frac{4}{45}-o(1)$. 
This parameter is chosen so that:

\begin{equation}\label{eq: delta}
n^{\delta}=\frac{n^{4\eps}}{4000\log n}
\end{equation}

holds,
and it will be used as the parameter for the $\delta$-subgraph oracle. Notice that from our assumptions:

\begin{equation}\label{ineq: delta2}
\delta\geq\frac{1}{\sqrt{\log n}}
\end{equation}

holds, as required by the algorithm for the subgraph oracle from \Cref{lem: oracle in bulk}.

Our algorithm consists of at most $z=\ceil{\log \left(\frac{8\wmax'\cdot n^2}{\gamma}\right )}\leq 4\log n\cdot \log (\wmax')$ phases. 
For all $0\leq i\leq z$, we let $M_i=\frac{\wmax'}{2^i\cdot \gamma}$. Notice that $M_0=\frac{\wmax'}{\gamma}$, and $M_z\leq \frac{1}{8n^2}$, and for all $i$, $M_i$ is an integral power of $2$.

Our algorithm may terminate at any time with a ``FAIL''; in this case we output $c_s=\wmax$. For all $i\geq 1$, as long as the algorithm did not terminate with a ``FAIL'', at the beginning of the $i$th phase, it computes an $s^{\out}$-$t$ flow $f_{i-1}$ in the current graph $G''$, such that the following invariants hold:

\begin{properties}{I}
	\item With probability at least $1-\frac{2i}{n^{\eps}\cdot \log^2(\wmax)}$, all shortcut operation performed before the start of Phase $i$ are valid; \label{inv: shortcut}

	\item If $(s,T)$ is a good pair, then the probability that the algorithm terminates prior to the beginning of Phase $i$ with a ``FAIL'' is at most $\frac{10i^2-2i}{n^{\eps}\cdot \log^2(\wmax)}$. \label{inv: fail prob}

	\item Flow $f_{i-1}$ is $M_{i-1}$-integral, and $\val(f_{i-1})\geq \opt_s-4nM_{i-1}$;  \label{inv: flow} 
 
\item The total number of edges $e\in E(G'')$ with $f_{i-1}(e)>0$ is at most $O\left(i\cdot n^{2-4\eps+o(1)}\cdot \log^4(\wmax)\right )$; and \label{inv: few edges carry flow}

	\item  For every vertex $v\in V(G)$ with $w(v)\geq M_{i-1}$, the flow $f_{i-1}$ on the corresponding special edge $(v^{\inn},v^{\out})$ is at most $w(v)-M_{i-1}$. \label{inv: residual cap large}
\end{properties}

At the end of the last phase, the algorithm produces a flow $f_z$, for which Invariants \ref{inv: shortcut}--\ref{inv: residual cap large} hold as well. For consistency of notation, we will sometimes view the end of Phase $z$ as the beginning of Phase $(z+1)$, even though the algorithm does not execute Phase $(z+1)$.

In addition to data structure $\DS_f$ for the current flow $f$, our  algorithm also maintains the current graph $G''$ and the residual flow network $H$ of $G''$ with respect to $f$, both in the modified adjacency-list representation. We note that shortcut operations may be performed over the course of each phase, and during the preprocessing step.
In every phase, we make at most one call to the $\delta$-subgraph oracle, and we make an additional call to the oracle in the preprocessing step. 
We ensure that the running time of every phase, and of the preprocessing step, is bounded by $O\left(n^{2-4\eps+o(1)}\cdot (\log(\wmax))^{O(1)}\right )$.

\paragraph{Bad events $\event_i$.}
For all $1\leq i\leq z+1$, we let $\event_i$ be the bad event that either (i) at least one shortcut operation executed before the beginning of Phase $i$ was invalid; or (ii) the algorithm terminated prior to the beginning of Phase $i$ with a ``FAIL''.
 From the above invariants, if $(s,T)$ is a good pair, then, for all $1\leq i\leq z+1$:

\begin{equation}\label{eq: bad event i}
\prob{\event_i}\leq \frac{10i^2}{n^{\eps}\cdot \log^2(\wmax)}.
\end{equation}

We start by providing an algorithm for the preprocessing step, whose purpose is to compute the initial flow $f_0$, that will serve as input to the first phase, together with the corresponding residual flow network.

\subsection{The Preprocessing Step}
\label{subsec: preprocessing}

Our algorithm starts with $G''=G'$.
The purpose of the preprocessing step is to compute an initial $s^{\out}$-$t$ flow $f_0$ in $G''$ (after possibly performing some shortcut operations in $G''$), so that Invariants \ref{inv: shortcut}--\ref{inv: residual cap large} hold at the  beginning of Phase 1. We will also compute the residual flow network $H=G''_{f_0}$ in the modified adjacency list representation, and update the data structure $\DS_f$ to represent the flow $f=f_0$. Flow $f_0$ and graphs $G''$ and $H$ will then serve as input to the first phase.

Let $\hat U_0=\set{v\in V(G)\mid w(v)\geq  M_0}$, and let $Z=\hat U_0\setminus \left(\set{s}\cup N_G(s)\right )$. Note that we can compute $\hat U_0$ and $Z_0$ in time $O(n)$ using the adjacency-list representation of $G$.

We start with intuition.
Assume that $(s,T)$ is a good pair, so $s\in L$ holds. Then, from \Cref{cor:  heavy neighbors in S}, $|Z\cap S|\leq |L|\cdot \gamma$, and so at most $|L|\cdot \gamma+|L|\leq 2|L|\cdot \gamma$ vertices of $Z$ may lie in $S\cup L$.
Consider now any vertex $v\in \hat U_0$. We say that $v$ is a \emph{bad vertex}, if $|N_G(v)\cap Z|>2|L|\cdot\gamma$, and otherwise we say that it is a \emph{good vertex}. Notice that, if $v$ is a bad vertex, then it must have a neighbor $u\in N_G(v)$ that lies in $Z\setminus (S\cup L)\subseteq R$, and therefore $v$ may not lie in $L$. Adding the edge $(v^{\out},t)$ to graph $G''$ is then a valid shortcut operation.

We will construct a sparsified subgraph $\hat G\subseteq G''$, as follows. We start by letting $\hat G$ be obtained from $G''$ by deleting from it all vertices in $\set{v^{\inn},v^{\out}\mid v\in V(G)\setminus \left(\hat U_0\cup \set{s}\right )}$. We then decrease the capacity of every special edge $e$ in the resulting graph by at least $M_0$ and at most $2M_0$ units, so that the resulting new capacity $c'(e)$ satisfies $\max\set{c(e)-2M_0,0}\leq c'(e)\leq c(e)-M_0$, and it is an integral multiple of $M_0$ (which may be equal to $0$). 
Next, for every vertex $v\in \hat U_0$, if $v$ is a bad vertex, then we add the shortcut edge $(v^{\out},t)$ to $\hat G$ (and to $G''$), and we delete all other edges leaving $v^{\out}$ from $\hat G$. If $v$ is a good vertex, then we delete all edges connecting $v^{\out}$ to the in-copies of the vertices in  $N_G(s)$. Consider now the resulting graph $\hat G$. On the one hand, this graph is quite sparse: it only contains copies of vertices $v\in \hat U_0\cup\set{s}$ (in addition to the vertex $t$), and, for each such vertex $v$, its in-copy has $1$ outgoing edge, and its out-copy has at most $2|L|\cdot \gamma$ outgoing edges in $\hat G$, so $|E(\hat G)|\leq O(n\cdot |L|\cdot \gamma)$ holds. On the other hand, it is not hard to see (and we prove it formally below) that there is an $s^{\out}$-$t$ flow in $\hat G$ of value at least $\opt_s-4n\cdot M_0$. Since $M_0$ and $\wmax'$ are both powers of $2$, we can ensure that this flow is $M_0$-integral, and it must satisfy Invariant \ref{inv: residual cap large}, since we decreased the capacity of every special edge in $\hat G$ by at least $M_0$ flow units.
 By computing the maximum $s^{\out}$-$t$ flow in graph $\hat G$ we  then obtain the desired flow $f_0$. 

Our algorithm follows the above high-level plan, except for minor changes that are required in order to ensure that it is efficient. There are two main difficulties with the plan we have described: (i) computing the set of all bad vertices efficiently; and (ii) constructing the graph $\hat G$ efficiently. We overcome both these challenges by employing the subgraph oracle.

We now describe the preprocessing step formally. Recall that we denoted by $\hat U_0\subseteq V(G)$ the set of all vertices $v\in V(G)$ with $w(v)\geq M_0=\frac{\wmax'}{\gamma}$, and we denoted by $Z=\hat U_0\setminus  \left(\set{s}\cup N_G(s)\right )$. 
We start by querying the $\delta$-subgraph oracle with the graph $G$ and the set $Z$ of its vertices. Let $(Y^{h},Y^{\ell},E')$ denote the response of the  oracle. Recall that $|E'|\leq n^{2-\delta}\cdot \log^2(\wmax)\leq O\left(n^{2-4\eps}\cdot \log n\cdot \log^2(\wmax)\right )$ holds (since $n^{\delta}=\frac{n^{4\eps}}{4000\log n}$ from Equation \ref{eq: delta}), and that $E'=E_G(Y^{\ell},Z)$.
Let $\event'_0$ be the bad event that the oracle erred. From the definition of the $\delta$-subgraph oracle, the probability that it errs is at most $\frac{1}{n^{5}\cdot\log^4(\wmax)}$, so $\prob{\event'_0}\leq \frac{1}{n^{5}\cdot\log^4(\wmax)}$.
We use the following observation.

\begin{observation}\label{obs: no error degree est}
If Event $\event'_0$ did not happen,  and if $(s,T)$ is a good pair, then $Y^h\subseteq S\cup R$ must hold.
\end{observation}
\begin{proof}
	Assume that Event $\event'_0$ did not happen, and that $(s,T)$ is a good pair. Consider some vertex $v\in Y^h$.
	Since we have assumed that the oracle did not err: 
	
	\[\begin{split} 
	|N_G(v)\cap Z|&\geq \frac{n^{1-\delta}}{1000\log n}\\
	&=4n^{1-4\eps}\\
	&> 2n^{\eps}\cdot \gamma\\
	&\geq 2|L|\cdot \gamma
	\end{split}
	\]

since $n^{\delta}=\frac{n^{4\eps}}{4000\log n}$ from Equation \ref{eq: delta}, $\gamma< 2n^{1-5\eps}$ from Inequality \ref{eq: gamma2}, and $|L|\leq n^{\eps}$ from  Inequality \ref{eq: small L}.
	
Observe that, from \Cref{cor:  heavy neighbors in S}, since all vertex weights in $G$ are bounded by $\wmax'$ and $M_0=\frac{\wmax'}{\gamma}$,
$|Z\cap S|\leq |L|\cdot \gamma$, and so at most $|L|\cdot \gamma+|L|< 2|L|\cdot \gamma$ vertices of $Z$ may lie in $S\cup L$. Since $v$ has at least $2|L|\cdot \gamma$ neighbors in $Z$, there is some vertex $u\in Z\setminus (S\cup L)$, that is a neighbor of $v$ in $G$. But then $u\in R$ must hold, and so $v\in S\cup R$. We conclude that $Y^h\subseteq S\cup R$.
\end{proof}

We perform the following sequence of shortcut operations: for every vertex $v\in Y^h$, we add the edge $(v^{\out},t)$ to $G''$, and we modify the adjacency-list representation of $G''$ accordingly. From \Cref{obs: no error degree est}, if bad event $\event'_0$ did not happen, then all shortcut operations that we perform are valid. From now on the algorithm for the preprocessing step may not terminate with a ``FAIL'', and it will not perform any additional shortcut operations. Since $\prob{\event'_0}\leq \frac{1}{n^5\cdot \log^4(\wmax)}$, this ensures that Invariants \ref{inv: shortcut} and \ref{inv: fail prob} will hold at the beginning of Phase 1.

Next, we construct a subgraph $\hat G\subseteq G''$. We will show that $|E(\hat G)|$ is sufficiently small, and we will provide an algorithm that constructs $\hat G$ efficiently. We will also show that the value of the maximum $s^{\out}$-$t$ flow in $\hat G$ is at least $\opt_s-4n\cdot M'$. Lastly, we will compute a maximum $s^{\out}$-$t$ flow in $\hat G$ that is $M_0$-integral, obtaining the desired flow $f_0$, and the corresponding residual flow network $H$. We start by defining the graph $\hat G$.

\paragraph{Definition of $\hat G$.}
We start by setting $\hat G=G''$. We then delete from $\hat G$ the copies of all vertices in $V(G)\setminus \left(\hat U_0\cup \set{s}\right)$, and all edges that are incident to them. We reduce the capacity of each remaining special edge $e$ by at least $M_0$ and at most $2M_0$, so that the resulting new capacity $c'(e)$ satisfies $\max\set{c(e)-2M_0,0}\leq c'(e)\leq c(e)-M_0$, and it is an integral multiple of $M_0$ (which may be equal to $0$). We denote by $\edel_1\subseteq E(G'')$ all edges that were deleted from $\hat G$ in this step. If $v$ is a vertex of $\hat U_0\cup \set{s}$, then we imagine splitting the special edge $e_v=(v^{\inn},v^{\out})$ of $G''$ into two copies. The first copy, denoted by $e'_v$, has capacity  $c(e_v)-c'(e_v)$ and is added to $\edel_1$; the second copy, denoted by $e''_v$, has capacity $c'(e_v)$, and is viewed as the special edge representing $v$ in graph $\hat G$.

 Next, for every vertex $v\in Y^h\setminus\set{s}$, we delete all edges leaving $v^{\out}$ from the resulting graph, except for the edge $(v^{\out},t)$. We denote the set of all edges deleted in this step by $\edel_2$. Lastly, for every vertex $v\in \hat U_0\setminus \set{s}$, we delete all edges connecting $v^{\out}$ to the vertices in the set $\set{u^{\inn}\mid u\in N_G(s)\cup\set{s}}$, and we denote by $\edel_3$ the set of all edges deleted in this step. This completes the definition of the graph $\hat G$. Let $\edel=\edel_1\cup \edel_2\cup \edel_3$. Then $E(\hat G)=E(G'')\setminus \edel$.

\paragraph{Efficient Construction of $\hat G$.}
We now provide an efficient algorithm for constructing the graph $\hat G$.
 First, for every vertex $v\in \hat U_0\cup \set{s}$, we include vertices $v^{\inn},v^{\out}$, and the special edge $e''_v=(v^{\inn},v^{\out})$  of capacity $c'(e_v)$ as defined above. We also add vertex $t$, and all regular edges that are incident to $t$ or to $s^{\out}$ in $G''$, whose other endpoint was already added to $\hat G$. Lastly, for every edge $e=(u,v)\in E'$ (here, $E'$ is the set of edges returned by the subgraph oracle), if $u\not\in Y^h$, then we add the regular edge $(u^{\out},v^{\inn})$ to $\hat G$, and, if $v\not\in Y^h$, then we add the regular edge $(v^{\out},u^{\inn})$ to $\hat G$. The capacities of all these edges are the same as in $G''$. It is easy to verify that the graph that we obtain is indeed $\hat G$, that $|E(\hat G)|\leq 
O(n)+O(|E'|)\leq  O\left(n^{2-4\eps}\cdot \log n\cdot \log^2(\wmax)\right )$, and that $\hat G$ can be computed in time $O(n)+O(|E'|)\leq  O\left(n^{2-4\eps}\cdot \log n\cdot \log^2(\wmax)\right )$.
Notice that the capacity of every edge in $\hat G$ is an integral multiple of $M_0$.

\paragraph{Maximum Flow in $\hat G$.}

We compute the maximum $s^{\out}$-$t$ flow in $\hat G$ , using the algorithm from \Cref{thm: maxflow}, where the resulting flow, denoted by $f_0$, is given via the edge-representation, and it is $M_0$-integral (since all edge capacities in $\hat G$ are $M_0$-integral).
The running time of the algorithm is:

\[O\left (|E(\hat G)|^{1+o(1)} \cdot \log^2(\wmax)\right )\leq O\left(n^{2-4\eps+o(1)}\cdot \log^5(\wmax)\right ).\]

We then update the modified adjacency-list representation of $H$, so that it becomes a valid modified adjecency-list representation of the residual flow network of $G''$ with respect to $f_0$, and the data structure $\DS_f$ to represent the flow $f_0$. It is easy to verfy that both can be done in time $O\left(n^{2-4\eps+o(1)}\cdot \log^4(\wmax)\right )$. 
This completes the description of the algorithm for the preprocessing step. We now complete its analysis.

\paragraph{Running time analysis.}
The sets $\hat U_0$ and $Z$ of vertices of $G$ can be computed in time $O(n)$ using the adjacency list representation of graph $G$. The time that is required to compute the graph $\hat G$, to compute the flow $f_0$, and to update the residual flow network $H$ is bounded by $O\left(n^{2-4\eps+o(1)}\cdot \log^5(\wmax)\right )$. Therefore, the total running time of the preprocessing step is: 

$$O\left(n^{2-4\eps+o(1)}\cdot \log^5(\wmax)\right ).$$

We have already established that, at the beginning of Phase 1, Invariants \ref{inv: shortcut} and \ref{inv: fail prob} hold. 
Invariant \ref{inv: residual cap large} follows immediately from the definition of the flow network $\hat G$, since for every vertex $v\in M_0$, the capacity of its corresponding edge $(v^{\inn},v^{\out})$ in $\hat G$ is set to be at most $w(v)-M_0$.
Notice that the total number of edges $e\in E(G'')$ with $f_0(e)>0$ is bounded by $|E(\hat G)|\leq O\left( n^{2-4\eps+o(1)}\cdot \log^4(\wmax)\right )$, so Invariant  \ref{inv: few edges carry flow} holds as well.
We use the following claim to prove that Invariant \ref{inv: flow} also holds.

\begin{claim}\label{claim: preproc: flow value}
	$\val(f_0)\ge \opt_s-4n\cdot M_0$.
\end{claim}
\begin{proof}
	Consider the graph $G''$ that is obtained at the end of the preprocessing step. For convenience, we take again the view that, for every vertex $v\in  \hat U_0$, its corresponding special edge $e_v=(v^{\inn},v^{\out})$ in $G''$ is replaced by two parallel copies: copy $e'_v$ of capacity  $c(e_v)-c'(e_v)\leq 2M_0$, and copy $e''_v$ of capacity $c'(e_v)$. This splitting of special edges does not change the value of the maximum $s^{\out}$-$t$ flow in $G''$.

	From \Cref{claim: shortcut}, there is an $s^{\out}$-$t$ flow $f$ in $G''$ of value at least $\opt_s$. Consider a flow-path decomposition of $f$, and let $\pset$ be the resulting collection of paths that carry non-zero flow. 
	We can assume w.l.o.g. that every path $P\in \pset$ contains at most one vertex in set $B=\set{v^{\inn}\mid v\in N_G(s)}$, and that this vertex is the second vertex on the path. Indeed, if this is not the case, then we can ``shortcut'' the path $P$ by connecting $s^{\out}$ directly to the last vertex of $B$ on the path $P$. Therefore, in particular, the paths in $\pset$ may not contain edges of $\edel_3$.
	
	We partition $\pset$ into two subsets: set $\pset_1$ contains all paths $P$ with $P\subseteq \hat G$, and $\pset_2$ contains all remaining paths.
	
	Consider now some path $P\in \pset_2$. Since $P\not\subseteq \hat G$, at least one edge of $P$ must lie in $\edel$. Let $e(P)$ be the first such edge on $P$. As discussed, already, $e(P)\not \in \edel_3$, so either $e(P)\in \edel_1$, or $e(P)\in \edel_2$ must hold. We further partition $\pset_2$ into two subsets $\pset_2^1$ and $\pset_2^2$, where a path $P\in \pset_2$ belongs to the former set if $e(P)\in \edel_1$, and to the latter set otherwise.
	
	Consider first a path $P\in \pset_2^2$, and recall that the corresponding edge $e(P)$ must originate at a vertex $v^{\out}$ for some $v\in Y^h$. Moreover, our algorithm inserted the edge $(v^{\out},t)$ of capacity $\wmax$ into $G''$, and this edge lies in $\hat G$ as well. We replace the path $P$ with a new path $P'\subseteq \hat G$, that follows $P$ from $s^{\out}$ to $v^{\out}$, and then uses the edge $(v^{\out},t)$. We send $f(P')=f(P)$ flow on path $P'$. By using the flow-paths in $\pset_1$ and in $\set{P'\mid P\in\pset_2^2}$, we obtain a valid $s^{\out}$-$t$ flow in $\hat G$, whose value is $\val(f)-\sum_{P\in \pset_2^1}f(P)$.
	
	It is easy to verify that $|\edel_1|\leq n$, since the set $\edel_1$ contains one special edge for each vertex $v\in V(G)$. Moreover, the capacity of every edges in $\edel_1$ is bounded by $2M_0$. Therefore, $\sum_{P\in \pset_2^1}f(P)\leq 2n\cdot M_0$, and so, from the above discussion, the value of the maximum $s^{\out}$-$t$ flow in $\hat G$ is at least $\val(f)-\sum_{P\in \pset_2^1}f(P)\geq \val(f)-2n\cdot M_0$.
\end{proof}

\subsection{Finishing the Algorithm}
\label{subsec: finishing}

So far we have provided an algorithm for the preprocessing step, that computes the initial flow $f_0$, and the corresponding residual flow network $H$, for which Invariants \ref{inv: shortcut}--\ref{inv: residual cap large} hold. The running time of the preprocessing step is $O\left(n^{2-4\eps+o(1)}\cdot \log^4(\wmax)\right )$.

The algorithm for implementing each phase is the most invovled technical part of our paper, and we provide it below. We assume for now that there exists an algortihm for implementing each phase, with the following properties:

\begin{itemize}
	\item At the end of each phase, Invariants \ref{inv: shortcut}--\ref{inv: residual cap large} hold; 
	
	\item The algorithm performs at most one call to the $\delta$-subgraph oracle in every phase; and
	
	\item The running time of the algorithm for a single phase is bounded by $O\left(n^{2-4\eps+o(1)}\cdot (\log(\wmax))^{O(1)}\right )$.
\end{itemize}

We provide an algorithm for executing each phase below, after we complete the proof of \Cref{thm: main subroutine} under these assumptions.

Recall that the number of phases in the algorithm is $z=\ceil{\log \left(\frac{8\wmax'\cdot n^2}{\gamma}\right )}\leq 4\log n\cdot \log (\wmax')$, and that $M_z\leq \frac{1}{8n^2}$.
The total running time of all $z$ phases and the preprocessing step is then bounded by:

\[O\left(z\cdot n^{2-4\eps+o(1)}\cdot (\log(\wmax))^{O(1)}\right )\leq O\left(n^{86/45+o(1)}\cdot (\log(\wmax))^{O(1)}\right ), \]

since $\eps=1/45$. The number of times the algorithm accesses the $\delta$-subgraph oracle is bounded by $z+1\leq 8\log n\cdot \log (\wmax)$.

If our algorithm ever terminates with a ``FAIL'', then we return $c_s=\wmax$. Clearly, $c_s$ is at least as high as the value of the minimum $s$-$T_s$ vertex cut in $G$. Let $\event'$ be the bad event that $(s,T)$ is a good pair, and the algorithm either terminated with a ``FAIL'', or at least one of the shortcut operations that it performed was invalid. From Invariants \ref{inv: shortcut}  and \ref{inv: fail prob}:

\[\prob{\event'}\leq \frac{10(z+1)^2}{n^{\eps}\cdot \log^2(\wmax)}\leq \frac{640\log^2 n}{n^{\eps}}\leq \frac{1}{16},\]

from Inequality \ref{eq: neps large}.
Notice that, from Invariant \ref{inv: flow}, $\val(f_{z})\geq \opt_s-4n\cdot M_{z}\geq \opt_s-\frac{1}{2n}$. Let $E'=E_{f_z}$ be the collection of all edges $e\in E(G'')$ with $f_z(e)>0$. From Invariant \ref{inv: few edges carry flow}, $|E'|\leq O\left(z\cdot n^{2-4\eps+o(1)}\cdot \log^4(\wmax)\right )\leq O\left( n^{2-4\eps+o(1)}\cdot \log^5(\wmax)\right )$.

We let $G^*\subseteq G''$ be the subgraph of $G''$ that only contains the edges of $E'$. The last step in our algorithm is to compute a maximum $s^{\out}$-$t$ flow in $G^*$, using the algorithm from \Cref{thm: maxflow}. We then output $c_s=\val(f^*)$. Observe that the running time of this last step is: 

\[O(|E'|^{1+o(1)}\cdot \log(\wmax))\leq O\left( n^{2-4\eps+o(1)}\cdot \log^7(\wmax)\right )\leq O\left( n^{86/45+o(1)}\cdot \log^7(\wmax)\right ),\]

since $\eps=1/45$, and so the total running time of the algorithm is 
$O\left(n^{86/45+o(1)}\cdot (\log(\wmax))^{O(1)}\right )$.

Since all edge capacities in $G^*$ are integral, flow $f^*$, and its value $\val(f^*)$ are integral. Since $\opt_s$ is an integer, and since
 $\val(f_z)\geq \opt_s-\frac{1}{2n}$, we get that $c_s=\val(f^*)\geq \ceil{\val(f_z)}\geq \opt_s$.
 
 Lastly, if $(s,T)$ is a good pair, and if Event $\event'$ did not happen, then all shortcut operations that the algorithm performed were valid, and so, from \Cref{claim: shortcut}, the value of the maximum $s^{\out}$-$t$ flow in $G''$, and hence in $G^*$, is bounded by $\opt_s$. Since $\prob{\event'}<\half$, we get that, if $(s,T)$ is a good pair, then, with probability at least $\half$, $c_s$ is equal to the value of the minimum $s$-$T_s$ vertex-cut in $G$.

It now remains to provide an algorithm for executing each phase, which we do next.

\subsection{The Algorithm for the $i$th Phase}
\label{subsec: one phase}

We now provide the description of the $i$th phase, for $1\leq i\leq z$.
For convenience, we denote $M'=M_i=\frac{\wmax'}{2^i\cdot \gamma}$, and we denote by $\hat U_i\subseteq V(G)$ the set of all vertices $v\in V(G)$ with $w(v)\geq M'$. 
 We assume that, at the beginning of the phase, Invariants \ref{inv: shortcut}--\ref{inv: residual cap large} hold. 
In particular, we are given an $s^{\out}$-$t$ flow $f_{i-1}$ in the current graph $G''$, that is $M_{i-1}=(2M')$-integral, and $\val(f_{i-1})\geq  \opt_s-4n\cdot M_{i-1}=\opt_s-8n\cdot M'$ holds. Additionally, we are given the residual flow network $H$ for the current graph $G''$, with respect to the flow $f_{i-1}$. From Invariant \ref{inv: residual cap large}, the residual flow network $H$ has the following property:

\begin{properties}{Q}
\item For every vertex $v\in V(G)$ with $w(v)\geq M_{i-1}=2M'$, there is a special edge $e=(v^{\inn},v^{\out})$ in $H$, whose residual capacity $\hat c(e)\geq M_{i-1}=2M'$. \label{prop: residual capacity}
\end{properties}

Notice that, in particular, from the above property, and since the flow $f_{i-1}$ is $2M'$-integral, for every vertex $v\in \hat U_i$, the residual capacity of the corresponding special edge 
$(v^{\inn},v^{\out})$  in $H$ is at least $M'$. 
We further partition $\hat U_i$ into two subsets: set $\hat U_i^h$ containing all vertices $v$ for wich the corresponding special edge $e=(v^{\inn},v^{\out})$ has residual capacity $\hat c(e)\geq M'\cdot \gamma$ in $H$, and set $\hat U_i^{\ell}=\hat U_i\setminus \hat U_i^h$.

Let $\Delta=\opt_s-\val(f_{i-1})$, so $\Delta\leq 8n\cdot M'$.
By the properties of the residual flow network, there is an $s^{\out}$-$t$ flow in $H$ of value $\Delta$. 
Moreover, if $(s,T)$ is a good pair, and if all shortcut operations so far were valid, then the value of the maximum $s^{\out}$-$t$ flow in $H$ is exactly $\Delta$.

We now provide the intuition for the algorithm for the $i$th phase. For the sake of the intuition, assume that $(s,T)$ is a good pair. 
The goal of the $i$th phase is to compute an $s^{\out}$--$t$ flow $f_{i}$ in $H$ of value at least $\Delta-4n\cdot M'$, after possibly performing some shortcut operations, such that the flow is $M'$-integral. In order to do so, like in the preprocessing step, we first sparsify the graph $H$, and then compute the flow in the sparsified graph. A major step towards the sparsification of $H$ is to compute a partition $(Q,Q')$ of $\hat U_i$ into two subsets, such that $Q'$ only contains relatively few vertices of $S\cup L$. Additionally, if $|Q|>n^{1-4\eps}$, then we will also compute an $s^{\out}$-$t$ cut $(\hat X,\hat Y)$ in $H$, such that, in any maximum $s^{\out}$-$t$ flow in $H$, the total flow over the edges in $E_H(\hat X,\hat Y)$ is close to their capacity, and for every vertex $v\in Q$, $v^{\inn}\in \hat X$ and $v^{\out}\in \hat Y$ holds. Intuitively, if a vertex $v\in \hat U_i$ has many edges connecting it to the vertices of $Q'$, then we are guaranteed that $v\in S\cup R$, and therefore, we can add a shortcut edge $(v^{\out},t)$, and delete all other edges leaving $v^{\out}$ in the sparsified graph. If a vertex $v\in \hat U_i$ only has few edges connecting it to vertices of $Q'$, then we can include all the corresponding edges of $H$ in the sparsified subgraph, but we will need to further sparsify the edges connecting copies of such vertices to the vertices of $Q^*=Q^{\inn}\cup Q^{\out}$, if $|Q|>n^{1-4\eps}$. For this last sparsification step, we exploit the cut $(\hat X,\hat Y)$ that our algorithm computes in this case. In order to compute the partition $(Q,Q')$ of $\hat U_i$ with the desired properties, we start by constructing an auxiliary graph $J\subseteq H$ with $s^{\out}\in V(J)$, $t\not\in V(J)$, and which has the additional property that for the vast majority of vertices $v\in S$, $v^{\inn}\in V(J)$ and $v^{\out}\not \in V(J)$ holds. We then let $\tilde H$ be a graph that is obtained from $H$ by contracting all vertices of $V(H)\setminus V(J)$ into a new sink vertex $t'$.
 We show that the desired partition $(Q,Q')$ of $\hat U_i$ can be computed by performing a series of min-cost flow computation in  $\tilde H$. However, in order to ensure that these min-cost flow computations can be executed efficiently, we need to ensure that the graph $\tilde H$ itself is sufficiently sparse. 
  Our algorithm for Phase $i$ then consists of three steps:

\begin{itemize}[leftmargin=1.34cm]
 \item[\emph{Step $1$.}]\customlabel{step1: local flow}{1}
 Perform iterative local flow computations, whose purpose is to ensure that the graph $\tilde H$ constructed at the end of this step is sufficiently sparse. During this step we may perform some shortcut operations, and augment the current flow $f_{i-1}$, obtaining a new $s^{\out}$-$t$ flow $f$ in $G''$. This step is described in  \Cref{sec: step 1}.
 \item[\emph{Step $2$.}]\customlabel{step2: partition}{2}
 Compute a partition $(Q,Q')$ of $\hat U_i$ by executing a number of min-cost flow computations in $\tilde H$, and compute the cut $(\hat X,\hat Y)$ in $\tilde H$ if $|Q|>n^{1-4\eps}$. This step is  described in \Cref{sec: step 2}.
 \item[\emph{Step $3$.}]\customlabel{step3: sparsify}{3}
Sparsify the graph $H$ and compute the maximum $s^{\out}$-$t$ flow in the resulting flow network. We also augment the current flow $f$, to obtain the final flow $f_i$. This step is described in \Cref{sec: step 3}.
\end{itemize}

%
\paragraph{Notation.}
For brevity of notation, we denote $\hat U_i$ by $U$, and the sets $\hat U_i^h$ and $\hat U_i^{\ell}$ of vertices by $U^h$ and $U^{\ell}$, respectively. We start with the flow $f=f_{i-1}$, and, as the algorithm in this phase progresses, we may augment the flow $f$, which may lead to changes in the residual flow network. The sets $U^h$ and $U^{\ell}$ of vertices are always defined with respect to the current flow network $H$.
Recall that we denoted by  $V^{\inn}=\set{v^{\inn}\mid v\in V(G)}$ and by $V^{\out}=\set{v^{\out}\mid v\in V(G)}$.

We now describe each of the steps in turn.

\section{Step 1: Local Flow Augmentations}
\label{sec: step 1}

Our algorithm for this step relies on the technique of local flow augmentation, that was first introduced by~\cite{CHILP17} in the context of the Maximal $2$-Connected Subgraph problem. 
The technique was later refined and applied to \emph{unweighted} and \emph{approximation} versions of the global minimum vertex-cut problem~\cite{NSY19,FNY20,CQ21}.
The algorithm in this section further expands and generalizes this technique, in particular by allowing local flow augmentations to be performed via general flows rather than just paths.

In order to provide intuition for this step, we consider a subgraph $H'\subseteq H$, that is defined via the following algorithmic process. Initially, graph $H'$ contains only the vertex $s^{\out}$, and then we iteratively expand it using the following two rules:

\begin{properties}{R}
	\item If there exists an edge $e=(u,v)\in E(H)$ whose capacity $\hat c(e)\geq M'\cdot \gamma$, such that $u\in V(H')$ and $v\not\in V( H')$, then add the edge $e$ and the vertex $v$ to $H'$.\label{rule: DFS extend}
	
	\item If there exists an integer $1\leq j\leq \log \gamma$, a vertex $v^{\out}\in V^{\out}\setminus V(H')$, and a collection $Y\subseteq V(H')\cap V^{\inn}$ of exactly $2^{j}$ distinct vertices, such that, for every vertex $u^{\inn}\in Y$, a regular edge $(u^{\inn},v^{\out})$ of capacity at least $\frac{M'\cdot \gamma}{2^j}$ is present in $H$, add $v^{\out}$ to $H'$, and, for all $u^{\inn}\in Y$, add the edge $(u^{\inn},v^{\out})$ to $H'$. \label{rule: j-bad}
\end{properties}

As we will show later, if graph $H'$ is constructed via the above rules, then for every vertex $v\in V(H')$, there is an $s^{\out}$-$v$ flow in $H'$ of value at least $\frac{M'\cdot \gamma}{2}$. Let $\Gamma=V(H')\cap V^{\out}$. Intuitively, the goal of the current step is to ensure that $|\Gamma|$ is small. In order to achieve this, we will gradually augment the flow $f$, that is initially set to $f=f_{i-1}$, which, in turn, will lead to changes in the graphs $H$ and $H'$. Once we ensure that $|\Gamma|$ is sufficiently small, we will construct a new graph $\tilde H\subseteq H$, obtained from $H$ by unifying all vertices of $V(H)\setminus V(H')$ into a destination vertex $t'$, which, in turn, will allow us, in Step~\ref{step2: partition}, to compute the  partition $(Q,Q')$ of $U$ with the desired properties. We will ensure that the resulting graph $\tilde H$ is sufficiently small.

We use a parameter $N=\frac{32n^{1+2\eps}\cdot  \log^2(\wmax)}{\gamma}$. As long as $|\Gamma|\geq N$, our algorithm chooses a random vertex $v^{\out}\in \Gamma$, computes a flow $f'$ of value at least $\frac{M'\cdot \gamma}{2}$ in $H'$ from $s^{\out}$ to $v^{\out}$, performs a shortcut operation by adding the edge $(v^{\out},t)$ to graphs $G''$ and $H$, and then augments the current flow $f$ with the flow $f'$ (that is extended via the edge $(v^{\out},t)$, so it terminates at $t$). The residual flow network $H$ is then updated with the new flow, and graph $H'$ is recomputed. As long as $|\Gamma|$ is sufficiently large, with high probability, the resulting shortcut operation is valid. Since, in each such iteration, we augment the flow by at least $\frac{M'\cdot \gamma}{2}$ flow units, while we are guaranteed that, from Invariant \ref{inv: flow}, $\opt_s-\val(f_{i-1})\leq 4n M_{i-1}=8nM'$, the number of such iterations  is bounded by $z'=\ceil{\frac{32n}{\gamma}}$. In order to ensure that the running time of every iteration is sufficiently low, we do not construct the entire graph $H'$; we terminate its construction once $|V(H')\cap V^{\out}|$ reaches $N$. Consider any vertex $v^{\out}\in V^{\out}$ that does not currently lie in $H'$, and an integer $1\leq j\leq \gamma$. For brevity, we say that vertex $v^{\out}$ is \emph{$j$-bad}, if there is a subset $Y\subseteq V^{\inn}\cap V(H')$ of $2^{j}$ distinct vertices, such that, for every vertex $u^{\inn}\in Y$, a regular edge $(u^{\inn},v^{\out})$ of capacity at least $\frac{M'\cdot \gamma}{2^j}$ lies in $H$. 
The execution of Rule \ref{rule: j-bad} requires an algorithm that can correctly identify a $j$-bad vertex, if such a vertex exists for any $1\leq j\leq \log \gamma$. While our algorithm can do this efficiently for relatively small values of the parameter $j$, for higher such values, for the sake of efficiency, we will employ the algorithm for the Heavy Vertex Problem from \Cref{claim: heavy}. This, however, may lead to some imprecisions: for example, the algorithm may terminate the construction of graph $H'$ while some $j$-bad vertex $v$ still exists, for some sufficiently large value of $j$. However, in this case, we are guaranteed that, if $Y\subseteq V^{\inn}\cap V(H')$ is a set of vertices, such that, for every vertex $u^{\inn}\in Y$, a regular edge $(u^{\inn},v^{\out})$ of capacity at least  $\frac{M'\cdot \gamma}{2^j}$ lies in $H$, then $|Y|\leq O( 2^j\cdot \log n\cdot \log(\wmax))$ holds. Therefore, the graph that our algorithm constructs in the last iteration may be a subgraph of the graph $H'$ as defined above, but its properties are sufficient for our purposes.

We now turn to provide a formal description of the algorithm for Step 1.
The algorithm starts with the flow $f=f_{i-1}$, and then gradually augments it over the course of at most $z'=\ceil{\frac{32n}{\gamma}}$ iterations, where in every iteration the flow is augmented by $\frac{M'\cdot \gamma}{2}$ flow units. Throughout, we will ensure that the flow $f$ remains $M'$-integral, and additionally, the following property holds:

\begin{properties}{Q'}
	\item For every vertex $v\in V(G)$ with $w(v)\geq M'$, there is a special edge $(v^{\inn},v^{\out})$ in $H$, whose residual capacity is at least $M'$. \label{prop: residual capacity weaker}
\end{properties}

At the beginning of the algorithm, we set $f=f_{i-1}$. As observed already, this flow is $(2M')$-integral, and hence $M'$-integral. Moreover, since the flow is $(2M')$-integral, and since Property \ref{prop: residual capacity} holds for it, Property \ref{prop: residual capacity weaker} must hold as well. Indeed, consider a vertex $v\in V(G)$ with $w(v)\geq M'$. If $w(v)\geq 2M'$, then, from Property \ref{prop: residual capacity}, the residual capacity of the special edge $(v^{\inn},v^{\out})$ in $H$ is at least $2M'\geq M'$. Otherwise, since the flow $f_{i-1}$ is $(2M')$-integral, no flow may be sent along the special edge  $(v^{\inn},v^{\out})$, so its residual capacity remains $w(v)\geq M'$.
We now provide an algorithm for executing a single iteration.

\subsection{The Execution of a Single Iteration}

As input to the $q$th iteration, we receive access to the modified adjacency list representation of the current graph $G''$, the current $s^{\out}$-$t$ flow $f$ in $G''$ (that was obtained by augmenting the flow $f_i$), together with access to the  modified adjacency-list representation of the residual flow network $H=G''_{f}$. We assume that the flow $f$ is $M'$-integral, and that Property  \ref{prop: residual capacity weaker} holds for it.
Recall that we have defined a parameter $N=\frac{32n^{1+2\eps}\cdot  \log^2(\wmax)}{\gamma}$. Let
$\tau^*=\frac{n}{\sqrt{\gamma}}$, and let $j^*=\ceil{\log(\tau^*)}$.
Throughout, we denote by $\hat T\subseteq V(H)$ the set of vertices $v\in V(H)$ for which the edge $(v,t)$ is present in $H$.

We are now ready to describe an algorithm for a single iteration of Step~\ref{step1: local flow}. 
The main subroutine of the algorithm is Procedure \explore, that is formally described in Figure \ref{fig:alg}. The procedure gradually constructs the graph $H'$
using Rules \ref{rule: DFS extend} and \ref{rule: j-bad}, denoting by $X^*$ the set of vertices of $H'$ that were already discovered by the algorithm. It terminates once either (i) $|X^*\cap V^{\out}|\geq N$ holds; or (ii) a vertex of $\hat T$ was reached; or (iii) the construction of the graph is complete. As mentioned already, in the latter case, the resulting graph may be a subgraph of the graph $H'$ defined via the rules  \ref{rule: DFS extend} and \ref{rule: j-bad}, but some $j$-bad vertices may still exist. We denote by $J\subseteq H'$ the subgraph of $H'$ that the procedure constructs. In addition to the set $X^*\subseteq  V(H')$ of vertices that were already discovered, the algorithm maintains a set $X\subseteq X^*$ of vertices that it still needs to explore, and the counter $N'=|X^*\cap V^{\out}|$. Finally, for every vertex $v^{\out}\in V^{\out}\setminus V(J)$, for all $1\leq j\leq j^*$, the algorithm maintains a counter $n_j(v^{\out})$, that counts the number of regular edges $(u^{\inn},v^{\out})$ with $u^{\inn}\in X^*$ whose capacity is at least $\frac{M'\cdot \gamma}{2^j}$, that the algorithm discovered. Once this counter reaches value $2^{j}$, vertex $v^{\out}$ is added to $X$ and to $X^*$, and the counter is no longer maintained.
 All such edges $(u^{\inn},v^{\out})$ that the algorithm discovers prior to adding $v^{\out}$ to $X$ are stored together with $v^{\out}$, and, once $v^{\out}$ is added to $X$, all these edges are added to graph $J$.
 The formal description of Procedure \explore appears in Figure \ref{fig:alg}.

\program{Procedure \explore}{fig:alg}
{
\begin{enumerate}
	\item Initialize the data structures:
	\begin{enumerate}
		\item Set $X^*=X=\set{s^{\out}}$, $N'=0$.
		\item Initialize the graph $J$ to contain the vertex $s^{\out}$ only.
		\item For all $1\leq j\leq j^*$, for every vertex $v^{\out}\in V^{\out}$, set $n_j(v^{\out})=0$.
	\end{enumerate}  
	
	\item While $X\neq \emptyset$, do: \label{step to go back}
	\begin{enumerate}
		\item Let $x\in X$ be any vertex. Delete $x$ from $X$.
		\item If $x\in \hat T$: terminate the algorithm and return $x$.
	
		Assume from now on that $x\not \in \hat T$.
		
		\item For every edge $e=(x,y)\in E(H)$ with $y\not\in X^*$, do:

			\begin{enumerate}
				\item If the capacity of $e$ in $H$ is at least $M'\cdot \gamma$:\label{procedure: add vertex type 1}
				\begin{itemize}
					\item add $y$ to $X$ and to $X^*$;
					\item add $y$ and the edge $(x,y)$ to $J$;
					\item if $y\in V^{\out}$, then increase $N'$ by $1$, and, if $N'=N$ holds, halt.
				\end{itemize}
			    \item Otherwise, if $x\in V^{\inn}$, $y\in V^{\out}\setminus X^*$, and $e$ is a regular edge, then, for all $1\leq j\leq j^*$, such that the capacity of $e$ in $H$ is at least $\frac{M'\cdot \gamma}{2^j}$ do: \label{procedure: add vertex type 2}
			    \begin{itemize}
			    	\item increase $n_j(y)$ by $1$ and store the edge $(x,y)$ with $y$.
			    	\item if $n_j(y)=2^{j}$, then add $y$ to $X$ and to $X^*$; add $y$ and the $2^{j}$ regular edges of capacity at least $\frac{M'\cdot \gamma}{2^j}$ stored with it to $J$; increase $N'$ by $1$, and, if $N'=N$ holds, halt.
			    \end{itemize}
			\end{enumerate}
	\end{enumerate}
\item For all $j^*< j\leq \log \gamma$, do: \label{step: bad vertices}
\begin{enumerate}
	\item Execute the algorithm from \Cref{claim: heavy} on the instance $(H,Z,Z',\tau,c^*)$ of the Heavy Vertex problem, where $Z=V(J)\cap V^{\inn}$; $Z'=V^{\out}\setminus V(J)$; $c^*=\frac{M'\cdot \gamma}{2^j}$; $\tau=2^{j}+1$; and $W=\wmax$. 
	
	\item If the algorithm returns bit $b=1$, a vertex $v^{\out}\in Z'$ and vertices $u^{\inn}_1,\ldots,u^{\inn}_{\tau+1}\in Z$:\label{step: addition3}

(Observe that there is at most one vertex $u^{\inn}_r$ such that edge $(u^{\inn}_r,v^{\out})$ is special; assume w.l.o.g. that, if such a vertex $u^{\inn}_r$ exists then $r=\tau+1$.)
	\begin{itemize}		
		\item Add $v^{\out}$ to $X$ and to $X^*$;
		\item Add $v^{\out}$ to $J$, and, for all $1\leq a\leq \tau$, add edge $(u^{\inn}_{a},v^{\out})$ to $J$;
		\item Increase $N'$ by $1$, and, if $N'=N$ holds, halt.
		\item go to Step (\ref{step to go back})
	\end{itemize}
\end{enumerate}
\end{enumerate}
}

Let $J'$ be the graph that is identical to $J$, except that, for every special edge $e$ that lies in $J$, we decrease its capacity $\hat c(e)$ in $J'$ by at least $M'$ units and at most $2M'$ units, so that the new capacity $\hat c'(e)$ becomes an integral multiple of $M'$, and $\hat c(e)-2M'\leq \hat c'(e)\leq \hat c(e)-M'$ holds. 
(Note that the capacity of each special edge in $J$ must be at least $\gamma\cdot M'$, from the description of Procedure \explore).
Observe that now all edge capacities in $J'$ are integral multiples of $M'$. Indeed, the above transformation ensures that the capacity of every special edge in $J'$ is an integral multiple of $M'$, and $M'$ is an integral power of $2$ from the definition. Since the capacity of every regular edge in $G''$ is $\wmax$, an integral power of $2$, 
we get that it is an integral multiple of $M'$. Lastly, since the flow $f$ is $M'$-integral, we get that the capacity of every regular edge in $J'$ is an integral multiple of $M'$. We use the following claim.

\begin{claim}\label{claim: flow}
	For every vertex $v$ that is ever added to graph $J$ by Procedure $\explore$, there is an $s^{\out}$-$v$ flow in $J'$ of value at least $\frac{M'\cdot \gamma}{2}$, that is $M'$-integral. Moreover, after $f$ is augmented with this flow, Property \ref{prop: residual capacity weaker} continues to hold.
\end{claim}
\begin{proof}
	 It is  sufficient to prove that, for every vertex $v$ that lies in $J$, there is
	 an $s^{\out}$-$v$ flow in $J'$ of value at least $\frac{M'\cdot \gamma}{2}$. Indeed, if this is the case then, since all edge capacities in $J'$ are integral multiples of $M'$, we get that there is
	 an $s^{\out}$-$v$ flow in $J'$ of value at least $\frac{M'\cdot \gamma}{2}$ that is $M'$-integral. After augmenting $f$ with this flow, Property \ref{prop: residual capacity weaker} will hold, since, for every vertex $u\in V(G)$ with $w(u)\geq M'$, we decreased the capacity of the special edge $(u^{\inn},u^{\out})$ in $J'$ by at least $M'$ units. In the remainder of the proof we show that, for every vertex 
	$v\in V(J)$, there is
	an $s^{\out}$-$v$ flow in $J'$ of value at least $\frac{M'\cdot \gamma}{2}$.

	The proof is by induction over the time when the vertex $v$ was added to $J$. The base of the induction is the beginning of the algorithm, when $J$ only contains the vertex $s^{\out}$, so the claim trivially holds.
	
	Consider now some time $\tset$ during Procedure \explore, when a vertex $v$ was added to $J$, and assume that the claim holds for all vertices that were added before $v$ to $J$. Vertex $v$ may only be added to $J$ either during Step \ref{procedure: add vertex type 1}, or Step \ref{procedure: add vertex type 2}, or Step \ref{step: addition3} of the procedure. Assume first that $v$ was added during  Step  \ref{procedure: add vertex type 1}. Then there is a vertex $u$ that already lied in $J$ before time $\tset$, such that an edge $(u,v)$ of capacity at least $M'\cdot \gamma$ is present in $H$, and is added to $J$ at time $\tset$. From the induction hypothesis, there is a flow $f'$ in $J'$, that sends $\frac{M'\cdot \gamma}2$ flow units from $s^{\out}$ to $u$. Since the capacity of edge $(u,v)$ in $J'$ is at least $M'\cdot \gamma-2M'\geq\frac{M'\cdot \gamma}{2}$, we can extend this flow by sending $\frac{M'\cdot \gamma}{2}$ flow units along the edge $(u,v)$, obtaining a flow in $J'$, that sends $\frac{M'\cdot \gamma}{2}$ flow units from $s^{\out}$ to $v$.
	
	Assume now that $v$ was added to $J$ during Step  \ref{procedure: add vertex type 2} or Step \ref{step: addition3}. Then $v\in V^{\out}$, and there is some integer $1\leq j\leq \gamma$, and a collection $Y$ of $2^{j}$ vertices of $V^{\inn}$ that lied in $J$ before time $\tset$, such that, for every verex $u\in Y$, a regular edge $(u,v)$ of capacity at least $\frac{M'\cdot \gamma}{2^j}$ was added to $J$ at time $\tset$. Denote $Y=\set{u_1,\ldots,u_{2^{j}}}$. From the induction hypothesis, for all $1\leq r\leq 2^{j}$, there is a flow $f_r$ of value $\frac{M'\cdot \gamma}{2}$ from $s^{\out}$ to $u_r$ in graph $J'$. For all $1\leq r\leq 2^{j}$, let $\pset_r$ be a flow-path decomposition of $f_r$, and, for every path $P\in \pset_r$, let $f_r(P)$ be the amount of flow sent via this path by $f_r$.
	Notice that, for all $1\leq r\leq 2^{j}$, edge $(u_r,v)$ is regular, and so its capacity in $J'$ remains the same as in $J$, and is at least $\frac{M'\cdot \gamma}{2^j}$.
	
	Let $f'$ be the flow obtained by taking the union of the flows $f_1,\ldots,f_{2^{j}}$, scaled down by factor $2^{j}$. In other words, for all $1\leq r\leq 2^{j}$, for every path $P\in \pset_r$, flow $f'$ sends $\frac{f_r(P)}{2^{j}}$ flow units via the path $P$, simultaneously. Then it is easy to verify that $f'$ is a valid flow in $J'$ (that is, it respects the edge capacities), the total amount of flow that $s^{\out}$ sends via $f'$ is $\frac{M'\cdot \gamma}{2}$, and, for all $1\leq r\leq 2^{j}$, vertex $u_r$ receives exactly $\frac{M'\cdot \gamma}{2^{j}}$ flow units. For all $1\leq r\leq 2^{j}$, we send the $\frac{M'\cdot \gamma}{2^{j}}$ flow units that terminate at $u_r$, via the edge $(u_r,v)$, to vertex $v$, obtaining a valid $s^{\out}$-$v$ flow in $J'$, of value $\frac{M'\cdot \gamma}{2}$. We conclude that there exists an $s^{\out}$-$v$ flow in $J'$ of value $\frac{M'\cdot \gamma}{2}$. 
\end{proof}

\paragraph{Running Time Analysis.}
We now analyze the running time of Procedure \explore. We first analyze the running time while ignoring the time spent on calls to the algorithm from \Cref{claim: heavy} for the Heavy Vertex problem. 
Let $\Gamma$ be the set of all vertices of  $V^{\out}\setminus\set{s^{\out}}$ that have been added to $J$ by Procedure \explore, so $|\Gamma|\leq N$ holds. The time required to initialize the data structures is $O(n\log n)$. 
Notice that the running time of Procedure \explore (excluding the time spent on calls to the algorithm from \Cref{claim: heavy}) can be asymptotically bounded by the total number of edges that the algorithm considered, times $O(\log n)$.
Recall that every edge of $G''$ that is not incident to $t$ has one of its endpoints in $V^{\out}$. We partition the edges that the procedure explored into two subsets: set $E^1$ contains all edges that have an endpoint in $\Gamma$, and set $E^2$ contains all remaining edges. Clearly, $|E^1|\leq n\cdot |\Gamma|\leq n\cdot N\leq O\left (\frac{n^{2+2\eps}\cdot \log^2(\wmax)}{\gamma} \right )$, since $N=\frac{32n^{1+2\eps}\cdot  \log^2(\wmax)}{\gamma}$. Next, we bound the number of edges of $E^2$ that the algorithm ever considered in Step \ref{procedure: add vertex type 2}. Observe that, for every vertex $v^{\out}\in V^{\out}$ and integer $1\leq j\leq j^*$, once $2^{j}$ edges from set $E^2$ that are incident to $v^{\out}$ are considered by the algorithm, vertex $v^{\out}$ is added to $J$ and to $\Gamma$. Therefore, the total number of edges of $E^2$ considered by the algorithm in this step is bounded by: 

\[O\left(n\cdot 2^{j^*}\right )\leq O\left(n\cdot \tau^*\right )\leq O\left (\frac{n^2}{\sqrt{\gamma}}\right ),\] 

since $j^*=\ceil{\log(\tau^*)}$, and  $\tau^*=\frac{n}{\sqrt{\gamma}}$.

 The time required to consider the remaining edges of $E^2$ is asymptotically bounded by the time spent on calls  to the algorithm from \Cref{claim: heavy}, that we bound next. Recall that the running time of the algorithm from  \Cref{claim: heavy}, on input $(G,Z,Z',\tau,c^*)$, is $O\left(\frac{n\cdot |Z|\cdot \log n\cdot\log(\wmax)}{\tau}\right )\leq O\left(\frac{n^2\cdot \log n\cdot\log(\wmax)}{\tau}\right )$. Since we only apply the claim to values $\tau=2^{j}$ where $j^*< j\leq \log \gamma$,
where 
$j^*=\ceil{\log(\tau^*)}$ and $\tau^*=\frac{n}{\sqrt{\gamma}}$, we get that the running time of a single application of the algorithm from  \Cref{claim: heavy} is: 

\[ O\left(\frac{n^2\cdot \log n\cdot\log(\wmax)}{\tau^*}\right )\leq O\left(n\cdot \sqrt{\gamma}\cdot \log n\cdot \log(\wmax)\right ).\]

We use the following observation to bound the number of calls to the algorithm from \Cref{claim: heavy} in a single iteration.

\begin{observation}\label{obs: num of calls}
	The number of times that Procedure \explore executes the algorithm from \Cref{claim: heavy} is bounded by $O(N\log n)\leq O\left(\frac{n^{1+2\eps}\cdot \log n\cdot \log^2(\wmax)}{\gamma}\right )$. 
\end{observation}
\begin{proof}
	Note that all calls to the algorithm from \Cref{claim: heavy} occur during Step \ref{step: bad vertices}, and every time Step \ref{step: bad vertices} is executed, the algorithm from \Cref{claim: heavy} is called at most $O(\log n)$ times.
	
	If  Procedure \explore does not terminate once Step \ref{step: bad vertices} is executed, then at least one vertex is added to $\Gamma$ during this step. Therefore, Step \ref{step: bad vertices}  may be executed at most $O(N)\leq O\left(\frac{n^{1+2\eps}\cdot \log^2(\wmax)}{\gamma}\right )$ times. Altogether, the total number of calls to the algorithm from   \Cref{claim: heavy} by Procedure \explore is bounded by $O(N\cdot \log n)\leq O\left(\frac{n^{1+2\eps}\cdot \log n\cdot \log^2(\wmax)}{\gamma}\right )$.
\end{proof}

We conclude that the total running time that Procedure \explore spends on  Step \ref{step: bad vertices} is bounded by:

\[ O\left(n\cdot \sqrt{\gamma}\cdot \log n\cdot \log(\wmax)\right )\cdot O\left(\frac{n^{1+2\eps}\cdot \log n\cdot  \log^2(\wmax)}{\gamma}\right )\leq O\left(\frac{n^{2+2\eps}\cdot  \log^2n\cdot \log^3(\wmax)}{\sqrt{\gamma}}\right ). \]

To summarize, the running time of Procedure \explore is bounded by:

\[\begin{split} 
&O\left (\frac{n^{2+2\eps}\cdot  \log^2(\wmax)}{\gamma} \right )+\left (\frac{n^2}{\sqrt{\gamma}}\right )+O\left(\frac{n^{2+2\eps}\cdot  \log^2n\cdot \log^3(\wmax)}{\sqrt{\gamma}}\right )\\
&\quad\quad\quad\quad\quad\quad \leq O\left(\frac{n^{2+2\eps}\cdot  \log^2n\cdot \log^3(\wmax)}{\sqrt{\gamma}}\right ).
\end{split}
 \]

\paragraph{Bounding $|E(J)|$.}
We use the following observation to bound $|E(J)|$.

\begin{observation}\label{obs: bound E(J)}
	$|E(J)|\leq  O\left(\frac{n^{2+2\eps}\cdot \log^2(\wmax)}{\gamma}\right )$.
\end{observation}
\begin{proof}
As observed already, every edge of $E(J)$ must be incident to a vertex of $\Gamma\cup \set{t}$. Since $|\Gamma|\leq N$, we get that $|E(J)|\leq O(n\cdot N)\leq O\left(\frac{n^{2+2\eps}\cdot \log^2(\wmax)}{\gamma}\right )$, since $N=\frac{32n^{1+2\eps}\cdot  \log^2(\wmax)}{\gamma}$.
\end{proof}

Once the algorithm from Procedure \explore terminates, we consider the following three cases.

\paragraph{Case 1: a vertex $x\in \hat T$ is returned.}\customlabel{case1: vertex_in_T}{1}
The first case happens if the algorithm for Procedure \explore terminates with a vertex $x\in \hat T$. 
Let $J'\subseteq J$ be the graph obtained from $J$ as before, by reducing the capacity of every special edge in $J$ by at least $M'$ and at most $2M'$ units, so it becomes an integral multiple of $M'$. From \Cref{claim: flow}, there is a flow $f'$ in $J'$ that is $M'$-integral, in which $s^{\out}$ sends $\frac{M'\cdot \gamma}{2}$ flow units to $x$. 
We use the algorithm from \Cref{thm: maxflow} to compute this flow in time $O(|E(J)|^{1+o(1)}\cdot \log \wmax)$. We extend this flow 
by sending $\frac{M'\cdot \gamma}{2}$ flow units via the edge $(x,t)$ obtaining an $s^{\out}$-$t$ flow $f''$ in $H$ of value $\frac{M'\cdot \gamma}{2}$, that is $M'$-integral. Finally, we augment the flow $f$ with this resulting flow $f''$, and update the residual flow network $H$. It is easy to verify that the resulting flow $f$ remains $M'$-integral, and that it has Property \ref{prop: residual capacity weaker}. The additional time that the algorithm spends in this case is bounded by $O(|E(J)|^{1+o(1)}\cdot \log (\wmax))$.

\paragraph{Case 2: $N'=N$.}\customlabel{case2: N}{2}
The second case happens if Procedure \explore terminated because $N'=N$ held, so, at the end of the procedure, $|\Gamma|=N$ holds.
We select a single vertex $v^{\out}\in \Gamma$ uniformly at random. We then perform a shortcut operation, by adding an edge $(v^{\out},t)$ of capacity $\wmax$ to graph $G''$ and to the residual flow network $H$. 
The remainder of the algorithm for this case is identical to the algorithm for Case~\ref{case1: vertex_in_T}, except that vertex $x$ is replaced with $v^{\out}$.
As before, from \Cref{claim: flow}, there is a flow $f'$ in $J'$ that is $M'$-integral, in which $s^{\out}$ sends $\frac{M'\cdot \gamma}{2}$ flow units to $v^{\out}$. 
We use the algorithm from \Cref{thm: maxflow} to compute such a flow in time $O(|E(J)|^{1+o(1)}\cdot \log (\wmax))$. We extend this flow 
by sending $\frac{M'\cdot \gamma}{2}$ flow units via the edge $(v^{\out},t)$, obtaining an $s^{\out}$-$t$ flow $f''$ in $H$ of value $\frac{M'\cdot \gamma}{2}$, that is $M'$-integral. Finally, we augment the flow $f$ with this resulting flow $f''$, and update the residual flow network $H$. It is easy to verify that the resulting flow $f$ remains $M'$-integral, and that it has Property \ref{prop: residual capacity weaker}. The additional time that the algorithm spends in this case is bounded by $O(|E(J)|^{1+o(1)}\cdot \log (\wmax))$.

This completes the description of Case~\ref{case2: N}. We use the following simple claim to show that the probability that the resulting shortcut operation was invalid is low.

\begin{claim}\label{claim: shortcut one iteration}
	If the algorithm chooses to perform a shortcut operation in an iteration, then the probability that this shortcut operation is invalid is at most $\frac{\gamma}{32n^{1+\eps}\cdot  \log^2(\wmax)}$.
\end{claim}
\begin{proof}
	Let $(v^{\out},t)$ be the shortcut edge that the algorithm has added. Note that the shortcut operation may only be invalid if (i) $(s,T)$ is a good pair; and (ii) $v\in L$ holds. Since
	$|\Gamma|=N=\frac{32n^{1+2\eps}\cdot  \log^2(\wmax)}{\gamma}$, and since, from Inequality \ref{eq: small L}, if $(s,T)$ is a good pair, $|L|\leq n^{\eps}$ holds. We then get that, if $(s,T)$ is a good pair, the probability that $v^{\out}\in L$ is bounded by:
	
	\[\prob{v\in L}\leq \frac{|L|}{|\Gamma|}\leq 
	\frac{\gamma}{32n^{1+\eps}\cdot \log^2(\wmax)}.\]
\end{proof}

\paragraph{Case 3.} The third case happens when neither of the first two cases happen. In this case, the current iteration becomes the last one, and we construct two graphs, $\tilde H$ and $\tilde H'$, to be used in Step~\ref{step2: partition} of the algorithm, whose construction we describe later.

\paragraph{Running time of an iteration.}
From our discussion, the running time of a single iteration, excluding the time required to construct the graph $\tilde H$, is bounded by:

\[\begin{split}
 &O\left(\frac{n^{2+2\eps}\cdot  \log^2n\cdot \log^3(\wmax)}{\sqrt{\gamma}}\right )+O(|E(J)|^{1+o(1)}\cdot \log \wmax)\\
 &\quad\quad\quad\quad\quad  \leq O\left(\frac{n^{2+2\eps}\cdot  \log^2n\cdot \log^3(\wmax)}{\sqrt{\gamma}}\right )+O\left(\frac{{n^{2+2\eps +o(1)}\cdot  \log^4(\wmax)}}{\gamma}\right )\\
& \quad\quad\quad\quad\quad  \leq O\left(\frac{n^{2+2\eps+o(1)}\cdot \log^4(\wmax)}{\sqrt{\gamma}}\right ),
\end{split} \]

from \Cref{obs: bound E(J)}.

\subsection{Completing Step 1}
The algorithm for Step 1 performs $z'= \ceil{\frac{32n}{\gamma}}$ iterations as described above. If, after $z'$ iterations, the algorithm for a single iteration still terminates with cases~\ref{case1: vertex_in_T} or~\ref{case2: N}, then we terminate the algorithm for the current phase with a ``FAIL''.

Note that the running time of the algorithm for Step 1 so far is:

\[\begin{split}
&O\left(\frac{n^{2+2\eps+o(1)}\cdot  \log^4(\wmax)}{\sqrt{\gamma}}\right )\cdot z'\\
&\quad\quad\quad\quad \leq 
O\left(\frac{n^{3+2\eps+o(1)}\cdot   \log^4(\wmax)}{\gamma^{3/2}}\right )\\
&\quad\quad\quad\quad \leq O\left (n^{2-4\eps+o(1)}\cdot \log^4(\wmax)\right ),
\end{split}\]

since $\gamma\geq n^{2/3+4\eps}$ from Inequality \ref{eq: gamma2}.

Let $\hat \event$\customlabel{event: heavy err}{$\hat \event$}be the bad event that the algorithm from \Cref{claim: heavy} for the Heavy Vertex problem erred at any time when it was called during the execution of Step 1.
Recall that the probability that the algorithm from \Cref{claim: heavy} errs in a single execution is at most $\frac{1}{n^{10}\cdot \log^6(\wmax) }$, and, from \Cref{obs: num of calls}, the algorithm is executed at most $O\left(\frac{n^{1+2\eps}\cdot \log n\cdot \log^2(\wmax)}{\gamma}\right )$ in every iteration. Since the number of iterations is bounded by $z'=\ceil{\frac{32n}{\gamma}}$, and since we assumed that $n$ is sufficiently large, from the Union Bound, we get that:

\begin{equation}\label{eq: prob error Heavy Vertex}
\prob{\hat \event}\leq \frac{1}{n^{10}\cdot \log^6(\wmax) }\cdot O\left(\frac{n^{2+2\eps}\cdot \log n\cdot \log^2(\wmax)}{\gamma^2}\right ) \leq \frac{1}{n^7\cdot \log^4(\wmax)}.
\end{equation}

Let $J$ be the graph constructed by Procedure \explore in the last iteration, and let $E^R\subseteq E(H)$ be the set of all {\bf regular} edges $(v^{\inn},u^{\out})$ with $v^{\inn}\in V(J)\cap V^{\inn}$ and $u^{\out}\in V^{\out}\setminus V(J)$. For all $1\leq j\leq \log \gamma$, let $E^R_j\subseteq E^R$ denote the set of all edges $e\in E^R$ with $\frac{M'\cdot \gamma}{2^j}\leq \hat c(e)< \frac{M'\cdot \gamma}{2^{j-1}}$, where $\hat c(e)$ is the capacity of edge $e$ in graph $H$. 
We use the following claim that bounds the number of regular edges of $H$ whose first endpoint lies in $J$.

\begin{claim}\label{claim: ER}
	If $e=(x,y)\in E(H)$ is a regular edge with $x\in V(J)$ and $y\not\in V(J)$, then $e\in E^R$. Additionally, 
	$E^R=\bigcup_{j=1}^{\log\gamma}E^R_j$ holds. Lastly, if Event~\ref{event: heavy err} did not happen, then,
	 for all $1\leq j\leq \log\gamma$, every vertex $u^{\out}\in V^{\out}\setminus V(J)$ is incident to at most $2^{j+11}\cdot \log n\cdot \log (\wmax)$ edges of $E^R_j$, and at most $O(\gamma\cdot \log n\cdot \log(\wmax))$ edges of $E^R$.
\end{claim}
\begin{proof}
Let $e=(x,y)\in E(H)$ be a regular edge with $x\in V(J)$ and $y\not\in V(J)$. Assume first that $x=v^{\out}$ and $y=u^{\inn}$, for some $v^{\out}\in V^{\out}$ and $u^{\inn}\in V^{\inn}$. Then edge $(v,u)$ lies in $G$, and so the regular edge $(v^{\out},u^{\inn})$ of capacity $\wmax$ was added to the split graph $G'$. Since the flow on this edge is bounded by $\max_{v\in V(G)}\set{w(v)}\leq W'_{\max}<\wmax/2$, we get that the capacity of the edge $e$ in $H$ must remain at least $\wmax/2\geq M'\cdot \gamma$. Therefore, Procedure \explore should have added $u^{\inn}$ to $J$ in Step~\ref{procedure: add vertex type 1}, a contradiction.
It then must be the case that $x=v^{\inn}$ and $y=u^{\out}$ for some $v^{\inn}\in V^{\inn}$ and $u^{\out}\in V^{\out}$, so $e\in E^R$.
	
Next, we prove that 	$E^R=\bigcup_{j=1}^{\log\gamma}E^R_j$.
Consider any regular edge $e=(v^{\inn},u^{\out})\in E^R$. Note that it is impossible that $\hat c(e)\geq M'\cdot \gamma$, since then $e$ and $u^{\out}$ should have been included in $J$ when Procedure \explore was applied for the last time.
We now claim that $\hat c(e)\geq M'$ must also hold. Indeed, assume otherwise. Notice that edge $e=(v^{\inn},u^{\out})$ does not exist in graph $G''$, from the definition of the split graph. If $\hat c(e)<M'$, then the current flow on edge $(u^{\out},v^{\inn})$ is non-zero, but below $M'$. Since the current flow is $M'$-integral, this is impossible.
We conclude that $E^R=\bigcup_{j=1}^{\log\gamma}E^R_j$.

Consider now an integer $1\leq j\leq \gamma$, and the set $E_j^R$ of edges; recall that it contains all edges $e\in E^R$ with $\frac{M'\cdot \gamma}{2^j}\leq \hat c(e)< \frac{M'\cdot \gamma}{2^{j-1}}$. Assume first that $j\leq j^*$. If some vertex $u^{\out}\in V^{\out}\setminus V(J)$ is incident to more than $2^j$ edges of $E^R_j$, then Step \ref{procedure: add vertex type 2} of Procedure \explore should have discovered it, when the procedure was executed for the last time, and $u^{\out}$ must have been added to $J$. Therefore, the number of edges of $E_j^R$ incident to any vertex $u^{\out}\in V^{\out}\setminus V(J)$ for $j\leq j^*$ is bounded by $2^j$. 

Lastly, consider an integer $j^*<j\leq \gamma$, and assume for contradiction that there is a vertex $u^{\out}\in V^{\out}\setminus V(J)$ that is incident to more than $2^{j+11}\cdot \log n\cdot \log(\wmax)$ edges of $E_j^R$. 
Consider the last time when Step \ref{step: addition3} of Procedure \explore was executed during the last iteration, and the call to the algorithm from  \Cref{claim: heavy} for the Heavy Vertex problem 
with parameters $c^*=\frac{M'\cdot \gamma}{2^j}$ and $\tau=2^j+1$ during the execution of that step. From our assumption, there are at least $2^{j+11}\cdot \log n\cdot \log(\wmax)>1000\tau\log n\log(\wmax)$ edges of capacity at least $c^*$ connecting vertices of $Z=V(J)\cap V^{\inn}$ to the vertex $u^{\out}$, that lies in $Z'$. Since this is the last time when Step \ref{procedure: add vertex type 2} was executed,  the algorithm from   \Cref{claim: heavy} must have returned bit $b=0$, and so it must have erred, contradicting the assumption that Event~\ref{event: heavy err} did not happen.
\end{proof}

As our last step, we verify that, for all $1\leq j\leq \log \gamma$, $|E_j^R|\leq |V^{\out}\setminus V(J)|\cdot 2^{j+11}\cdot \log n\cdot \log(\wmax)$. If we discover that this is not the case for any such value $j$, then we terminate the algorithm for the current phase with a ``FAIL''. Notice that, from  \Cref{claim: ER}, this may only happen if Event~\ref{event: heavy err} occurred. 

 In order to execute this last step efficiently,
 for all $1\leq j\leq \log \gamma$, we maintain a counter $N_j$, that counts the number of edges in $E_j^R$; we initialize all these counters to $0$.
  We then consider every vertex $u^{\inn}\in V^{\inn}\cap V(J)$ in turn. When such a vertex $u^{\inn}$ is considered, we use the modified adjacency-list representation of graph $H$ in order to count all regular edges $e=(u^{\inn},v^{\out})$ leaving $u^{\inn}$, that lie in $E^R_j$, for all $1\leq j\leq \log \gamma$, and update the counter $N_j$ accordingly. Once the counter $N_j$ for any such integer $1\leq j\leq \log \gamma$ becomes greater than $|V^{\out}\setminus V(J)|\cdot 2^{j+11}\cdot \log n\cdot \log(\wmax)$, we
terminate the algorithm with a ``FAIL''. Note that, for every vertex $u^{\inn} \in V^{\inn}\cap V(J)$, for every regular edge $e=(u^{\inn},v^{\out})$ whose capacity is at least $M'$, if $e\not\in E^R$, then $v^{\out}\in \Gamma$ holds. Therefore, the time required to execute this step is bounded by:

\[O(N\cdot n)+O(n\gamma\log n\cdot \log^2(\wmax))\le O\left (n^{2-4\eps+o(1)}\cdot \log^2(\wmax)\right ),  \]

since $N=\frac{32n^{1+2\eps}\cdot  \log^2(\wmax)}{\gamma}$ and $n^{6\eps}\le \gamma\leq n^{1-4\eps}$ from Inequality \ref{eq: gamma2}.

\paragraph{Construction of Graphs $\tilde H'$ and $\tilde H$.}
At the end of Step~\ref{step1: local flow} we construct the graph $\tilde H'$, that is defined as follows. Graph $\tilde H'$ is obtained from $H$, by contracting all vertices of $V(H)\setminus V(J)$ into a supersink $t'$, while keeping the parallel edges, but deleting all edges that leave $t'$, including self-loops.
For every edge $e=(x,y)$ of $H$ with $x\in V(J)$ and $y\not\in V(J)$, there is a unique edge $e'=(x,t)$ corresponding to $e$ in $\tilde H'$, whose capacity in $\tilde H'$ is equal to the capacity of $e$ in $H$. We sometimes say that $e'$ \emph{represents} $e$ in $\tilde H'$, or that $e'$ is a \emph{copy} of $e$ in $\tilde H'$, and we may not distinguish between $e$ and $e'$.

From \Cref{claim: ER}, it is easy to verify that every edge $e\in E(\tilde H')$ belongs to one of the following three cathegories:

\begin{itemize}
	\item either one endpoint of $e$ lies in $\Gamma$;
	\item or $e$ is a copy of an edge $e'\in E^R$; 
	\item or $e$ is a copy of a special edge $e'=(v^{\inn},v^{\out})$, with $v^{\inn}\in V^{\inn}\cap V(J)$.
\end{itemize}

Notice that the number of all edges incident to the vertices of $\Gamma$ is bounded by $|\Gamma|\cdot n\leq O\left(\frac{n^{2+2\eps}\cdot \log^2(\wmax)}{\gamma}\right )$, and the set of all such edges can be constructed in time $O\left(\frac{n^{2+2\eps}\cdot \log^2(\wmax)}{\gamma}\right )\leq O\left (n^{2-4\eps+o(1)}\cdot \log^2(\wmax)\right )$ by simply inspecting all edges that are incident to the vertices of $\Gamma$ using the adjacency-list representation of $H$. The edges of $E^R$, and all special edges $e=(v^{\inn},v^{\out})$ with $v^{\inn}\in V^{\inn}\cap V(J)$
 can be computed in time $O(n\gamma\log n\cdot \log^2(\wmax))\leq O\left (n^{2-4\eps+o(1)}\cdot \log^2(\wmax)\right )$ using the 
procedure described above for verifying that $|E^R_j|$ is not too large for all $1\leq j\leq \log \gamma$. Altogether, graph $\tilde H'$ can be constructed in time:

\[O\left (n^{2-4\eps+o(1)}\cdot \log^2(\wmax)\right ).  \]

Finally, we construct the graph $\tilde H$, that is identical to $\tilde H'$, except that we unify parallel edges into a single edge, whose capacity is equal to the total capacity of the corresponding parallel edges.

Altogether, the total running time of Step 1 then remains bounded by $O\left (n^{2-4\eps+o(1)}\cdot \log^4(\wmax)\right )$. We summarize the properties of graph $\tilde H'$ that will be useful for us in the next step, in the following observation.

\begin{observation}\label{obs: props of tilde H}
If the algorithm for Step~\ref{step1: local flow} did not terminate with a ``FAIL'', then the graph $\tilde H'$ that it constructed has the following properties: 
	
\begin{properties}{P}
\item	The total number of edges $e\in E(\tilde H')$ that have one  endpoint in $V(J)\cap V^{\out}$ and another in $V(J)\setminus V^{\out}$ is bounded by $O(n\cdot N)\le O\left(\frac{{n^{2+2\eps}\cdot  \log^2(\wmax)}}{\gamma}\right )$. 

\item Every other edge of $E(\tilde H')$ that is not special, is of the form $(v^{\inn},t')$, where $v^{\inn}\in V(J)\cap V^{\inn}$, and corresponds to an edge of $E^R$, that is, a regular edge $e=(v^{\inn},u^{\out})\in E(H)$, where $u^{\out}\in V^{\out}\setminus V(J)$, $u\neq v$, and $M'\leq \hat c(e)<M'\cdot \gamma$, where $\hat c(e)$ is the capacity of $e$ in $H$. \label{prop: regular incident to t'}

\item If we denote,	for all
	$1\leq j\leq \log \gamma$, by $E_j^R$ the set of all edges $e\in E^R$ with $\frac{M'\cdot \gamma}{2^j}\leq \hat c(e)< \frac{M'\cdot \gamma}{2^{j-1}}$, then $|E_j^R|\leq n\cdot 2^{j+11}\cdot \log n\cdot \log(\wmax)$ holds. Moreover,  if bad event~\ref{event: heavy err} did not happen, then, for all $1\leq j\leq \log \gamma$, every vertex $v^{\out}\in V^{\out}\setminus V(J)$ is incident to at most  $2^{j+11}\log n\cdot\log(\wmax)$ edges of $E_j^R$. \label{prop: no bad event}

\item \label{prop: few sticking out} Altogether, $|E^R|\leq O\left (n\cdot \gamma\cdot \log n\cdot \log(\wmax)\right )$, and the total capacity of the edges of $E^R$ in $H$ is at most $2^{12}n\cdot M'\cdot \gamma\cdot \log^2n\log(\wmax)$.
Also: $$|E(\hat H')|\leq  O\left(\frac{{n^{2+2\eps}\cdot \log^2(\wmax)}}{\gamma}\right )+O\left (n\cdot \gamma\cdot \log n\cdot \log(\wmax)\right )\leq O\left(n^{2-4\eps+o(1)}\cdot \log^2(\wmax)\right ).$$
\end{properties}
\end{observation}

We also obtain the following useful corollary of the observation:

\begin{corollary}\label{cor: edges entering L}
	Assume that $(s,T)$ is a good pair, that all shortcut operations so far have been valid, and that Event~\ref{event: heavy err} did not happen. Let $A^*$ be the set of all edges $(x,y)\in E(H)$, where $x\in V(J)$ and $y\in L^*\setminus V(J)$. Then $\sum_{e\in A^*}\hat c(e)\leq  2^{14}\cdot n^{\eps}\cdot \gamma\cdot M'\cdot \log n\cdot \log(\wmax).$
\end{corollary}

\begin{proof}
	We partition the set $A^*$ of edges into two subsets: set $A_1$ contains all edges $(x,y)\in A^*$ with $y\in L^{\inn}$, and $A_2$ contains all remaining edges.
	
	Consider first some edge $e=(x,y)\in A_1$. Since $y\in L^{\inn}$, from the construction of the split graph $G'$, it must be the case that $x\in V^{\out}$. We denote $x=v^{\out}$ and $y=u^{\inn}$, where $u\in L$ and $v\in V(G)$. From \Cref{claim: ER},  $e$ may not be a regular edge. Therefore, $e$ must be a special edge, $e=(v^{\out},v^{\inn})$, for some vertex $v\in L$. Since $v^{\inn}$ was not added to $J$ during the last application of Procedure \explore, we get that $\hat c(e)<M'\cdot \gamma$. Therefore, $\sum_{e\in A_1}\hat c(e)\leq |L|\cdot M'\cdot \gamma$.
	
Next, we bound the total capacity of the edges in $A_2$.
	For every vertex $v^{\out}\in L^{\out}\setminus V(J)$, let $A_2(v^{\out})\subseteq A_2$ denote the subset of all edges $e\in A_2$ that are incident to $v^{\out}$.
	From Property \ref{prop: no bad event} of \Cref{obs: props of tilde H}, for all 
	$1\leq j\leq \gamma$, $v^{\out}$ may be incident to at most  $2^{j+11}\cdot\log n\cdot\log(\wmax)$ edges of $E_j^R$: regular edges $(a,v^{\out})$ with $a\in V(J)$, whose capacity is at most $\frac{M'\gamma}{2^{j-1}}$. Altogether, the total capacity of all regular edges $(a,v^{\out})\in A_2(v^{\out})$ is at most $2^{12}\cdot M'\cdot\gamma\cdot\log n\cdot\log(\wmax)$. Only one special edge $(b,v^{\out})$ with $b\in V(J)$ may be incident to $v^{\out}$ in $H$, that must be a forward edge. Since vertex $v^{\out}$ was not added to $J$ when Procedure \explore was executed for the last time, the capacity of this edge is bounded by $M'\gamma$. Overall, we get that $\sum_{e\in A_2(v^{\out})}\hat c(e)\leq 2^{13}\cdot\gamma\cdot M'\cdot \log n\cdot \log(\wmax)$, and $\sum_{e\in A_2}\hat c(e)\leq |L|\cdot 2^{13}\gamma\cdot M'\cdot \log n\cdot \log(\wmax)$.
	
	Since, from Inequality \ref{eq: small L}, $|L|\leq n^{\eps}$, we get that: $\sum_{e\in A_1\cup A_2}\hat c(e)\leq 2^{14}\cdot n^{\eps}\cdot \gamma\cdot M'\cdot \log n\cdot \log(\wmax)$.
	\end{proof}


The next claim bounds the probability that one of the shortcut steps executed in Step~\ref{step1: local flow} was invalid, or that the algorithm for this step terminated with a ``FAIL''.

\begin{claim}\label{claim: step 1 analysis}
	The probability that any shortcut operation executed in the current step is invalid is at most $\frac{1}{n^{\eps}\cdot \log^2(\wmax)}$.
	Moreover, if $(s,T)$ is a good pair, and if all shortcut operations performed before Phase $i$ were valid, then the probability that the algorithm terminates with a ``FAIL'' in the first step is at most $\frac{2}{n^{\eps}\cdot \log^2(\wmax)}$.
\end{claim}
\begin{proof}
	Since the number of iterations is at most $z'=\frac{32n}{\gamma}$, and since, from \Cref{claim: shortcut one iteration}, for every iteration, the probability that the algorithm performs an invalid shortcut operation over the course of the iteration is at most  $\frac{\gamma}{32n^{1+\eps}\cdot  \log^2(\wmax)}$, we get that overall, the probability that any shortcut operation executed in Step~\ref{step1: local flow} is invalid is bounded by:
	
	\[ \frac{\gamma}{32n^{1+\eps}\cdot  \log^2(\wmax)}\cdot z'\leq \frac{1}{n^{\eps}\cdot \log^2(\wmax)}.\]

Assume now that $(s,T)$ is a good pair, and that all shortcut operations performed before Phase $i$ were valid. We claim that, in this case, if the algorithm for Step~\ref{step1: local flow} terminates with a ``FAIL'', then either Event~\ref{event: heavy err} has happened, or  at least one shortcut operation performed in the current step was invalid. Indeed, assume otherwise, so all shortcut operations executed by the algorithm were valid, and Event~\ref{event: heavy err} did not happen. 
The only way that the algorithm for Step~\ref{step1: local flow} terminates with a ``FAIL'' in this case is if the algorithm did not terminate after $z'$ iterations.

Observe that, since we assumed that all shortcut operations were valid, the value of the maximum $s^{\out}$-$t$ flow in the graph $G''$ remains equal to $\opt_s$ throughout this step. From Invariant \ref{inv: flow}, at the beginning of the phase, we were given a flow $f_{i-1}$ with $\val(f_{i-1})\geq \opt_s-4n\cdot M_{i-1}=\opt_s-8n\cdot M'$. Our algorithm performed $z'=\frac{32n}{\gamma}$ iterations, and in each iteration it augmented the flow by $\frac{M'\cdot \gamma}{2}$ flow units. Therefore, at the end of Step~\ref{step1: local flow}, we obtain a valid $s^{\out}$-$t$ flow in $G''$, whose value is: 

\[\val(f_{i-1})+z'\cdot \frac{M'\cdot \gamma}{2}>\val(f_{i-1})+8n\cdot M'\geq \opt_s,\]

a contradiction. Therefore, if $(s,T)$ is a good pair, and if all shortcut operations performed before Phase $i$ were valid, the algorithm for Step~\ref{step1: local flow} may only terminate with a ``FAIL'' if at least one of the shortcut operations it performed was invalid, or Event~\ref{event: heavy err} has happened. The probability of the former is bounded by $\frac{1}{n^{\eps}\cdot \log^2(\wmax)}$, and, from Inequality \ref{eq: prob error Heavy Vertex}, $\prob{\hat \event}\leq \frac{1}{n^7\cdot \log^4(\wmax)}$. Altogether, we get that, if all shortcut operations performed before Phase $i$ were valid, then the probability that the algorithm for Step~\ref{step1: local flow} terminates with a ``FAIL'' is at most:

\[\frac{1}{n^{\eps}\cdot \log^2(\wmax)}+\frac{1}{n^7\cdot \log^4(\wmax)}\leq  \frac{2}{n^{\eps}\cdot \log^2(\wmax)}.\]
\end{proof}

\section{Step 2: Computing the Partition $(Q,Q')$ of $U$}
\label{sec: step 2}

In this step, we will use the graphs $\tilde H$ and $\tilde H'$ that were constructed in the previous step.
The goal of this step is to compute a partition $(Q,Q')$ of the set $U$ of vertices, such that such that $Q'$ only contains relatively few vertices of $R\cup L$. Additionally, if $|Q|>n^{1-4\eps}$, then we will also compute an $s^{\out}$-$t$ cut $(\hat X,\hat Y)$ in $H$, such that, in any maximum $s^{\out}$-$t$ flow in $H$, the total flow over the edges in $E_H(\hat X,\hat Y)$ is close to their capacity, and for every vertex $v\in Q$, $v^{\inn}\in \hat X$ and $v^{\out}\in \hat Y$ holds.
Recall that, for every edge $e\in E(H)$, we denoted by $\hat c(e)$ its capacity in $H$.
Throughout this step, we denote by $F$ the value of the maximum $s^{\out}$-$t$ flow in $H$, and by $F'$ the value of the maximum $s^{\out}$-$t'$ flow in $\tilde H'$. We start by exploring some properties of the graphs $H$ and $\tilde H'$.

\subsection{Properties of Maximum $s^{\out}$-$t$ Flow in $H$}

We define several subsets of edges of $H$, and then express the value $F$ of the maximum $s^{\out}$-$t$ flow in $H$ in terms of their capacities. We use the following definition:

\begin{definition}[Set $E'_S$ of edges.]\label{def: set E'S}
Consider the following three subsets of edges of $H$: set

$$E^0_S=\set{(v^{\inn},v^{\out})\mid v\in S}$$ 

that contains all forward special edges representing the vertices of $S$; set 

$$E^1_S=\set{(v^{\inn},u^{\out})\mid v\in L\cup S, u\in S\setminus\set{v}}$$ 

that contains all regular edges connecting in-copies of vertices of $L\cup S$ to out-copies of vertices of $S$; and set 

$$E^2_S=\set{(v^{\inn},u^{\out})\mid v\in S, u\in R}$$ 

that contains all regular edges connecting in-copies of vertices of $S$ to out-copies of vertices of $R$. 
We set  $E'_S=E^0_S\cup E^1_S\cup E^2_S$, and we denote by $\hat c_S=\sum_{e\in E'_S}\hat c(e)$ the total capacity of the edges of $E'_S$ in $H$.
\end{definition}

Note that every edge of $E'_S$ has an endpoint in $S^{\out}\cup R^*$.

 We start with the following simple observation regarding the structure of $H$.

\begin{observation}\label{obs: edges crossing}
	Assume that $(s,T)$ is a good pair, and that all shortcut operations so far have been valid.
Consider the partition $(X,Y)$ of the vertices of $G''$ (and of $H$), where $X=L^* \cup S^{\inn}$, and $Y=V(G'')\setminus X=S^{\out}\cup R^*\cup \set{t}$. Then	$E_H(X,Y)=E'_S$.
\end{observation}
\begin{proof}
	It is immediate to verify that $E'_S\subseteq E_H(X,Y)$. We now prove that $E_H(X,Y)\subseteq E'_S$.
	
	Consider any edge $e=(x,y)$ of $H$ with $x\in X$ and $y\in Y$. Since $(L,S,R)$ is a vertex cut in $G$, and since we have assumed that $(s,T)$ is a good pair, and all shortcut operations so far had been valid, $y\neq t$, and moreover, either $x\not\in L^*$, or $y\not\in R^*$ holds.
	
	Assume first that $x\not \in L^*$, so $x=v^{\inn}$, for some $v\in S$. If $e$ is a special edge, then $e=(v^{\inn},v^{\out})\in E^0_S$ must hold. If $e$ is a regular edge, then $e=(v^{\inn},u^{\out})$, where $u\in (S\cup R)\setminus\set{v}$ must hold, so $e\in E^1_S\cup E^2_S$.
	
	Assume now that $y\not\in R^*$, so $y=v^{\out}$ for some $v\in S$. As before, if $e$ is a special edge, then $e=(v^{\inn},v^{\out})\in E^0_S$ must hold. If $e$ is a regular edge, then $e=(u^{\inn},v^{\out})$, where $u\in S\cup L$ must hold, and $u\neq v$, so $e\in E^1_S$. In any case, $e\in E'_S$ must hold.
\end{proof}

Note that, from \Cref{obs: edges crossing}, if $(s,T)$ is a good pair, and if all shortcut operations so far have been valid, then graph $H\setminus E'_S$ contains no $s^{\out}$-$t$ path. Next, we express the value $F$ of the maximum $s^{\out}$-$t$ flow in $H$, in terms of $\hat c_S$.

\begin{claim}\label{claim: max flow in H}
Assume that $(s,T)$ is a good pair, and that all shortcut operations so far have been valid. Then $F=\hat c_S$ and $F\leq 8n\cdot  M'$ must hold. Moreover, if $f'$ is a maximum $s^{\out}$-$t$ flow in $H$, and $\pset$ is its flow-path decomposition, then every path in $\pset$ contains exactly one edge of $E'_S$.
\end{claim}


\begin{proof}
Assume that $(s,T)$ is a good pair, and that all shortcut operations so far have been valid.
Let $(X,Y)$ be the partition of $V(G'')$ that is defined as before:  $X=L^*\cup S^{\inn}$, and $Y=V(G'')\setminus X=S^{\out}\cup R^*\cup \set{t}$.

Consider the current $s^{\out}$-$t$ flow $f$ in $G''$, and recall that its value is equal to the total amount of flow leaving $X$, minus the total amount of flow entering $X$, so $\val(f)=f^{\out}(X)-f^{\inn}(X)$, where  $f^{\out}(X)=\sum_{e\in E_{G''}(X,Y)}f(e)$ and $f^{\inn}(X)=\sum_{e\in E_{G''}(Y,X)}f(e)$. From the definition of the split graph $G''$, and since $(L,S,R)$ is a vertex cut in $G$, it is immediate to verify that $E_{G''}(X,Y)=\set{(v^{\inn},v^{\out})\mid v\in S}=E^0_S$ (we have used the fact that, from Property \ref{prop: residual capacity weaker}, the residual capacity of each original special edge of $G''$ in $H$ is greater than $0$). Therefore, $f^{\out}(X)=\sum_{e\in E^0_S}f(e)$.
It is easy to verify that, if $e=(x,y)\in E_{G''}(Y,X)$, then $e$ is a regular edge, where $x=v^{\out}$ for some vertex $v^{\out}\in V^{\out}$, and $y=u^{\inn}$ for some vertex $u^{\inn}\in V^{\inn}\setminus\set{v^{\inn}}$; moreover,  either (i) $v\in R\cup S$ and $u\in S$; or (ii) $v\in S$ and $u\in L\cup S$  must hold. In other words, each such edge must be anti-parallel to some edge in $E_S^1\cup E_S^2$. From the definition of the split graph $G''$, the edges of $E_S^1\cup E_S^2$ may not lie in $G''$. So for every edge $(u^{\inn},v^{\out})\in E_S^1\cup E_S^2$, its residual capacity in $H$ is: $\hat c(u^{\inn},v^{\out})=f(v^{\out},u^{\inn})$.
Therefore:

\[f^{\inn}(X)=\sum_{e\in E_S^1\cup E_S^2}\hat c(e).\]

Altogether, we get that:

\[\val(f)=f^{\out}(X)-f^{\inn}(X)=\sum_{e\in E^0_S}f(e)-\sum_{e\in E_S^1\cup E_S^2}\hat c(e).\]
	
Recall that the value of the maximum $s^{\out}$-$t$ flow in $G''$ is $\opt_s=\sum_{e\in E^0_S}c(e)$, and the value of the maximum $s^{\out}$-$t$ flow in $H$ is:

\[\begin{split}
F&=\opt_s-\val(f)\\
&= \sum_{e\in E^0_S}c(e)-\sum_{e\in E^0_S}f(e)+\sum_{e\in E_S^1\cup E_S^2}\hat c(e)\\
&=\sum_{e\in E_S^0}\hat c(e)+\sum_{e\in E_S^1\cup E_S^2}\hat c(e)\\
&=\hat c_S.
\end{split} \]

Consider now the maximum $s^{\out}$-$t$ flow $f'$ in $H$, whose value is $F=\hat c_S=\sum_{e\in E'_S}\hat c(e)$, and let $\pset$ be its flow-path decomposition. Since, from \Cref{obs: edges crossing}, the edges of $E'_S$ separate $s^{\out}$ from $t$ in $H$, we get that every flow-path $P\in \pset$ must contain at least one edge of $E'_S$. Since the total flow on the paths in $\pset$ is equal to the total capacity of the edges of $E'_S$, we get that every path in $\pset$ must contain exactly one edge of $E'_S$.

It remains to show that
$F\leq 8n\cdot M'$. Recall that, if  $(s,T)$ is a good pair, and all shortcut operations so far have been valid, then, from \Cref{claim: shortcut}, the value of the maximum $s^{\out}$-$t$ flow in $G'$ is $\opt_s$, and so the value of the maximum $s^{\out}$-$t$ flow in the residual flow network $H$ is equal to $\opt_s-\val(f)$. Since flow $f$ was obtained by augmenting the flow $f_{i-1}$, from Invariant \ref{inv: flow}, $\val(f)\geq \val(f_{i-1})> \opt_s-4nM_{i-1}=\opt_s-8nM'$, since
$M'=M_i=M_{i-1}/2$. Therefore, $F=\opt_s-\val(f)\leq 8nM'$.
\end{proof}

\subsection{Properties of the Maximum $s^{\out}$--$t'$ Flow in $\tilde H'$}

Next, we explore the properties of the maximum $s^{\out}$-$t'$ flow in $\tilde H'$, whose value is denoted by $F'$. We start by defining and analyzing a special type of $s^{\out}$-$t'$ paths in $\tilde H'$, that we call \emph{problematic} paths.

\begin{definition}[Problematic path]\label{def: problematic path}
Consider any $s^{\out}$-$t'$ path $P$ in $\tilde H'$, and denote its last edge by $e=(x,t')$. Let $e'=(x,y)\in E(H)$ be the edge of $H$ that edge $e$ represents. We say that path $P$ is \emph{problematic}, if $E(P)\cap E'_S=\emptyset$, and $e'\not\in E'_S$.
\end{definition}

We will use the following observation in order to bound the total flow carried by problematic paths in any $s^{\out}$-$t'$ flow in $\tilde H'$.

\begin{observation}\label{obs: y in L}
Assume that $(s,T)$ is a good pair, that all shortcut operations so far have been valid, and that Event~\ref{event: heavy err} did not happen.
	Let $P$ be any $s^{\out}$-$t'$ path in $\tilde H'$ that is problematic, let $e=(x,t')$ be the last edge on $P$, and let $e'=(x,y)$ be the edge of $H$ that edge $e$ represents. Then either (i) $x=v^{\out}$ and $y=v^{\inn}$ for some vertex $v\in L$, so $e'$ is a backward special edge representing a vertex in $L$;
 or  (ii) $y\in L^{\out}$.
\end{observation}
\begin{proof}
	Let $P'\subseteq H$ be the path obtained from $P$ by replacing its last edge with edge $e'=(x,y)$. Since $P'$ contains no edges of $E'_S$, and its first vertex is $s^{\out}\in L^{\out}$, from \Cref{obs: edges crossing}, all vertices on $P'$ must lie in $L^*\cup S^{\inn}$.
	
	Assume first that $e'$ is a special edge.
	If $e'$ a backward special edge, that is, $e'=(v^{\out},v^{\inn})$ for some vertex $v\in V(G)$, then, since all vertices of $P'$ are contained in $L^*\cup S^{\inn}$, it must be the case that $v^{\out}\in L^{\out}$, and so $v\in L$. 
	
	 If $e'$ is a forward special edge, that is, $e'=(v^{\inn},v^{\out})$ for some vertex $v$, then, since $e'\not\in E'_S$, $v\in L$ and $v^{\out}\in L^{\out}$ must hold. 
	
Lastly, assume that $e'$ is a regular edge. Then from Property \ref{prop: regular incident to t'} of \Cref{obs: props of tilde H}, $x= v^{\inn}$ for some vertex $v^{\inn}\in V^{\inn}$ and $y=u^{\out}$ for some vertex $u^{\out}\in V^{\out}$. Since $y\in L^*\cup S^{\inn}$, we get that $y\in L^{\out}$ must hold. 
\end{proof}

We now obtain the following immediate corollary of the observation.

\begin{corollary}\label{cor: flow on special paths}
	Assume that $(s,T)$ is a good pair, that all shortcut operations so far have been valid, and that Event~\ref{event: heavy err} did not happen.
	Let $f''$ be any $s^{\out}$-$t'$ flow in $\tilde H'$, let $\pset$ be a flow-path decomposition of $f''$, and let $\pset'\subseteq \pset$ be the set of problematic paths. Then $\sum_{P\in \pset'}f''(P)\leq 2^{14}n^{\eps}\cdot M'\cdot \gamma \cdot \log^2 n\cdot \log(\wmax)$.
\end{corollary}

\begin{proof}
	We assume that $(s,T)$ is a good pair, that all shortcut operations so far have been valid, and that Event~\ref{event: heavy err} did not happen.
	Let $A^*$ be the set of all edges $(x,y)\in E(H)$, where $x\in V(J)$ and $y\in L^*\setminus V(J)$. From \Cref{cor: edges entering L}, 
	$\sum_{e\in A^*}\hat c(e)\leq  2^{14}\cdot n^{\eps}\cdot \gamma\cdot M'\cdot \log^2 n\cdot \log(\wmax)$.

Consider now any $s^{\out}$-$t'$ flow $f''$ in $\tilde H'$, and let $\pset$ be a flow-path decomposition of $f''$. Let $\pset'\subseteq \pset$ be the set of all problematic paths. Then from \Cref{obs: y in L}, for every path $P\in \pset'$, the last edge on $P$ is a copy of an edge in $A^*$, and so $\sum_{P\in \pset'}f''(P)\leq \sum_{e\in A^*}\hat c(e)\leq 2^{14}\cdot\gamma\cdot M'\cdot n^{\eps}\cdot\log^2 n\cdot\log(\wmax)$.	
\end{proof}

For brevity, we denote by $\eta=2^{14}\cdot\gamma\cdot M'\cdot n^{\eps}\cdot\log^2 n\cdot\log(\wmax)$.
Recall that we denoted by $U\subseteq V(G)$ the set of all vertices $v\in V(G)$, for which $w(v)\geq M'$ holds.
Let $Q'_0$ denote the set of all vertices $v\in U$ with $v^{\inn}\not\in V(J)$. 
The next claim bounds the value $F'$ of the maximum $s^{\out}$--$t'$ flow in $\tilde H'$ in terms of $F$, and also shows that
 only a small number of vertices of $Q'_0$ may lie in $S\cup L$, provided that $(s,T)$ is a good pair, all shortcut operations so far have been valid, and bad event~\ref{event: heavy err} did not happen.

\begin{claim}\label{claim: max flow in tilde H}
Assume that $(s,T)$ is a good pair, that all shortcut operations so far have been valid, and that Event~\ref{event: heavy err} did not happen. Then $F\leq F'\leq F+\eta$, and $|Q'_0\cap (S\cup L)|\leq \frac{2\eta}{M'}\leq n^{1-8\eps}$.
\end{claim}
\begin{proof}
	Recall that graph $\tilde H'$ is obtained from $H$ by contracting all vertices of $V(H)\setminus V(J)$ into the sink $t'$, and recall that $t\not\in V(J)$. It is then immediate to see that any $s^{\out}$-$t$ flow in $H$ defines an $s^{\out}$-$t'$ flow in $\tilde H'$ of the same value, so $F'\geq F$ must hold.
	
Next, we show that $|Q'_0\cap (S\cup L)|\leq \frac{2\eta}{M'}\leq n^{1-8\eps}$.	
	Consider  an optimal $s^{\out}$-$t$ flow $f'$ in $H$, whose value is $F$, and let $\pset$ be a flow-path decomposition of this flow. 
Let $E''_S=\set{(v^{\inn},v^{\out})\mid v\in Q'_0\cap S}$ be the set of all forward special edges of $H$ that correspond to the vertices of $Q'_0\cap S$, so $E''_S\subseteq E^0_S$.
For every edge $e\in E''_S$, let $\pset(e)\subseteq \pset$ be the set of flow-paths containing the edge $e$.
	From \Cref{claim: max flow in H}, $F=\hat c_S$, and every flow-path in $\pset$ contains exactly one edge of $E'_S$. Therefore, for every edge $e\in E''_S$, $\sum_{P\in \pset(e)}f'(P)=\hat c(e)$.
	
Flow $f'$ naturally defines an $s^{\out}$-$t'$ flow $\hat f$ of value $F$ in graph $\tilde H'$, as follows. For every path $P\in \pset$, let $a_P$ be the first vertex on $P$ that does not belong to $V(J)$, and let $P'$ be obtained by taking the subpath of $P$ from $s^{\out}$ to $a_P$, and replacing $a_P$ with $t'$. We then set the flow $\hat f(P')=f'(P)$. It is immediate to verify that $\hat f$ is a valid $s^{\out}$--$t'$ flow of value $F$ in graph $\tilde H'$. Consider now any edge $e=(v^{\inn},v^{\out})\in E''_S$, and a path $P\in \pset(e)$. Since $v^{\inn}\not\in V(J)$, the corresponding path $P'$ must be problematic.
From \Cref{cor: flow on special paths}: 

\[\sum_{e\in E''_S}\sum_{P\in \pset(e)}\hat f(P')\leq 2^{14}n^{\eps}\cdot M'\cdot \gamma \cdot \log^2 n\cdot \log(\wmax)=\eta.\]

On the other hand, since, from Property \ref{prop: residual capacity weaker}, the capacity of every edge $e\in E''_S$ is at least $M'$, we get that:

\[ \sum_{e\in E''_S}\sum_{P\in \pset(e)}\hat f(P')= \sum_{e\in E''_S}\sum_{P\in \pset(e)}f''(P)=\sum_{e\in E''_S}\hat c(e)\geq M'\cdot |E''_S|=M'\cdot |Q'_0\cap S|.\]

Altogether, we get that: 

\[|Q'_0\cap (S\cup L)|\leq |Q'_0\cap S|+|L|\leq \frac{\eta}{M'}+n^{\eps}\leq \frac{2\eta}{M'}\leq n^{1-8\eps},\]

from Inequality \ref{eq: small L}. We have also used the fact that $\eta=2^{14}n^{\eps}\cdot M'\cdot \gamma \cdot \log^2 n\cdot \log(\wmax)$, $\gamma\leq 2n^{1-10\eps}$ from Inequality \ref{eq: gamma2}, and $n^{\eps}>2^{20}\cdot\log^2 n\cdot\log(\wmax)$ from Inequality \ref{eq: neps large}.

	Consider now a maximum $s^{\out}$-$t'$ flow $f''$ in $\tilde H'$, whose value is $F'$, and let $\pset'$ be a flow-path decomposition of this flow. We partition the flow-paths in $\pset'$ into two subsets: set $\pset'_1$ containing all paths $P$ that are non-problematic, and set $\pset'_2$ containing all problematic paths. Since every path of $\pset'_1$ contains an edge of $E'_S$ (or an edge representing it), $\sum_{P\in \pset'_1}f''(P)\leq \sum_{e\in E'_S}\hat c(e)=\hat c_S=F$  from \Cref{claim: max flow in H}. 
	From \Cref{cor: flow on special paths}, $\sum_{P\in \pset'_2}f''(P)\leq 2^{14}\cdot\gamma \cdot M'\cdot n^{\eps}\cdot\log^2 n\cdot\log(\wmax)=\eta$.
	Therefore: 
	
	\[\val(f)\leq \sum_{P\in \pset'_1}f''(P)+\sum_{P\in \pset'_2}f''(P)\leq F+\eta.\]
\end{proof}

\subsection{Flow Deficit and Bad Batches of Vertices.}
Recall that we have denoted by $Q'_0$ the set of all vertices $v\in U$ with $v^{\inn}\not\in V(J)$, and that, from \Cref{claim: max flow in tilde H} if $(s,T)$ is a good pair, all shortcut operations so far are valid, and Event~\ref{event: heavy err} did not happen, then $|Q'_0\cap (S\cup L)|\leq  n^{1-8\eps}$. 
We now denote by $Q_0$ all remaining vertices of $U$, that is, vertices $v\in U$ with $v^{\inn}\in V(J)$. 
Our goal is to compute a partition $(Q,Q')$ of $U$ that has the following property: on the one hand, only a small number of vertices of $Q'$ belong to $S\cup L$; on the other hand, in any maximum $s^{\out}$-$t'$ flow in $\tilde H'$, for most vertices $v\in Q$, the total flow on the parallel edges $(v,t')$ is close to their capacity. We will ensure that $Q'_0\subseteq Q'$, so it now remains to compute the partition of $Q_0$. A central notion that we use in in order to compute the partition is that of a \emph{flow deficit}, which we define next.

\paragraph{Flow deficit.}
For every vertex $v\in Q_0$, we denote by $A(v)$ the set of all edges $(v^{\inn},x)\in E(H)$ with $x\not\in V(J)$, and by $\beta(v)$ their total capacity in $H$. Notice that every edge $(v^{\inn},x)\in A(v)$ is represented by a distinct edge $(v^{\inn},t')$ in $\tilde H'$; we do not distinguish between these edges, so sometimes we may equivalently view $A(v)$ as the set of all parallel edges $(v^{\inn},t')$ of $\tilde H'$.
Given an $s^{\out}$-$t'$ flow $f'$ in $\tilde H'$, for every vertex $v\in Q_0$, the \emph{flow deficit of $v$}, denoted by $\Delta_{f'}(v)$, is $\beta(v)-\sum_{e\in A(v)}f'(e)$.

For convenience, for any subset $\hat Q\subseteq Q_0$ of vertices, we will denote by $\beta(\hat Q)=\sum_{v\in \hat Q}\beta(v)$, and by $A(\hat Q)=\bigcup_{v\in \hat Q}A(v)$.

 Intuitively, the flow deficit measures the amount of the capacity of the edges of $A(v)$ that is unused by the flow. The following claim is  key to the second step. Intuitively, the claim shows that, if we are given any maximum $s^{\out}$-$t'$ flow in $\tilde H'$, then the total deficit of the vertices of $Q_0\cap S$ must be small. This property will be used in order to identify subsets $\hat B\subseteq Q_0$ of vertices that only contain relatively few vertices of $S$; we will refer to such subsets as \emph{bad batches} of vertices.
 
\begin{claim}\label{claim: small deficit}
	Assume that $(s,T)$ is a good pair, that all shortcut operations so far have been valid, and that Event~\ref{event: heavy err} did not happen. Let $f'$ be any maximum $s^{\out}$-$t'$ flow in $\tilde H'$. Then the total flow deficit of the vertices of $Q_0\cap S$ is $\sum_{v\in Q_0\cap S}\Delta_{f'}(v)\leq 2\eta$.
\end{claim}
\begin{proof}
	Assume that $(s,T)$ is a good pair, that all shortcut operations so far have been valid, and that Event~\ref{event: heavy err} did not happen.
	Let $f'$ be any maximum $s^{\out}$-$t'$ flow in $\tilde H'$. From \Cref{claim: max flow in tilde H} and \Cref{claim: max flow in H}, the value of the flow is $F'\geq F=\hat c_S=\sum_{e\in E'_S}\hat c(e)$. Let $\pset$ be a flow-path decomposition of $f'$, and let $(\pset',\pset'')$ be a partition of $\pset$, where set $\pset''$ contains all problematic paths, and set $\pset'$ contains all remaining paths. From \Cref{cor: flow on special paths}, $\sum_{P\in \pset''}f'(P)\leq 2^{14}n^{\eps}\cdot M'\cdot \gamma \cdot \log^2 n\cdot \log(\wmax)=\eta$.
	Therefore, $\sum_{P\in \pset'}f'(P)\geq F'-\eta \geq \sum_{e\in E'_S}\hat c(e)-\eta$.
	Note that, from the definition of problematic paths, every path $P\in\pset'$ contains an edge of $E'_S$ or its copy; in the discussion below, we do not distinguish between edges of $E'_S$ and their copies in $\tilde H'$. Therefore:
	$\sum_{e\in E'_S}f'(e)\geq \sum_{P\in \pset'}f'(P)\geq \sum_{e\in E'_S}\hat c(e)-\eta$. Altogether, we get that:
	
	\[\sum_{e\in E'_S}(\hat c(e)-f'(e))\leq \eta.\]
	
	Let $A'=\bigcup_{v\in Q_0\cap S}A(v)$, and let $E''_S=A'\cap E'_S$.
	Since, for every edge $e\in E'_S$, $f'(e)\leq \hat c(e)$, we get that:
	
	\[ \sum_{e\in E''_S}(\hat c(e)-f'(e))\leq \sum_{e\in E'_S}(\hat c(e)-f'(e))\leq \eta. \]
	
	Since $E''_S\subseteq A'$, we get that:

\begin{equation}\label{eq: deficit}	
\sum_{e\in A'}(\hat c(e)-f'(e))= \sum_{e\in A'\setminus E''_S}(\hat c(e)-f'(e))+\sum_{e\in E''_S}(\hat c(e)-f'(e))\leq \sum_{e\in A'\setminus E''_S}\hat c(e)+\eta.  
\end{equation}

	Finally, we bound $\sum_{e\in A'\setminus E''_S}\hat c(e)$ in the following observation.

	\begin{observation}\label{obs: compare capacities}
			Assume that $(s,T)$ is a good pair, that all shortcut operations so far have been valid, and that Event~\ref{event: heavy err} did not happen. Then
		$\sum_{e\in A'\setminus E''_S}\hat c(e)\leq \eta$.
	\end{observation}
\begin{proof}
	Consider some vertex $v\in Q_0\cap S$, and recall that we have denoted by $A(v)$ the set of all edges $(v^{\inn},x)\in E(H)$ with $x\not\in V(J)$.
	Consider now any edge $e=(v^{\inn},x)\in A(v)$. From the construction of the split graph, $x\in V^{\out}$ must hold, so we denote $x=u^{\out}$ for some vertex $u\in V(G)$. If $u=v$, then edge $e$ lies in $E^0_S$, and if $u\in (R\cup S)\setminus\set{v}$, then $e\in E^1_S\cup E^2_S$. Therefore, the only way that $e\not\in E'_S$ is if $u\in L$. We conclude that for every edge $e=(x,y)$ with $e\in A'\setminus E''_S$, it must be the case that $x\in V(J)$ and $y\in L^*\setminus V(J)$. From \Cref{cor: edges entering L}, $\sum_{e\in A'\setminus E''_S}\hat c(e)\leq  2^{14}\cdot n^{\eps}\cdot \gamma\cdot M'\cdot \log^2 n\cdot \log(\wmax)=\eta$.
\end{proof}
	
Altogether, we get that:

\[\sum_{v\in Q_0\cap S}\Delta_{f'}(v)=\sum_{e\in A'}(\hat c(e)-f'(e))\leq \sum_{e\in A'\setminus E''_S}\hat c(e)+\eta\leq 2\eta. \]
\end{proof}

Recall that $M'=M_i=\frac{\wmax'}{2^i\cdot \gamma}$, and that the number of phases in the algorithm is bounded by $z=\ceil{\log \left(\frac{8\wmax'\cdot n^2}{\gamma}\right )}$. 
Therefore,
$M'\geq M_z\geq \frac{\wmax'}{\gamma\cdot 2^z}$, and $\log(M')\geq -8\log(\wmax)$. We define a parameter $j_0$ as follows: if $M'\geq 1$, then $j_0=0$, and otherwise, $j_0=\log(M')$. From our discussion:

\begin{equation}\label{eq: bound on j0}
-8\log(\wmax)\leq j_0\leq 0.
\end{equation}

Notice that, since the capacity of every regular edge in $G''$ is $\wmax$, from Property \ref{prop: residual capacity weaker}, and since the flow $f$ is $M'$-integral, the capacity of every edge in $H$ is at least $2^{j_0}$.

For an integer $j\geq j_0$, let $B_j=\set{v\in Q_0\mid \beta(v)\geq 2^j}$; for brevity, we will refer to the vertices of $B_j$ as \emph{$j$-interesting} vertices. We next bound the number of $j$-interesting vertices that may lie in $S$, for all such $j$. The following claim shows that for sufficiently large values of $j$, $|B_j\cap S|$ must be small, and so in particular such a set $B_j$ of vertices may only contain few vertices of $S\cup L$.

\begin{claim}\label{claim: bound j-interesting}
	Assume that $(s,T)$ is a good pair, that all shortcut operations so far have been valid, and that Event~\ref{event: heavy err} did not happen.
	Then for all $j\geq j_0$, $|B_j\cap S|\leq \frac{9n\cdot M'}{2^{j}}$.
\end{claim}
\begin{proof}
		Assume that $(s,T)$ is a good pair, that all shortcut operations so far have been valid, and that Event~\ref{event: heavy err} did not happen.
	Fix an integer $j\geq j_0$, and consider a maximum $s^{\out}$-$t'$ flow $f'$ in graph $\tilde H'$. Let $\pset$ be a flow-path decomposition of $f'$. 
	
	Let $A'=A(B_j\cap S)=\bigcup_{v\in B_j\cap S}A(v)$. Note that, in graph $\tilde H'$, every edge of $A'$ connects a vertex $v^{\inn}$, for $v\in B_j\cap S$, to the vertex $t'$. 
	Therefore, we can assume w.l.o.g. that every flow-path in $\pset$ contains at most one edge of $A'$. From \Cref{claim: max flow in tilde H} and \Cref{claim: max flow in H}:

	\[\val(f')=F'\leq F+\eta\leq  8n M'+\eta.\]
	
	Therefore:
	
	\[\sum_{e\in A'}f'(e)\leq \val(f')\leq 8n M'+\eta.\]
	
We can then lower-bound the total deficit of the vertices of $B_j\cap S$ with respect to $f'$ as follows:
	
	\[
	\begin{split}
	\sum_{v\in B_j\cap S}\Delta_{f'}(v)&=\sum_{v\in B_j\cap S}\left (\beta(v)-\sum_{e\in A(v)}f'(e)\right )\\
	&\geq |B_j\cap S|\cdot 2^j-\sum_{e\in A'}f'(e)\\
	&\geq |B_j\cap S|\cdot 2^j-8n M'-\eta.
	\end{split}
	\]
	
	On the other hand, from \Cref{claim: small deficit}, $\sum_{v\in B_j\cap S}\Delta_{f'}(v)\leq 2\eta$ must hold.
	
	Therefore, we get that:
	
	\[|B_j\cap S|\leq \frac{3\eta}{2^j}+\frac{8n M'}{2^j}\leq \frac{9n M'}{2^{j}},\]
	
	since $\eta=2^{20}n^{\eps}\cdot M'\cdot \gamma \cdot \log^2 n\cdot \log(\wmax)\leq nM'$, as $\gamma\leq 2n^{1-10\eps}$ from Inequality \ref{eq: gamma2} and $n^{\eps}\geq 2^{14}\log^2 n\log(\wmax)$ from Inequality \ref{eq: neps large}.
\end{proof}

Let $j^*$ be the smallest integer for which $2^{j^*}\geq 4 n^{9\eps}\cdot M'$ holds (and notice that $j^*$ may be negative). Clearly:

\begin{equation}
4 n^{9\eps}\cdot M'\leq 2^{j^*}\leq 8 n^{9\eps}\cdot M'.\label{eq: bound on j*}
\end{equation} 

Additionally, since $\gamma\ge n^{10\eps}$ from Inequality \ref{eq: gamma3}, we get that:

\begin{equation}
 2^{j^*}\leq \frac{\gamma\cdot M'}{16}.\label{eq: bound2 on j*}
\end{equation}

Let $Q'_1\subseteq Q_0$ be the collection of all vertices $v\in Q_0$ with $\beta(v)\geq 2^{j^*}$; equivalently, $Q'_1=B_{j^*}$. Then, from \Cref{claim: bound j-interesting}, 
if $(s,T)$ is a good pair, all shortcut operations so far have been valid, and Event~\ref{event: heavy err} did not happen, it must be the case that: 

\begin{equation}\label{eq: bound1 q'}
|Q'_1\cap (S\cup L)|\leq |Q'_1\cap S|+|Q'_1\cap L|\leq \frac{9nM'}{2^{j^*}}+n^{\eps}\leq 2n^{1-8\eps}.
\end{equation}


We are now ready to define a central notion for the current step: the $j$-bad batches of vertices.

\begin{definition}[$j$-bad batch of vertices.]
	Let $f'$ be any maximum $s^{\out}$--$t'$ flow in $\tilde H'$, and let $j_0\leq j< j^*$ be an integer. A subset $\hat B\subseteq Q_0\setminus Q'_1$ of vertices is a \emph{$j$-bad batch} with respect to $f'$ if the following hold:
	
	\begin{itemize}
		\item for every vertex $v\in \hat B$, $\Delta_{f'}(v)\geq 2^j$; and
		
		\item $|\hat B|\geq \frac{n^{5.5\eps}\cdot \gamma\cdot M'}{2^j}$.
	\end{itemize}
\end{definition}

In the next key claim we show that only a small fraction of vertices in a $j$-bad batch may lie in $S\cup L$.

\begin{claim}\label{claim: bad batch not in SL}
		Assume that $(s,T)$ is a good pair, that all shortcut operations so far have been valid, and that Event~\ref{event: heavy err} did not happen. Let $\hat B\subseteq Q_0\setminus Q'_1$ be a $j$-bad batch of vertices with respect to some maximum $s^{\out}$--$t'$ flow  $f'$ in $\tilde H'$, for some $j_0\leq j<j^*$. Then $|\hat B\cap(S\cup L)|\leq \frac{|\hat B|}{n^{4\eps}}$.
\end{claim}
\begin{proof}
	Assume that $(s,T)$ is a good pair, that all shortcut operations so far have been valid, and that Event~\ref{event: heavy err} did not happen. Let $\hat B\subseteq Q_0\setminus Q'_1$ be a $j$-bad batch of vertices with respect to some maximum $s^{\out}$--$t'$ flow $f'$ in $\tilde H'$, for some $j_0\leq j<j^*$. Denote by $B'=\hat B\cap(S\cup L)$, and assume for contradiction that $|B'|>\frac{|\hat B|}{n^{4\eps}}$. Then:
	
	\[|B'|>\frac{|\hat B|}{n^{4\eps}}\geq  \frac{n^{1.5\eps}\cdot \gamma\cdot M'}{2^j}.\]
	
	Recall that $\gamma\cdot M'\geq 2^{j^*}\geq 2^j$, from Inequality \ref{eq: bound2 on j*}. Therefore, $|B'\cap L|\leq n^{\eps}\leq \frac{|B'|}{2}$ (we have used Inequality \ref{eq: small L}). But then we get that:
	
	\[\sum_{v\in B'\cap S}\Delta_{f'}(v)\geq 2^j\cdot |B'\cap S|\geq  \frac{n^{1.5\eps}\gamma\cdot M'}{2}>2\eta, \]
	
	contradicting \Cref{claim: small deficit} (we have used the fact that $\eta=2^{14}n^{\eps}\cdot M'\cdot \gamma \cdot \log^2 n\cdot \log(\wmax)$, and that $n^{\eps/10}>2^{20}\log^2 n\cdot \log (\wmax)$ from Inequality \ref{eq: neps large})
\end{proof}

Let $Q'_2\subseteq Q_0$ contain all vertices $v\in Q_0$ with $v^{\out}\in V(J)$. It is easy to verify that:

\begin{equation}\label{eq: start q2}
|Q'_2|\leq N=\frac{{32n^{1+2\eps}\cdot \log^2(\wmax)}}{\gamma}\leq n^{1-5\eps},
\end{equation} 

since $N=\frac{32n^{1+2\eps}\cdot  \log^2(\wmax)}{\gamma}$, 
$\gamma\geq n^{8\eps}$ from Inequality \ref{eq: gamma3}, and $n^{\eps}\leq 2^{20}\log n\log (\wmax)$ from Inequality \ref{eq: neps large}.

\subsection{The Algorithm}

We are now ready to describe the algorithm for Step~\ref{step2: partition}. Our algorithm works with graph $\tilde H$ -- the graph that is obtained  from $\tilde H'$ by unifying the sets of parallel edges that are incident to $t'$. Note that, for every vertex $v\in Q_0$, $\beta(v)$ is equal to the capacity of the edge $(v,t')$ in graph $\tilde H$. We compute the set $Q'_1=B_{j^*}$ of vertices, and we  
initialize $Q=Q_0\setminus (Q'_1\cup Q'_2)$ and $Q'=Q'_0\cup Q'_1\cup Q'_2$. Notice that $(Q,Q')$ is a partition of $U$.

Our algorithm performs a number of iterations. In every iteration, we identify some $j$-bad batches $\hat B\subseteq Q$ of vertices, for some $j_0\leq j<j^*$, and move the vertices of each such batch $\hat B$ from $Q$ to $Q'$. We ensure that all such batches of vertices are disjoint.

Specifically, in order to execute a single iteration, we set up an instance of maximum min-cost flow in graph $\tilde H$, with source $s^{\out}$ and destination $t'$. For every vertex $v$ that currently lies in $Q$, we set the \emph{cost} of the edge $(v,t')$ to be $1$, and the costs of all other edges of $\tilde H$ are set to $0$. We then compute a maximum min-cost flow from $s^{\out}$ to $t'$ in $\tilde H$, using the algorithm from \Cref{thm: maxflow}, obtaining an integral flow $f'$. For every vertex $v\in Q$, we compute its deficit $\Delta_{f'}(v)=\beta(v)-f'(v,t')$. For all $j_0\leq j<j^*$, we let $\hat B_j\subseteq Q$ be the set of all vertices $v\in Q$ with $2^j\leq \Delta_{f'}(v)<2^{j+1}$. For every integer $j_0\leq j<j^*$ for which $|\hat B_j|\geq\frac{n^{5.5\eps}\cdot \gamma\cdot M'}{2^j}$ holds, we move all vertices of $\hat B_j$ from $Q$ to $Q'$; notice that, in this case, $\hat B_j$ is a $j$-bad batch with respect to $f'$.
If, for any integer $j_0\leq j<j^*$, $|\hat B_j|\geq\frac{n^{5.5\eps}\cdot \gamma\cdot M'}{2^j}$ holds, then we say that the current iteration is of \emph{type 1}. Otherwise, 
 for all $j_0\leq j<j^*$, $|\hat B_j|<\frac{n^{5.5\eps}\cdot \gamma\cdot M'}{2^j}$ holds, and we say that the current iteration is of \emph{type 2}. In this case, the current iteration becomes the last one.
Notice that the time required to execute every iteration is bounded by $O\left(|E(\tilde H)|^{1+o(1)}\cdot \log (\wmax)\right )$. Since every edge of $\tilde H$ is either incident to $t'$ or has an endpoint in $\Gamma$, we get that: $|E(\tilde H)|\leq O\left (N\cdot n\right )\leq  O\left(\frac{n^{2+2\eps}\cdot \log^2(\wmax)}{\gamma}\right )$, since $N=\frac{32n^{1+2\eps}\cdot  \log^2(\wmax)}{\gamma}$. Altogether, the running time of a single iteration is $O\left(\frac{n^{2+2\eps+o(1)}\cdot \log^3(\wmax)}{\gamma}\right )$.
We bound the number of iterations in the following claim.

\begin{claim}\label{claim: number of iterations}
If $(s,T)$ is a good pair, all shortcut operations so far have been valid, and the event $\hat \event$ did not happen, then the total number of type-1 iterations is at most $\frac{216n^{1-5.5\eps}\cdot \log n\cdot \log(\wmax)}{\gamma}$.
\end{claim}
\begin{proof}
	For every integer $j_0\leq j<j^*$, let $B'_j=B_j\setminus B_{j+1}$, so $B'_j$ contains all vertices $v\in Q_0$, with $2^j\leq \beta(v)<2^{j+1}$. The following observation is key to proving the claim.
	
	\begin{observation}\label{obs: after first iteration}
		Assume that $(s,T)$ is a good pair, all shortcut operations so far have been valid, and the event $\hat \event$ did not happen.
		Then for every integer $j_0\leq j<j^*$, after the first iteration is completed, $|B'_j\cap Q|\leq \frac{64nM'}{2^{j}}$ holds.
	\end{observation}
\begin{proof}
		Assume that $(s,T)$ is a good pair, all shortcut operations so far have been valid, and the event $\hat \event$ did not happen.
	We fix an integer $j_0\leq j<j^*$ and prove the claim for this integer. Let $f'$ be the flow computed in the first iteration. We partition the vertices of $B'_j\cap Q$ into two subsets: set $X$ containing all vertices $v$ with $\Delta_{f'}(v)\geq 2^{j-1}$, and set $X'$ containing all remaining vertices. 
	
	Observe first that $|X|\leq \frac{4n^{5.5\eps}\cdot \gamma\cdot M'}{2^{j-1}}$ must hold. Indeed, assume otherwise, so $|X|> \frac{4n^{5.5\eps}\cdot \gamma\cdot M'}{2^{j-1}}$. Since, for every vertex $v\in B'_j$, $\beta(v)<2^{j+1}$, we get that $\Delta_{f'}(v)<2^{j+1}$ must hold as well. Therefore, for every vertex $v\in X$, $2^{j-1}\leq \Delta_{f'}(v)<2^{j+1}$ holds. We let $X_1\subseteq X$ be the set of all vertices $v\in X$ with $2^{j-1}\leq \Delta_{f'}(v)<2^j$, and we let $X_2$ be the set of all remaining vertices, so for each vertex $v\in X_2$, $2^{j}\leq \Delta_{f'}(v)<2^{j+1}$ holds. If $|X_1|\ge \frac{|X|}{2}>\frac{2n^{5.5\eps}\cdot \gamma\cdot M'}{2^{j-1}}$ holds, then, since $X_1\subseteq \hat B_{j-1}$, the set $\hat B_{j-1}$ of vertices should have been moved from $Q$ to $Q'$, including all vertices of $X_1$. Otherwise, $|X_2|\geq \frac{|X|}{2}>\frac{4n^{5.5\eps}\cdot \gamma\cdot M'}{2^{j}}$ holds, then, since $X_1\subseteq \hat B_{j-1}$, the set $\hat B_{j-1}$ of vertices should have been moved from $Q$ to $Q'$, a contradiction.
	We conclude that $|X|\leq \frac{4n^{5.5\eps}\cdot \gamma\cdot M'}{2^{j-1}}$.


	Consider now the set $X'$ of vertices. Observe that, for every vertex $v\in X'$, $\Delta_{f'}(v)< 2^{j-1}$, and, since $\beta(v)\geq 2^j$, we get that $f'(v,t')\geq 2^{j-1}$. Overall, we get that:
	
	\[\val(f')\geq \sum_{v\in X'}f'(v,t')\geq |X'|\cdot 2^{j-1}.\]
	
	On the other hand, 
	from \Cref{claim: max flow in tilde H} and \Cref{claim: max flow in H}, if $(s,T)$ is a good pair, all shortcut operations so far have been valid and Event $\hat \event$ did not happen:

	\[\val(f')=F'\leq F+\eta\leq  8n M'+\eta\leq 16n M',\]
	
	since $\eta=2^{14}n^{\eps}\cdot M'\cdot \gamma \cdot \log^2 n\cdot \log(\wmax)$, $\gamma\leq n^{1-2\eps}$ from Inequality \ref{eq: gamma2}, and $2^{14}\cdot \log^2 n\cdot\log(\wmax)\leq n^{\eps}$ from Inequality \ref{eq: neps large}.
	Altogether we get that $|X'|\leq \frac{16n M'}{2^{j-1}}$, and:
	
	\[|B'_j\cap Q|\leq |X|+|X'|\leq  \frac{4n^{5.5\eps}\cdot \gamma\cdot M'}{2^{j-1}}+ \frac{16n M'}{2^{j-1}}\leq \frac{64n M'}{2^{j}},\]
	since $\gamma\leq 2n^{1-10\eps}$ from Inequality \ref{eq: gamma2}.
\end{proof}

Since, for all $j_0\leq j<j^*$, $B_j=\bigcup_{j'=j}^{j^*-1}B'_j$, we get that, after the first iteration, for all $j_0\leq j<j^*$, $|B_j\cap Q|\leq \frac{128nM'}{2^{j}}$ holds. Consider now some iteration of the algorithm that is a type-1 iteration. 
Then there must be at least one integer $j_0\leq j<j^*$, for which, in the current iteration, $|\hat B_j|\geq\frac{n^{5.5\eps}\cdot \gamma\cdot M'}{2^j}$ held. We then say that this iteration is a \emph{type-$1(j)$ iteration} (if there are several such integers $j$, we choose one arbitrarily.)
Recall that $\hat B_j$ contains all vertices $v\in Q$ with $2^j\leq \Delta_{f'}(v)<2^{j+1}$, so in particular $\beta(v)\geq 2^j$ for each such vertex $v$, and therefore, $\hat B_j\subseteq B_j$ must hold. If some iteration $r$ is a type-$1(j)$ iteration, then the vertices of $\hat B_j$ are deleted from $Q$ during that iteration. Since, after the first iteration, $|B_j\cap Q|\leq \frac{128n M'}{2^{j}}$, and since each type-$1(j)$ iteration moves at least $\frac{n^{5.5\eps}\cdot \gamma\cdot M'}{2^j}$ vertices of $B_j$ from $Q$ to $Q'$, we get that the total number of type-$1(j)$ iterations is bounded by:

\[ \frac{128 n M'}{2^{j}}\cdot \frac{2^j}{n^{5.5\eps}\cdot \gamma\cdot M'}\leq \frac{128n^{1-5.5\eps}}{\gamma}.\]

Summing this up over all $j_0\leq j<j^*$, and recalling that $j^*\leq \log n\log(\wmax)$ from Inequality \ref{eq: bound2 on j*}, and $j_0\geq -8\log(\wmax)$, from Inequality \ref{eq: bound on j0}, we get that the total number of iterations is bounded by $\frac{216n^{1-5.5\eps}\cdot\log n\cdot \log(\wmax)}{\gamma}$. 
\end{proof}

We execute the iterations of min-cost flow computations until either a type-2 iteration is encountered, or more than  $\frac{256n^{1-5.5\eps}\cdot \log n\cdot \log(\wmax)}{\gamma}$ type-1 iterations occur. In the latter case, we terminate the algorithm for the current phase with a ``FAIL''. The following observation is immediate.

\begin{observation}\label{obs: fail in step 2}
	The algorithm for Step 2 may only terminate with a ``FAIL'' if either $(s,T)$ is not a good pair, or at least one shortcut operation executed before the current phase or during Step 1 of the current phase is invalid, or Event $\hat \event$ has happened.
\end{observation}

Since the running time of every iteration is $O\left(\frac{n^{2+2\eps+o(1)}\cdot \log^3(\wmax)}{\gamma}\right )$, we get that the total running time of Step~\ref{step2: partition} so far is:

\[\begin{split}
&O\left(\frac{n^{2+2\eps+o(1)}\cdot \log^3(\wmax)}{\gamma}\right )\cdot O\left(\frac{n^{1-5.5\eps}\cdot \log n\cdot \log(\wmax)}{\gamma}\right )\\
&\quad\quad\quad\quad\quad\quad\quad\quad\leq  O\left(\frac{n^{3-3.5\eps+o(1)}\cdot \log^4(\wmax)}{\gamma^2}\right )\\
&\quad\quad\quad\quad\quad\quad\quad\quad\leq  O\left(n^{2-4\eps+o(1)}\cdot \log^4(\wmax)\right ),
\end{split}\]

since $\gamma\geq n^{2/3+5\eps}$ from Inequality \ref{eq: gamma2}.

We now summarize the properties of the final partition $(Q,Q')$ of $U$ that the algorithm obtains. We start by showing that $Q'$ may only contain a relatively small number of vertices from $S\cup L$.

\begin{claim}\label{claim: Q' few S and L}
	Assume that $(s,T)$ is a good pair, that all shortcut operations so far have been valid, and that Event~\ref{event: heavy err} did not happen. Then	$|Q'\cap (S\cup L)|\leq 2n^{1-4\eps}$.
\end{claim}
\begin{proof}
	We assume  that $(s,T)$ is a good pair, that all shortcut operations so far have been valid, and that Event~\ref{event: heavy err} did not happen. Then from \Cref{claim: max flow in tilde H} $|Q'_0\cap (S\cup L)|\leq  n^{1-8\eps}$. Additionally, from Inequality \ref{eq: bound1 q'},  $|Q'_1\cap (S\cup L)|\leq 2n^{1-8\eps}$, and from Inequality \ref{eq: start q2}, $|Q'_2|\leq  n^{1-5\eps}$. Altogether:
	
	\[ |(Q_0'\cup Q'_1\cup Q_2')\cap (S\cup L)|\leq 2n^{1-5\eps}. \]
	
	For conveninence, denote $Q'_3=Q'\setminus (Q_0'\cup Q'_1\cup Q_2')$. From our algorithm, there is a partition $\bset$ of $Q'_3$ that has the following property: for every set $\hat B\in \bset$, there is an integer $j_0\le j<j^*$ and a maximum $s^{\out}$-$t'$ flow $f'$ in $\tilde H'$, such that $\hat B$ is a $j$-bad batch for $f'$. From \Cref{claim: bad batch not in SL}, for each such batch $\hat B\in \bset$, $|\hat B\cap(S\cup L)|\leq \frac{|\hat B|}{n^{4\eps}}$.
	Altogether, we get that:
	
	\[|Q'_3\cap (S\cup L)|=\sum_{\hat B\in \bset}|\hat B\cap (S\cup L)|\leq \sum_{\hat B\in \bset}\frac{|\hat B|}{n^{4\eps}}\leq n^{1-4\eps}.\]
	
	Altogether:
	
	\[|Q'\cap (S\cup L)|= |(Q_0'\cup Q'_1\cup Q_2')\cap (S\cup L)|+|Q'_3\cap (S\cup L)|\leq 2n^{1-5\eps}+n^{1-4\eps}\leq 2n^{1-4\eps}. \]
\end{proof}

Let $f^*$ be the min-cost maximum $s^{\out}$-$t'$ flow in $\tilde H$ computed in the last iteration, so the value of the flow must be $F'$. This flow naturally defines an $s^{\out}$-$t'$ flow in $\tilde H'$ of the same value, that we also denote by $f^*$. We also denote by $F^*=\sum_{e\in A(Q)}f^*(e)$ -- the total flow that $f^*$ sends on the edges of $A(Q)=\bigcup_{v\in Q}A(v)$. We summarize the central properties of the set $Q$ of vertices in the following claim.

\begin{claim}\label{claim: prop of Q}
	$F^*\geq \beta(Q)-4M'\cdot n^{5.5\eps}\cdot\gamma\cdot \log n\cdot\log(\wmax)$. Moreover, if $\hat f$ is any maximum $s^{\out}$-$t'$ flow in $\tilde H'$, then $\sum_{e\in A(Q)}\hat f(e)\geq F^*$.
\end{claim}
\begin{proof}
	The second assertion follows immediately from the definition of the min-cost flow: recall that the costs of all edges of $\tilde H'$ that represent the edges of $A(Q)$ are set to $1$, while the costs of all other edges are $0$. If there exists a  maximum $s^{\out}$-$t'$ flow $\hat f$ in $\tilde H'$ with $\sum_{e\in A(Q)}\hat f(e)< F^*$, then the cost of this flow would be lower than that of $f^*$, a contradiction.
	
	We now turn to prove the first assertion. As before, for all $j_0\leq j<j^*$, we let $\hat B_j\subseteq Q$ be the set of all vertices $v\in Q$ with $2^j\leq \Delta_{f^*}(v)<2^{j+1}$. Since the iteration where $f^*$ was computed is of type 2, we get that, for all $j_0\leq j<j^*$, $|\hat B_j|<\frac{n^{5.5\eps}\cdot \gamma\cdot M'}{2^j}$ holds.
	Therefore, for all $j_0\leq j<j^*$:
	
	\begin{equation}\label{eq: deficit in Bj}
		\sum_{v\in \hat B_j}\left (\beta(v)-\sum_{e\in A(v)}f^*(e)\right )\leq 2^{j+1}\cdot \frac{n^{5.5\eps}\cdot \gamma\cdot M'}{2^j}\leq 2M'n^{5.5\eps}\gamma
		\end{equation}
	
	Let $\hat Q=Q\setminus\left(\bigcup_{j=j_0}^{j^*-1}\hat B_j\right )$.
	Recall that, if $v\in Q$, then $\beta(v)<2^{j^*}$ must hold (as otherwise $v$ should have been added to $Q'_1$), and so $\Delta_{f^*}(v)<2^{j^*}$ must hold. Since the flow $f^*$ is $M'$-integral, and $j_0=\log M'$, we get that, for every vertex $v\in \hat Q$, $\Delta_{f^*}(v)=0$, and $\beta(v)=\sum_{e\in A(v)}f^*(e)$. By combining this with Inequality \ref{eq: deficit in Bj} for all $j_0\leq j<j^*$, and recalling that $j^*\leq \log n\cdot \log(\wmax)$, while $j_0\geq -8\log(\wmax)$ from Inequalities \ref{eq: bound on j0} and \ref{eq: bound2 on j*}, we get that:
	
	\[\sum_{v\in Q}\left (\beta(v)-\sum_{e\in A(v)}f^*(e)\right )\leq  2M'\cdot n^{5.5\eps}\cdot\gamma\cdot (j^*-j_0)\leq 4M'\cdot n^{5.5\eps}\cdot\gamma \cdot \log n\cdot\log(\wmax),
		\]
		
	and so:
	
	\[F^*=\sum_{v\in Q}\sum_{e\in A(v)}f^*(e)\geq \beta(Q)-4M'\cdot n^{5.5\eps}\cdot\gamma\cdot\log n\cdot\log(\wmax).
	\]	
\end{proof}

\subsection{The Final Cut}

We say that Case 1 happens if $|Q|> n^{1-4\eps}$, and otherwise we say that Case 2 happens. If Case 1 happens, then we need to perform one additional step that we describe below, after we prove the following simple observation.

\begin{observation}\label{obs: case 1}
	If Case 1 happened, then $F^*\geq n^{1-4\eps}\cdot M'/2$.
\end{observation}

\begin{proof}
	Recall that for every vertex $v\in Q$, $v^{\inn}\in V(J)$ and $v^{\out}\not\in V(J)$ holds (by the definition of the sets $Q_0'$ and $Q'_2$ of vertices), so the special edge $(v^{\inn},v^{\out})$ belongs to $A(v)$. From Property \ref{prop: residual capacity weaker}, the capacity of this edge in $H$ is at least $M'$, and so $\beta(v)\geq M'$ must hold. Since we assumed that Case 1 happened, $|Q|\geq n^{1-4\eps}$, so $\beta(Q)\geq n^{1-4\eps}\cdot M'$. From 
\Cref{claim: prop of Q}:

\[F^*\geq \beta(Q)-4M'n^{5.5\eps}\cdot \gamma\cdot \log n\log(\wmax)\geq \frac{n^{1-4\eps}\cdot M'} 2,\]

since $n^{\eps/10}>16\log n\cdot\log(\wmax)$ from Inequality \ref{eq: neps large} and $\gamma\leq 2n^{1-10\eps}$ from Inequality \ref{eq: gamma2}.
\end{proof}

Assume that Case 1 happened.
We start by providing an intuitive motivation and explanation of  the additional step required in this case.
For the sake of this informal exposition, we assume that $(s,T)$ is a good pair, that all shortcut operations so far have been valid, and that Event~\ref{event: heavy err} did not happen.

In Step~\ref{step3: sparsify}, we will sparsify the graph $H$, to compute a graph $H'\subseteq H$ (after possibly performing some shortcut operations). The goal of that step is to ensure that, on the one hand, $|E(H')|$ is quite small, while, on the other hand, the value of the maximum $s^{\out}$-$t$ flow in $H'$ is close to $F$ -- the value of the maximum $s^{\out}$-$t$ flow in $H$. One of the main ideas in sparsifying the graph $H$ is to use the fact that set $Q'$ contains only a small number of vertices of $S\cup L$. Therefore, if a vertex $v\in V(G)$ has many neighbors in $Q'$, then $v\in S\cup R$ must hold. We can then safely add a shortcut edge $(v^{\out},t')$, and delete all other edges leaving $v^{\out}$ from $H$. Additionally, in the sparsification step, we rely on the strong properties of the graph $J$ that are summarized in \Cref{obs: props of tilde H}, and that ensure that the total number of edges of $H$ with both endpoints in $J$, and of edges $(x,y)$ with $x\in J$ and $y\not\in J$, is relatively small. However, one major issue still remains, and it involves the set of regular edges $E^*=\set{(u^{\out},v^{\inn})\mid u^{\out}\not\in J,v^{\inn}\in J}$. If $|Q|$ remains sufficiently large, then the number of such edges may be large, and it is not immediately clear how to sparsify them.

For further intuition, consider any maximum $s^{\out}$-$t$ flow $f$ in $H$, and let $\pset$ be a flow-path decomposition of $f$. For every path $P\in \pset$, let $P'\subseteq P$ be the subpath of $P$ starting from $s^{\out}$ and terminating at the first vertex of $P$ that lies outside of $V(J)$. Let $P''$ be the subpath of $P$ starting from the last vertex of $P'$ and terminating at $t$. We say that path $P$ is \emph{meandering} with respect to $V(J)$, if $P''$ contains some vertex of $V(J)$. Notice that all paths in $\pset$ that use the edges of $E^*$ must be meandering. We also define a flow $\hat f$ in $H$, where for every path $P\in \pset$, we send $f(P)$ flow units via the path $P'$ (so in other words, we only send flow via the prefixes of the paths in $\pset$).

Let $A'=A(Q)$, and let $A''$ contain all other edges $e=(x,y)\not\in A'$ with $x\in V(J)$ and $y\not\in V(J)$. Notice that our algorithm ensures that, in any maximum $s^{\out}$-$t'$ flow in $\tilde H'$, the total flow on the edges of $A'$ is close to their total capacity. Assume for now that we could achieve the same property for the edges of $A''$. Consider now a maximum $s^{\out}$-$t$ flow $f$ in $H$, and let the set $\pset$ of paths and the flow $\hat f$ be defined as before. Since the value $F$ of the flow $f$ is quite close to the value $F'$ of the maximum $s^{\out}$-$t'$ flow in $\tilde H$, the edges of $A'\cup A''$ must be almost saturated by the flow $\hat f$. If $P\in \pset$ is a meandering path, then it needs to use at least two edges of $A'\cup A''$. It would then follow that the total flow on the meandering paths must be low, and so simply deleting the edges of $E^*$ from the graph $H$ would not significantly decrease the value of the maximum $s^{\out}$-$t$ flow in $H$.

Unfortunately, our algorithm only guarantees that, in any maximum $s^{\out}$-$t'$ flow $f'$ in $\tilde H'$, the total flow on the edges of $A'$ is close to the total capacity of these edges, but the flow on the edges of $A''$ may be much lower than their total capacity. This may potentially leave a lot of capacity that the meandering paths may exploit, and so we cannot claim directly that the total flow on the meandering paths is low. Therefore, deleting the edges of $E^*$ from $H$ may significantly reduce the value of the maximum $s^{\out}$-$t$ flow in $ H$. In fact it is not difficult to show that, in any maximum $s^{\out}$-$t$ flow in $H$, there is only a small subset $\hat Q\subseteq Q$ of vertices, such that the edges of $E^*$ incident to the vertices of $\set{v^{\inn}\mid v\in \hat Q}$ carry any significant amount of flow (this follows from the properties of the min-cost flow). So only a relatively small number of edges of $E^*$ are essential for approximately preserving the maximum $s^{\out}$-$t$ flow in $ H$. But unfortunately it is not clear how to identify such edges.

In order to overcome this difficulty, we compute a minimum cut in graph $\tilde H$, separating $s^{\out}$ from the edges of $A''$; denote this cut by $(\hat X,\hat Y)$, with $s^{\out}\in \hat X$. We denote by $Q'_3\subseteq Q$ the set of vertices $v\in Q$ with $v^{\inn}\in \hat Y$, and we show that $|Q'_3|$ is small. The vertices of $Q'_3$ are then deleted from $Q$ and added to $Q'$. If we now consider the graph $H^*$, that is obtained from $H$ by unifying all vertices of $V(H)\setminus \hat X$ into a vertex $t''$, then, as we show later, this graph now has the desired properties: the value of the maximum $s^{\out}$--$t''$ flow in $H^*$ is close to $F$, and in any such flow, all edges $(x,y)$ with $x\in \hat X$ and $y\not\in \hat X$ are close to being saturated. This allows us to delete from $H$ all edges $(u^{\out},v^{\inn})$ with $u^{\out}\not\in \hat X$ and $v^{\inn}\in \hat X$, in order to sparsify the graph $H$, without significantly decreasing the value of the maximum $s^{\out}$-$t$ flow in it. We now describe the algorithm for the final cut more formally. This algorithm is only executed if Case 1 happens, that is, $|Q|>n^{1-4\eps}$ holds after the last iteration of the min-cost flow calculations.

Recall that we denoted by $A'=A(Q)=\bigcup_{v\in Q}A(v)$, and by $A''$ all edges $(x,y)\in E(H)\setminus A'$ with $x\in V(J)$ and $y\not\in V(J)$.
We let $\tilde H''$ be the graph that is very similar to $\tilde H'$, except for the following changes. We delete the sink $t'$ and instead add two new sinks $t_1$ and $t_2$ to $\tilde H''$. For every edge $e=(x,y)\in A'$, we add the edge $(x,t_1)$ of capacity $\hat c(e)$ representing $e$ to $\tilde H''$, and for every edge $e'=(x',y')\in A''$, we add the edge $(x',t_2)$ of capacity $\hat c(e')$ representing $e'$ to $\tilde H''$. 

We then compute a minimum $s^{\out}$-$t_2$ cut $(\hat X,\hat Y)$ (with respect to edge capacities) in $\tilde H''$, where $s^{\out}\in \hat X$ and $t_2\in \hat Y$. 
In order to compute the cut, we unify all parallel edges, obtaining a graph with $O(|E(\tilde H)|)$ edges, and then apply the algorithm from \Cref{cor: mincut} to the resulting graph.
Note that the running time for this step is bounded by:

\[O(|E(\tilde H'')|)+O\left (|E(\tilde H|)^{1+o(1)}\cdot \log \wmax\right )\leq O(|E(\tilde H')|)+O\left (|E(\tilde H|)^{1+o(1)}\cdot \log \wmax\right ).\]

It is easy to see that this step does not increase the asymptotic running time of Step~\ref{step2: partition}, which remains bounded by $O\left(n^{2-4\eps+o(1)}\cdot \log^4(\wmax)\right )$.

Let $Q'_3\subseteq Q$ be the set of all vertices $v\in Q$ with $v^{\inn}\in \hat Y$. Additionally, let $\hat A$ be the set of all edges $e=(x,y)\in E(H)$, with $x\in \hat X$ and $y\not\in \hat X$.
In the following claim, we summarize the main properties of the cut $(\hat X,\hat Y)$.

\begin{claim}\label{claim: final cut}
	The total capacity $\sum_{e\in \hat A}\hat c(e)\leq F'+M'\cdot n^{1-4.4\eps}$. Additionally,
	$t_1\in \hat X$, and $|Q'_3|\leq n^{1-4.4\eps}$. 
\end{claim}
\begin{proof}
The following claim is central to the proof of \Cref{claim: final cut}.

\begin{claim}\label{claim: partial flow}
The value of the maximum $s^{\out}$-$t_2$ flow in $\tilde H''$ is equal to $F'-F^*$.
\end{claim}
\begin{proof}
Recall that we have computed an $s^{\out}$-$t'$ flow $f^*$ in $\tilde H'$ of value $F'$, where $\sum_{e\in A'}f^*(e)=F^*$, and so $\sum_{e\in A''}f^*(e)=F'-F^*$. This flow naturally defines an $s^{\out}$-$t_2$ flow in $\tilde H''$ of value $F'-F^*$. Therefore, the value of the maximum $s^{\out}$-$t_2$ flow in $\tilde H''$ is at least $F'-F^*$. It now remains to prove that it is at most $F'-F^*$.

Assume otherwise, that is, that the value of the maximum $s^{\out}$-$t_2$ flow in $\tilde H''$ is greater than $F'-F^*$. Let $r=\min\set{1,M'}$. Since all edge capacities in $G''$ are integral, and the flow $f$ is $M'$-integral, the value of the flow must be at least $F'-F^*+r$.
	
Let $Z$ be the graph obtained from $\tilde H''$ by adding a single destination vertex $t$, and adding an edge $(t_1,t)$ of capacity $F^*-r$, and an edge $(t_2,t)$ of capacity $F'-F^*+r$ (see \Cref{fig: Z}).
\begin{figure}[t]
	\centering
	\includegraphics[width=0.8\textwidth]{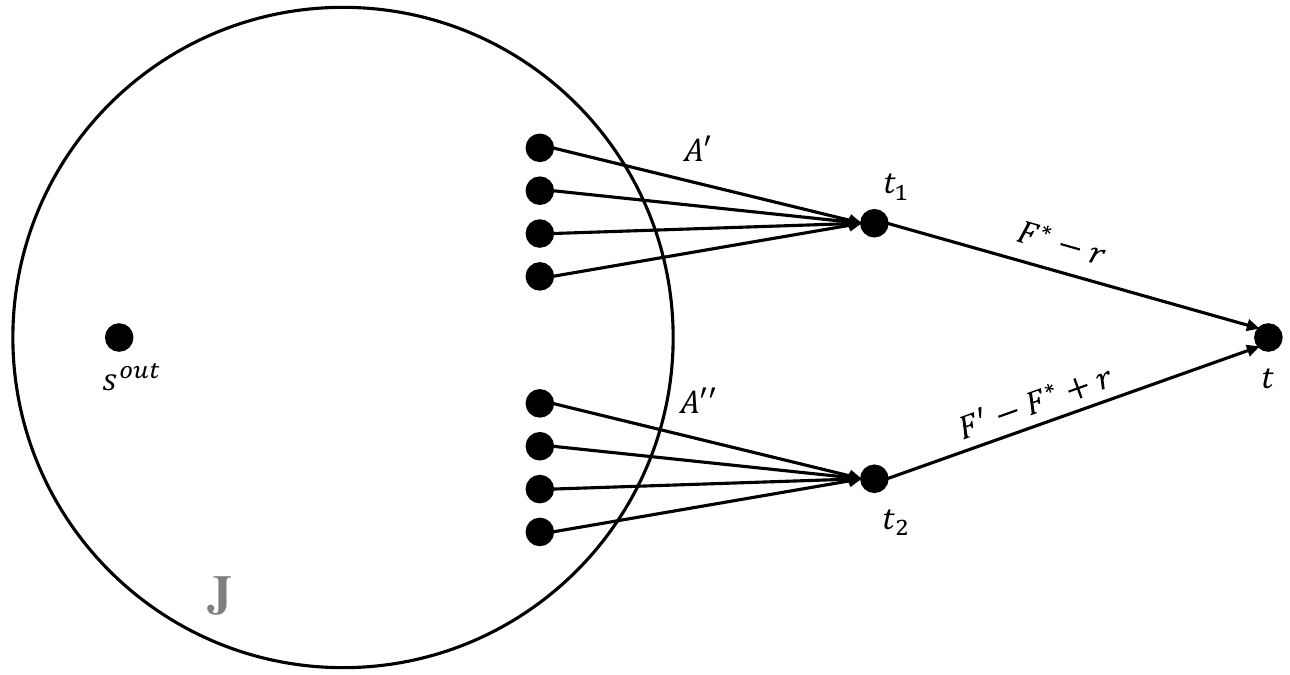}
	\caption[-]{A schematic view of graph $Z$.} \label{fig: Z}
\end{figure}
We use the following observation.
	
\begin{observation}\label{obs: violating flow}
If the value of the maximum $s^{\out}$-$t_2$ flow in $\tilde H''$ is at least  $F'-F^*+r$, then there is an $s^{\out}$-$t$ flow in $Z$ of value $F'$.
\end{observation}

The proof of \Cref{claim: partial flow} follows immediately from \Cref{obs: violating flow}. Indeed, assume for contradiction that the value of the maximum $s^{\out}$-$t_2$ flow in $\tilde H''$ is greater than $F'-F^*$, so it is at least $F'-F^*+r$.
Then, from \Cref{obs: violating flow}, there is an $s^{\out}$-$t$ flow in $Z$ of value $F'$. This flow naturally defines an $s^{\out}$-$t'$ flow $\hat f$ in $\tilde H'$ of value $F'$, in which $\sum_{e\in A'}\hat f(e)<F^*$, contradicting \Cref{claim: prop of Q}. In order to complete the proof of 
\Cref{claim: partial flow}, it is now enough to prove \Cref{obs: violating flow}, that we do next.

\begin{proofof}{\Cref{obs: violating flow}}
	For convenience, and in order to implify the notation, in this proof we denote $s^{\out}$ by $s$, and we denote the capacities of the edges $e\in E(Z)$ by $c(e)$.
	The proof provided here follows standard arguments, and it is very similar to the proof of Lemma 4.5 from \cite{chuzhoy2016improved}, which, in turn, relies on arguments suggested by Paul Seymour.

	Assume for contradiction that the value of the maximum $s$-$t$ flow in $Z$ is less than  $F'$. Then from the Max-Flow / Min-Cut Theorem,
	 there is an $s$-$t$ cut $(X,Y)$ in $Z$ with $\sum_{e\in E_Z(X,Y)} c(e)< F'$. Since the total capacity of the edges $(t_1,t)$ and $(t_2,t)$ is $F'$, it must be the case that either $t_1\in Y$, or $t_2\in Y$ (or both).
	
	Assume first that $t_1\in X$ and $t_2\in Y$. Then edge $(t_1,t)$ contributes its capacity $F^*-r$ to the cut. Let $E'$ denote the set of all edges $E_Z(X,Y)$, excluding the edge $(t_1,t)$. Then $\sum_{e\in E'}c(e)=\sum_{e\in E_Z(X,Y)} c(e)-(F^*-r)<F'-F^*+r$. However, from our assumption, there exists an $s$-$t_2$ flow $\hat f$ in $Z$, whose value is at least $F'-F^*+r$. If we consider a flow-path decomposition $\qset$ of this flow, it is immediate to verify that every flow-path in the decomposition needs to use an edge of $E'$, a contradiction.

	Assume next that $t_2\in X$ and $t_1\in Y$. Then edge $(t_2,t)$ contributes its capacity $F'-F^*+r$ to the cut. Let $E''$ denote the set of all edges $E_Z(X,Y)$, excluding the edge $(t_2,t)$. Then $\sum_{e\in E''}c(e)=\sum_{e\in E_Z(X,Y)} c(e)-(F'-F^*+1)<F^*-r$.
	Recall however that we have computed an $s^{\out}$-$t'$ flow in $\tilde H'$ of value $F'$, where the total flow on the edges of $A'$ is $F^*$. This flow can be naturally used to define an $s$-$t_1$ flow in $Z$ of value $F^*$. If we consider a flow-path decomposition $\qset'$ of this flow, then every flow-path has to use an edge of $E''$, a contradiction.
	
	Finally, assume that $t_1,t_2\in Y$. Denote $E'''=E_Z(X,Y)$. Recall that graph $\tilde H'$ contains an $s^{\out}$-$t'$ flow $f^*$ of value $F'$. This flow naturally defines a flow $\tilde f$ in $Z$ of value $F'$ from $s$ to $t_1$ and $t_2$. In other words, if we consider a flow-path decomposition $\qset''$ of this flow, then every path originates at $s$ and terminates at either $t_1$ or $t_2$. It is easy to verify that every flow-path in $\qset''$ must contain an edge of $E'''$. Since the value of the flow is $F'$, while the total capacity of the edges in $E'''$ is less than $F'$, this is a contradiction.
\end{proofof}

This completes the proof of \Cref{claim: partial flow}.
\end{proof}

We are now ready to prove that $t_1 \in \hat X$. Indeed, assume for contradiction that $t_1\not\in \hat X$. In this case, $\hat A=E_{\tilde H''}(\hat X,\hat Y)$ holds; recall all edges of $\hat A$ lie in $\tilde H'$. Moreover, the total capacity of the edges of $\hat A$ in $\tilde H'$ is $F'-F^*$ by \Cref{claim: partial flow}, and from the Max-Flow / Min-Cut theorem. Additionally, from \Cref{obs: case 1}, $F^*>0$, and so $\sum_{e\in \hat A}\hat c(e)=F'-F^*<F'$. Notice however that the set $\hat A$ of edges separates $s^{\out}$ from $t'$ in graph $\tilde H'$, contradicting the fact that there is an $s^{\out}$-$t'$ flow in $\tilde H'$ of value $F'$. We conclude that $t_1\in \hat X$ must hold.

Let $J'$ be the subgraph of $\tilde H''$ induced by the vertices of $\hat X$.
In order to bound the cardinality of $Q'_3$, we need the following simple observation.

\begin{observation}\label{obs: flow in J'}
	There is an $s^{\out}$--$t_1$ flow in $J'$ of value at least $F^*$.
\end{observation}
\begin{proof}
	For convenience, in this proof, we denote the capacities of the edges $e\in \tilde H''$ by $c(e)$.
	Assume otherwise. From the Max-Flow / Min-Cut theorem, there is a collection $E_1$ of edges in graph $\tilde H''$ with $\sum_{e\in E_1}c(e)<F^*$, such that in $J'\setminus E_1$ there is no path connecting $s^{\out}$ to $t_1$. Let $E_2=E_{\tilde H''}(\hat X,\hat Y)$, and recall that, from \Cref{claim: partial flow} and the Max-Flow / Min-Cut theorem, $\sum_{e\in E_2}c(e)=F'-F^*$. Clearly, $\tilde H''\setminus(E_1\cup E_2)$ contains no paths connecting $s$ to either $t_1$ or $t_2$, and, from our discussion, $\sum_{e\in E_1\cup E_2}c(e)<F^*+F'-F^*=F'$. However, there is an $s^{\out}$-$t'$ flow $f^*$ in $\tilde H'$ of value $F'$, and this flow naturally defines a flow in $\tilde H''$ of value $F'$ from $s^{\out}$ to $t_1$ and $t_2$. In other words, if $\qset$ is a flow-path decomposition of this flow, then every path in $\qset$ originates at $s^{\out}$ and terminates at $t_1$ or at $t_2$. Clearly, every path in $\qset$ has to contain an edge of $E_1\cup E_2$. Since the total flow on the paths in $\qset$ is $F'$, while the total capacity of the edges in $E_1\cup E_2$ is less than $F'$, this is a contradiction.
\end{proof}

Consider an $s^{\out}$-$t_1$ flow $\hat f$ in  $J'$ of value at least $F^*$, and let $\qset$ be a flow-path decomposition of this flow. Note that, for every path $P\in \qset$, its penultimate vertex must belong to set $\set{v^{\inn}\mid v\in Q\setminus Q'_3}$, and so its last edge must lie in set $\bigcup_{v\in Q\setminus Q'_3}A(v)$. Therefore, $\beta(Q\setminus Q'_3)\geq F^*$ must hold. From \Cref{claim: prop of Q}, $\beta(Q)\leq F^*+4M'\cdot n^{5.5\eps}\cdot\gamma\cdot \log n\cdot\log(\wmax)$. Therefore, we get that $\beta(Q'_3)=\beta(Q)-\beta(Q\setminus Q'_3)\leq 4M'\cdot n^{5.5\eps}\cdot\gamma\cdot \log n\cdot \log(\wmax)$. Recall that, for every vertex $v\in Q$, $v^{\out}\not\in V(J)$, and so the special edge $(v^{\inn},v^{\out})$, whose capacity in $H$ is at least $M'$ (from Property \ref{prop: residual capacity weaker}), belongs to $A(v)$. Therefore, for every vertex $v\in Q$, $\beta(v)\geq M'$. We conclude that:

\[ |Q'_3|\leq \frac{\beta(Q'_3)}{M'}\leq 4n^{5.5\eps}\cdot\gamma\cdot \log n\cdot \log(\wmax)\leq n^{1-4.4\eps}, \]

since $n^{\eps/10}\geq 2^{20}\cdot \log n \cdot \log (\wmax)$ from Inequality \ref{eq: neps large} and $\gamma\leq 2n^{1-10\eps}$ from Inequality \ref{eq: gamma2}.

Lastly, recall that we denoted by $\hat A$ all edges $(x,y)$ of $H$ with $x\in \hat X$ and $y\not \in \hat X$. Since $t_2\not\in \hat X$ and $t_1\in \hat X$, the set $\hat A$ of edges contains all edges of $E_{\tilde H''}(\hat X,\hat Y)$ (where we do not distinguish between the edges of $H$ and their copies in $\tilde H''$), and all edges in $\bigcup_{v\in Q\setminus Q'_3}A(v)$. By combining the Max-Flow / Min-Cut theorem with \Cref{claim: partial flow}, we get that the total capacity of the edges in $E_{\tilde H''}(\hat X,\hat Y)$ is equal to $F'-F^*$. The total capacity of the edges in $\bigcup_{v\in Q\setminus Q'_3}A(v)$ is: 

\[\beta(Q\setminus Q'_3)\leq \beta(Q)\leq  F^*+4M'\cdot n^{5.5\eps}\cdot \gamma\cdot \log n\cdot \log(\wmax)\leq F^*+M'\cdot n^{5.6\eps}\cdot \gamma\leq F^*+M'\cdot n^{1-4.4\eps}.\]

(we used \Cref{claim: prop of Q} and the fact that $n^{\eps/10}\geq 2^{20}\cdot \log n\cdot \log(\wmax)$ from Inequality \ref{eq: neps large} and that $\gamma\leq 2n^{1-10\eps}$ from Inequality \ref{eq: gamma2}).

Altogether, we get that the total capacity of the edges of $\hat A$ in $H$ is at most $F'+M'\cdot n^{1-4.4\eps}$. 
\end{proof}

If Case 1 happens, then we compute the cut $(\hat X,\hat Y)$ in $\tilde H''$ as described above, and then move the vertices of $Q'_3$ from $Q$ to $Q'$. Since, from \Cref{claim: final cut}, $|Q'_3|\leq n^{1-4.4\eps}$, from \Cref{claim: Q' few S and L}, 
if $(s,T)$ is a good pair, all shortcut operations so far have been valid, and Event~\ref{event: heavy err} did not happen, then $|Q'\cap (S\cup L)|\leq  3n^{1-4\eps}$ holds. We also let $X^*=\hat X\cap V(H)$. Then we are guaranteed that $X^*\subseteq V(J)$, and $\sum_{e\in E_H(X^*,V(H)\setminus X^*)}\hat c(e)\leq F'+M'\cdot n^{1-4.4\eps}$. Recall that, from \Cref{claim: max flow in tilde H}, if $(s,T)$ is a good pair,  all shortcut operations so far have been valid, and  Event~\ref{event: heavy err} did not happen, then:

\[F'\leq F+\eta=F+2^{14}n^{\eps}\cdot M'\cdot \gamma \cdot \log^2 n\cdot \log(\wmax)\leq F+n^{1-8\eps}\cdot M'.\]


We now summarize the outcome of Step~\ref{step2: partition} in the following observation, whose proof follows immediately from our discussion.

\begin{observation}\label{obs: Step 2 summary}
	Let $(Q,Q')$ be the partition of $U$ computed in Step~\ref{step2: partition}. Then the following hold:
	
	\begin{itemize}
		\item If $(s,T)$ is a good pair,  all shortcut operations so far have been valid, and  Event~\ref{event: heavy err} did not happen, then $|Q'\cap (S\cup L)|\leq  3n^{1-4\eps}$; and
		
		\item If $|Q|\geq n^{1-4\eps}$, then the algorithm computed a subset $X^*\subseteq V(J)$ of vertices, such that, for every vertex $v\in Q$, $v^{\inn}\in X^*$ and $v^{\out}\not\in X^*$ holds. 
		Additionally, if $(s,T)$ is a good pair,  all shortcut operations so far have been valid, and Event~\ref{event: heavy err} did not happen, then $\sum_{e\in E_H(X^*,V(H)\setminus X^*)}\hat c(e)\leq F+2n^{1-4.4\eps}\cdot M'$.
	\end{itemize} 
\end{observation}

\section{Step 3: Sparsification}
\label{sec: step 3}

The goal of this step is to compute a subgraph $H'\subseteq H$, that is relatively sparse, so that the value of the maximum $s^{\out}$-$t$ flow in $H'$ is close to that in $H$. We will then compute the maximum $s^{\out}$-$t$ flow in $H'$, and augment $f$ with this flow, obtaining the desired flow $f_{i}$ that will serve as the output of the current phase.

We start by defining a partition of the edges of $E(H)$ into six subsets $E_1,\ldots,E_6$, and analyzing each subset separately. For some of these subsets we also describe the intuition for its sparsification. After that we formally define the sparsified graph $H'$, and provide an algorithm for constructing $H'$ efficiently. We also formally prove that, with high probability, the value of the maximum $s^{\out}$-$t$ flow in $H'$ is sufficiently high. Lastly, we provide an algorithm for computing the desired output of the current phase.

\subsection{Partition of $E(H)$}

We start by defining a partition of $E(H)$ into six subsets $E_1,\ldots,E_6$, and by providing intuition for sparsifying some of these subsets. It may be convenient for now to think about constructing the graph $H'$ by starting with $H'=H$ and then gradually modifying it via the steps described below, though the eventual efficient algorithm for constructing $H'$ will be different. We now define and analyze each of the sets in the partition in turn.

\paragraph{Set $E_1$:} Recall that, for every vertex $v\in V(G)\setminus U$, $w(v)<M'$ must hold. We let $E_1$ be the set of all edges of $H$ that are incident to the vertices of $\set{v^{\inn},v^{\out}\mid v\in V(G)\setminus U}$. Note that, since the current flow $f$ is $M'$-integral, for every vertex  $v\in V(G)\setminus U$, the flow on the special edge $(v^{\inn},v^{\out})$ must be $0$. It is then easy to verify that the flow on every regular edge of $G''$ that enters $v^{\inn}$ or leaves $v^{\out}$ is also $0$. Therefore, every regular edge of $E_1$ must be of the form $(x^{\out},y^{\inn})$, where either $x\in V(G)\setminus U$, or $y\in V(G)\setminus U$, and this edge must also be present in $G''$. 
We will delete the edges of $E_1$ from the graph $H'$. It is easy to verify (and we do so formally later) that this may only decrease the value of the maximum $s^{\out}$-$t$ flow in $H'$ by at most $M'\cdot n$.

\paragraph{Set $E_2$:} Set $E_2$ contains all special edges of $H$ (both the forward copies $(v^{\inn},v^{\out})$, and their backward counterparts $(v^{\out},v^{\inn})$), excluding those that lie in $E_1$. Additionally, $E_2$ contains all edges of $H$ that are incident to $t$. Clearly, $|E_2|\leq 3n$. All edges of $E_2$ will be included in $H'$.

\paragraph{Set $E_3$:} Set $E_3$ contains all regular edges $(x,y)$ of $H$ that do not lie in $E_1$, such that either (i) $x\in V(J)$; or (ii) at least one of the vertices $x,y$ lies in $\Gamma=V(J)\cap V^{\out}$
 (note that both conditions may hold simultaneously). Consider any  edge $e=(x,y)\in E_3$. Recall that $|\Gamma|\leq N= \frac{32n^{1+2\eps}\cdot \log^2(\wmax)}{\gamma}$.
Therefore, the number of edges of $E_3$ that are incident to the vertices of $\Gamma$ is bounded by $\frac{64n^{2+2\eps} \cdot \log^2(\wmax)}{\gamma}\leq O(n^{2-4\eps})$, since $64 \log^2(\wmax)\leq n^{\eps}$ from Inequality \ref{eq: neps large}, and $\gamma\geq n^{7\eps}$ from Inequality \ref{eq: gamma3}. 
Note that every edge with both endpoints in $J$ must be incident to a vertex of $\Gamma$. It is then easy to verify that, if $e$ is an edge of $E_3$ that is not incident to a vertex of $\Gamma$, then it must be the case that  
$e=(v^{\inn},u^{\out})$, where $v^{\inn}\in V^{\inn}\cap V(J)$ and $u^{\out}\in V^{\out}\setminus V(J)$. From \Cref{claim: ER}, each such edge  must lie in set $E^R$, and, from Property \ref{prop: few sticking out} of  \Cref{obs: props of tilde H}, $|E^R|\leq O\left (n\cdot \gamma\cdot \log n\cdot \log(\wmax)\right )\leq O(n^{2-4\eps})$, since $\gamma\leq n^{1-5\eps}$ from Inequality \ref{eq: gamma2}, and $\log n\cdot \log(\wmax)\le n^{\eps}$.
Altogether, we get that $|E_3|\leq O(n^{2-4\eps})$. All edges of $E_3$ will be included in $H'$.

\paragraph{Set $E_4$:} Set $E_4$ contains all regular edges $(x,y)$ that do not lie in $E_1\cup E_3$, such that $x\in V^{\out}$ and $y\in \set{v^{\inn}\mid v\in Q'}$. Recall that, from \Cref{obs: Step 2 summary}, if $(s,T)$ is a good pair,  all shortcut operations so far have been valid, and  Event $\hat \event$ did not happen, then $|Q'\cap (S\cup L)|\leq  3n^{1-4\eps}$. Therefore, if $u^{\out}\in V^{\out}$ is a vertex that is incident to more than $3n^{1-4\eps}$ edges of $E_4$, then it must be the case that $u$ has at least one neighbor in $G$ that lies in $R$, and so $u\in R\cup S$ must hold. For each such vertex $u^{\out}$, we can then add a shortcut edge $(u^{\out},t)$ to $G''$, to $H$, and to $H'$, and delete all other edges leaving $u^{\out}$ from $H'$. In order to execute this step efficiently, we will employ the $\delta$-subgraph oracle with the set $Z=Q'$ of vertices; recall that the parameter $\delta$ was defined so that $n^{\delta}=\frac{n^{4\eps}}{4000\log n}$ (see Equation \ref{eq: delta}).
Recall that the oracle must return a partition $(Y^h,Y^{\ell})$ of $V(G)$, and a collection $E'$ of at most $n^{2-\delta}\cdot \log^2(\wmax)=n^{2-4\eps}\cdot 4000\log n\cdot \log^2\wmax$ edges of $G$, such that 
$E'=E_G(Y^{\ell},Z)$ holds. If the oracle does not err,
then every vertex of $Y^h$ has at least $\frac{n^{1-\delta}}{1000\log n}\geq 4n^{1-4\eps}$ neighbors in $Q'$. Therefore, if $(s,T)$ is a good pair,  all shortcut operations so far have been valid,  Event $\hat \event$ did not happen and the oracle did not err, then every vertex $v\in Y^h$ must lie in $S\cup R$. For each such vertex $v\in Y^h$, we will add the shortcut edge $(v^{\out},t)$ to $G''$, $H$, and $H'$, and delete all edges of $E_4$ that leave the vertex $v^{\out}$ from $H'$. From the above discussion, the number of edges of $E_4$ that remain in $H'$ is then bounded by $O(n^{2-4\eps}\cdot\log n\cdot \log^2(\wmax))$.

\paragraph{Set $E_5$.} Set $E_5$ contains all regular edges $(x,y)\in E(H)$ that do not lie in $E_1\cup E_3\cup E_4$, such that $y=v^{\inn}$ for some vertex $v\in Q$, and $x\in V^{\out}\setminus V(J)$. If $|Q|\leq n^{1-4\eps}$, then $|E_5|\leq |Q|\cdot n\leq n^{2-4\eps}$, and we keep all edges of $E_5$ in graph $H'$. Otherwise, all edges of $E_5$ will be deleted from $H'$. We will show that, if $(s,T)$ is a good pair,  all shortcut operations so far have been valid, and  Event $\hat \event$ did not happen, then the deletion of the edges of $E_5$ from $H'$ may only reduce the value of the maximum $s^{\out}$-$t$ flow in $H'$ slightly. Intuitively, if $\hat f$ is the maximum  $s^{\out}$-$t$ flow in $H'$ before the edges of $E_5$ are deleted from it, and $\pset$ is its flow-path decomosition, then every flow-path $P\in \pset$ that contains an edge of $E_5$ must use at least two edges of $E_H(X^*,Y^*)$, where $Y^*=V(H)\setminus X^*$. However, from \Cref{obs: Step 2 summary}, the total capacity of the edges in $E_H(X^*,Y^*)$ is close to $F$ -- the value of the maximum $s^{\out}$-$t$ flow in $H$. It then follows that the paths of $\pset$ that contain edges of $E_5$ cannot carry large amount of flow in $\hat f$.

\paragraph{Set $E_6$.} This set contains all remaining edges of $H$. The following observation provides a characterization of the edges of $E_6$.

\begin{observation}\label{obs: edges of E6}
	Let $e=(x,y)$ be an edge of $E_6$. Then $e$ is a regular edge, with $x\in V^{\inn}\setminus V(J)$ and $y\in V^{\out}\setminus V(J)$.
\end{observation}
\begin{proof}
Consider any edge $e=(x,y)\in E_6$. From the definition of set $E_2$, $e$ must be a regular edge, and from the definition of set $E_3$, $x\not\in V(J)$ must hold. 

Assume first that $x=v^{\out}$ and $y=u^{\inn}$ for some vertices $v^{\out}\in V^{\out}$ and $u^{\inn}\in V^{\inn}$. From the definition of set $E_5$, and since $x\not\in V(J)$, we get that $u\not\in Q$. From the definition of set $E_1$, and since $(Q,Q')$ is a partition of $U$, we get that $u\in Q'$. But then the edge must lie in $E_4$.

We conclude that it must be the case that $x\in V^{\inn}$ and $y\in V^{\out}$; denote $x=v^{\inn},y=u^{\out}$ for some $v,u\in V(G)$. From the definition of set $E_3$, we get that $v^{\inn}\not \in J$ and $u^{\out}\not\in \Gamma$. We conclude that neither of the vertices $v^{\inn},u^{\out}$ may lie in $J$. 
\end{proof}

Consider now any vertex $a^{\inn}\in V^{\inn}\setminus V(J)$. 
We say that this vertex is \emph{bad} if it is incident to at least $n^{1-4\eps}$ edges of $E_6$, and that it is \emph{good} otherwise. 
The following observation is central to dealing with the edges of $E_6$.

\begin{observation}\label{obs: bad vertices}
	Assume that $(s,T)$ is a good pair, that all shortcut operations performed in previous phases and in Step 1 of the current phase were valid, and that Event $\hat \event$ did not happen. Let $a^{\inn}\in V^{\inn}\setminus V(J)$ be a bad vertex. Then $a\in R$ must hold.
\end{observation}

We provide the proof of the observation below, after providing additional intuition for sparsification of the set $E_6$ of edges. Intuitively, we will set up an instance of the  Directed Heavy Degree Estimation problem (see \Cref{def: heavy weight est}), whose input is the graph $H$, the sets $Z=V^{\inn}\setminus V(J)$ and $Z'=V^{\out}\setminus V(J)$ of vertices, threshold value $\tau=n^{1-4\eps}$, capacity thershold $c^*=M'$, and parameter $W=\wmax$.
We will then use the algorithm from \Cref{claim: heavy weight est} to this problem, and denote by $\hat Y\subseteq Z$ the outcome of the algorithm. Note that if the algorithm from \Cref{claim: heavy weight est} does not err, then we are guaranteed that, for every vertex $x^{\inn}\in \hat Y$, at least $n^{1-4\eps}$ edges of $E_6$ are incident to $x^{\inn}$, so every vertex in $\hat Y$ is a bad vertex. Additionally,  every vertex $a^{\inn}\in Z\setminus \hat Y$ is incident to at most $1000\tau\log n\cdot\log(\wmax)$ edges of $E_6$.
For every vertex $a^{\inn}\in \hat Y$, 
we will add a shortcut edge $(a^{\inn},t)$ to $G''$, $H$ and $H'$; from \Cref{obs: bad vertices}, if $(s,T)$ is a good pair,  all shortcut operations executed before Step 3 of the current phase have been valid, and Event $\hat \event$ did not happen, then this is a valid shortcut edge. We will then delete from $H'$ all edges of $E_6$ leaving $a^{\inn}$. The edges of $E_6$ that are incident to the vertices of $Z\setminus \hat Y$ remain in $H'$; if the algorithm from \Cref{claim: heavy weight est} does not err, their number is bounded by $O\left(n\cdot \tau\log n\cdot\log(\wmax)\right )\leq O\left(n^{2-4\eps}\log n\cdot\log(\wmax)\right )$.
 We now prove \Cref{obs: bad vertices}.

\begin{proofof}{\Cref{obs: bad vertices}}
	Assume that $(s,T)$ is a good pair, that all shortcut operations performed in previous phases and in Step 1 of the current phase are valid, and that Event $\hat \event$ did not happen. Let $a^{\inn}\in V^{\inn}\setminus V(J)$ be a bad vertex, and assume for contradiction that $a\not\in R$, so $a\in L\cup S$.

	We denote by $\hat E(a)$ the set of all edges of $E_6$ incident to  $a^{\inn}$. Let $e=(a^{\inn},v^{\out})\in \hat E(a)$ be any such edge. From the construction of the split graph, edge $e$ does not lie in $G''$, so $e$ is a backward edge corresponding to the edge $(v^{\out},a^{\inn})$ of $G''$. Since the current flow $f$ is $M'$-integral, the flow on the edge $(v^{\out},a^{\inn})$ is at least $M'$, and so the capacity of $e$ in $H$ is at least $M'$. 
	
	We further partition set $\hat E(a)$ into two subsets: set $\hat E'(a)$ containing all edges $(a^{\inn},v^{\out})$ where $v\in L$, and set $\hat E''(a)$ containing all remaining edges. From Inequality \ref{eq: small L}, 
	$|L|\leq n^{\eps}$, and so $|\hat E''(a)|\geq n^{1-4\eps}-n^{\eps}\geq \frac{n^{1-4\eps}}{2}$. 
	
	Consider any edge $e=(a^{\inn},v^{\out})\in \hat E''(a)$. From the definition of set $E''(a)$, $v\in S^{\out}\cup R^{\out}$, and from our assumption, $a^{\inn}\in S^{\inn}\cup L^{\inn}$. It is then immediate to verify that $e\in E'_S$ must hold (see \Cref{def: set E'S}), and so $E''(a)\subseteq E'_S$.
	
	Let $\hat f$ be a maximum $s^{\out}$-$t$ flow in $H$, so $\val(\hat f)=F$, and let $\pset$ be the flow-path decomposition of $F$. From \Cref{claim: max flow in H}, $F=\hat c_S$ holds, and moreover, every path in $\pset$ contains exactly one edge of $E'_S$. Since $E''(a)\subseteq E'_S$, we get that, for every edge $e\in E''(a)$,  $\hat f(e)=\hat c(e)$ holds.
	
	Recall that, from \Cref{obs: edges crossing}, if we consider the partition $(X,Y)$ of the vertices of $H$, where $X=L^* \cup S^{\inn}$, and $Y=V(H)\setminus X=S^{\out}\cup R^*\cup \set{t}$, then	$E_H(X,Y)=E'_S$.
It then follows that, for every flow-path $P\in \pset$, if we delete the unique edge of $E'_S$ from $P$, then we obtain two subpaths of $P$: subpath $P_1$ whose vertices are contained in  $L^*\cup S^{\inn}$, and subpath $P_2$, whose vertices are contained in  $S^{\out}\cup R^*$.

Let $\pset'\subseteq \pset$ be the set of paths $P\in \pset$ with $E(P)\cap \hat E''(a)\neq \emptyset$. From our discussion, every path $P\in \pset'$ contains exactly one edge of $\hat E''(a)$, and all vertices preceding this edge on the path lie in $L^*\cup S^{\inn}$.

Consider now any path $P\in \pset'$, and let $v(P)\in V(P)$ be the first vertex of $P$ that does not lie in $V(J)$. Since $a^{\inn}\not\in V(J)$, $v(P)$ appears before $a^{\inn}$ on $P$ (or $v(P)=a^{\inn}$), and so $v(P)\in L^*\cup S^{\inn}$ must hold. We denote by $e(P)$ the edge of $P$ immediately preceding $v(P)$ on the path.
From our discussion so far:

\[\sum_{P\in \pset'}\hat f(P)=\sum_{e\in \hat E''(a)}\hat c(e)\geq |E''(a)|\cdot M'\ge \frac{n^{1-4\eps}\cdot M'}{2}.\]

Recall that graph $\tilde H'$ was obtained from $H$ by contracting all vertices of $V(H)\setminus V(J)$ into the vertex $t'$. Consider a path $P\in \pset'$, and let $P'$ be the subpath of $P$ between $s^{\out}$ and $v(P)$. Let $P''$ be the path obtained from $P'$ by replacing the last vertex $v(P)$ with $t'$. Note that $P''$ is a valid $s^{\out}$-$t'$ path in $\tilde H'$.
From our discussion, path $P'$ may not contain a vertex of $S^{\out}\cup R^*$, and so it may not contain an edge of $E'_S$ (see \Cref{def: set E'S}). Therefore,
path $P''$ must be problematic (see \Cref{def: problematic path}).
We define an $s^{\out}$-$t'$ flow $f'$ in graph $\tilde H'$ as follows:
for every path $P\in \pset'$, we set $f'(P'')=\hat f(P)$. Notice that $f'$ is a valid $s^{\out}$-$t'$ flow in $\tilde H'$, that sends at least 
$\sum_{P\in \pset'}\hat f(P)\ge \frac{n^{1-4\eps}\cdot M'}{2}$ flow units via problematic paths. But from \Cref{cor: flow on special paths}, if $(s,T)$ is a good pair, all shortcut operations so far have been valid, and that Event $\hat \event$ did not happen, the total flow that $f'$ may send via problematic paths is:

\[\sum_{P\in \pset'}f'(P'')\leq 2^{14}n^{\eps}\cdot M'\cdot \gamma \cdot \log^2 n\cdot \log(\wmax)<\frac{n^{1-4\eps}\cdot M'}{2},\]

since $2^{20}\log^2 n\cdot \log(\wmax)\leq n^{\eps}$ from Inequality \ref{eq: neps large} and $\gamma<n^{1-6\eps}$ from Inequality \ref{eq: gamma2}, a contradiction. We conclude that $a\in R$ must hold.
\end{proofof}

 Let $\tilde \event_i$ be the bad event that any of the shortcut operations performed in previous phases were invalid, and let $\tilde \event'_i$ be the bad event that any of the shortcut operations performed in Step 1 of the current phase was invalid. 
 From Invariant \ref{inv: shortcut}, $\prob{\tilde \event_i}\leq \frac{2i}{n^{\eps}\cdot \log^2(\wmax)}$, and from \Cref{claim: step 1 analysis}, $\prob{\tilde \event_i'}\leq \frac{1}{n^{\eps}\cdot \log^2(\wmax)}$.

\subsection{Constructing the graph $H'$}

We now provide an efficient algorithm for constructing the graph $H'$. 
We start with $V(H')=V(H)$ and $E(H')=\emptyset$, and then take care of the sets $E_2,\ldots,E_6$ of edges in turn.

\paragraph{Edges of $E_2$.} Recall that set $E_2\subseteq E(H)$ contains all special edges of $H$, and all edges of $H$ that are incident to $t'$. We can compute the edges of $E_2$ by considering every vertex $v\in V(G)$ with $w(v)\geq M'$ in turn; for each such vertex $v$, we can use the data structure $\DS_f$ to determine the flow value $f(v^{\inn},v^{\out})$, which in turn allows us to determine whether the edges  $(v^{\inn},v^{\out})$, $(v^{\out},v^{\inn})$  lie in $H$, and if so, to compute their residual capacity. It is then easy to see that set $E_2$ of edges can be computed in time $O(n)$; recall that $|E_2|\leq O(n)$ holds. We add the edges of $E_2$ to graph $H'$.

\paragraph{Edges of $E_3$.} Next, we add all edges of $E_3$ to $H'$. Recall that set $E_3$ contains all regular edges $(x,y)$ of $H$,  such that either $x\in V(J)$, or at least one of the vertices $x,y$ lies in $\Gamma$, excluding the edges that are incident to the vertices of $\set{v^{\inn},v^{\out}\mid v\in V(G)\setminus U}$. Recall that $|\Gamma|\leq \frac{32n^{1+2\eps}\cdot \log^2(\wmax)}{\gamma}$. By inspecting the arrays corresponding to each of the vertices of $\Gamma$ in the modified adjacency-list representation of $H$, we can compute all edges of $E_3$ that are incident to the vertices of $\Gamma$, in time $O(|\Gamma|\cdot n)\leq O\left(\frac{n^{2+2\eps}\cdot \log^2(\wmax)}{\gamma}\right )\leq O(n^{2-4\eps})$ (we have used the fact that $\log^2(\wmax)\leq n^{\eps}$ from Inequality \ref{eq: neps large}, and $\gamma\geq n^{9\eps}$ from Inequality \ref{eq: gamma3}). As observed already, if $e$ is an edge of $E_3$ that is not incident to a vertex of $\Gamma$, then it must be the case that  
$e=(v^{\inn},u^{\out})$, where $v^{\inn}\in V^{\inn}\cap V(J)$ and $u^{\out}\in V^{\out}\setminus V(J)$, and moreover $e\in E^R$.
For each such edge $e=(v^{\inn},u^{\out})$, the corresponding edge $(v^{\inn},t')$ lies in $\tilde H'$, and every edges incident to $t'$ in $\tilde H'$ represents either an edge of $E^R$, or a special edge, so the number of such edges is bounded by $n+|E^R|$. We can then compute all edges of $E^R$ by inspecting all edges that are incident to $t'$ in $\tilde H'$, and considering, for each such edge $e'$, the edge $e$ of $H$ that it represents. The time required to do so is bounded by 
$O(n+|E^R|)\leq O\left (n\cdot \gamma\cdot \log n\cdot \log(\wmax)\right )\leq O(n^{2-4\eps})$, since $\gamma\leq n^{1-5\eps}$ from Inequality \ref{eq: gamma2}, and $\log n\cdot \log(\wmax)\le n^{\eps}$ from Inequality \ref{eq: neps large}; we have also used  Property \ref{prop: few sticking out} of  \Cref{obs: props of tilde H} to bound $|E^R|$.
Overall, we get that $|E_3|\leq O(|\Gamma|\cdot n)+O(|E^R|)\leq O(n^{2-4\eps})$, and our algorithm computes the set $E_3$ of edges in time $O(n^{2-4\eps})$. We add all edges of $E_3$ into $H'$.

\paragraph{Edges of $E_4$.}
Recall that set $E_4$ contains all regular edges $(x,y)$ that do not lie in $E_1\cup E_3$, such that $x\in V^{\out}$ and $y\in \set{v^{\inn}\mid v\in Q'}$. Recall that we have set the value of $\delta$ so that $n^{\delta}=\frac{n^{4\eps}}{4000\log n}$ holds (see Equation \ref{eq: delta}), and from Inequality \ref{ineq: delta2},
$\delta\geq\frac{1}{\sqrt{\log n}}$. We perform a call to the $\delta$-subgraph oracle with graph $G$ and the set $Z=Q'$ of its vertices.
Recall that the oracle must return a partition $(Y^h,Y^{\ell})$ of $V(G)$, and a collection $E'$ of at most $n^{2-\delta}\cdot \log^2(\wmax)=n^{2-4\eps}\cdot 4000\log n\cdot \log^2(\wmax)$ edges of $G$, such that 
$E'=E_G(Y^{\ell},Z)$ holds. Recall that, if the oracle does not err,
then every vertex of $Y^h$ has at least $\frac{n^{1-\delta}}{1000\log n}\geq 4n^{1-4\eps}$ neighbors in $Q'$. 
We let $\hat \event'$ be the bad event that the oracle erred. From the definition of the subgraph oracle, $\prob{\hat \event'}\leq \frac{1}{n^5\cdot \log^4(\wmax)}$.
For every edge $(u,v)\in E'$ with $u\in U$, we add the corresponding edge $(u^{\out},v^{\inn})$, that must lie in $E_4$, to graph $H'$. Additionally, for every vertex $x\in Y^h$, we add the shortcut edge $(x^{\out},t)$ to $G''$, $H$, and $H'$. 
We will use the following observation.

\begin{observation}\label{obs: valid shortcuts E4}
If $(s,T)$ is a good pair and neither of the events $\tilde \event_i,\tilde \event'_i, \hat \event, \hat \event'$ happened, then all shortuct operations performed while processing the set $E_4$ of edges are valid.
\end{observation}
\begin{proof}
Assume that $(s,T)$ is a good pair, and that neither of the events $\tilde \event_i,\tilde \event'_i, \hat \event, \hat \event'$ has happened, so in particular, all shortcut operations prior to Step 3 of the current phase were valid. Then, from \Cref{obs: Step 2 summary} $|Q'\cap (S\cup L)|\leq  3n^{1-4\eps}$ holds. 
Consider now some vertex $x\in Y^h$. Since we have assumed that Event $\hat \event'$ did not happen, $x$ has  at least $4n^{1-4\eps}$ neighbors in $Q'$ in graph $G$, and so in particular, $x$ must have a neighbor that lies in $R$. We conclude that $x\in S\cup R$ must hold, and so the shortcut edge $(x^{\out},t)$ is valid.
\end{proof}

To summarize, we add to $H'$ all edges of $E_4$, except for the edges $(u^{\out},v^{\inn})$ with $u\in Y^h$. We also add, for every vertex $u\in Y^h$, the shortcut edge $(u^{\out},t)$ to $H'$. The number of edges added to $H'$ while processing $E_4$ is bounded by $O(|E'|+n)\leq O\left (n^{2-4\eps}\cdot \log n\cdot \log^2(\wmax)\right )$. The time that the algorithm spends on processing the edges of $E_4$ (excluding the running time of the algorithm that implements the oracle) is bounded by $O(|E'|+n)\leq O\left (n^{2-4\eps}\cdot \log n\cdot \log^2(\wmax)\right )$.

\paragraph{Edges of $E_5$.} Recall that set $E_5$ contains all regular edges $(x,y)\in E(H)$ that do not lie in $E_1\cup E_3\cup E_4$, such that $y=v^{\inn}$ for some vertex $v\in Q$, and $x\in V^{\out}\setminus V(J)$. If $|Q|> n^{1-4\eps}$, then we do not include any edges of $E_5$ in $H'$. Otherwise, we include all edges of $E_5$ in $H'$. In the latter case, in order to compute the edges of $E_5$, we simply inspect the arrays in the modified adjacency list representation of $H$ that correspond to the vertices in $\set{v^{\inn}\mid v\in Q}$. Since $|Q|\leq n^{1-4\eps}$, we can compute the set $E_5$ of edges in time $O(|Q|\cdot n)\leq O(n^{2-4\eps})$, and $|E_5|\leq O(n^{2-4\eps})$ holds in this case.

\paragraph{Edges of $E_6$.} Recall that, if $e=(x,y)\in E_6$, then, from \Cref{obs: edges of E6},  $x\in V^{\inn}\setminus V(J)$ and $y\in V^{\out}\setminus V(J)$ holds. We set up an instance of the Directed Heavy Degree Estimation problem (see \Cref{def: heavy weight est}), whose input is the graph $H$, the sets $Z=V^{\inn}\setminus V(J)$ and $Z'=V^{\out}\setminus V(J)$ of vertices, threshold value $\tau=n^{1-4\eps}$,  capacity thershold $c^*=M'$, and parameter $W=\wmax$. We apply the algorithm from \Cref{claim: heavy weight est} to this problem, and denote by $\hat Y\subseteq Z$ the outcome of the algorithm. Let $\hat \event''$ be the event that the algorithm from \Cref{claim: heavy weight est} errs. Then, from \Cref{claim: heavy weight est}, $\prob{\hat \event''}\leq \frac{1}{n^{10}\cdot \log^4(\wmax)}$. If Event $\hat \event''$ did not happen, then for every vertex $x^{\inn}\in \hat Y$, at least $n^{1-4\eps}$ edges of $E_6$ are incident to $x^{\inn}$, so every vertex in $\hat Y$ is a bad vertex. Additionally,  every vertex $a^{\inn}\in Z\setminus \hat Y$, is incident to at most $1000\tau\log n\cdot\log(\wmax)$ edges of $E_6$.
For every vertex $x^{\inn}\in \hat Y$, we add the shortcut edge $(x^{\inn},t)$ to $G''$, $H$, and $H'$. Note that, from \Cref{obs: bad vertices}, if $(s,T)$ is a good pair, and neither of the events $\tilde \event_i,\tilde \event'_i,\hat \event,\hat \event''$ happened, then for every vertex $x^{\inn}\in \hat Y$, $x\in R$ holds, and so all shortcut operations performed while processing the edges of $E_6$ are valid.

Next, we consider every vertex $x^{\inn}\in Z\setminus \hat Y$ in turn. As observed already, if Event $\hat \event''$ did not happen, then $x^{\inn}$ is incident to at most $1000\tau\log n\cdot\log(\wmax)=1000n^{1-4\eps}\log n\cdot \log(\wmax)$ edges of $E_6$. Notice that the only regular edges that leave $x^{\inn}$ are the edges of $E_6$, or the edges that are incident to vertices of $\Gamma=V^{\out}\cap V(J)$ (this is since every edge $e\in E_6$ is a backward copy of a regular edge of $G''$, and so it may not lie in $E_1$). Since $|\Gamma|\le N= \frac{32n^{1+2\eps}\cdot \log^2(\wmax)}{\gamma}<n^{1-4\eps}$ (as $2^{10}\log n \cdot \log(\wmax)<n^{\eps}$ from Inequality \ref{eq: neps large} and $\gamma\geq n^{10\eps}$ from Inequality \ref{eq: gamma3}), we get that, if Event $\hat \event''$ did not happen, the total number of edges leaving the vertex $x^{\inn}$ in $H$ is bounded by $2000n^{1-4\eps}\log n\cdot \log(\wmax)$. For every vertex $x^{\inn}\in Z\setminus \hat Y$, we process the lists of all edges leaving $x^{\inn}$ in the modified adjacency-list representation of $H$. As we process these edges, we add each regular edge that is not incident to vertices of $\Gamma$ to $H'$. If the number of such edges becomes greater than $1000n^{1-4\eps}\log n\cdot \log(\wmax)$, then we terminate the algorithm and return ``FAIL''; in this case, bad event $\hat \event''$ must have happened. 

Recall that the running time of the algorithm for the Directed Heavy Degree Estimation problem from \Cref{claim: heavy weight est} is $O\left(\frac{n\cdot |Z'|\cdot \log n\cdot\log(\wmax)}{\tau}\right )\leq O\left(n^{1+4\eps}\cdot \log n\cdot\log(\wmax)\right )\leq O(n^{2-4\eps})$, since $\tau=n^{1-4\eps}$, $\log n\cdot \log(\wmax)\leq n^{\eps}$, and $\eps=\frac 1 {45}$.
The remaining time required in order to recover the edges of $E_6$ that are incident to the vertices of $Z\setminus \hat Y$ is bounded by $O(|Z|\cdot n^{1-4\eps}\cdot \log n\cdot \log(\wmax))\leq O( n^{2-4\eps}\cdot \log n\cdot \log(\wmax))$. Note that we insert into $H'$ all edges of $E_6$ that are incident to the vertices of $Z\setminus \hat Y$, and, for each vertex $x\in \hat Y$, we insert a shortcut edge $(x,t)$ into $H'$.

This completes the construction of the graph $H'$. 
Consider now some edge $e\in E(H')$. If $e$ is a regular edge, or a backward special edge, then we set its capacity  $\hat c'(e)$ in $H'$ to be $\hat c(e)$, the same as its capacity in $H$. 
Otherwise, $e$ is a  forward special edge, and we reduce its capacity by at least $M'$ and at most $2M'$ units, to ensure that the resulting capacity $\max\set{\hat c(e)-2M',0}\leq \hat c'(e)\leq \hat c(e)-M'$ is an integral multiple of $M'$ (recall that, from Property \ref{prop: residual capacity weaker}, $\hat c(e)\geq M'$ must hold).
Since $\wmax,\wmax'$ and $M'$ are all integral powers of $2$, and since the current flow $f$ is $M'$-integral, it is easy to verify that for every edge $e$ of $H'$, its capacity in $H'$ is an integral multiple of $M'$.
From the above discussion, it is immediate to verify that: 

\begin{equation}\label{eq: size of H'}
|E(H')|\leq O(n^{2-4\eps}\cdot \log n\cdot \log^3(\wmax)),
\end{equation}

 and the time that our algorithm spends on constructing $H'$ is bounded by $O(n^{2-4\eps}\cdot \log n\cdot \log^3(\wmax))$.

Recall that, if $(s,T)$ is not a good pair, then any shortcut operation is valid by definition. 
From our discussion, if $(s,T)$ is a good pair, and neither of the events $\tilde \event_i,\tilde \event'_i, \hat \event, \hat \event', \hat \event''$  happened, then all shortcut operations performed in Step 3 are valid ones.  

Recall that  $\tilde \event_i$ is the bad event that any of the shortcut operations performed in previous phases were invalid, 
and, from Invariant \ref{inv: shortcut}, $\prob{\tilde \event_i}\leq \frac{2i}{n^{\eps}\cdot \log^2(\wmax)}$. Recall that  $\tilde \event'_i$ is the bad event that any of the shortcut operations performed in Step 1 of the current phase was invalid, and, from \Cref{claim: step 1 analysis}, $\prob{\tilde \event_i'}\leq \frac{1}{n^{\eps}\cdot \log^2(\wmax)}$. Let $\tilde \event''_i$ be the bad event that at least one shortcut operation performed over the course of Step 3 is invalid. Notice that shortuct operations may only be performed during Step 3 when processing edges of $E_4$ or of $E_6$. From \Cref{obs: valid shortcuts E4}, and our discussion above, if $(s,T)$ is a good pair, and if neither of the events  $\tilde \event_i,\tilde \event'_i$ has happened, then Event $\tilde \event''_i$ may only happen if either of the events $\hat \event, \hat \event'$, or $\hat \event''$ has happened. In other words, if $(s,T)$ is a good pair, then:

\[\begin{split}
\prob{\tilde \event''_i\mid \neg \tilde \event_i\band \neg\tilde \event'_i}& \leq \prob{\hat \event\bor \hat \event'\bor\hat\event''\mid \neg \tilde \event_i\band \neg\tilde \event'_i}\\
&\leq \frac{1}{n^7\cdot \log^4(\wmax)}+ \frac{1}{n^5\cdot \log^4(\wmax)} +\frac{1}{n^{10}\cdot \log^4(\wmax)}\\
&\leq \frac{3}{n^5\cdot \log^4(\wmax)}.
\end{split}
\]

(we have used Inequality \ref{eq: prob error Heavy Vertex} to bound $\prob{\hat \event}$; the bounds on the probabilities of Events $\hat\event'$ and $\hat \event''$ follow from the above discussion.)

Let $\tilde \event_{i+1}$ be the bad event that any of the shortcut operations executed over the course of the first $i$ phases was invalid. Then:

\begin{equation}\label{eq: invalid shortcuts}
\begin{split}
\prob{\tilde \event_{i+1}}&\leq \prob{\tilde \event_i}+\prob{\tilde \event'_i}+\prob{\tilde \event''_i\mid  \neg \tilde \event_i\band \neg\tilde \event'_i }\\
&\leq \frac{2i}{n^{\eps}\cdot \log^2(\wmax)}+\frac{1}{n^{\eps}\cdot \log^2(\wmax)}+\frac{3}{n^5\cdot \log^4(\wmax)}\\
&\leq \frac{2(i+1)}{n^{\eps}\cdot \log^2(\wmax)}.
\end{split}
\end{equation}

Our algorithm will perform no further shortcut operations in the current phase, and so Invariant  \ref{inv: shortcut} holds at the end of the phase.

\subsection{Properties of Graph $H'$}
In this subsection we establish the main property of $H'$, namely, that the value of the maximum $s^{\out}$-$t$ flow in $H'$ is close to that in $H$, in the following claim.

\begin{claim}\label{claim: max flow value}
	Let $F$ be the value of the maximum $s^{\out}$-$t$ flow in $H$, and let $F'$ be the value of the maximum $s^{\out}$-$t$ flow in $H'$.
	If $|Q|\leq n^{1-4\eps}$, then $F'\geq F-2n\cdot M'$ holds. Additionally, if $(s,T)$ is a good pair, and neither of the bad events $\event_i, \tilde \event_{i+1},\hat \event,\hat \event',\hat \event''$ has happened, then $F'\geq F-2n\cdot M'-2n^{1-4.4\eps}\cdot M'$ holds even if $|Q|>n^{1-4\eps}$. 
\end{claim}
\begin{proof}
We define an intermediate graph $H''$, that is obtained from $H'$ by adding all edges of $E_1$ to it, and setting all edge capacities in $H''$ to be identical to those in $H$. Let $F''$ denote the value of the maximum $s^{\out}$-$t$ flow in $H''$. We use the following claim to lower-bound $F''$.
	
	\begin{claim}\label{claim: flow in H''}
		If $|Q|\leq n^{1-4\eps}$, then $F''=F$. Additionally,
		if $(s,T)$ is a good pair and  neither of the bad events $\event_i, \tilde \event_{i+1},\hat \event,\hat \event',\hat \event''$ happened, then $F''\geq F-2n^{1-4.4\eps}\cdot M'$ even if $|Q|>n^{1-4\eps}$.
	\end{claim}
\begin{proof}
	Let $f^*$ be the maximum $s^{\out}$-$t$ flow in $H$, and let $\pset$ be its flow-path decomposition. Let $\pset'\subseteq \pset$ be the collection of all paths $P\in \pset$ that are contained in $H''$, and let $\pset''=\pset\setminus \pset'$. Consider now a path $P\in \pset''$, and let $e(P)\in E(P)$ be the first edge on $P$ that does not lie in $H''$. Note that the edges of $E_1,E_2$ and $E_3$ are all present in $H''$, so $e(P)\in E_4\cup E_5\cup E_6$ must hold. We further partition the set $\pset''$ into two subsets: set $\pset''_1$ containing all paths $P\in \pset''$ with $e(P)\in E_4\cup E_6$, and set $\pset''_2=\pset''\setminus \pset''_1$.
	
	Consider first a path $P\in \pset''_1$, and denote $e(P)=(x,y)$. Recall that $e(P)\in (E_4\cup E_6)\setminus E(H'')$ holds. This can only happen if shortcut edge $(x,t)$ was added to $G'',H$ and $H'$, so in particular this edge lies in $H''$. We let $P'$ be the path obtained from $P$ by contcatenating the subpath of $P$ from $s^{\out}$ to $x$ and the edge $(x,t)$.
	
	Let $f'$ be the $s^{\out}$-$t$ flow in $H''$ obtained by sending, for every path $P\in \pset'$, $f^*(P)$ flow units on $P$, and, for every path $P\in \pset'_1$, $f^*(P)$ flow units on $P'$. Notice that $f'$ is a valid $s^{\out}$-$t$ flow in $H''$, and its value is $\sum_{P\in \pset'\cup \pset''_1}f^*(P)=F-\sum_{P\in \pset''_2}f^*(P)$.
Observe that, if $|Q|\leq n^{1-4\eps}$, then all edges of $E_5$ were added to $H'$, and so $\pset''_2=\emptyset$. Therefore, in this case, $F'\geq \sum_{P\in \pset'\cup \pset''_1}f^*(P)=F$. 

Assume from now on that $|Q|>n^{1-4\eps}$, and additionally that $(s,T)$ is a good pair and  neither of the bad events $\event_i, \tilde \event_{i+1},\hat \event,\hat \event',\hat \event''$ happened.
It is now sufficient to prove that, in this case, $\sum_{P\in \pset''_2}f^*(P)\leq 2n^{1-4.4\eps}\cdot M'$ holds.
	
	Recall that for every edge $e=(x,y)\in E_5$, $y=v^{\inn}$ for some vertex $v\in Q$ and $x\in V^{\out}\setminus V(J)$. In particular, $y\in  X^*$ and $x\not\in  X^*$ holds. Since $s^{\out}\in X^*$, we get that every path in $\pset''_2$ must contain at least two edges of $E_H(X^*,V(H)\setminus X^*)$. Clearly, every path in $\pset$ must contain at least one such edge. Therefore:
	
	\[ \sum_{e\in E_H(X^*,V(H)\setminus X^*)}f^*(e)\geq F+\sum_{P\in \pset''_2}f^*(P).\]
	
	However, from \Cref{obs: Step 2 summary}, $\sum_{e\in E_H(X^*,V(H)\setminus X^*)}\hat c(e)\leq F+2n^{1-4.4\eps}\cdot M'$ must hold, and so:
	
		\[ \sum_{e\in E_H(X^*,V(H)\setminus X^*)}f^*(e)\leq \sum_{e\in E_H(X^*,V(H)\setminus X^*)}\hat c(e)\leq F+2n^{1-4.4\eps}\cdot M'.\]
		
Altogether, we get that $\sum_{P\in \pset''_2}f^*(P)\leq 2n^{1-4.4\eps}\cdot M'$, and so: 

$$F''\geq \sum_{P\in \pset'\cup \pset''_1}f^*(P)=F-\sum_{P\in \pset''_2}f^*(P)\geq F-2n^{1-4.4\eps}\cdot M'.$$
\end{proof}

Observe that graph $H'$ can be thought of as being obtained from $H''$ in two steps: first, we decrease the capacity of every forward special edge by at least $M'$ and at most $2M'$ flow units. It is immediate to verify that this may only decrease the value of the maximum flow by at most $2nM'$ units, by inspecting the minimum $s^{\out}$-$t$ cut in $H''$, whose capacity may decrease by at most $2nM'$. Let $\hat H$ be the resulting flow network. Lastly, for every vertex $v\in V(G)$ with $w(v)<M'$, we delete all regular edges that are incident to $v^{\out}$ or $v^{\inn}$ in $\hat H$, obtaining the final graph $H'$. We claim that this last step does not affect the value of the maximum $s^{\out}$-$t$ flow. Indeed, let $\hat f$ be the maximum $s^{\out}$-$t$ flow in $\hat H$, and denote its value by $\hat F$. Let $e=(v^{\inn},v^{\out})$ be a special edge corresponding to a vertex $v\in V(G)$ with $w(v)<M'$, and let $e'=(x,y)$ be a regular edge of $H''$ that is incident to either $v^{\inn}$ or $v^{\out}$. We claim that $\hat f(e')=0$ must hold, so deleting this edge will not influence the flow $\hat f$. Indeed, recall that, from Property \ref{prop: residual capacity weaker}, $w(v)<M'$ must hold. Since the current flow $f$ is $M'$-integral, $f(v^{\inn},v^{\out})=0$ must hold. Since the only edge leaving $v^{\inn}$ in $G''$ is the edge $(v^{\inn},v^{\out})$, it is immediate to verify that, for every regular edge $(u^{\out},v^{\inn})\in E(G'')$ that enters $v^{\inn}$, $f(u^{\out},v^{\inn})=0$. Therefore, the only regular edges incident to $v^{\inn}$ in $H$ are the forward regular edges that enter $v^{\inn}$. Since the capacity of the edge $(v^{\inn},v^{\out})$ in $\hat H$ is $0$, we get that the flow $\hat f(e')$ on every edge $e'\in E(\hat H)$ that is incident to $v^{\inn}$ must be $0$. Using a similar reasoning, the flow $\hat f(e'')$ on every edge $e''\in E(\hat H)$ that is incident to $v^{\out}$ must be $0$. We conclude that $F'\geq \hat F\geq F''-2nM'$. From \Cref{claim: flow in H''}, if $|Q|\leq n^{1-4\eps}$, then $F''=F$, and so $F'\geq F-2n\cdot M$. Moreover, 
if $(s,T)$ is a good pair and if neither of the bad events $\event_i, \tilde \event_{i+1},\hat \event,\hat \event',\hat \event''$ happened, then $F''\geq F-2n^{1-4.4\eps}\cdot M'$ holds even if $|Q|>n^{1-4\eps}$. In this case, $F'\geq F-2nM'-2n^{1-4.4\eps}\cdot M'$.
\end{proof}

\subsection{Completing the Phase}

In order to complete the algorithm for Step 3, we use the algorithm from \Cref{thm: maxflow} in order to compute a maximum $s^{\out}$-$t$ flow $f''$ in graph $H'$; since all edge capacities in $H'$ are integral multiples of $M'$, the resulting flow $f''$ is guaranteed to be $M'$-integral.
The running time of the algorithm from \Cref{thm: maxflow} is $O\left (|E(H')|^{1+o(1)}\cdot \log(\wmax)\right )\leq O\left(n^{2-4\eps+o(1)}\cdot \log^4(\wmax)\right )$, from Inequality \ref{eq: size of H'}.

Next, we consider two cases. The first case happens if $|Q|\leq n^{1-4\eps}$. In this case, we let $f_i$ be the 
$s^{\out}$-$t$ flow in $G''$ obtained by augmenting the current flow $f$ with the flow 
$f''$. From \Cref{claim: max flow value}, and because $F\geq \opt_s-\val(f)$ must hold, we get that:

\begin{equation}\label{eq: flow value case 1}
\val(f_i)=\val(f)+\val(f'')=\val(f)+F'\geq \val(f)+F-2nM'\geq \opt_s-2nM'.
\end{equation}

Note that the resulting flow $f_i$ must be $M'$-integral. Given the flow $f$, the flow $f_i$ can be computed in time $O(|E(H')|)\leq O(n^{2-4\eps}\cdot \log n\cdot \log^3(\wmax))$, and we can update the data structure $\DS_f$, and the residual flow network $H$ within the same asymptotic running time.

Assume now that $|Q|>n^{1-4\eps}$. In this case, we check whether $\val(f'')\geq \sum_{e\in E_H(X^*,V(H)\setminus X^*)}\hat c(e)-4n^{1-4.4\eps}\cdot M'$ holds, and, this is not the case, then we terminate the algorithm with a ``FAIL''. Notice that the time to compute the value $\sum_{e\in E_H(X^*,V(H)\setminus X^*)}\hat c(e)$ is bounded by $O(|E(\tilde H')|)$, and it is subsumed by the running time of the algorithm for Step 3 so far. From \Cref{obs: Step 2 summary}, if $(s,T)$ is a good pair and neither of the bad events $\tilde \event_{i+1}$, \ref{event: heavy err} happened, then $F\geq \sum_{e\in E_H(X^*,V(H)\setminus X^*)}\hat c(e)-2n^{1-4.4\eps}\cdot M'$. Additionally, from \Cref{claim: max flow value}, 
$(s,T)$ is a good pair, and neither of the bad events $\event_i, \tilde \event_{i+1},\hat \event,\hat \event',\hat \event''$ happened, then:

\[\val(f'')=F'\geq F-2n\cdot M'-2n^{1-4.4\eps}\cdot M'\geq \sum_{e\in E_H(X^*,V(H)\setminus X^*)}\hat c(e)-2n\cdot M'-4n^{1-4.4\eps}\cdot M'.
\]

Therefore, if $|Q|>n^{1-4\eps}$, $(s,T)$ is a good pair, and neither of the bad events $\event_i, \tilde \event_{i+1},\hat \event,\hat \event',\hat \event''$ has happened, we are guaranteed that $\val(f'')\geq \sum_{e\in E_H(X^*,V(H)\setminus X^*)}\hat c(e)-2n\cdot M'-4n^{1-4.4\eps}\cdot M'$ must hold. So our algorithm may only terminate with a ``FAIL'' in this case if $(s,T)$ is not a good pair, or at least one of the events $\event_i, \tilde \event_{i+1},\hat \event,\hat \event',\hat \event''$ has happened.

Finally, assume that $|Q|>n^{1-4\eps}$, and that $\val(f'')\geq \sum_{e\in E_H(X^*,V(H)\setminus X^*)}\hat c(e)-4n^{1-4.4\eps}\cdot M'$. Clearly, $F\leq \sum_{e\in E_H(X^*,V(H)\setminus X^*)}\hat c(e)$ must hold. Therefore: 

\begin{equation}\label{eq: flow value 2}
\val(f'')\geq F-4n^{1-4.4\eps}\cdot M'-2nM'\geq F-4nM'\geq \opt_s-\val(f)-4nM'.
\end{equation}


By augmenting the flow $f$ with $f''$, exactly as before, we the obtain a new flow $f_i$, with $\val(f_i)=\val(f)+\val(f'')\geq \opt_s-4nM'$, that is $M'$-integral.

To summarize, the running time of the algorithm for Step 3 is bounded by  $O(n^{2-4\eps+o(1)}\cdot \log^4(\wmax))$, and the total running time of the entire phase is bounded by $O(n^{2-4\eps+o(1)}\cdot \left(\log(\wmax)\right)^{O(1)})$. If the algorithm for the $i$th phase does not terminate with a ``FAIL'', then it computes a flow $f_i$ in $G''$, that is $M'$-integral, where $M'=M_i$, 
with $\val(f_i)\geq \opt_s-4nM'=\opt_s-4nM_i$, establishing Invariant \ref{inv: flow}. We have already established Invariant \ref{inv: shortcut}, and Invariant \ref{inv: residual cap large} follows from the fact that we decrease the residual capacity of every forward special edge in $H'$ by at least $M'=M_i$ units.
From Invariant \ref{inv: few edges carry flow}, the total number of edges $e\in E(G'')$ with $f_{i-1}(e)>0$ was at most $O\left(i\cdot n^{2-4\eps+o(1)}\cdot \log^4(\wmax)\right )$. We only augmented the initial flow $f=f_{i-1}$ in Steps 1 and 3. Recall that the number of iterations in Step 1 is bounded by $z'=\frac{32n}{\gamma}$. 
In each iteration $j$, we may compute a flow  $\hat f_j$ in the graph $J$ that is constructed in the iteration, and then augment $f$ with $\hat f_j$. Since, from \Cref{obs: bound E(J)}, $|E(J)|\leq  O\left(\frac{n^{2+2\eps}\cdot \log^2(\wmax)}{\gamma}\right )$, we get that the total number of edges $e\in E(G'')$ with $f_{i-1}(e)=0$, such that the flow $f(e)>0$ holds at the end of Step 1 is bounded by:

\[
 O\left(\frac{n^{2+2\eps}\cdot \log^2(\wmax)}{\gamma}\right )\cdot z'\leq O\left(\frac{n^{3+2\eps}\cdot \log^2(\wmax)}{\gamma^2}\right )\leq O\left(n^{2-4\eps+o(1)}\cdot \log^2(\wmax)\right ),
\]

since $\gamma\geq n^{2/3+5\eps}$ from Inequality \ref{eq: gamma2}. In Step 3 we perform a single augmentation of the flow $f$, by using the maximum $s^{\out}$-$t$ flow $f''$ in graph $H'$. Since $|E(H')|\leq O(n^{2-4\eps}\cdot \log n\cdot \log^3(\wmax))$ from Inequality \ref{eq: size of H'}, this augmentation may affect at most $O(n^{2-4\eps}\cdot \log n\cdot \log^3(\wmax))$ edges of $G''$. Overall, the number of edges $e\in E(G'')$ with $f_{i-1}(e)=0$ and $f_i(e)>0$ is bounded by $O\left(n^{2-4\eps+o(1)}\cdot \log^3(\wmax)\right )$, and so Invariant 
\ref{inv: few edges carry flow} continues to hold. It now only remains to establish Invariant \ref{inv: fail prob}.

Assume that $(s,T)$ is a good pair.
The algorithm may terminate Phase $i$ with a ``FAIL'' either in Step 1, or Step 2, or in Step 3. From \Cref{claim: step 1 analysis}, if Event $\tilde \event_i$ did not happen, then the probability that the algorithm terminates with a ``FAIL'' in Step 1 of the current phase is at most $\frac{2}{n^{\eps}\cdot \log^2(\wmax)}$.
From \Cref{obs: fail in step 2}, the algorithm for Step 2 may only terminate with a ``FAIL'' if at least one of the events $\tilde \event_{i+1}$ or $\hat \event$ has happened. 
The algorithm may only terminate in Step 3 of the current phase with a ``FAIL'' if one  of the bad events $\event_i, \tilde \event_{i+1},\hat \event,\hat \event',\hat \event''$ happened. Overall, if $(s,T)$ is a good pair, then the probability that the current phase terminates with a ``FAIL'' is bounded by:

\[
\frac{2}{n^{\eps}\cdot \log^2(\wmax)}+\prob{\tilde \event_{i+1}}+\prob{\hat \event}+\prob{\hat \event'}+\prob{\hat \event''}.
\]

Recall that 
$\prob{\tilde \event_{i+1}}\leq \frac{2(i+1)}{n^{\eps}\cdot \log^2(\wmax)}$ from Inequality \ref{eq: invalid shortcuts}, and $\prob{\hat \event}\leq  \frac{1}{n^7\cdot \log^4(\wmax)}$ from Inequality \ref{eq: prob error Heavy Vertex}. Additionally, from our discussion, $\prob{\hat \event'}\leq  \frac{1}{n^5\cdot \log^4(\wmax)}$ and $\prob{\hat \event''}\leq \frac{1}{n^{10}\cdot \log^4(\wmax)}$.
Altogether, we get that, if $(s,T)$ is a good pair, then the probability that the current phase terminates with a ``FAIL'' is bounded by:

\[
\begin{split}
\frac{2}{n^{\eps}\cdot \log^2(\wmax)}&+\frac{2(i+1)}{n^{\eps}\cdot \log^2(\wmax)}+ \frac{1}{n^5\cdot \log^4(\wmax)}+\frac{1}{n^{10}\cdot \log^4(\wmax)}\\
&\leq \frac{2i+5}{n^{\eps}\cdot \log^2(\wmax)}\\
&\leq \frac{7i}{n^{\eps}\cdot \log^2(\wmax)}.
\end{split}
\]

From Invariant \ref{inv: fail prob}, if
 $(s,T)$ is a good pair, then the probability that the algorithm terminates prior to the beginning of Phase $i$ with a ``FAIL'' is at most $\frac{10i^2-2i}{n^{\eps}\cdot \log^2(\wmax)}$. We then get that, if $(s,T)$ is a good pair, then the probability that  the algorithm terminates prior to the beginning of Phase $i$ with a ``FAIL'' is at most $\frac{10i^2+5i}{n^{\eps}\cdot \log^2(\wmax)}\leq \frac{10(i+1)^2-2(i+1)}{n^{\eps}\cdot \log^2(\wmax)}$, so Invariant  \ref{inv: fail prob} continues to hold at the end of the phase.

\paragraph{Acknowledgement.} The authors thank Ron Mosenzon for pointing out a minor error in the definition of a good set of terminals (\Cref{def: good set of terminals}) in a previous version of the paper.

\newpage
\appendix
\section{Isolating Cuts: Proof of \Cref{thm: min cuts via isolating}}
\label{sec: appendix_isolating}

In this section we provide the proof of Theorem~\ref{thm: min cuts via isolating}. We restate this theorem for the convenience of the reader.

\isolating*

In this proof, for a vertex set $X\subseteq V(G)$, we denote by $N_G(X)=\left(\bigcup_{v\in X}N_G(v)\right )\setminus X$.
It will be convenient for us to assume that the vertices of $T$ form an independent set, that is, no edge of $G$ connects a pair of such vertices. In order to ensure this property, we simply subdivide every edge of $G$ with a new vertex of weight $2nW$. It is easy to verify that, for every vertex $v\in T$, the value of the minimum vertex-cut separating $v$ from $T\setminus\set{v}$ does not change following this transformation; if there is no vertex-cut separating $v$ from $T\setminus\set{v}$ in $G$, then the value of the minimum vertex-cut separating $v$ from $T\setminus\set{v}$ in the new graph is at least $2Wn$, higher than the value of any proper vertex-cut in $G$. Additionally, the number of edges in $G$ only grows by factor $2$. Therefore, we assume from now on that $T$ is an independent set.

Let $r=\lceil \log |T|\rceil$. We assign, to every vertex $v\in T$, a binary string $\beta_v$ of length $r$, so that strings assigned to distinct vertices of $T$ are  different.
Next, we perform $r$ iterations. For $1\leq i\leq r$, in order to perform the $i$th iteration, we let $A_i\subseteq T$ contain all vertices $v\in T$ for which the $i$th bit in $\beta_v$ is equal to $0$, and we let $B_i=T\setminus A_i$. 
We then compute a minimum $A_i$-$B_i$ vertex-cut $(X_i,Y_i,Z_i)$ in $G$, in time $O\left(m^{1+o(1)}\cdot\log W\right )$, using the algorithm from \Cref{cor: min_vertex_cut}.
Consider now the graph $G\setminus\left(\bigcup_{i=1}^rY_i\right )$, and, for every vertex $v\in T$, denote by $C_v$ the connected component of this graph that contains $v$, by $U_v=V(C_v)$, and by $m_v=\sum_{u\in U_v}\deg_G(u)$. 
Observe that, for all $1\leq i\leq r$, $Y_i\cap T=\emptyset$, and for all $v\in T$, $N_G(U_v)\subseteq \bigcup_{i=1}^rY_i$. Therefore, for every vertex $v\in T$, $N_G(U_v)\cap T=\emptyset$.
We start with the following simple observation.

\begin{observation}\label{obs: individual clusters}
	For every vertex $v\in T$, $U_v\cap T=\set{v}$.
\end{observation}
\begin{proof}
	Consider a vertex $v\in T$.
	By definition, $v\in U_v\cap T$. Assume for contradiction that $U_v$ contains another vertex $u \neq v$ with $u\in T$.
	Since the binary strings assigned to $u$ and $v$ are different, there is an index $1\leq j\leq r$, such that $\beta_v$ and $\beta_u$ differ in the $j$th bit.
	Assume without loss of generality that $u\in A_j$ and $v\in B_j$. Since $Y_j \subseteq V(G)$ separates $A_j$ from $B_j$ in $G$, every path $P$ connecting $u$ to $v$ in $G$ must contain a vertex of $Y_j$. Therefore, it is impossible that $u$ and $v$ lie in the same connected component of $G\setminus\left(\bigcup_{i=1}^rY_i\right )$.
\end{proof}

From \Cref{obs: individual clusters}, we get that for every pair $u,v\in T$ of distinct vertices, $C_u\neq C_v$, and so $\sum_{v\in T}m_v\leq O(m)$.

For every vertex $v\in T$, we then construct a
graph $G_v$ as follows. We start from the subgraph of $G$ induced by the set $U_v\cup N_G(U_v)$ of vertices, and then remove all edges $e=(x,y)$ with $x,y\in N_G(U_v)$. We then add a vertex $t$, and connect it to every vertex in $N_G(U_v)$ with an edge. This finishes the construction of the graph $G_v$. Note that $|E(G_v)|\leq O(m_v)$, and moreover, given the adjacency-list representaton of $G$, we can compute $G_v$ in time $O(m_v)$, by inspecting the adjacency lists of the vertices of $U_v$.
From our observation, since $(U_v\cup N_G(U_v))\cap T=\set{v}$, we get that the $V(G_v)\cap T=\set{v}$.
We use the algorithm from \Cref{cor: min_vertex_cut} to compute, in time $O\left (|E(G_v)|^{1+o(1)}\cdot \log W\right )\leq O\left (m_v^{1+o(1)}\cdot \log W \right )$, a minimum $v$-$t$ vertex-cut in $G_v$, and denote by $c_v$ its value. We then return, for every vertex $v\in T$, the value $c_v$ as the value of the minimum vertex cut separating $v$ from $T\setminus\set{v}$ in $G$.
From our discussion, the total running time of the algorithm is bounded by $O\left(m^{1+o(1)}\cdot\log W\right )+\sum_{v\in T}O\left (m_v^{1+o(1)}\cdot \log W\right )\leq O\left(m^{1+o(1)}\cdot\log W\right )$. It now remains to show that, for every vertex $v$, the value $c_v$ that the algorithm outputs is indeed the value of the minimum vertex-cut separating $v$ from $T\setminus\set{v}$ in $G$.

Consider any vertex $v\in T$, and let $(L_v,S_v,R_v)$ be a minimum vertex-cut separating $v$ from $T\setminus \set{v}$ in $G$, so that $v\in L_v$; from among all such cuts, choose the one minimizing $|L_v|$. We denote $w(S_v)$ by $\lambda_v$, so $\lambda_v$ is the value of the  minimum vertex-cut separating $v$ from $T\setminus\set{v}$ in $G$.
The following claim is key to the correctness of the algorithm.

\begin{claim}\label{claim: single vertex separator}
	For every vertex $v\in T$, $L_v\subseteq U_v$ must hold.
\end{claim}

We prove the claim below, after we complete the proof of \Cref{thm: min cuts via isolating} using it. Fix a vertex $v\in T$, and
notice that, since $L_v\subseteq U_v$, $N_G(L_v)\subseteq U_v\cup N_G(U_v)$ must hold.
 Consider the  partition $(L',S',R')$ of the vertices of $G_v$, where $L'=L_v$; $S'=N_G(L_v)$; and $R'$ contains the remaining vertices of $G_v$, so in particular $t\in R'$. Clearly, $S'\subseteq S_v$, and moreover, no edge of $G_v$ may connect a vertex of $L'$ to a vertex of $R'$. Therefore, 
$(L',S',R')$ is a valid $v$-$t$ vertex-cut in $G_v$, and its value is at most $\lambda_v=w(S_v)$. We conclude that $c_v\leq \lambda_v$. On the other hand, consider any $v$-$t$ vertex-cut $(A,B,C)$ in $G_v$, so $v\in A$ and $t\in C$ holds.
Note that $A\subseteq U_v$ must hold, since every vertex $u\in V(G_v)\setminus\set{t}$ that does not lie in $U_v$ is connected to $t$ with an edge.
In particular, this means that $N_G(A)=N_{G_v}(A)\subseteq B$.
Consider now a partition $(A',B',C')$ of $V(G)$, where $A'=A$, $B'=N_G(A)$, and $C'$ contains the remaining vertices of $G$. From the above discussion, $B'\subseteq B$ must hold. Clearly, $(A',B',C')$ is a vertex-cut in $G$, whose value is at most $w(B)$. Since $V(G_v)\cap T=\set{v}$, the only terminal that lies in $A\cup B$ is $v$, and so this cut separates $v$ from $T\setminus\set{v}$. We conclude that $c_v\geq \lambda_v$, and altogether, $c_v=\lambda_v$. In order to complete the proof of \Cref{thm: min cuts via isolating}  it now remains to prove \Cref{claim: single vertex separator}.

\begin{proofof}{\Cref{claim: single vertex separator}}
	We fix a vertex $v\in T$. For all $1\leq i\leq r$, we define a set $\Lambda_i$ of vertices as follows. If $v\in A_i$, then we set $\Lambda_i=X_i$; note that $\Lambda_i\cap T=A_i$ in this case. Otherwise, $v\in B_i$, and we set $\Lambda_i=Z_i$; in this case, $\Lambda_i\cap T=B_i$.
	In the following claim we show that, for all $1\leq i\leq r$, $L_v\subseteq \Lambda_i$. 
	
	\begin{claim}\label{claim: containment for one iteration}
		For all $1\leq i\leq r$, $L_v\subseteq \Lambda_i$. 
	\end{claim}
	\begin{proof}
	The proof of the claim relies on \emph{submodularity of vertex-cuts}, that is summarized in the following observation.

	\begin{observation}[Submodularity of weighted vertex-cuts.]\label{obs:submodularity} Let $G=(V,E)$ be an undirected graph with weights $w(v)\geq 0$ on its vertices $v\in V$, and let $A,B\subseteq V$ be a pair of vertex subsets. Then:
		
		$$w(N(A)) + w(N(B)) \geq w(N(A \cup B)) + w(N(A \cap B)).$$
	\end{observation}
	\begin{proof}
		We consider the contribution of every vertex $v\in V$ to the expressions: 
		
		\begin{equation}\label{eq 1}
		w(N(A))+w(N(B))
		\end{equation}
		
		and
		
		\begin{equation}\label{eq: 2}
		w(N(A \cup B)) + w(N(A \cap B)).
		\end{equation}
		

Notice that a vertex $v \notin N(A) \cup N(B)$ contributes $0$
to expression \ref{eq 1}. We now show that it also contributes $0$ to expression \ref{eq: 2}. Indeed, if no neighbor of $v$ lies in $A\cup B$, then $v$ may not lie in either of the sets $N(A \cup B)$ or $N(A \cap B)$, so its contribution to expression \ref{eq: 2} is $0$. Assume now that some neighbor $u\in N(v)$ of $v$ lies in $A$, but no neighbor of $v$ lies in $B$. Since $v\not\in N(A)$, $v\in A$ must hold. But then $v\in  A\cup B$, and no neighbor of $v$ may lie in $A\cap B$. Therefore, $v\not\in N(A \cup B)$ and $v\not \in N(A \cap B)$ must hold. The case where some neighbor of $v$ lies in $B$ but no neighbor of $v$ lies in $A$ is symmetric. Finally, assume that some neighbor of $v$ lies in $A$, and the same holds for $B$. Since $v\not \in  N(A) \cup N(B)$, it must be the case that $v\in A$ and $v\in B$. But then $v\in A\cap B$, so it may not lie in either of the sets $N(A \cup B)$ or in $N(A \cap B)$. To conclude, if $v \notin N(A) \cup N(B)$, its contribution to both expressions is $0$.

Assume now that $v\in N(A)\cup N(B)$. 
		Observe first that, if $v \in N(A) \cap N(B)$, then it contributes $2w(v)$ to the first expression, and its contribution to the second expression  is at most $2w(v)$.

		Assume next that $v\in N(A) \setminus N(B)$, so $v$ contributes $w(v)$ to expression \ref{eq 1}. 
		We show that in this case, either $v\not \in N(A\cap B)$, or $v\not \in N(A\cup B)$ must hold, so $v$ contributes at most $w(v)$ to the expression \ref{eq: 2}. In order to do so, we consider two cases. The first case happens if  $v\not\in B$. In this case, no vertex of $N(v)$ may lie in $B$, so $v\not \in N(A\cap B)$. Assume now that the second case happens, so $v\in B$. In this case, $v\not\in N(A\cup B)$ must hold, since, from the definition, $N(A\cup B)$ does not contain vertices of $A\cup B$.
		
		The only remaining case is where  $v \in N(B) \setminus N(A)$, and it follows a symmetric argument.
	\end{proof}
	
We are now ready to complete the proof of 	\Cref{claim: containment for one iteration}.	
Assume for contradiction that the claim does not hold, so $L_v \setminus \Lambda_i \neq \emptyset$. 
		Consider a vertex-cut $(\hat X,\hat Y,\hat Z)$ in $G$, where $\hat X=L_v \cap \Lambda_i$ and $\hat Y=N(L_v\cap \Lambda_i)$, and $\hat Z$ contains all remaining vertices of $G$. It is immediate to verify that this cut separates $v$ from $T\setminus\set{v}$. Indeed, since $L_v\cap T=\set{v}$ and $v\in \Lambda_i$, we get that $\hat X\cap T=\set{v}$. Additionally, since $\hat Y\subseteq L_v\cup N(L_v)$ and $(L_v\cup N(L_v))\cap T=\set{v}$, we get that $\hat Y\cap T=\emptyset$. Moreover, since $|\hat X|<|L_v|$, from the definition of the cut $(L_v,S_v,R_v)$, we get that $w(\hat Y)>w(S_v)$ must hold. 
		However, by \Cref{obs:submodularity}:

		\begin{equation}\label{eq: submod}
		w(N(L_v \cup \Lambda_i)) + w(N(L_v \cap \Lambda_i)) \leq w(N(L_v)) + w(N(\Lambda_i)).
		\end{equation}

		Recall that $(L_v,S_v,R_v)$ is a vertex-cut in $G$, so $N_G(L_v)\subseteq S_v$ must hold, and therefore, $w(N(L_v))\leq w(S_v)$. We then get that:
		
		\[w(N(L_v \cap \Lambda_i))=w(N(\hat X))=w(\hat Y)>w(S_v)\geq w(N(L_v)).\]

	Combining this with Inequality \ref{eq: submod}, we get that:

		\[
		w(N(L_v \cup \Lambda_i)) < w(N(\Lambda_i))
		\]
		
		must hold.
		Since $L_v\cap T=\set{v}$ and $v\in \Lambda_i$, we get that $(L_v \cup \Lambda_i) \cap T = \Lambda_i \cap T$. 
		In particular, $L_v \cup \Lambda_i$ contains all vertices of $A_i$ and no vertices of $B_i$ (or the other way around). 
		Moreover, since $N(L_v) \cap T = \emptyset$ and $N(\Lambda_i) \cap T = \emptyset$, we also get that $N(L_v \cup \Lambda_i) \cap T = \emptyset$. 
		Therefore, the vertices of $N(L_v\cup \Lambda_i)$ separate the vertices of $A_i$ from the vertices of $B_i$, and their weight is smaller than $w(N(\Lambda_i))\leq w(Y_i)$, contradicting the choice of the vertex-cut $(X_i,Y_i,Z_i)$.
	\end{proof}

	Let $U^*_v= \bigcap_{i=1}^r\Lambda_i$.
	We conclude that $L_v\subseteq U^*_v$. 
	Recall that $C_v$ is the connected component of $G\setminus\left(\bigcup_{i=1}^rY_i\right)$ containing the vertex $v$. It is also easy to verify that $C_v\subseteq U^*_v$, and that equivalently, $C_v$ can be defined as the connected component of $G[\bigcap_{i=1}^r\Lambda_i]$ containing $v$. Since $G[L_v]$ must be a connected graph from the minimality of $|L_v|$, it must be a subset of $U_v=V(C_v)$.
\end{proofof}

\section{Degree Estimation and Heavy Vertex Problems}
In this section we provide algorithms for both Undirected and Directed Degree Estimation problems, as well as an algorithm for the  Heavy Vertex Problem.
We start with an algorithm for the Undirected Degree Estimation problem.

\subsection{Algorithms for Degree Estimation Problems: Proof of Claims \ref{claim: degree est} and \ref{claim: heavy weight est}}
\label{sec: appendix_degree}


We start by presenting an algorithm for the Undirected Degree Estimation problem, proving \Cref{claim: degree est}, whose statement we restate here for the convenience of the reader.

\degree*

\paragraph{Proof of Claim~\ref{claim: degree est}.}
Assume first that $\tau> \frac{|Z|}{100\log n\cdot\log(\wmax)}$. In this case, we return 
	$A=\emptyset$. It is easy to verify that in this case the algorithm does not err, since no vertex in $Z'$ may have at least $1000\tau\log n\cdot \log(\wmax)>|Z|$ neighbors in $Z$. We assume from now on that $\tau \leq \frac{|Z|}{100\log n\cdot \log(\wmax)}$ holds.
	
	We start by computing a random set $T\subseteq Z$ of vertices as follows: every vertex $v\in Z$ is added to $T$ independently with probability $\frac{10\log n\cdot\log(\wmax)}{\tau}$. Let $\event$ be the bad event that $|T|>\frac{100|Z|\log n\cdot\log(\wmax)}{\tau}$. 
	Since $\expect{|T|}=\frac{10|Z|\log n\cdot\log(\wmax)}{\tau}$, by applying the first Chernoff bound from \Cref{lem: Chernoff} with $t=\frac{100|Z|\log n\cdot\log(\wmax)}{\tau}$, we get that $\prob{\event}\leq 2^{-(100|Z|\log n\cdot\log(\wmax))/\tau}\leq n^{-100\log(\wmax)}$, since $\tau\leq |Z|$. 
	Notice that we can check whether Event $\event$ happened in time $O\left (\frac{|Z|\cdot \log n\log(\wmax)}{\tau}\right )$. If we discover that the event happened, then we terminate the algorithm and return $A=\emptyset$.
	
	Assume now that Event $\event$ did not happen.
	We let $A\subseteq Z'$ be the set that contains all vertices $v\in Z'$ with $|N_{G}(v)\cap T|\geq 100\log n\log(\wmax)$. Note that set $A$ of vertices can be computed in time $O(|T|\cdot n)\leq O\left(\frac{n\cdot |Z|\cdot \log n\cdot \log(\wmax)}{\tau}\right )$, as follows: for every vertex $u\in T$, we consider all its $O(n)$ neighbors in $G$. For each such neighbor $v$, if $v\not\in Z'$, then no further processing of $v$ is needed. Otherwise, if $v\in Z'$, and $v$ is processed for the first time by our algorithm, then we initialize the counter $n_v=1$, and mark $v$ as being already considered by the algorithm. We also add $v$ to the list $\Lambda\subseteq Z'$ of vertices that our algorithm considered. Otherwise, we simply increase the counter $n_v$ by $1$. Clearly, the algorithm spends $O(n)$ time for processing every vertex $v\in T$, and $O(n|T|)\leq O\left(\frac{n\cdot |Z|\cdot \log n\cdot\log(\wmax)}{\tau}\right )$ time overall. We then consider the vertices $v\in \Lambda$ one by one, and for each such vertex $v$, if $n_v\geq 100\log n\cdot\log(\wmax)$, then we add $v$ to $A$.
	
	Consider now some vertex $v\in Z'$. We say that a bad event $\event'(v)$ happened if $|N_G(v)\cap Z|<\tau$ and $v\in A$, and we say that a bad event $\event''(v)$ happened if $|N_G(v)\cap Z|\geq 1000\tau\log n\log(\wmax)$ and $v\not\in A$. We now bound the probabilities of events $\event'(v)$ and $\event''(v)$.
	
	For convenience, denote $N_G(v)\cap Z=\set{y_1,\ldots,y_r}$. For all $1\leq i\leq r$, let $x_i$ be the random variable that is equal to $1$ if $y_i\in T$ and to $0$ otherwise, and let $X_v=\sum_{i=1}^rx_i$. Notice that $\expect{X_v}=\frac{10r\cdot \log n\cdot\log(\wmax)}{\tau}$.
	
	Assume first that $r<\tau$, so $\expect{X_v}<10\log n\log(\wmax)$, and recall that $v$ is only added to $A$ if $X_v>100\log n\cdot\log(\wmax)$. In this case, by using the first part of the Chernoff bound from \Cref{lem: Chernoff} with $t=100\log n\cdot\log(\wmax)$, we get that:
	
	\[\prob{\event'(v)}=\prob{X_v>100\log n\cdot\log(\wmax)}\leq 2^{-100\log n\cdot\log(\wmax)}=n^{-100\cdot\log(\wmax)}.\]
	
	Assume now that $r\geq 1000\tau\log n\cdot\log(\wmax)$, so $\expect{X_v}\geq 1000\log^2 n\log^2(\wmax)$. In this case, $\event''(v)$ may only happen if $X_v\leq 100\log n\cdot\log(\wmax)$. Using the second part of the Chernoff bound from \Cref{lem: Chernoff} with $\delta=0.9$, we get that:

	\[\begin{split}
	\prob{\event''(v)}&=\prob{X_v< 100\log n\cdot\log(\wmax)}\\
	&\leq \prob{X_v< (1-\delta)\cdot 1000\log^2 n\cdot\log^2(\wmax)}\\
	&\leq e^{-0.81\cdot 500\log^2 n\cdot\log^2(\wmax)}\\
	&\leq n^{-400\cdot\log(\wmax)}.
	\end{split}\]

	Note that our algorithm may only err if the event $\event$ happens, or if, for some $v\in Z'$, at least one of the events $\event'(v)$ or $\event''(v)$ happen. Using the Union Bound, we conclude that the probability that the algorithm errs is at most
\[\begin{split}
n^{-100\log(\wmax)}&+\sum_{v\in V}\left(n^{-100\cdot\log(\wmax)}+n^{-400\cdot\log(\wmax)}\right)\\
&\leq n^{-50\cdot \log(\wmax)} \\
&\leq n^{-50-\log(\wmax)} \\
&\leq \frac{1}{n^{10}\cdot \log^6(\wmax)}.
\end{split}\]

This concludes the proof of \Cref{claim: degree est}. We note that in the above proof we could have replaced $\wmax$ with any other parameter $\hat W$ that is greater than a large enough constant, so that $n^{\log(\hat W)}\geq \log^6\hat W$ holds.

We now restate Claim~\ref{claim: heavy weight est}.
\heavyweight*

The proof of \Cref{claim: heavy weight est} easily follows by adapting the algorithm from the proof of \Cref{claim: degree est}. The only difference is that we replace the parameter $\wmax$ with $W$ and  switch the roles of vertex sets $Z$ and $Z'$, so the set $T$ of vertices is subsampled from $Z'$.  We then add to the set $A$ every vertex $v\in Z$, such that there are at least $100\log n\log W$ edges $(v,u)$ with $u\in Z'$, whose capacity is at least $c^*$. The analysis of the algorithm remains essentially identical. 

\subsection{Algorithm for the Heavy Vertex Problem: Proof of \Cref{claim: heavy}}
\label{sec: appendix_heavy}

In this section we provide a proof of Claim~\ref{claim: heavy}. First, we restate this claim for the convenience of the reader.

\heavy*

\paragraph{Proof of Claim~\ref{claim: heavy}.}
The algorithm closely follows the algorithm from the proof of \Cref{claim: heavy weight est}, except that we switch the roles of the sets $Z$ and $Z'$ of vertices, and reverse the directions of the edges. 
As before, if $\tau> \frac{|Z|}{100\log n\cdot\log W}$, then we return $b=0$; it is easy to verify that in this case the algorithm does not err. Otherwise, we select a subset $T\subseteq Z$ of vertices by adding every vertex $v\in Z$ to the set $T$ independently with probability $\frac{10\log n\cdot\log W}{\tau}$ as before. If $|T|\geq\frac{100|Z|\log n\cdot\log }{\tau}$, then we terminate the algorithm and return $b=0$ as before. We assume from now on that $|T|<\frac{100|Z|\log n\cdot\log }{\tau}$. By inspecting the edges leaving the vertices of $T$, we identify a subset $A\subseteq Z'$, containing all vertices $v\in Z'$, such that the number of edges $e=(u,v)$ with $u\in T$ and $c(e)\geq c^*$ is at least $100\log n\cdot \log W$. If $A=\emptyset$, then we return $b=0$. Otherwise, we select an arbitrary vertex $v\in A$, and inspect all edges entering $v$ in $G$, in time $O(n)$. If we find a set $E'\subseteq \delta^-(v)$ of at least $\tau$ edges $(u,v)$ with $u\in Z$, whose capacities are at least $c^*$, then we select arbitrary $\tau$ such edges, and return $b=1$ together with the endpoints of these edges that are distinct from $v$. Otherwise, we return $b=0$. 
Following the same analysis as in the proofs of \Cref{claim: heavy weight est} and \Cref{claim: degree est}, it is easy to verify that the probability that the algorithm errs is at most  $\frac{1}{n^{10}\cdot \log^6W}$. Since $\tau\leq |Z|$, we get that $O(n)\leq O\left(\frac{n\cdot |Z|\cdot \log n\cdot\log W}{\tau}\right )$, so the running time of the algorithm remains bounded by $O\left(\frac{n\cdot |Z|\cdot \log n\cdot\log W}{\tau}\right )$.

%
%

\newpage
\bibliographystyle{alphaurlinit}
\bibliography{global-min-cut}



\end{document}